\documentclass[twoside,11pt]{article}
\pdfoutput=1

\usepackage[tbtags,fleqn]{amsmath}
\usepackage[pdfa,unicode]{hyperref}
\usepackage{jair, rawfonts}
\usepackage[
backend=biber,
style=numeric,
]{biblatex}
\usepackage{array}
\makeatletter
\newcounter{subterm}[equation] % Define a new counter, reset with every equation
\renewcommand{\thesubterm}{\theequation\alph{subterm}} % New subterm counter is \theequation with alphabetic subterm numbering suffix

\newcommand{\subtermeqlabel}[1]{%
    \refstepcounter{subterm}% Step the counter
    \textup{\tagform@{\thesubterm}}% Display the tag
    \protected@write\@auxout{}{%
        \string\newlabel{#1}{{\thesubterm}{\thepage}{\@currentlabelname}{\@currentHref}{}}%
    }%
}
\makeatother
\usepackage{amsthm}
\usepackage{amssymb}
\usepackage{environ}
\usepackage{xspace}
\usepackage{xcolor}
\usepackage{graphicx}
\usepackage{ifthen}
\usepackage{algorithm}
\usepackage[noend]{algorithmic}
\usepackage[capitalise, noabbrev]{cleveref}
\usepackage{subcaption}
\usepackage{placeins}
\usepackage{pifont}
\usepackage{url}
\usepackage{tabularx}
\usepackage{tikz}
\usepackage{multirow}
\usepackage[draft]{todonotes}
\usepackage[shortlabels, inline]{enumitem}

\usetikzlibrary{decorations.pathmorphing}
\usetikzlibrary{positioning}
\usetikzlibrary{fit}

\definecolor {snow}                {rgb}{1.00,0.98,0.98}
\definecolor {ghostwhite}          {rgb}{0.97,0.97,1.00}
\definecolor {whitesmoke}          {rgb}{0.96,0.96,0.96}
\definecolor {gainsboro}           {rgb}{0.86,0.86,0.86}
\definecolor {floralwhite}         {rgb}{1.00,0.98,0.94}
\definecolor {oldlace}             {rgb}{0.99,0.96,0.90}
\definecolor {linen}               {rgb}{0.98,0.94,0.90}
\definecolor {antiquewhite}        {rgb}{0.98,0.92,0.84}
\definecolor {papayawhip}          {rgb}{1.00,0.94,0.84}
\definecolor {blanchedalmond}      {rgb}{1.00,0.92,0.80}
\definecolor {bisque}              {rgb}{1.00,0.89,0.77}
\definecolor {peachpuff}           {rgb}{1.00,0.85,0.73}
\definecolor {navajowhite}         {rgb}{1.00,0.87,0.68}
\definecolor {moccasin}            {rgb}{1.00,0.89,0.71}
\definecolor {cornsilk}            {rgb}{1.00,0.97,0.86}
\definecolor {ivory}               {rgb}{1.00,1.00,0.94}
\definecolor {lemonchiffon}        {rgb}{1.00,0.98,0.80}
\definecolor {seashell}            {rgb}{1.00,0.96,0.93}
\definecolor {honeydew}            {rgb}{0.94,1.00,0.94}
\definecolor {mintcream}           {rgb}{0.96,1.00,0.98}
\definecolor {azure}               {rgb}{0.94,1.00,1.00}
\definecolor {aliceblue}           {rgb}{0.94,0.97,1.00}
\definecolor {lavender}            {rgb}{0.90,0.90,0.98}
\definecolor {lavenderblush}       {rgb}{1.00,0.94,0.96}
\definecolor {mistyrose}           {rgb}{1.00,0.89,0.88}
\definecolor {white}               {rgb}{1.00,1.00,1.00}
\definecolor {black}               {rgb}{0.00,0.00,0.00}
\definecolor {darkslategray}       {rgb}{0.18,0.31,0.31}
\definecolor {dimgray}             {rgb}{0.41,0.41,0.41}
\definecolor {slategray}           {rgb}{0.44,0.50,0.56}
\definecolor {lightslategray}      {rgb}{0.47,0.53,0.60}
\definecolor {gray}                {rgb}{0.75,0.75,0.75}
\definecolor {lightgrey}           {rgb}{0.83,0.83,0.83}
\definecolor {midnightblue}        {rgb}{0.10,0.10,0.44}
\definecolor {navy}                {rgb}{0.00,0.00,0.50}
\definecolor {cornflowerblue}      {rgb}{0.39,0.58,0.93}
\definecolor {darkslateblue}       {rgb}{0.28,0.24,0.55}
\definecolor {slateblue}           {rgb}{0.42,0.35,0.80}
\definecolor {mediumslateblue}     {rgb}{0.48,0.41,0.93}
\definecolor {lightslateblue}      {rgb}{0.52,0.44,1.00}
\definecolor {mediumblue}          {rgb}{0.00,0.00,0.80}
\definecolor {royalblue}           {rgb}{0.25,0.41,0.88}
\definecolor {blue}                {rgb}{0.00,0.00,1.00}
\definecolor {dodgerblue}          {rgb}{0.12,0.56,1.00}
\definecolor {deepskyblue}         {rgb}{0.00,0.75,1.00}
\definecolor {skyblue}             {rgb}{0.53,0.81,0.92}
\definecolor {lightskyblue}        {rgb}{0.53,0.81,0.98}
\definecolor {steelblue}           {rgb}{0.27,0.51,0.71}
\definecolor {lightsteelblue}      {rgb}{0.69,0.77,0.87}
\definecolor {lightblue}           {rgb}{0.68,0.85,0.90}
\definecolor {powderblue}          {rgb}{0.69,0.88,0.90}
\definecolor {paleturquoise}       {rgb}{0.69,0.93,0.93}
\definecolor {darkturquoise}       {rgb}{0.00,0.81,0.82}
\definecolor {mediumturquoise}     {rgb}{0.28,0.82,0.80}
\definecolor {turquoise}           {rgb}{0.25,0.88,0.82}
\definecolor {cyan}                {rgb}{0.00,1.00,1.00}
\definecolor {lightcyan}           {rgb}{0.88,1.00,1.00}
\definecolor {cadetblue}           {rgb}{0.37,0.62,0.63}
\definecolor {mediumaquamarine}    {rgb}{0.40,0.80,0.67}
\definecolor {aquamarine}          {rgb}{0.50,1.00,0.83}
\definecolor {darkgreen}           {rgb}{0.00,0.39,0.00}
\definecolor {darkolivegreen}      {rgb}{0.33,0.42,0.18}
\definecolor {darkseagreen}        {rgb}{0.56,0.74,0.56}
\definecolor {seagreen}            {rgb}{0.18,0.55,0.34}
\definecolor {mediumseagreen}      {rgb}{0.24,0.70,0.44}
\definecolor {lightseagreen}       {rgb}{0.13,0.70,0.67}
\definecolor {palegreen}           {rgb}{0.60,0.98,0.60}
\definecolor {springgreen}         {rgb}{0.00,1.00,0.50}
\definecolor {lawngreen}           {rgb}{0.49,0.99,0.00}
\definecolor {green}               {rgb}{0.00,1.00,0.00}
\definecolor {chartreuse}          {rgb}{0.50,1.00,0.00}
\definecolor {mediumspringgreen}   {rgb}{0.00,0.98,0.60}
\definecolor {greenyellow}         {rgb}{0.68,1.00,0.18}
\definecolor {limegreen}           {rgb}{0.20,0.80,0.20}
\definecolor {yellowgreen}         {rgb}{0.60,0.80,0.20}
\definecolor {forestgreen}         {rgb}{0.13,0.55,0.13}
\definecolor {olivedrab}           {rgb}{0.42,0.56,0.14}
\definecolor {darkkhaki}           {rgb}{0.74,0.72,0.42}
\definecolor {khaki}               {rgb}{0.94,0.90,0.55}
\definecolor {palegoldenrod}       {rgb}{0.93,0.91,0.67}
\definecolor {lightgoldenrodyellow} {rgb}{0.98,0.98,0.82}
\definecolor {lightyellow}         {rgb}{1.00,1.00,0.88}
\definecolor {yellow}              {rgb}{1.00,1.00,0.00}
\definecolor {gold}                {rgb}{1.00,0.84,0.00}
\definecolor {lightgoldenrod}      {rgb}{0.93,0.87,0.51}
\definecolor {goldenrod}           {rgb}{0.85,0.65,0.13}
\definecolor {darkgoldenrod}       {rgb}{0.72,0.53,0.04}
\definecolor {rosybrown}           {rgb}{0.74,0.56,0.56}
\definecolor {indianred}           {rgb}{0.80,0.36,0.36}
\definecolor {saddlebrown}         {rgb}{0.55,0.27,0.07}
\definecolor {sienna}              {rgb}{0.63,0.32,0.18}
\definecolor {peru}                {rgb}{0.80,0.52,0.25}
\definecolor {burlywood}           {rgb}{0.87,0.72,0.53}
\definecolor {beige}               {rgb}{0.96,0.96,0.86}
\definecolor {wheat}               {rgb}{0.96,0.87,0.70}
\definecolor {sandybrown}          {rgb}{0.96,0.64,0.38}
\definecolor {tan}                 {rgb}{0.82,0.71,0.55}
\definecolor {chocolate}           {rgb}{0.82,0.41,0.12}
\definecolor {firebrick}           {rgb}{0.70,0.13,0.13}
\definecolor {brown}               {rgb}{0.65,0.16,0.16}
\definecolor {darksalmon}          {rgb}{0.91,0.59,0.48}
\definecolor {salmon}              {rgb}{0.98,0.50,0.45}
\definecolor {lightsalmon}         {rgb}{1.00,0.63,0.48}
\definecolor {orange}              {rgb}{1.00,0.65,0.00}
\definecolor {darkorange}          {rgb}{1.00,0.55,0.00}
\definecolor {coral}               {rgb}{1.00,0.50,0.31}
\definecolor {lightcoral}          {rgb}{0.94,0.50,0.50}
\definecolor {tomato}              {rgb}{1.00,0.39,0.28}
\definecolor {orangered}           {rgb}{1.00,0.27,0.00}
\definecolor {red}                 {rgb}{1.00,0.00,0.00}
\definecolor {hotpink}             {rgb}{1.00,0.41,0.71}
\definecolor {deeppink}            {rgb}{1.00,0.08,0.58}
\definecolor {pink}                {rgb}{1.00,0.75,0.80}
\definecolor {lightpink}           {rgb}{1.00,0.71,0.76}
\definecolor {palevioletred}       {rgb}{0.86,0.44,0.58}
\definecolor {maroon}              {rgb}{0.69,0.19,0.38}
\definecolor {mediumvioletred}     {rgb}{0.78,0.08,0.52}
\definecolor {violetred}           {rgb}{0.82,0.13,0.56}
\definecolor {magenta}             {rgb}{1.00,0.00,1.00}
\definecolor {violet}              {rgb}{0.93,0.51,0.93}
\definecolor {plum}                {rgb}{0.87,0.63,0.87}
\definecolor {orchid}              {rgb}{0.85,0.44,0.84}
\definecolor {mediumorchid}        {rgb}{0.73,0.33,0.83}
\definecolor {darkorchid}          {rgb}{0.60,0.20,0.80}
\definecolor {darkviolet}          {rgb}{0.58,0.00,0.83}
\definecolor {blueviolet}          {rgb}{0.54,0.17,0.89}
\definecolor {purple}              {rgb}{0.63,0.13,0.94}
\definecolor {mediumpurple}        {rgb}{0.58,0.44,0.86}
\definecolor {thistle}             {rgb}{0.85,0.75,0.85}
\definecolor {snow2}               {rgb}{0.93,0.91,0.91}
\definecolor {snow3}               {rgb}{0.80,0.79,0.79}
\definecolor {snow4}               {rgb}{0.55,0.54,0.54}
\definecolor {seashell2}           {rgb}{0.93,0.90,0.87}
\definecolor {seashell3}           {rgb}{0.80,0.77,0.75}
\definecolor {seashell4}           {rgb}{0.55,0.53,0.51}
\definecolor {antiquewhite1}       {rgb}{1.00,0.94,0.86}
\definecolor {antiquewhite2}       {rgb}{0.93,0.87,0.80}
\definecolor {antiquewhite3}       {rgb}{0.80,0.75,0.69}
\definecolor {antiquewhite4}       {rgb}{0.55,0.51,0.47}
\definecolor {bisque2}             {rgb}{0.93,0.84,0.72}
\definecolor {bisque3}             {rgb}{0.80,0.72,0.62}
\definecolor {bisque4}             {rgb}{0.55,0.49,0.42}
\definecolor {peachpuff2}          {rgb}{0.93,0.80,0.68}
\definecolor {peachpuff3}          {rgb}{0.80,0.69,0.58}
\definecolor {peachpuff4}          {rgb}{0.55,0.47,0.40}
\definecolor {navajowhite2}        {rgb}{0.93,0.81,0.63}
\definecolor {navajowhite3}        {rgb}{0.80,0.70,0.55}
\definecolor {navajowhite4}        {rgb}{0.55,0.47,0.37}
\definecolor {lemonchiffon2}       {rgb}{0.93,0.91,0.75}
\definecolor {lemonchiffon3}       {rgb}{0.80,0.79,0.65}
\definecolor {lemonchiffon4}       {rgb}{0.55,0.54,0.44}
\definecolor {cornsilk2}           {rgb}{0.93,0.91,0.80}
\definecolor {cornsilk3}           {rgb}{0.80,0.78,0.69}
\definecolor {cornsilk4}           {rgb}{0.55,0.53,0.47}
\definecolor {ivory2}              {rgb}{0.93,0.93,0.88}
\definecolor {ivory3}              {rgb}{0.80,0.80,0.76}
\definecolor {ivory4}              {rgb}{0.55,0.55,0.51}
\definecolor {honeydew2}           {rgb}{0.88,0.93,0.88}
\definecolor {honeydew3}           {rgb}{0.76,0.80,0.76}
\definecolor {honeydew4}           {rgb}{0.51,0.55,0.51}
\definecolor {lavenderblush2}      {rgb}{0.93,0.88,0.90}
\definecolor {lavenderblush3}      {rgb}{0.80,0.76,0.77}
\definecolor {lavenderblush4}      {rgb}{0.55,0.51,0.53}
\definecolor {mistyrose2}          {rgb}{0.93,0.84,0.82}
\definecolor {mistyrose3}          {rgb}{0.80,0.72,0.71}
\definecolor {mistyrose4}          {rgb}{0.55,0.49,0.48}
\definecolor {azure2}              {rgb}{0.88,0.93,0.93}
\definecolor {azure3}              {rgb}{0.76,0.80,0.80}
\definecolor {azure4}              {rgb}{0.51,0.55,0.55}
\definecolor {slateblue1}          {rgb}{0.51,0.44,1.00}
\definecolor {slateblue2}          {rgb}{0.48,0.40,0.93}
\definecolor {slateblue3}          {rgb}{0.41,0.35,0.80}
\definecolor {slateblue4}          {rgb}{0.28,0.24,0.55}
\definecolor {royalblue1}          {rgb}{0.28,0.46,1.00}
\definecolor {royalblue2}          {rgb}{0.26,0.43,0.93}
\definecolor {royalblue3}          {rgb}{0.23,0.37,0.80}
\definecolor {royalblue4}          {rgb}{0.15,0.25,0.55}
\definecolor {blue2}               {rgb}{0.00,0.00,0.93}
\definecolor {blue4}               {rgb}{0.00,0.00,0.55}
\definecolor {dodgerblue2}         {rgb}{0.11,0.53,0.93}
\definecolor {dodgerblue3}         {rgb}{0.09,0.45,0.80}
\definecolor {dodgerblue4}         {rgb}{0.06,0.31,0.55}
\definecolor {steelblue1}          {rgb}{0.39,0.72,1.00}
\definecolor {steelblue2}          {rgb}{0.36,0.67,0.93}
\definecolor {steelblue3}          {rgb}{0.31,0.58,0.80}
\definecolor {steelblue4}          {rgb}{0.21,0.39,0.55}
\definecolor {deepskyblue2}        {rgb}{0.00,0.70,0.93}
\definecolor {deepskyblue3}        {rgb}{0.00,0.60,0.80}
\definecolor {deepskyblue4}        {rgb}{0.00,0.41,0.55}
\definecolor {skyblue1}            {rgb}{0.53,0.81,1.00}
\definecolor {skyblue2}            {rgb}{0.49,0.75,0.93}
\definecolor {skyblue3}            {rgb}{0.42,0.65,0.80}
\definecolor {skyblue4}            {rgb}{0.29,0.44,0.55}
\definecolor {lightskyblue1}       {rgb}{0.69,0.89,1.00}
\definecolor {lightskyblue2}       {rgb}{0.64,0.83,0.93}
\definecolor {lightskyblue3}       {rgb}{0.55,0.71,0.80}
\definecolor {lightskyblue4}       {rgb}{0.38,0.48,0.55}
\definecolor {slategray1}          {rgb}{0.78,0.89,1.00}
\definecolor {slategray2}          {rgb}{0.73,0.83,0.93}
\definecolor {slategray3}          {rgb}{0.62,0.71,0.80}
\definecolor {slategray4}          {rgb}{0.42,0.48,0.55}
\definecolor {lightsteelblue1}     {rgb}{0.79,0.88,1.00}
\definecolor {lightsteelblue2}     {rgb}{0.74,0.82,0.93}
\definecolor {lightsteelblue3}     {rgb}{0.64,0.71,0.80}
\definecolor {lightsteelblue4}     {rgb}{0.43,0.48,0.55}
\definecolor {lightblue1}          {rgb}{0.75,0.94,1.00}
\definecolor {lightblue2}          {rgb}{0.70,0.87,0.93}
\definecolor {lightblue3}          {rgb}{0.60,0.75,0.80}
\definecolor {lightblue4}          {rgb}{0.41,0.51,0.55}
\definecolor {lightcyan2}          {rgb}{0.82,0.93,0.93}
\definecolor {lightcyan3}          {rgb}{0.71,0.80,0.80}
\definecolor {lightcyan4}          {rgb}{0.48,0.55,0.55}
\definecolor {paleturquoise1}      {rgb}{0.73,1.00,1.00}
\definecolor {paleturquoise2}      {rgb}{0.68,0.93,0.93}
\definecolor {paleturquoise3}      {rgb}{0.59,0.80,0.80}
\definecolor {paleturquoise4}      {rgb}{0.40,0.55,0.55}
\definecolor {cadetblue1}          {rgb}{0.60,0.96,1.00}
\definecolor {cadetblue2}          {rgb}{0.56,0.90,0.93}
\definecolor {cadetblue3}          {rgb}{0.48,0.77,0.80}
\definecolor {cadetblue4}          {rgb}{0.33,0.53,0.55}
\definecolor {turquoise1}          {rgb}{0.00,0.96,1.00}
\definecolor {turquoise2}          {rgb}{0.00,0.90,0.93}
\definecolor {turquoise3}          {rgb}{0.00,0.77,0.80}
\definecolor {turquoise4}          {rgb}{0.00,0.53,0.55}
\definecolor {cyan2}               {rgb}{0.00,0.93,0.93}
\definecolor {cyan3}               {rgb}{0.00,0.80,0.80}
\definecolor {cyan4}               {rgb}{0.00,0.55,0.55}
\definecolor {darkslategray1}      {rgb}{0.59,1.00,1.00}
\definecolor {darkslategray2}      {rgb}{0.55,0.93,0.93}
\definecolor {darkslategray3}      {rgb}{0.47,0.80,0.80}
\definecolor {darkslategray4}      {rgb}{0.32,0.55,0.55}
\definecolor {aquamarine2}         {rgb}{0.46,0.93,0.78}
\definecolor {aquamarine4}         {rgb}{0.27,0.55,0.45}
\definecolor {darkseagreen1}       {rgb}{0.76,1.00,0.76}
\definecolor {darkseagreen2}       {rgb}{0.71,0.93,0.71}
\definecolor {darkseagreen3}       {rgb}{0.61,0.80,0.61}
\definecolor {darkseagreen4}       {rgb}{0.41,0.55,0.41}
\definecolor {seagreen1}           {rgb}{0.33,1.00,0.62}
\definecolor {seagreen2}           {rgb}{0.31,0.93,0.58}
\definecolor {seagreen3}           {rgb}{0.26,0.80,0.50}
\definecolor {palegreen1}          {rgb}{0.60,1.00,0.60}
\definecolor {palegreen2}          {rgb}{0.56,0.93,0.56}
\definecolor {palegreen3}          {rgb}{0.49,0.80,0.49}
\definecolor {palegreen4}          {rgb}{0.33,0.55,0.33}
\definecolor {springgreen2}        {rgb}{0.00,0.93,0.46}
\definecolor {springgreen3}        {rgb}{0.00,0.80,0.40}
\definecolor {springgreen4}        {rgb}{0.00,0.55,0.27}
\definecolor {green2}              {rgb}{0.00,0.93,0.00}
\definecolor {green3}              {rgb}{0.00,0.80,0.00}
\definecolor {green4}              {rgb}{0.00,0.55,0.00}
\definecolor {chartreuse2}         {rgb}{0.46,0.93,0.00}
\definecolor {chartreuse3}         {rgb}{0.40,0.80,0.00}
\definecolor {chartreuse4}         {rgb}{0.27,0.55,0.00}
\definecolor {olivedrab1}          {rgb}{0.75,1.00,0.24}
\definecolor {olivedrab2}          {rgb}{0.70,0.93,0.23}
\definecolor {olivedrab4}          {rgb}{0.41,0.55,0.13}
\definecolor {darkolivegreen1}     {rgb}{0.79,1.00,0.44}
\definecolor {darkolivegreen2}     {rgb}{0.74,0.93,0.41}
\definecolor {darkolivegreen3}     {rgb}{0.64,0.80,0.35}
\definecolor {darkolivegreen4}     {rgb}{0.43,0.55,0.24}
\definecolor {khaki1}              {rgb}{1.00,0.96,0.56}
\definecolor {khaki2}              {rgb}{0.93,0.90,0.52}
\definecolor {khaki3}              {rgb}{0.80,0.78,0.45}
\definecolor {khaki4}              {rgb}{0.55,0.53,0.31}
\definecolor {lightgoldenrod1}     {rgb}{1.00,0.93,0.55}
\definecolor {lightgoldenrod2}     {rgb}{0.93,0.86,0.51}
\definecolor {lightgoldenrod3}     {rgb}{0.80,0.75,0.44}
\definecolor {lightgoldenrod4}     {rgb}{0.55,0.51,0.30}
\definecolor {lightyellow2}        {rgb}{0.93,0.93,0.82}
\definecolor {lightyellow3}        {rgb}{0.80,0.80,0.71}
\definecolor {lightyellow4}        {rgb}{0.55,0.55,0.48}
\definecolor {yellow2}             {rgb}{0.93,0.93,0.00}
\definecolor {yellow3}             {rgb}{0.80,0.80,0.00}
\definecolor {yellow4}             {rgb}{0.55,0.55,0.00}
\definecolor {gold2}               {rgb}{0.93,0.79,0.00}
\definecolor {gold3}               {rgb}{0.80,0.68,0.00}
\definecolor {gold4}               {rgb}{0.55,0.46,0.00}
\definecolor {goldenrod1}          {rgb}{1.00,0.76,0.15}
\definecolor {goldenrod2}          {rgb}{0.93,0.71,0.13}
\definecolor {goldenrod3}          {rgb}{0.80,0.61,0.11}
\definecolor {goldenrod4}          {rgb}{0.55,0.41,0.08}
\definecolor {darkgoldenrod1}      {rgb}{1.00,0.73,0.06}
\definecolor {darkgoldenrod2}      {rgb}{0.93,0.68,0.05}
\definecolor {darkgoldenrod3}      {rgb}{0.80,0.58,0.05}
\definecolor {darkgoldenrod4}      {rgb}{0.55,0.40,0.03}
\definecolor {rosybrown1}          {rgb}{1.00,0.76,0.76}
\definecolor {rosybrown2}          {rgb}{0.93,0.71,0.71}
\definecolor {rosybrown3}          {rgb}{0.80,0.61,0.61}
\definecolor {rosybrown4}          {rgb}{0.55,0.41,0.41}
\definecolor {indianred1}          {rgb}{1.00,0.42,0.42}
\definecolor {indianred2}          {rgb}{0.93,0.39,0.39}
\definecolor {indianred3}          {rgb}{0.80,0.33,0.33}
\definecolor {indianred4}          {rgb}{0.55,0.23,0.23}
\definecolor {sienna1}             {rgb}{1.00,0.51,0.28}
\definecolor {sienna2}             {rgb}{0.93,0.47,0.26}
\definecolor {sienna3}             {rgb}{0.80,0.41,0.22}
\definecolor {sienna4}             {rgb}{0.55,0.28,0.15}
\definecolor {burlywood1}          {rgb}{1.00,0.83,0.61}
\definecolor {burlywood2}          {rgb}{0.93,0.77,0.57}
\definecolor {burlywood3}          {rgb}{0.80,0.67,0.49}
\definecolor {burlywood4}          {rgb}{0.55,0.45,0.33}
\definecolor {wheat1}              {rgb}{1.00,0.91,0.73}
\definecolor {wheat2}              {rgb}{0.93,0.85,0.68}
\definecolor {wheat3}              {rgb}{0.80,0.73,0.59}
\definecolor {wheat4}              {rgb}{0.55,0.49,0.40}
\definecolor {tan1}                {rgb}{1.00,0.65,0.31}
\definecolor {tan2}                {rgb}{0.93,0.60,0.29}
\definecolor {tan4}                {rgb}{0.55,0.35,0.17}
\definecolor {chocolate1}          {rgb}{1.00,0.50,0.14}
\definecolor {chocolate2}          {rgb}{0.93,0.46,0.13}
\definecolor {chocolate3}          {rgb}{0.80,0.40,0.11}
\definecolor {firebrick1}          {rgb}{1.00,0.19,0.19}
\definecolor {firebrick2}          {rgb}{0.93,0.17,0.17}
\definecolor {firebrick3}          {rgb}{0.80,0.15,0.15}
\definecolor {firebrick4}          {rgb}{0.55,0.10,0.10}
\definecolor {brown1}              {rgb}{1.00,0.25,0.25}
\definecolor {brown2}              {rgb}{0.93,0.23,0.23}
\definecolor {brown3}              {rgb}{0.80,0.20,0.20}
\definecolor {brown4}              {rgb}{0.55,0.14,0.14}
\definecolor {salmon1}             {rgb}{1.00,0.55,0.41}
\definecolor {salmon2}             {rgb}{0.93,0.51,0.38}
\definecolor {salmon3}             {rgb}{0.80,0.44,0.33}
\definecolor {salmon4}             {rgb}{0.55,0.30,0.22}
\definecolor {lightsalmon2}        {rgb}{0.93,0.58,0.45}
\definecolor {lightsalmon3}        {rgb}{0.80,0.51,0.38}
\definecolor {lightsalmon4}        {rgb}{0.55,0.34,0.26}
\definecolor {orange2}             {rgb}{0.93,0.60,0.00}
\definecolor {orange3}             {rgb}{0.80,0.52,0.00}
\definecolor {orange4}             {rgb}{0.55,0.35,0.00}
\definecolor {darkorange1}         {rgb}{1.00,0.50,0.00}
\definecolor {darkorange2}         {rgb}{0.93,0.46,0.00}
\definecolor {darkorange3}         {rgb}{0.80,0.40,0.00}
\definecolor {darkorange4}         {rgb}{0.55,0.27,0.00}
\definecolor {coral1}              {rgb}{1.00,0.45,0.34}
\definecolor {coral2}              {rgb}{0.93,0.42,0.31}
\definecolor {coral3}              {rgb}{0.80,0.36,0.27}
\definecolor {coral4}              {rgb}{0.55,0.24,0.18}
\definecolor {tomato2}             {rgb}{0.93,0.36,0.26}
\definecolor {tomato3}             {rgb}{0.80,0.31,0.22}
\definecolor {tomato4}             {rgb}{0.55,0.21,0.15}
\definecolor {orangered2}          {rgb}{0.93,0.25,0.00}
\definecolor {orangered3}          {rgb}{0.80,0.22,0.00}
\definecolor {orangered4}          {rgb}{0.55,0.15,0.00}
\definecolor {red2}                {rgb}{0.93,0.00,0.00}
\definecolor {red3}                {rgb}{0.80,0.00,0.00}
\definecolor {red4}                {rgb}{0.55,0.00,0.00}
\definecolor {deeppink2}           {rgb}{0.93,0.07,0.54}
\definecolor {deeppink3}           {rgb}{0.80,0.06,0.46}
\definecolor {deeppink4}           {rgb}{0.55,0.04,0.31}
\definecolor {hotpink1}            {rgb}{1.00,0.43,0.71}
\definecolor {hotpink2}            {rgb}{0.93,0.42,0.65}
\definecolor {hotpink3}            {rgb}{0.80,0.38,0.56}
\definecolor {hotpink4}            {rgb}{0.55,0.23,0.38}
\definecolor {pink1}               {rgb}{1.00,0.71,0.77}
\definecolor {pink2}               {rgb}{0.93,0.66,0.72}
\definecolor {pink3}               {rgb}{0.80,0.57,0.62}
\definecolor {pink4}               {rgb}{0.55,0.39,0.42}
\definecolor {lightpink1}          {rgb}{1.00,0.68,0.73}
\definecolor {lightpink2}          {rgb}{0.93,0.64,0.68}
\definecolor {lightpink3}          {rgb}{0.80,0.55,0.58}
\definecolor {lightpink4}          {rgb}{0.55,0.37,0.40}
\definecolor {palevioletred1}      {rgb}{1.00,0.51,0.67}
\definecolor {palevioletred2}      {rgb}{0.93,0.47,0.62}
\definecolor {palevioletred3}      {rgb}{0.80,0.41,0.54}
\definecolor {palevioletred4}      {rgb}{0.55,0.28,0.36}
\definecolor {maroon1}             {rgb}{1.00,0.20,0.70}
\definecolor {maroon2}             {rgb}{0.93,0.19,0.65}
\definecolor {maroon3}             {rgb}{0.80,0.16,0.56}
\definecolor {maroon4}             {rgb}{0.55,0.11,0.38}
\definecolor {violetred1}          {rgb}{1.00,0.24,0.59}
\definecolor {violetred2}          {rgb}{0.93,0.23,0.55}
\definecolor {violetred3}          {rgb}{0.80,0.20,0.47}
\definecolor {violetred4}          {rgb}{0.55,0.13,0.32}
\definecolor {magenta2}            {rgb}{0.93,0.00,0.93}
\definecolor {magenta3}            {rgb}{0.80,0.00,0.80}
\definecolor {magenta4}            {rgb}{0.55,0.00,0.55}
\definecolor {orchid1}             {rgb}{1.00,0.51,0.98}
\definecolor {orchid2}             {rgb}{0.93,0.48,0.91}
\definecolor {orchid3}             {rgb}{0.80,0.41,0.79}
\definecolor {orchid4}             {rgb}{0.55,0.28,0.54}
\definecolor {plum1}               {rgb}{1.00,0.73,1.00}
\definecolor {plum2}               {rgb}{0.93,0.68,0.93}
\definecolor {plum3}               {rgb}{0.80,0.59,0.80}
\definecolor {plum4}               {rgb}{0.55,0.40,0.55}
\definecolor {mediumorchid1}       {rgb}{0.88,0.40,1.00}
\definecolor {mediumorchid2}       {rgb}{0.82,0.37,0.93}
\definecolor {mediumorchid3}       {rgb}{0.71,0.32,0.80}
\definecolor {mediumorchid4}       {rgb}{0.48,0.22,0.55}
\definecolor {darkorchid1}         {rgb}{0.75,0.24,1.00}
\definecolor {darkorchid2}         {rgb}{0.70,0.23,0.93}
\definecolor {darkorchid3}         {rgb}{0.60,0.20,0.80}
\definecolor {darkorchid4}         {rgb}{0.41,0.13,0.55}
\definecolor {purple1}             {rgb}{0.61,0.19,1.00}
\definecolor {purple2}             {rgb}{0.57,0.17,0.93}
\definecolor {purple3}             {rgb}{0.49,0.15,0.80}
\definecolor {purple4}             {rgb}{0.33,0.10,0.55}
\definecolor {mediumpurple1}       {rgb}{0.67,0.51,1.00}
\definecolor {mediumpurple2}       {rgb}{0.62,0.47,0.93}
\definecolor {mediumpurple3}       {rgb}{0.54,0.41,0.80}
\definecolor {mediumpurple4}       {rgb}{0.36,0.28,0.55}
\definecolor {thistle1}            {rgb}{1.00,0.88,1.00}
\definecolor {thistle2}            {rgb}{0.93,0.82,0.93}
\definecolor {thistle3}            {rgb}{0.80,0.71,0.80}
\definecolor {thistle4}            {rgb}{0.55,0.48,0.55}
\definecolor {gray1}               {rgb}{0.01,0.01,0.01}
\definecolor {gray2}               {rgb}{0.02,0.02,0.02}
\definecolor {gray3}               {rgb}{0.03,0.03,0.03}
\definecolor {gray4}               {rgb}{0.04,0.04,0.04}
\definecolor {gray5}               {rgb}{0.05,0.05,0.05}
\definecolor {gray6}               {rgb}{0.06,0.06,0.06}
\definecolor {gray7}               {rgb}{0.07,0.07,0.07}
\definecolor {gray8}               {rgb}{0.08,0.08,0.08}
\definecolor {gray9}               {rgb}{0.09,0.09,0.09}
\definecolor {gray10}              {rgb}{0.10,0.10,0.10}
\definecolor {gray11}              {rgb}{0.11,0.11,0.11}
\definecolor {gray12}              {rgb}{0.12,0.12,0.12}
\definecolor {gray13}              {rgb}{0.13,0.13,0.13}
\definecolor {gray14}              {rgb}{0.14,0.14,0.14}
\definecolor {gray15}              {rgb}{0.15,0.15,0.15}
\definecolor {gray16}              {rgb}{0.16,0.16,0.16}
\definecolor {gray17}              {rgb}{0.17,0.17,0.17}
\definecolor {gray18}              {rgb}{0.18,0.18,0.18}
\definecolor {gray19}              {rgb}{0.19,0.19,0.19}
\definecolor {gray20}              {rgb}{0.20,0.20,0.20}
\definecolor {gray21}              {rgb}{0.21,0.21,0.21}
\definecolor {gray22}              {rgb}{0.22,0.22,0.22}
\definecolor {gray23}              {rgb}{0.23,0.23,0.23}
\definecolor {gray24}              {rgb}{0.24,0.24,0.24}
\definecolor {gray25}              {rgb}{0.25,0.25,0.25}
\definecolor {gray26}              {rgb}{0.26,0.26,0.26}
\definecolor {gray27}              {rgb}{0.27,0.27,0.27}
\definecolor {gray28}              {rgb}{0.28,0.28,0.28}
\definecolor {gray29}              {rgb}{0.29,0.29,0.29}
\definecolor {gray30}              {rgb}{0.30,0.30,0.30}
\definecolor {gray31}              {rgb}{0.31,0.31,0.31}
\definecolor {gray32}              {rgb}{0.32,0.32,0.32}
\definecolor {gray33}              {rgb}{0.33,0.33,0.33}
\definecolor {gray34}              {rgb}{0.34,0.34,0.34}
\definecolor {gray35}              {rgb}{0.35,0.35,0.35}
\definecolor {gray36}              {rgb}{0.36,0.36,0.36}
\definecolor {gray37}              {rgb}{0.37,0.37,0.37}
\definecolor {gray38}              {rgb}{0.38,0.38,0.38}
\definecolor {gray39}              {rgb}{0.39,0.39,0.39}
\definecolor {gray40}              {rgb}{0.40,0.40,0.40}
\definecolor {gray42}              {rgb}{0.42,0.42,0.42}
\definecolor {gray43}              {rgb}{0.43,0.43,0.43}
\definecolor {gray44}              {rgb}{0.44,0.44,0.44}
\definecolor {gray45}              {rgb}{0.45,0.45,0.45}
\definecolor {gray46}              {rgb}{0.46,0.46,0.46}
\definecolor {gray47}              {rgb}{0.47,0.47,0.47}
\definecolor {gray48}              {rgb}{0.48,0.48,0.48}
\definecolor {gray49}              {rgb}{0.49,0.49,0.49}
\definecolor {gray50}              {rgb}{0.50,0.50,0.50}
\definecolor {gray51}              {rgb}{0.51,0.51,0.51}
\definecolor {gray52}              {rgb}{0.52,0.52,0.52}
\definecolor {gray53}              {rgb}{0.53,0.53,0.53}
\definecolor {gray54}              {rgb}{0.54,0.54,0.54}
\definecolor {gray55}              {rgb}{0.55,0.55,0.55}
\definecolor {gray56}              {rgb}{0.56,0.56,0.56}
\definecolor {gray57}              {rgb}{0.57,0.57,0.57}
\definecolor {gray58}              {rgb}{0.58,0.58,0.58}
\definecolor {gray59}              {rgb}{0.59,0.59,0.59}
\definecolor {gray60}              {rgb}{0.60,0.60,0.60}
\definecolor {gray61}              {rgb}{0.61,0.61,0.61}
\definecolor {gray62}              {rgb}{0.62,0.62,0.62}
\definecolor {gray63}              {rgb}{0.63,0.63,0.63}
\definecolor {gray64}              {rgb}{0.64,0.64,0.64}
\definecolor {gray65}              {rgb}{0.65,0.65,0.65}
\definecolor {gray66}              {rgb}{0.66,0.66,0.66}
\definecolor {gray67}              {rgb}{0.67,0.67,0.67}
\definecolor {gray68}              {rgb}{0.68,0.68,0.68}
\definecolor {gray69}              {rgb}{0.69,0.69,0.69}
\definecolor {gray70}              {rgb}{0.70,0.70,0.70}
\definecolor {gray71}              {rgb}{0.71,0.71,0.71}
\definecolor {gray72}              {rgb}{0.72,0.72,0.72}
\definecolor {gray73}              {rgb}{0.73,0.73,0.73}
\definecolor {gray74}              {rgb}{0.74,0.74,0.74}
\definecolor {gray75}              {rgb}{0.75,0.75,0.75}
\definecolor {gray76}              {rgb}{0.76,0.76,0.76}
\definecolor {gray77}              {rgb}{0.77,0.77,0.77}
\definecolor {gray78}              {rgb}{0.78,0.78,0.78}
\definecolor {gray79}              {rgb}{0.79,0.79,0.79}
\definecolor {gray80}              {rgb}{0.80,0.80,0.80}
\definecolor {gray81}              {rgb}{0.81,0.81,0.81}
\definecolor {gray82}              {rgb}{0.82,0.82,0.82}
\definecolor {gray83}              {rgb}{0.83,0.83,0.83}
\definecolor {gray84}              {rgb}{0.84,0.84,0.84}
\definecolor {gray85}              {rgb}{0.85,0.85,0.85}
\definecolor {gray86}              {rgb}{0.86,0.86,0.86}
\definecolor {gray87}              {rgb}{0.87,0.87,0.87}
\definecolor {gray88}              {rgb}{0.88,0.88,0.88}
\definecolor {gray89}              {rgb}{0.89,0.89,0.89}
\definecolor {gray90}              {rgb}{0.90,0.90,0.90}
\definecolor {gray91}              {rgb}{0.91,0.91,0.91}
\definecolor {gray92}              {rgb}{0.92,0.92,0.92}
\definecolor {gray93}              {rgb}{0.93,0.93,0.93}
\definecolor {gray94}              {rgb}{0.94,0.94,0.94}
\definecolor {gray95}              {rgb}{0.95,0.95,0.95}
\definecolor {gray97}              {rgb}{0.97,0.97,0.97}
\definecolor {gray98}              {rgb}{0.98,0.98,0.98}
\definecolor {gray99}              {rgb}{0.99,0.99,0.99}
\definecolor {darkgrey}            {rgb}{0.66,0.66,0.66}

\newcommand{\resp}[1]{[resp.\ #1]}

\newcommand{\TODO}[1]{{}}

\newcommand{\ignore}[1]{}

\newcommand{\RSTODO}[1]{{\bf \textcolor{darkgreen}{{\fbox{RS TODO:} #1}}}}
\renewcommand{\RSTODO}[1]{}

\newenvironment{rschange}{\color{darkviolet}}{\normalcolor}

\newcommand{\ignoreinshort}[1]{}
\newcommand{\ignoreinlong}[1]{{#1}}
\providecommand{\longversion}{true}
\ifthenelse{\equal{\longversion}{true}}
{%
    \renewcommand{\ignoreinshort}[1]{\textcolor{blue}{#1}}
    \newcommand{\ignoreinshortnc}[1]{{#1}}
    \renewcommand{\ignoreinlong}[1]{}
    \newcommand{\ignoreinlongnc}[1]{}
    
    \NewEnviron{ignoreinlongenv}{}
    \NewEnviron{IGNORE}{}
}%
{%
    \renewcommand{\ignoreinshort}[1]{}
    \newcommand{\ignoreinshortnc}[1]{}
    \renewcommand{\ignoreinlong}[1]{\textcolor{blue}{#1}}
    \newcommand{\ignoreinlongnc}[1]{{#1}}
    \NewEnviron{ignoreinshortenv}{}
    
}
\def\makenewenumerate#1#2{%
    \newcounter{cnt#1}
    \newenvironment{#1}%
    {\begin{list}{\makebox[0pt][r]{#2}}%
            {\setlength{\itemsep}{0pt}% 
                \setlength{\parsep}{.2em}%
                \setlength{\leftmargin}{1.5em}%
                \setlength{\labelwidth}{.4em}%
                \usecounter{cnt#1}}}
            {\end{list}}}

\makenewenumerate{myenumerate}{\arabic{cntmyenumerate}.}

\makenewenumerate{renumerate}{\rm(\roman{cntrenumerate})}
\makenewenumerate{renumerateprime}{\rm(\roman{cntrenumerateprime}$'$)}
\makenewenumerate{renumeratesecond}{\rm(\roman{cntrenumeratesecond}$''$)}

\makenewenumerate{aenumerate}{({\em\alph{cntaenumerate}})}
\makenewenumerate{aenumerateprime}{({\em\alph{cntaenumerateprime}$'$})}
\makenewenumerate{aenumeratesecond}{({\em\alph{cntaenumeratesecond}$''$})}

\newcommand{\sref}[1]{\S{}\ref{#1}}

\newcommand{\pair}[2]{\ensuremath{\langle{#1},{#2}\rangle}\xspace}

\newcommand{\tuple}[1]{\ensuremath{\langle{#1}\rangle}\xspace}
\newcommand{\set}[1]{\ensuremath{\{{#1}\}}\xspace}
\newcommand{\imp}{\ensuremath{\rightarrow}\xspace}
\newcommand{\limp}{\ensuremath{\leftarrow}\xspace}
\renewcommand{\iff}{\ensuremath{\leftrightarrow}\xspace}
\newcommand{\defas}{\ensuremath{\stackrel{\text{\scalebox{.7}{def}}}{=}}\xspace}

\newcommand{\pos}{\phantom{\neg}}

\newcommand\cala{\ensuremath{\mathcal{A}}\xspace}

\newcommand\calt{\ensuremath{\mathcal{T}}\xspace}

\newcommand{\topt}{\ensuremath{\top}\xspace}

\newcommand\mysout{\bgroup \markoverwith{{-}}\ULon}
\newcommand\nosout{\bgroup \markoverwith{{ }}\ULon}
\definecolor{mygray}{rgb}{0.90,0.90,0.90}
\definecolor{mywhite}{rgb}{1.00,1.00,1.00}

\newcommand{\B}{\ensuremath{\mathcal{B}}\xspace}
\newcommand{\T}{\ensuremath{\mathcal{T}}\xspace}

\newcommand{\smt}{SMT\xspace}
\newcommand{\smtt}{\ensuremath{\text{SMT}(\T)}\xspace}
\newcommand{\smttt}[1]{\ensuremath{\text{SMT}(#1)}\xspace}

\newcommand{\euf}{\ensuremath{\mathcal{EUF}}\xspace}

\newcommand{\larat}{\ensuremath{\mathcal{LA}(\mathbb{Q})}\xspace}
\newcommand{\laint}{\ensuremath{\mathcal{LA}(\mathbb{Z})}\xspace}

\renewcommand{\larat}{\ensuremath{\mathcal{LRA}}\xspace}
\renewcommand{\laint}{\ensuremath{\mathcal{LIA}}\xspace}

\newcommand{\smtlarat}{\smttt{\larat}}

\newcommand{\pmodels}{\models_p}

\newcommand{\mathsat}{\textsc{MathSAT}\xspace}

\newcommand{\mathsatfive}{\textsc{MathSAT5}\xspace}

\renewcommand{\TODO}[1]{\todo[inline,color=green!40]{{\small{#1}}}}

\renewcommand{\RSTODO}[1]{\todo[inline,color=green!40]{{\small{RS TODO: #1}}}}

\newcommand{\cnfnamed}[2][]{\ensuremath{\mathsf{CNF_{#2}^{#1}}}}
\newcommand{\DeMorganCNF}{\cnfnamed{DM}}
\newcommand{\TseitinCNF}{\cnfnamed{TS}}
\renewcommand{\TseitinCNF}{\cnfnamed{Ts}}
\newcommand{\PlaistedCNF}{\cnfnamed{PG}}

\newcommand{\NNF}[1]{\ensuremath{\mathsf{NNF}(#1)}}
\newcommand{\NNFna}{\ensuremath{\mathsf{NNF}}}
\newcommand{\NNFPlaisted}{\NNFna{}+\PlaistedCNF{}}
\newcommand{\allA}{\ensuremath{\mathbf{A}}\xspace}
\newcommand{\allalpha}{\ensuremath{\boldsymbol{\alpha}}\xspace}
\newcommand{\allB}{\ensuremath{\mathbf{B}}\xspace}

\newcommand{\TA}[1]{\ensuremath{\calt\hspace{-.1cm}\cala(#1)}\xspace}
\newcommand{\TAna}{\ensuremath{\calt\hspace{-.1cm}\cala}\xspace}
\newcommand{\TTA}[1]{\ensuremath{\calt\hspace{-.1cm}\calt\hspace{-.1cm}\cala(#1)}\xspace}
\newcommand{\muA}{\ensuremath{\mu^\allA}\xspace}
\newcommand{\mualpha}{\ensuremath{\mu^{\allalpha}}\xspace}
\newcommand{\etaA}{\ensuremath{\eta^\allA}\xspace}
\newcommand{\muAprime}{\ensuremath{\muA{}'}\xspace}

\newcommand{\muB}{\ensuremath{\mu^\allB}\xspace}
\newcommand{\etaB}{\ensuremath{\eta^\allB}\xspace}

\newcommand{\residual}[2]{\ensuremath{#1|_{#2}}\xspace}
\newcommand{\exdone}{\ensuremath{\hfill\diamond}}
\newcommand{\vi}{\ensuremath{\varphi}}

\newcommand{\vicnf}{\ensuremath{\psi}}
\newcommand{\vicnfts}{\ensuremath{\TseitinCNF(\vi)}}

\newcommand{\vicnfpg}{\ensuremath{\PlaistedCNF(\vi)}}

\newcommand{\vinnfcnfpg}{\ensuremath{\PlaistedCNF(\NNF{\vi})}}
\newcommand{\any}{\ensuremath{*}}

\newcommand{\minus}{\scalebox{0.5}[0.5]{\(-\)}}
\newcommand{\plus}{\scalebox{0.5}[0.5]{\(+\)}}
\newcommand{\poslab}[1]{\ensuremath{#1^{\plus}}}
\newcommand{\neglab}[1]{\ensuremath{#1^{\minus}}}

\newcommand{\eqcomment}[1]{\textcolor{darkgreen}{//#1}}

\newcommand{\Iff}{\ensuremath{\Leftrightarrow}\xspace}
\newcommand{\eqind}{\ensuremath{\stackrel{\text{\scalebox{.7}{ind}}}{=}}\xspace}
\newcommand{\Iffind}{\ensuremath{\stackrel{\text{\scalebox{.7}{ind}}}{\Iff}}\xspace}

\renewcommand{\B}{\ensuremath{\mathbb{B}}\xspace}
\newcommand{\pequiv}{\equiv_p}

\newtheorem{fact}{Fact}%[section]
\newtheorem{theorem}{Theorem}

\newtheorem{lemma}{Lemma}
\newtheorem{example}{Example}
\newtheorem{remark}{Remark}
\def\newplaintheorem#1#2{%
\newtheorem{#1plain}{#2}%
\newenvironment{#1}{\begin{#1plain} \em}{\end{#1plain}}}
\newplaintheorem{generalremark}{General Remark}
\crefname{generalremark}{General Remark}{General Remarks}

\newcommand{\tabularallsat}{\textsc{TabularAllSAT}\xspace}
\newcommand{\tabularallsmt}{\textsc{TabularAllSMT}\xspace}
\newcommand{\df}{\textsc{d4}}
\newcommand{\decdnnf}{\textsc{model-graph}\xspace}
\newcommand{\dfdecdnnf}{\textsc{\df+\decdnnf}\xspace}

\ShortHeadings{On CNF Conversion for %
SAT %\ignoreinshort{
and SMT %}  
Enumeration}
{Masina, Spallitta, \& Sebastiani}
\firstpageno{1}

\begin{document}

\title{On CNF Conversion for SAT
  and SMT
  Enumeration}

\author{\name Gabriele Masina \email gabriele.masina@unitn.it \\
  \name Giuseppe Spallitta \email giuseppe.spallitta@unitn.it \\
  %\addr NASA Ames Research Center, Mail Stop: 244-7,\\
  %Moffett Field, CA  94035 USA
  %\AND
  \name Roberto Sebastiani \email roberto.sebastiani@unitn.it \\
  \addr DISI, University of Trento,
  Via Sommarive 9, 38123 Povo, Trento, Italy
  %  \AND
  %  \name Philip Laird \email laird@ptolemy.arc.nasa.gov \\
  %  \addr NASA Ames Research Center,
  %  AI Research Branch, Mail Stop: 269-2,\\
  %  Moffett Field, CA  94035 USA
}

% For research notes, remove the comment character in the line below.
% \researchnote

\maketitle

\begin{abstract}
  Modern SAT %\ignoreinshort{
  and SMT %} 
  solvers are designed to handle problems expressed in
  Conjunctive Normal Form (CNF) so that non-CNF problems must be
  CNF-ized upfront, typically by using variants of either Tseitin
  or Plaisted and Greenbaum transformations.
  %converted into CNF upfront.
  % \ignore{The Tseitin transformation is a widely used procedure for this conversion, since it generates an equisatisfiable formula without causing an exponential increase in the number of clauses\ignoreinshort{, by introducing fresh Boolean atoms to label sub-formulas}.
  % This transformation can also be employed for SAT enumeration (AllSAT) by projecting the enumeration on original variables only.}
  When passing from plain solving to enumeration, however, the capability of
  producing partial satisfying assignments that are as small as possible becomes
  crucial, which raises the question of whether such CNF encodings are also
  effective for enumeration.

  In this paper, we investigate both theoretically and empirically the
  effectiveness of CNF conversions for
  %\ignoreinlong{disjoint} 
  SAT and SMT enumeration. On the negative side, we show that: (i)~Tseitin
  transformation prevents the solver from producing short partial assignments,
  thus seriously affecting the effectiveness of enumeration; (ii)~Plaisted and
  Greenbaum transformation overcomes this problem only in part. On the positive
  side, we prove theoretically and we show empirically that combining Plaisted
  and Greenbaum transformation with NNF preprocessing upfront \mbox{---which} is
  typically not used in solving--- can fully overcome the problem and can
  drastically reduce both the number of partial assignments and the execution
  time. \ignore{This analysis is validated by an experimental evaluation on
    non-CNF problems originating from both synthetic and real-world-inspired
    applications.}%\ignoreinshort{%} % mia aggiunta%significantly
\end{abstract}

\section{Introduction}%
State-of-the-art SAT and SMT solvers deal very efficiently with formulas
expressed in Conjunctive Normal Form (CNF). In real-world scenarios, however,
it is common for problems to be expressed as non-CNF formulas. Hence, these
problems are converted into CNF before being processed by the solver. This
conversion is generally done by using variants of
the Tseitin~\cite{tseitinComplexityDerivationPropositional1983} or
the Plaisted and Greenbaum~\cite{plaistedStructurepreservingClauseForm1986} transformations, which
generate a linear-size equisatisfiable CNF formula by labeling sub-formulas
with fresh Boolean atoms. These transformations can also be employed for SAT
and SMT enumeration (AllSAT and AllSMT, respectively), by projecting the truth
assignments
onto the original atoms only.

When passing from plain \emph{solving} to \emph{enumeration}, however, the
capability of enumerating \emph{partial} satisfying assignments that are as
small as possible is crucial, because each prevents from enumerating a number
of total assignments that is exponential w.r.t.\ the number of unassigned
atoms. This raises the question of whether CNF encodings conceived for solving
are also effective for enumeration.
%A study on the impact of the CNF encoding on feature model analysis has been done by Kuiter et al.~\cite{kuiterTseitinNotTseitin2022}. 
To the best of our knowledge, however, no research has yet been published to
analyze how the different CNF encodings may affect the effectiveness of the
solvers for AllSAT and AllSMT.

\paragraph*{Contributions.}
In this paper, we investigate, both theoretically and empirically, the effectiveness of different CNF transformations for %\ignoreinshort{
SAT and SMT %} 
enumeration, %
%\ignoreinshort{
both in the disjoint and non-disjoint cases. %}. 
%\ignoreinlong{We focus on AllSAT, restricting to disjoint enumeration. We expect analogous results for AllSMT.}\@ 
The contribution of this paper is twofold.

First, on the negative side, we show that the commonly-employed CNF
transformations for solving are not suitable for enumeration.\@ In particular,
we notice that the Tseitin transformation introduces top-level label
definitions for sub-formulas with double implications, which need to be
satisfied as well, and thus prevent the solver from producing short partial
assignments. We also notice that the Plaisted and Greenbaum transformation
solves this problem only in part by labeling sub-formulas only with one-way
implications if they occur with single polarity, but it has similar issues to
the Tseitin transformation when sub-formulas occur with both polarities.

Second, on the positive side, we prove theoretically and we show empirically
that converting the formula into Negation Normal Form (NNF) before applying the
Plaisted and Greenbaum transformation can fix the problem and drastically
improve the effectiveness of the enumeration process by up to orders of
magnitude.

This analysis is confirmed by an experimental evaluation of non-CNF problems
originating from both synthetic and real-world-inspired applications.
%We preprocess each formula using different CNF transformations, and measure the number of partial assignments and the execution times for each transformation. 
The results confirm the theoretical analysis, showing that the proposed
combination of NNF with the Plaisted and Greenbaum CNF allows for a
drastic reduction in both the number of partial assignments and the execution
time. %} 

% \begin{ignoreinshortenv}

% \GSTODO{Chiamerei "Contributions" il paragrafo successivo? Okay che partiamo con il lavoor precedente, ma cosí sembra che tutto sia nel lavoro precedente.}
% \GMNOTE{L'ho cambiato su richiesta di un reviewer, possiamo tenere ``Contributions'' se preferite}

\paragraph*{Previous works.}
A preliminary %and much shorter 
version of this paper was presented at SAT 2023~\cite{masinaCNFConversionDisjoint2023}.
In this paper, we present the following novel contributions:
\begin{itemize}
  \item we formalize the main claim of the paper and formally prove it
        (\Cref{th:existsetaB} in \sref{sec:solution});
  \item we extend the analysis to non-disjoint AllSAT and disjoint and non-disjoint
        AllSMT;
  \item we extend the empirical evaluation to a much broader set of benchmarks,
        including also non-disjoint AllSAT and disjoint and non-disjoint AllSMT, which
        confirm the theoretical results. We extend the evaluation by using also the
        novel AllSAT and AllSMT enumerators \tabularallsat{} and \tabularallsmt{},
        obtaining similar results as with \mathsat{}. Moreover, we extend the timeout
        of each job-pair from 1200s to 3600s, providing thus a more informative
        comparison;
  \item we present a much more detailed related work section.
\end{itemize}
\paragraph*{Organization.}
The paper is organized as follows. In~\sref{sec:background} we introduce the theoretical background necessary to understand the rest of the paper. In~\sref{sec:problem} we analyze the problem of the classical CNF transformations when used for AllSAT
%\ignoreinshort{\ 
and AllSMT.\@
%}.
In~\sref{sec:solution} we propose one possible solution, whose effectiveness is
evaluated on both synthetic and real-world-inspired benchmarks
in~\sref{sec:experiments}. The related work is presented
in~\sref{sec:related-work}. We conclude the paper in~\sref{sec:conclusions},
drawing some final remarks and indicating possible future work. %\ignoreinshort{%} 

\section{Background}%
\label{sec:background}
This section introduces the notation and the theoretical background necessary
to understand the content of this paper.
%\ignoreinshort{
We assume the reader is familiar with the basic syntax, semantics, and results
of propositional and first-order logics.\@ We briefly summarize the main
concepts and results of SAT and SMT, and the fundamental ideas behind SAT and
SMT enumeration and projected enumeration implemented by modern AllSAT and
AllSMT solvers.
%}
%\ignoreinlong{We recall the standard syntax, semantics, and results of propositional logic and Satisfiability Modulo Theories (SMT), and the fundamental ideas behind SAT enumeration and projected enumeration implemented by modern AllSAT solvers.}

\subsection{%\ignoreinlong{Propositional Logic}\ignoreinshort{
    SAT and SAT Modulo Theories}%}%
\label{sec:background:propositional-logic}
%In this section, we summarize some basic definitions and results of
%propositional logic \ignoreinshort{\GMCHANGE{ and SMT}}.

\paragraph*{Notation.}
%\ignoreinlong{ and Terminology}.
In the paper, we adopt the following conventions. We refer to Boolean atoms
with capital letters, such as $A$ and $B$,
%\ignoreinshort{
and to first-order atoms (including Boolean atoms), with Greek letters, e.g.,
$\alpha$.\@
% and to  atoms, both Boolean and first-order, with Greek letters such as $\alpha$. %}.
The symbols $\allA\defas\set{A_1, \dots, A_N }$ and $\allB\defas\set{B_1,
        \dots, B_M}$ denote disjoint sets of Boolean atoms,
%\ignoreinshort{, 
and $\allalpha{}\defas\set{\alpha_1,\dots,\alpha_k}$ denotes a set of
first-order (including Boolean) atoms. %generic 
%}
% \ignoreinshort{\GMCHANGE{The symbol $\allalpha\defas\set{\alpha_1,\dots,\alpha_M}$ denotes a set of generic (Boolean or \T{}-) atoms s.t.\ \allalpha{} and \allB{} are disjoint.}}
Propositional
%\ignoreinshort{
and first-order
%} 
formulas are referred to with Greek letters $\vi, \psi$.
We write $\vi(\allalpha)$ to denote that $\allalpha$ is the set of atoms
occurring in $\vi$.
We denote Boolean constants by $\B\defas\set{\top, \bot}$. Total truth
assignments are denoted by $\eta$, whereas partial truth assignments are
denoted by $\mu$, possibly annotated with superscripts. %\ignoreinshort{, %}.

\paragraph*{%\ignoreinshort{
    Propositional Satisfiability.}%}
A \emph{propositional} (also \emph{Boolean}) \emph{formula} $\vi$ can be defined recursively as follows. The Boolean constants $\top$ and $\bot$ are formulas; a Boolean atom $A$ and its negation $\neg A$ are formulas, also referred to as \emph{literals}; a connection of two formulas by one of the Boolean connectives $\wedge, \vee, \imp, \iff$ is a formula.
%\ignoreinshort{
A disjunction $(\vee)$ of literals is called a \emph{clause}. %A conjunction
%$(\wedge)$ of literals is called a \emph{cube}.
%\ignoreinshort{
\emph{Propositional satisfiability (SAT)} is the problem of deciding the satisfiability of a propositional formula, i.e., if there is a way to assign truth values to the atoms such that the formula evaluates to $\top$.
We refer to standard literature (e.g.,~\cite{biereHandbookSatisfiability2021})
for details.
%
%}
% \ignoreinshort{\GMCHANGE{Some formulas can also present quantifiers to express more complex conditions on Boolean parameters. In particular, the existential quantifier ($\exists B$), means that for at least one truth value of $B$ the formula should be satisfied.}}

% \begin{ignoreinshortenv}
\paragraph*{Satisfiability Modulo Theories.}
% \paragraph*{Satisfiability Modulo Theories}
%Let $\Sigma$ be a first-order signature containing function and predicate symbols. A $\Sigma$-term is either a variable or it is built by applying a function symbol in $\Sigma$ to $\Sigma$-terms. If $t_1, \dots , t_n$ are $\Sigma$-terms and $P$ is a predicate symbol, then $P(t_1, \dots , t_n)$ is a $\Sigma$-atom. Boolean atoms are defined as 0-ary predicate symbols. 
% Given a first-order signature $\Sigma$, i.e., a set of predicate, function and constant symbols, a term is either a variable $v$, a constant symbol, or a combination of terms using function symbols. A $\Sigma$-atom is a predicate symbol applied to a tuple of terms (Boolean atoms can be viewed as zero-arity predicates).
% A $\Sigma$-formula $\vi$ is built in the usual way out of the Boolean connectives, and $\Sigma$-atoms. A $\Sigma$-theory is a set of first-order sentences with signature $\Sigma$.
% Examples of theories are bit-vector arithmetic (\bv{}), floating-point arithmetic (\fl{}), linear arithmetic over integers (\laint{}) and real numbers (\larat{}), and arrays (\mem{}).\@ % and their combinations.
% With a little abuse of notation, we often refer to a theory \T{} instead of its corresponding signature, so that by \T-formulas \resp{literals, atoms} we mean formulas \resp{literals, atoms} in the signature of \T{}.
As it is standard in most \smt{} literature, we restrict to quantifier-free
first-order formulas. A first-order term is either a variable, a constant
symbol, or a combination of terms using function symbols. A first-order atom is
a predicate symbol applied to a tuple of terms (Boolean atoms can be viewed as
zero-arity predicates). A first-order formula is either a first-order atom, or
a connection of two formulas by one of the Boolean connectives. %
% is built in the usual way out of the existential and universal quantifiers, the
% Boolean connectives, and first-order atoms,
% i.e., predicate symbols applied to tuples of terms (Boolean atoms can be viewed as zero-arity predicates).
A first-order theory \T{} is a (possibly infinite) set of first-order formulas,
that provides an intended interpretation of constant, function, and predicate
symbols. Examples of theories of practical interest are those of equality and
uninterpreted functions (\euf{}), of linear arithmetic over integer (\laint{})
or real numbers (\larat{}), and combinations thereof. We refer to formulas and
atoms over symbols defined by \T{} as \T-formulas and \T-atoms, respectively.

\emph{Satisfiability Modulo the Theory $\T$, also $\smtt$,} is the problem of deciding the satisfiability of a first-order formula with respect to some background theory \T{}.
% With a little abuse of notation, we often refer to a theory \T{} instead of its corresponding signature, so that by \T-formulas \resp{literals, atoms} we mean formulas \resp{literals, atoms} in the signature of \T{}.
%~\cite{barrettSatisfiabilityModuloTheories2018} %, i.e.\ a set of closed first-order formulas. 
% In \smtt{} formulas, atoms are either Boolean atoms or \T{}-atoms.
A formula $\vi$ is \emph{\T{}-satisfiable}
%(also \emph{\T{}-consistent}) 
if $\vi\wedge\T{}$ is satisfiable in the first-order sense. Otherwise, it is
\emph{\T{}-unsatisfiable}. % (also \emph{\T{}-inconsistent}).
% The notation and terminology introduced for propositional logic can be straightforwardly extended for SMT formulas.
%
We refer to standard literature
(e.g.,~\cite{barrettSatisfiabilityModuloTheories2021}) for details.
% In this paper, we focus on quantifier-free SMT formulas.
% Since propositional logic can be seen as a strict subset of first-order logic, by considering

Hereafter, unless otherwise indicated, by \emph{formulas} we mean both Boolean
and \T-formulas, and by \emph{atoms} we mean both Boolean and \T-atoms, for a
generic theory \T{}.
We assume that formulas are internally represented as single-rooted directed acyclic
graphs (DAGs) %, i.e.\ arborescences, 
where internal nodes are labeled
with Boolean connectives, and leaves are labeled with literals or
atoms. (Some authors call them "circuits", e.g.~\cite{friedAllSATCombinationalCircuits2023,friedEntailingGeneralizationBoosts2024}.)
%% RS:  il concetto di "syntactically equivalent formulas" non mi
%% risulta esistere.
%In this way, syntactically equivalent formulas are represented by the same DAG.\@ 
In this way, multiple occurrences of the same sub-formula in one formula are
represented by only one shared sub-DAG. The \emph{size} of a formula is the
number of nodes and arcs of its DAG representation.

% \end{ignoreinshortenv}

% \begin{definition}
% \ignoreinlong{A sub-formula occurs with \emph{positive} \resp{\emph{negative}} \emph{polarity} (also \emph{positively} \resp{\emph{negatively}}) if it occurs under an even \resp{odd} number of nested negations. Specifically, $\vi$ occurs positively in $\vi$; if $\neg\vi_1$ occurs positively \resp{negatively} in $\vi$, then $\vi_1$ occurs negatively \resp{positively} in $\vi$; if $\vi_1\wedge \vi_2$ or $\vi_1\vee\vi_2$ occur positively \resp{negatively} in $\vi$, then $\vi_1$ and $\vi_2$ occur positively \resp{negatively} in $\vi$; if $\vi_1\imp\vi_2$ occurs positively \resp{negatively} in $\vi$, then $\vi_1$ occurs negatively \resp{positively} and $\vi_2$ occurs positively \resp{negatively} in $\vi$; if $\vi_1\iff\vi_2$ occurs in $\vi$,
%     then $\vi_1$ and $\vi_2$ occur both positively and negatively in $\vi$.
% }
% \end{definition}
% \begin{definition}

% \begin{ignoreinshortenv}
\paragraph*{Total and partial truth assignments.}
%%%%%%%% SPOSTATO QUI NELLA VERSIONE LONG %%%%%%%%%%%%%
% \begin{definition}
Given a set $\allalpha$ of atoms, a \emph{total truth assignment} is a total
map $\eta: \allalpha\longmapsto\B$. A \emph{partial truth assignment} is a
partial map $\mu: \allalpha\longmapsto\B$.
% \end{definition}
Notice that a partial truth assignment represents (aka ``covers'') $2^K$ total
truth assignments extending it, where $K$ is the number of unassigned atoms
%\ignoreinlong{variables}\ignoreinshort{
% atoms %} 
in $\mu$.
%With a little abuse of notation,
We often represent a truth assignment either as a set, s.t.\
$\mu\defas\set{\alpha\ |\ \mu(\alpha)=\top}\cup\set{\neg \alpha\ |\
        \mu(\alpha)=\bot}$, or as a conjunction of literals, s.t.\
$\mu\defas\bigwedge_{\mu(\alpha)=\top}{\alpha}\wedge\bigwedge_{\mu(\alpha)=\bot}{\neg
        \alpha}$. If $\mu_1\subseteq\mu_2$ \resp{$\mu_1\subset\mu_2$} we say that
$\mu_1$ is a \emph{sub-assignment} \resp{\emph{strict sub-assignment}} of
$\mu_2$ and that $\mu_2$ is a \emph{super-assignment} \resp{\emph{strict
        super-assignment}} of $\mu_1$. We denote with $\residual{\vi}{\mu}$ the
\emph{residual of $\vi$ under $\mu$}, i.e.\ the formula obtained by
substituting in $\vi$ each $\alpha_i\in\allalpha$ with $\mu(\alpha_i)$, and by
recursively applying the standard propagation rules of truth values through
Boolean operators.

% \begin{definition}
Given a set $\allalpha$ of atoms and a formula $\vi(\allalpha)$, we say that a
\emph{[partial or total] truth assignment $\mu:\allalpha\longmapsto\B$
    propositionally satisfies} $\vi$, denoted as $\mu\pmodels\vi$, iff
$\residual{\vi}{\mu}=\top$.~\footnote{ The definition of satisfiability by
    partial assignment may present some ambiguities for non-CNF and
    existentially-quantified
    formulas~\cite{sebastianiAreYouSatisfied2020,mohleFourFlavorsEntailment2020,sebastianiEntailmentVsVerification2025}.
    Here we adopt the above definition because it is the easiest to implement, the
    most efficient to compute, and it is the one typically used by state-of-the-art
    SAT solvers. {We refer to~\cite{sebastianiAreYouSatisfied2020,mohleFourFlavorsEntailment2020,sebastianiEntailmentVsVerification2025} for
            details.} }
            
            %\ignoreinshort{%} %\ignoreinshort{%} 
% \end{definition}
% \footnote{
%}
A partial truth assignment $\mu$ is \emph{minimal} for $\vi$ iff
$\mu\pmodels\vi$ and every strict sub-assignment $\mu{}'\subset\mu$ is such
that $\mu{}'\not\pmodels\vi$.
%\ignoreinshort{
A Boolean formula is satisfiable iff there exists a truth assignment
propositionally satisfying it. A \T-formula is \T-satisfiable iff there exists
a \T-satisfiable truth assignment propositionally satisfying it. %}

Given two disjoint sets $\allalpha, \allB$ of atoms, and a%\ignoreinlong{ Boolean} 
%\ignoreinlong{CNF }
formula $\vicnf(\allalpha\cup\allB)$, we say that a \emph{[partial or total]
    truth assignment $\mualpha:\allalpha\longmapsto\B$ propositionally%\ignoreinshort{
    %} 
    satisfies} $\exists\allB.\vicnf$ iff there exists a total truth assignment
$\etaB:\allB\longmapsto\B$ such that $\mualpha\cup\etaB:
    \allalpha\cup\allB\longmapsto\B$ propositionally satisfies $\vicnf$. %\ignoreinshort{%} 

%\ignoreinshort{
Two formulas $\vi{}$ and $\psi{}$ are \emph{propositionally equivalent},
denoted as $\vi\pequiv\psi$, iff every total truth assignment propositionally
satisfying $\vi$ also propositionally satisfies $\psi$, and vice versa. %}

%%% SPOTATO QUI NELLA VERSIONE LONG %%%%%%%%%%%%%

% \end{ignoreinshortenv}

\paragraph*{Negation Normal Form.}
%%%%%%%% SPOSTATO QUI NELLA VERSIONE LONG %%%%%%%%%%%%%
% \ignoreinshort{
Given a formula $\vi$, a sub-formula occurs with \emph{positive}
\resp{\emph{negative}} \emph{polarity} (also \emph{positively}
\resp{\emph{negatively}}) if it occurs under an even \resp{odd} number of
nested negations. Specifically, $\vi$ occurs positively in $\vi$; if
$\neg\vi_1$ occurs positively \resp{negatively} in $\vi$, then $\vi_1$ occurs
negatively \resp{positively} in $\vi$; if $\vi_1\wedge \vi_2$ or
$\vi_1\vee\vi_2$ occur positively \resp{negatively} in $\vi$, then $\vi_1$ and
$\vi_2$ occur positively \resp{negatively} in $\vi$; if $\vi_1\imp\vi_2$ occurs
positively \resp{negatively} in $\vi$, then $\vi_1$ occurs negatively
\resp{positively} and $\vi_2$ occurs positively \resp{negatively} in $\vi$; if
$\vi_1\iff\vi_2$ occurs in $\vi$, 
then $\vi_1$ and $\vi_2$ both occur positively and negatively. % in $\vi$.
% then $\vi_1$ and $\vi_2$ occur both positively and negatively in $\vi$.
% }

%\begin{rschange}
A formula is in \emph{Negation Normal Form (NNF)} iff it is given only by the
recursive applications of $\wedge$ and $\vee$ to literals, i.e., iff all its
sub-formulas occur positively, except for literals. %\ignoreinlong{ Boolean} %\ignoreinshort{, 
%An NNF formula can be conveniently represented as a  Directed Acyclic Graph (DAG)
%---that is, as a single-root and/or DAG with literals as leaves.
%\ignoreinshort{
%We call the {\em size} of an NNF DAG the sum of the numbers of
%its nodes and arcs. %}
% \end{definition}
A formula can be converted into a propositionally-equivalent NNF formula by
recursively rewriting implications $(\vi_1\imp\vi_2)$ as $(\neg\vi_1\vee\vi_2)$
and equivalences $(\vi_1\iff\vi_2)$ as
$(\neg\vi_1\vee\vi_2)\wedge(\vi_1\vee\neg\vi_2)$, and then by recursively
``pushing down'' the negations: $\neg(\vi_1\vee\vi_2)$ as
$(\neg\vi_1\wedge\neg\vi_2)$, $\neg(\vi_1\wedge\vi_2)$ as
$(\neg\vi_1\vee\neg\vi_2)$ and $\neg\neg\vi_1$ as $\vi_1$. %\ignoreinshort{%}%\ignoreinshort{%}
The following fact holds.
%\begin{gmchange}
\begin{lemma}%
    \label{th:nnfdaglinear}
    Let \vi{} be a %\ignoreinlong{ Boolean} 
    formula and \NNF{\vi} be the NNF
    formula resulting from converting \vi{} into NNF as described
    above. Then the size of \NNF{\vi} is linear w.r.t.\ the size of \vi.
    % If $\NNF{\vi}$ is represented as a DAG,  then its size is linear w.r.t.\ the original one.
\end{lemma}
\noindent(Although we believe this fact is well-known,
% we failed to find a formal proof
% in the literature. Hence,
we provide a formal proof in~\sref{sec:proofnnfdaglinear}.)
%\end{gmchange}
% }
Intuitively, we only need at most 2 nodes for each sub-formula $\vi_i$ of
$\vi$, representing $\NNF{\vi_i}$ and $\NNF{\neg\vi_i}$ for positive and
negative occurrences of $\vi_i$ respectively. These nodes are shared among up
to exponentially-many branches generated by expanding the nested iffs.~%
%\footnote{%
Notice that \Cref{th:nnfdaglinear} holds {\em because a DAG representation of
        NNF is used}, so that we need at most two nodes for each sub-formula, one for
each polarity. (If instead one represented the NNF formula as a tree, then the
NNF representation would blow up exponentially in the number of nested $\iff$'s
in the original formula, because each sub-formula $(\vi_i \iff \vi_j)$ is
recursively expanded into $(\neg\vi_i \vee \vi_j)\wedge (\vi_i \vee
    \neg\vi_j)$.)
%}
%\end{rschange}

% \ignoreinshort{
We have the following fact, for which we provide a complete proof
in~\sref{sec:proofmunnf}.
\begin{lemma}%
    \label{th:munnf}
    Consider a formula $\vi$, %, and let $\NNF{\vi}$ be its NNF. % DAG.\@
    and a partial truth assignment $\mu$. % on $\GMCHANGEp{\allalpha $. 
    Then $\residual{\vi}{\mu}=v$ iff $\residual{\NNF{\vi}}{\mu}=v$, for
    $v\in\set{\top,\bot}$.
    % , where $\any$ represents the case in which the residual is neither $\top$ nor $\bot$.
\end{lemma}
Notice that \Cref{th:munnf} \emph{is not} a direct consequence of the fact that \NNF{\vi} is equivalence-preserving, because the above notion of satisfiability by partial assignment is such that if $\vi_1\pequiv\vi_2$, then $\mu\pmodels\vi_1$ \emph{does not} imply that $\mu\pmodels\vi_2$~\cite{sebastianiAreYouSatisfied2020,mohleFourFlavorsEntailment2020,sebastianiEntailmentVsVerification2025}. Consider, e.g., $\vi_1\defas(A\vee{}B)\wedge(A\vee\neg{}B)$ and $\vi_2\defas(A\wedge{}B)\vee(A\wedge{\neg{}B})$, and the partial assignment $\mu\defas\set{A}$.
Although $\vi_1\pequiv\vi_2$, we have that $\residual{\vi_1}{\mu}=\top$, but $\residual{\vi_2}{\mu}=B\vee\neg B$, which, although \emph{logically} equivalent to $\top{}$, is \emph{syntactically} different from it.
% }

\paragraph*{CNF Transformations.}%
\label{sec:bg:cnf}
A %\ignoreinlong{ Boolean} 
formula is in \emph{Conjunctive Normal Form (CNF)} iff it is a conjunction %\ignoreinlong{ $(\wedge)$} 
of clauses. %\ignoreinlong{, where a clause is a disjunction $(\vee)$ of literals}.
% \end{definition}
Numerous CNF transformation procedures, commonly referred to as CNF-izations,
have been proposed in the literature.
%\ignoreinlong{In the next paragraph, we}\ignoreinshort{
We summarize the three most frequently employed techniques. %} 

The \emph{Classic CNF-ization} (\DeMorganCNF{}) converts a formula into a
propositionally-equivalent formula in CNF by applying DeMorgan's rules. First,
it converts the formula into NNF.\@ Second, it recursively rewrites
sub-formulas $\vi_1\vee(\vi_2\wedge\vi_3)$ as
$(\vi_1\vee\vi_2)\wedge(\vi_1\vee\vi_3)$ to distribute $\vee$ over $\wedge$,
until the formula is in CNF.\@ The principal limitation of this transformation
lies in the possible exponential growth of the resulting formula compared to
the original (e.g., when the formula is a disjunction of conjunctions of
sub-formulas), making it unsuitable for modern SAT and SMT solvers~(see e.g.,
\cite{prestwichCNFEncodings2021}). %\ignoreinlong{in DNF}\ignoreinshort{%}), %\ignoreinshort{ %} 

The \emph{Tseitin CNF-ization}
(\TseitinCNF{})~\cite{tseitinComplexityDerivationPropositional1983} avoids this
exponential blow-up by labeling each sub-formula $\vi_i$ with a fresh Boolean
atom $B_i$, which is used as a placeholder for the sub-formula. Specifically,
it consists in applying recursively bottom-up the rewriting rule
$\vi\Longrightarrow\vi[\vi_i|B_i] \wedge \DeMorganCNF(B_i\iff\vi_i)$ until the
resulting formula is in CNF, where $\vi[\vi_i|B_i]$ is the formula obtained by
substituting in $\vi$ every occurrence of $\vi_i$ with $B_i$.

The \emph{Plaisted and Greenbaum CNF-ization}
(\PlaistedCNF{})~\cite{plaistedStructurepreservingClauseForm1986} is a variant
of the \TseitinCNF{} that exploits the polarity of sub-formulas to reduce the
number of clauses of the final formula. If a sub-formula $\vi_i$
appears only with positive \resp{negative} polarity, then it can be labeled
with a one-way implication as $\DeMorganCNF(B_i\imp\vi_i)$
\resp{$\DeMorganCNF(B_i\limp\vi_i)$}; if $\vi_i$ occurs with both polarities,
then it is labeled as $\DeMorganCNF(B_i\iff\vi_i)$, as
with \TseitinCNF{}.

With both \TseitinCNF{} and \PlaistedCNF{}, due to the introduction of the
label variables, the final formula does not preserve the propositional
equivalence with the original formula but only the equisatisfiability.
Moreover, they also have a stronger property. If $\vi(\allalpha)$ is a non-CNF
formula and $\vicnf(\allalpha\cup\allB)$ is either the \TseitinCNF{} or the
\PlaistedCNF{} encoding of $\vi$, where $\allB$ are the fresh Boolean atoms
introduced by the transformation, then
$\vi(\allalpha)\pequiv\exists\allB.\vicnf(\allalpha\cup\allB)$. %\ignoreinshort{%}

Most of the modern SAT and SMT solvers do not deal directly with non-CNF
formulas, rather they convert them into CNF by using either \TseitinCNF{} or
\PlaistedCNF{},
% \ignoreinlong{. As seen in the previous paragraph, since these transformations introduce fresh atoms into the resulting formulas, a model of $\vi(\allalpha)$ can be found as a truth assignment satisfying $\exists\allB.\vicnf(\allalpha\cup\allB)$.}%
%\ignoreinshort{, 
and then find truth assignments propositionally satisfying $\vi(\allalpha)$ by
finding truth assignments propositionally satisfying
$\exists\allB.\vicnf(\allalpha\cup\allB)$. %}
% \begin{definition}
% \end{definition}

% \RSNOTE{@Gabriele: nota che qui ho tolto la parte relativa a ``\any'',
%   perche' non l'abbiamo definita chiaramente, e perche' e' inutile, in
% quanto deriva dal fatto che la proprietà e' in ``iff''.}
% GM: Ok

% The work of this paper focuses on the following result from~\cite{sebastianiAreYouSatisfied2020}.
% \begin{lemma}%
%     \label{lemma:noncnf-sat}
%     Let $\vi(\allA)$ be a non-CNF formula and let $\vicnf(\allA\cup\allB)$ be the result of applying either the \TseitinCNF{} or the \PlaistedCNF{} transformation to $\vi$.\\
%     Given a total truth assignment $\etaA:\allA\longmapsto\B$, then  $\etaA\pmodels\vi$ iff $\etaA\pmodels\exists\allB.\vicnf$. However, given a partial truth assignment $\muA:\allA\longmapsto\B$, then  $\muA\pmodels\vi$ if $\muA\pmodels\exists\allB.\vicnf$, but the reverse implication does not hold. 
% \end{lemma}
% \GMNOTE{Dobbiamo dimostrarlo per Plaisted?}

\subsection{AllSAT,
    AllSMT,
    Projected AllSAT
    and Projected AllSMT}%}%
\label{sec:background:allsat}

% \RSTODO{riscrivi \S2.2. di
%   conseguenza, spiegando che l'algoritmo e' solo un esempio e che la
%   disjointness e la minimalità non sono strettamente necessari.}

AllSAT is the task of enumerating all the truth assignments propositionally
satisfying a Boolean formula.\@ %\ignoreinshort{%}%\ignoreinlong{In this paper, we focus on the enumeration of disjoint models, that is, pairwise mutually-inconsistent models.}
% \begin{ignoreinshortenv}
The task can be found in the literature in two versions: \emph{disjoint}
AllSAT, where the assignments are required to be pairwise mutually
inconsistent, and \emph{non-disjoint} AllSAT, where they are not. A
generalization to the \smtt{} case is All\smtt{}, defined as the task of
enumerating all the \T{}{\em -satisfiable} truth assignments propositionally
satisfying a \smtt{} formula. Also in this case, both disjoint or non-disjoint
All\smtt{} are possible. %~\cite{yuAllSATUsingMinimal2014}. %\footnote{In the SMT literature the word ``AllSMT'' is used with slightly distinct meanings~\cite{lahiriSMTTechniquesFast2006,mathsat5_tacas13,phanAllSolutionSatisfiabilityModulo2015f}}. 

In the following, for the sake of compactness, we present definitions and
algorithms referring to AllSAT. An extension to All\smtt{} can be obtained
straightforwardly, by substituting ``\allA'' with ``\allalpha'' and ``truth
assignments'' with ``\T{}-satisfiable truth assignments''.
%restricting to \T{}-satisfiable truth assignments.
% \end{ignoreinshortenv}

% \begin{definition}
Given a formula $\vi$, we denote with
$\TTA{\vi}\defas\set{\eta_1,\dots,\eta_j\dots,\eta_M}$ the set of all total
truth assignments propositionally satisfying $\vi$. %\ignoreinlong{ Boolean} %\ignoreinshort{%}
% \end{definition}
% \begin{definition}
We denote with $\TA{\vi}\defas\set{\mu_1,\dots,\mu_i\dots,\mu_N}$ a set of
partial truth assignments propositionally satisfying $\vi$ s.t.:%\ignoreinshort{%}
% \ignoreinlong{
%     \begin{enumerate}[(a)]
%         \item every $\eta\in\TTA{\vi}$ is a super-assignment of some $\mu\in\TA{\vi}$;
%         \item every pair $\mu_i,\mu_j\in\TA{\vi}$ assigns opposite truth value to at least one atom.
%     \end{enumerate}
% }
% \ignoreinshort{
\begin{enumerate}[(a)]
    \item\label{item:ta:complete} every $\eta\in\TTA{\vi}$ is a super-assignment of some $\mu\in\TA{\vi}$;
\end{enumerate}
and, only in the \emph{disjoint} case:
\begin{enumerate}[(b)]
    \item\label{item:ta:disjoint} every pair $\mu_i,\mu_j\in\TA{\vi}$ assigns opposite truth value to at least one atom.
\end{enumerate}
% }

% \end{definition}
Notice that, whereas $\TTA{\vi}$ is unique, multiple $\TA{\vi}$s are admissible
for the same formula $\vi$, including $\TTA{\vi}$. AllSAT is the task of
enumerating either $\TTA{\vi}$ or a set $\TA{\vi}$. Typically, AllSAT solvers
aim at enumerating a set $\TA{\vi}$ which is as small as possible, since every
partial truth assignment prevents from enumerating a number of total truth
assignments that is exponential w.r.t.\ the number of unassigned atoms, so that
to save large amounts of computational space and time. %\ignoreinshort{%}%\ignoreinlong{model}\ignoreinshort{%} %\ignoreinlong{models}\ignoreinshort{%} 

%\mathsat{}~\cite{mathsat5_tacas13}, the SAT/SMT solver that we used for our experiments.\GMSIDENOTE{Riscrivere non riferendosi solo a mathsat}
The enumeration of a $\TA{\vi}$ for a non-CNF formula $\vi$ is typically
implemented by first converting it into CNF, and then by enumerating its
satisfying assignments by means of \emph{Projected AllSAT}. %\ignoreinlong{models}\ignoreinshort{%} 
%\ignoreinshort{\GMCHANGE{/AllSMT}}}.
Specifically, let $\vi(\allA)$ be a non-CNF formula and let
$\vicnf(\allA\cup\allB)$ be the result of applying either \TseitinCNF{} or
\PlaistedCNF{} to $\vi$, where $\allB$ is the set of Boolean atoms introduced
by either transformation. $\TA{\vi}$ is enumerated via Projected AllSAT as
$\TA{\exists\allB.\vicnf}$, i.e.\ as a set of (partial) truth assignments over
$\allA$ that can be extended to total truth assignments satisfying $\vicnf$
over $\allA\cup\allB$. We refer the reader to the general schema described
in~\cite{lahiriSMTTechniquesFast2006}, which we briefly recap here for
completeness of narration. %\ignoreinshort{\GMCHANGE{/AllSMT}} %\ignoreinlong{models of}\ignoreinshort{%} %, as a consequence of \cref{lemma:noncnf-sat}. 

Let $\vicnf(\allA\cup\allB)$ be a CNF formula over two disjoint sets
%of %\ignoreinlong{Boolean variables}\ignoreinshort{
%atoms %} 
$\allA, \allB$ of atoms, where $\allA$ is a set of
\emph{relevant atoms} s.t.\ we want to enumerate a
$\TA{\exists\allB.\vicnf}$.
%\begin{rschange}
%\ignoreinshort{
For disjoint AllSAT with minimal models, the solver enumerates one-by-one
partial truth assignments $\mu_1,\dots,\mu_i,\dots\mu_N$ which comply with
point~\ref{item:ta:complete} above where each $\mu_i\defas\muA_i\cup\etaB_i$ is
s.t.:%}\ignoreinlong{The} 
\begin{enumerate}
    \item\label{item:mumodelsvi} (\emph{satisfiability}) $\mu_i\pmodels\vicnf$;
    \item\label{item:musdisjoint} (\emph{disjointness}) for each
          $j<i$, $\muA_i,\muA_j$ assign opposite truth values to some atom in $\allA$;
    \item\label{item:muAminimal} (\emph{minimality}) $\muA_i$ is \emph{minimal}, meaning that no literal can be dropped from it without losing properties~\ref{item:mumodelsvi} and~\ref{item:musdisjoint}.
          % \GMSIDENOTE{la property~\ref{item:muAminimal} può essere ``rilassata''? E.g., se uso chronological backtracking dovrei riuscire a produrre assignment parziali più corti anche se non sono necessariamente ``minimali'' nel senso in cui è inteso qui. Magari solo una footnote.
          % }
\end{enumerate}
%\ignoreinshort{
For the non-disjoint case, property~\ref{item:musdisjoint} is omitted, and the
reference to it in~\ref{item:muAminimal} is dropped. If the minimality of
models is not required, property~\ref{item:muAminimal} is omitted. %}\RSSIDENOTE{Da rivedere}

An example of a basic projected AllSAT procedure producing minimal models
(implemented e.g.\ in \mathsatfive{}~\cite{mathsat5_tacas13}) works as
follows. At each step $i$, it finds a total truth assignment
$\eta_i\defas\etaA_i\cup\etaB_i$ s.t.\ $\eta_i\pmodels\psi_i$, where
$\psi_i\defas\psi\wedge\bigwedge_{j=1}^{i-1}\neg\muA_j$, and then invokes a
minimization procedure on $\etaA_i$ to compute a partial truth assignment
$\muA_i$ satisfying properties~\ref{item:mumodelsvi},~\ref{item:musdisjoint}
and~\ref{item:muAminimal}. Then, the solver adds the blocking clause
$\neg\muA_i$ ---to ensure property~\ref{item:musdisjoint} for the disjoint
version and to ensure an exhaustive exploration of the solutions space for the
non-disjoint version--- and it continues the search. This process is iterated
until {$\psi_{N+1}$ }is found to be unsatisfiable for some $N$, and the set
$\set{\muA_i}_{i=1}^N$ is returned. %\ignoreinshort{---}%\ignoreinshort{ %} %the formula 
The minimization procedure consists in iteratively dropping literals one-by-one
from $\etaA_i$, checking if it still satisfies the formula. The outline of this
minimization procedure is shown in \cref{alg:minimize}. Each minimization step
is $O(\#\mathit{clauses}\cdot\#\mathit{vars})$.

\begin{algorithm}[t]
    \begin{algorithmic}[1]
        %\begin{rschange}
        \caption[A]{{\sc Minimize-Assignment}($\psi_i, \eta_i, \allA$)\\
            \hspace*{\algorithmicindent}\textbf{Input}:
                CNF formula $\psi_i$, total truth assignment $\eta_i$ s.t.\ $\eta_i\pmodels\psi_i$, set of relevant atoms $\allA$:\\
                % \ignoreinlong{$\psi_i\defas\psi\wedge\bigwedge_{j=1}^{i-1}\neg\muA_j$,}%
                % say that if disjoint then ... otherwise only \psi
                % \ignoreinshort{
                \hspace*{\algorithmicindent}\null\quad $\psi_i(\allA\cup\allB)\defas\psi\wedge\bigwedge_{j=1}^{i-1}\neg\muA_j$ if disjoint, $\psi_i\defas\psi$ if non-disjoint \\%
                % }%
                \hspace*{\algorithmicindent}\null\quad $\eta_i = \etaA_i \cup \etaB_i$\\
            \hspace*{\algorithmicindent}\textbf{Output}: minimal partial truth assignment $\muA_i$ s.t.\ $\muA_i\cup\etaB_i\pmodels\psi_i$%
            \label{alg:minimize}}
        \STATE $\muA_i \leftarrow \etaA_i$
        \FOR{$\ell\in\muA_i$}
        \IF{$\residual{\psi_i}{\muA_i\setminus\set{\ell}\ \cup\ \etaB_i} = \top$}
        \STATE $\muA_i \leftarrow \muA_i \setminus \set{\ell}$
        \ENDIF
        \ENDFOR
        \RETURN $\muA_i$
        %\end{rschange}
    \end{algorithmic}
\end{algorithm}

Notice that, since we are in the context of {\em projected} AllSAT, the
minimization algorithm only minimizes the relevant atoms in $\allA$, %\ignoreinshort{%}\ignoreinlong{variables} 
%although 
and the truth value of existentially quantified variables in $\allB$ is still
used to check the satisfiability of the formula by the current partial
assignment.\@ Moreover, in the disjoint case, to enforce the pairwise
disjointness between the assignments, $\psi_i$ in \cref{alg:minimize} refers to
the original formula conjoined with all current blocking clauses
$\bigwedge_{j=1}^{i-1}\neg\muA_j$, whereas in the non-disjoint case $\psi_i$
refers to the original formula only, allowing for a more effective minimization
while renouncing the disjointness property. %\ignoreinshort{%} %\ignoreinshort{%}. 
%\GMSIDENOTE{Ho tolto la frase "Conflict clauses are excluded by the minimization, being redundant." perché il reviewer E ha chiesto chiarimenti, e nonmi sembra necessaria per il background}
% Conflict clauses are excluded by the minimization, being redundant.
% \ignoreinshort{

The procedure reduces to non-projected AllSAT if $\allB=\emptyset$. The same
procedure can be easily generalized to disjoint and non-disjoint All\smtt{},
with the only difference that only \T-satisfiable total truth assignments
$\eta_i$ are considered. % and the minimization procedure is applied to the Boolean abstractions of $\psi_i$ and $\eta_i$.
% }

We stress the fact that the above procedure is reported just as an example, and
that the work described in this paper is agnostic w.r.t.\ the AllSAT or AllSMT
procedure used, provided that it is able to produce \emph{partial} assignments. %\ignoreinlong{disjoint }%\ignoreinshort{%}
%enumeration strategy matches the above conditions. %}

% \GMNOTE{Mention different enumeration approaches, e.g., blocking/non-blocking~\cite{todaImplementingEfficientAll2016}, different minimization techniques (see~\cite{todaImplementingEfficientAll2016,morgadoGoodLearningImplicit2005,raviMinimalAssignmentsBounded2004}). Mention methods not matching these conditions, e.g., BDD-based, Dualiza}

%\end{rschange}

\section{The impact of CNF transformations on %\ignoreinlong{AllSAT} \ignoreinshort{
  Enumeration} %}%
\label{sec:problem}
% \RSTODO{riscrivi \S3 e \S4 di
% conseguenza, spiegando prima l'intuiizione e poi l'esempio.}
In this section, we present a theoretical analysis of the impact of different
CNF-izations on the enumeration of short partial truth assignments. In
particular, we focus on
\TseitinCNF~\cite{tseitinComplexityDerivationPropositional1983} and
\PlaistedCNF~\cite{plaistedStructurepreservingClauseForm1986}. %\ignoreinlong{the AllSAT task}\ignoreinshort{%}. 
% \ignoreinshort{ %
% }
We point out how CNF-izing AllSAT problems using these transformations can
introduce unexpected drawbacks for enumeration. In fact, we show that the
resulting encodings can force the enumerator to produce partial assignments
that are larger in size and in number than necessary. % when they are used to preprocess non-CNF formulas. 
In our analysis we refer to AllSAT, but it applies to All\smt{} as well by
restricting to theory-satisfiable truth assignments.\@ Moreover, the analysis
applies to both disjoint and non-disjoint enumeration.

\subsection{The impact of Tseitin CNF transformation}%
\label{sec:problem:label}

We show that using the \TseitinCNF{}
transformation~\cite{tseitinComplexityDerivationPropositional1983} can be
problematic for enumeration. %preprocessing the input formula
In particular, we point out a fundamental weakness of \TseitinCNF{}:
%
%\ADDED{
%\GMSIDENOTE{RSTODO: Change ``suffices to satisfy''?}
\begin{fact}%
    \label{fact:tseitin}
    If a partial assignment \muA{} satisfies $\vi$, this does not imply that \muA{} satisfies $\exists
        \allB.\vicnfts$.
\end{fact}
\noindent
In fact, we recall that \TseitinCNF{} works by applying recursively the rewriting step (\sref{sec:bg:cnf}):
\begin{eqnarray}
    \label{eq:rewritingTseitin}
    \vi\Longrightarrow\vi[\vi_i|B_i] \wedge (B_i\iff\vi_i)
\end{eqnarray}
\noindent and then by recursively CNF-izing the two conjuncts.
A {\em partial} assignment \muA{} may satisfy a non-CNF formula $\vi(\allA)$
because it does not need to assign a truth value to the atoms in {\em all}
sub-formulas of $\vi$. (E.g., \muA{} can satisfy $\vi\defas\vi_1\vee\vi_2$
without assigning values to the atoms in $\vi_2$ if it satisfies $\vi_1$.)
Consider \eqref{eq:rewritingTseitin} s.t.\ $\vi_i$ is some sub-formula of \vi{}
whose atoms are not assigned by \muA{}. Although \muA{} satisfies $\vi$, \muA{}
does not satisfy $\exists B_i.( \vi[\vi_i|B_i]\wedge (B_i\iff\vi_i))$. In fact,
to satisfy the second conjunct it is necessary to assign some truth value not
only to $B_i$ but also to some of the unassigned atoms in $\vi_i$, so that to
make $\vi_i$ evaluate to the same truth value assigned to $B_i$.

%   \noindent
% %    In general, \cref{fact:tseitin} can be explained as follows.
%     In fact, a
%     {\em partial} assignment \muA{} may suffice to satisfy a non-CNF formula 
%     $\vi(\allA)$  because it does not need to assign a truth value to the
%     atoms in \emph{all}
%     subformulas of $\vi$. (E.g., \muA{} can satisfy
%     $\vi\defas\vi_1\vee\vi_2$ without assigning values to the atoms in
%     $\vi_2$ if it satisfies $\vi_1$.)
%     %
%     Consider some subformula  $\vi_i$ of \vi{} whose atoms are not assigned by \muA{}.
%     %
%     We recall that \TseitinCNF{} works by applying recursively the rewriting step (\sref{sec:bg:cnf}):
%     %\GMSIDENOTE{In \sref{sec:bg:cnf} abbiamo usato \DeMorganCNF($B_i\iff\vi_i$)}
%     \begin{eqnarray}
%     \label{eq:rewritingTseitin}
%     \vi\Longrightarrow\vi[\vi_i|B_i] \wedge (B_i\iff\vi_i)
%     \end{eqnarray}
%     %\vi\Longrightarrow\vi[\vi_i|B_i] \wedge \DeMorganCNF(B_i\iff\vi_i)$$
%     %
%     %$$\TseitinCNF{(\vi)}=\TseitinCNF{(\vi[\vi|B_i])}\wedge\DeMorganCNF{(B_i\iff\vi_i)}$$.
%     \noindent and then by recursively CNF-izing the two conjuncts.
%     Unfortunately, although \muA{} suffices to satisfy $\vi$,
%     \muA{} does not
%     suffice to satisfy $\exists B_i.( \vi[\vi_i|B_i]\wedge
%     (B_i\iff\vi_i))$, because to satisfy the second conjunct it
%     is necessary to assign some truth value not only to $B_i$ but also
%     to some of the atoms in $\vi_i$, so
%     that to make $\vi_i$ evaluate to the same truth value given to $B_i$.

As a consequence of \cref{fact:tseitin}, given $\muA$ satisfying \vi{}, in
order to produce an assignment \muAprime{} satisfying $\exists \allB.\vicnfts$
the enumerator is most often forced to assign other atoms in \allA{}, so that
$\muAprime\supset\muA$. %(unnecessarily)
Given the fact that the amount of total assignments covered by a partial
assignment decreases exponentially with its length (see
\sref{sec:background:propositional-logic}), the above weakness causes a blow-up
in the number of partial assignments
%which we need generating
needed to cover all models. This may drastically affect the effectiveness and
efficiency of the enumeration.

This is illustrated in the following example, where instead of one single short
partial assignment the enumerator is forced to enumerate 9 longer ones.

%We first illustrate this issue with an example.
%\newpage
\begin{example}%
    \label{ex1}
    Consider the propositional formula  over %the set of atoms
    $\allA\defas\set{A_1, A_2, A_3, A_4, A_5, A_6, A_7}$:
    \begin{equation}
        \label{eq:ex1:vi}
        % \vi \defas 
        %     \overbrace{(\underbrace{(A_1 \wedge A_2)}_{B_1} \vee  A_3)}^{B_2} \iff
        %     \overbrace{ (\underbrace{(A_4 \wedge A_5 )}_{B_3} \vee  A_6) }^{B_4}
        % \overbrace{(A_1 \wedge A_2)}^{B_1} \vee  
        % \overbrace{(A_3 \iff \underbrace{((A_4\vee A_5) \wedge 
        % \overbrace{(A_6 \vee A_7)}^{B_2})}_{B_3})}^{B_4}
        \vi \defas
        \overbrace{(A_1\wedge A_2)}^{B_1}\vee
        \overbrace{(
            \overbrace{(
                \overbrace{(A_3\vee A_4)}^{B_2}\wedge
                \overbrace{(A_5\vee A_6)}^{B_3}
                )}^{B_4}\iff
            A_7
            )}^{B_5}.
    \end{equation}
    \noindent
    $\vi$ is not in CNF, and thus it must be CNF-ized before starting the enumeration process.
    If \TseitinCNF{} is used, then the following CNF formula is obtained:

    \begin{subequations}%
        \label{eq:ex1:vicnf}
        \begin{alignat}{2}
            % \vicnf \defas
            % &(\neg B_1\vee\pos A_1)\wedge(\neg B_1\vee\pos A_2)\wedge(\pos B_1\vee\neg A_1\vee\neg A_2)&\wedge\label{eq:ex1:vicnf:line1}\\
            % &(\pos B_2\vee\neg B_1)\wedge(\pos B_2\vee\neg A_3)\wedge(\neg B_2\vee\pos B_1\vee\pos A_3)&\wedge\label{eq:ex1:vicnf:line2}\\
            % &(\neg B_3\vee\pos A_4)\wedge(\neg B_3\vee\pos A_5)\wedge(\pos B_3\vee\neg A_4\vee\neg A_5)&\wedge\label{eq:ex1:vicnf:line3}\\
            % &(\pos B_2\vee\neg B_3)\wedge(\pos B_2\vee\neg A_6)\wedge(\neg B_2\vee\pos B_3\vee\pos A_6)&\wedge\label{eq:ex1:vicnf:line4}\\
            % &(\neg B_4 \vee\pos B_2)\wedge(\pos B_4\vee\neg B_2)&\label{eq:ex1:vicnf:line5}
            %-------------------------------
            % \psi\defas
            % &(\neg B_1\vee\pos A_1)\wedge(\neg B_1\vee\pos A_2)\wedge(\pos B_1\vee\neg A_1\vee\neg A_2)&\wedge\label{eq:ex1:vicnf:line1}\\
            % &(\pos B_2\vee\neg A_5)\wedge(\pos B_2\vee\neg A_6)\wedge(\neg B_2\vee\pos A_5\vee\pos A_6)&\wedge\label{eq:ex1:vicnf:line2}\\
            % &(\neg B_3\vee\pos A_4)\wedge(\neg B_3\vee\pos B_2)\wedge(\pos B_3\vee\neg A_4\vee\neg B_2)&\wedge\label{eq:ex1:vicnf:line3}\\
            % &(\neg B_4\vee\neg A_3\vee B_3)\wedge(\neg B_4\vee A_3\vee\neg B_3)\wedge
            % ( B_4\vee A_3\vee B_3)\wedge( B_4\vee\neg A_3\vee\neg B_3)
            % &\wedge\label{eq:ex1:vicnf:line4}\\
            % &(\pos B_1\vee\pos B_4)&
            % --------------------------------------
             & \vicnfts\defas\nonumber                                                                                                         \\
             & \enspace(\neg B_1\vee\pos A_1)\wedge(\neg B_1\vee\pos A_2)\wedge(\pos B_1\vee\neg A_1\vee\neg A_2) & \wedge
             & \quad\eqcomment{(B_1\iff (A_1\wedge A_2))}\label{eq:ex1:vicnf:line1}                                                            \\
             & \enspace(\pos B_2\vee\neg A_3)\wedge(\pos B_2\vee\neg A_4)\wedge(\neg B_2\vee\pos A_3\vee\pos A_4) & \wedge
             & \quad\eqcomment{(B_2\iff (A_3\vee A_4))}\label{eq:ex1:vicnf:line2}                                                              \\
             & \enspace(\pos B_3\vee\neg A_5)\wedge(\pos B_3\vee\neg A_6)\wedge(\neg B_3\vee\pos A_5\vee\pos A_6) & \wedge
             & \quad\eqcomment{(B_3\iff (A_5\vee A_6))}\label{eq:ex1:vicnf:line3}                                                              \\
             & \enspace(\neg B_4\vee\pos B_2)\wedge(\neg B_4\vee\pos B_3)\wedge(\pos B_4\vee\neg B_2\vee\neg B_3) & \wedge
             & \quad\eqcomment{(B_4\iff (B_2\wedge B_3))}\label{eq:ex1:vicnf:line4}                                                            \\
             & \enspace(\neg B_5\vee\pos B_4\vee\neg A_7)\wedge(\neg B_5\vee\neg B_4\vee\pos A_7)                 & \wedge
             & \quad\eqcomment{(B_5\iff (B_4\iff A_7))}\label{eq:ex1:vicnf:line5}                                                              \\
             & \enspace(\pos B_5\vee\pos B_4\vee\pos A_7)\wedge(\pos B_5\vee\neg B_4\vee\neg A_7)                 & \wedge\nonumber            %\tag{...}
            \\
             & \enspace(\pos B_1\vee\pos B_5)                                                                     & \label{eq:ex1:vicnf:line6}
        \end{alignat}
    \end{subequations}
    where the fresh atoms $\allB\defas\set{B_1, B_2, B_3, B_4, B_5}$ label
    sub-formulas as in~\eqref{eq:ex1:vi}.

    Consider the (minimal) partial truth assignment:
    \begin{equation}
        \label{eq:ex1:muA}
        \muA\defas\set{\neg A_3,\neg A_4,\neg A_7}.
    \end{equation}
    \muA{} satisfies $\vi$ \eqref{eq:ex1:vi}, {even though it does not assign a truth value to the sub-formulas $(A_1\wedge A_2)$ and $(A_5\vee A_6)$.} % since the atoms $A_1, A_2, A_5, A_6$ are not assigned.
        %
        %Unfortunately, 
        {Yet, }$\muA$ does not
    satisfy $\exists\allB.\vicnfts$.
    In fact,  {there is no total truth
            assignment $\etaB$ on \allB{} such that
            $\muA\cup\etaB\pmodels\vicnfts$},
    because~\eqref{eq:ex1:vicnf:line1} and~\eqref{eq:ex1:vicnf:line3}
    cannot be satisfied by assigning only variables in $\allB$;
    rather, it is necessary to further assign  at least one atom in
    $\set{A_1,A_2}$ to satisfy~\eqref{eq:ex1:vicnf:line1} and at least one
    in $\set{A_5,A_6}$ to satisfy~\eqref{eq:ex1:vicnf:line3}.

    Suppose an enumerator assigns first the literals in \muA~\eqref{eq:ex1:muA},
    which force to assign also $\muB\defas\set{\neg B_2, \neg B_4, B_5}$ due
    to~\eqref{eq:ex1:vicnf:line2},~\eqref{eq:ex1:vicnf:line4},
    and~\eqref{eq:ex1:vicnf:line5} respectively. Since $\muA\cup\muB$ satisfies all
    clauses except those in~\eqref{eq:ex1:vicnf:line1}
    and~\eqref{eq:ex1:vicnf:line3}, then the enumerator needs extending
    $\muA\cup\muB$ by enumerating the partial assignments on the unassigned atoms
    \set{A_1,A_2,A_5,A_6,B_1,B_3} which satisfy~\eqref{eq:ex1:vicnf:line1}
    and~\eqref{eq:ex1:vicnf:line3}. Regardless of the search strategy adopted, this
    requires generating no less than $9$ satisfying partial assignments on
    $\set{A_1,A_2,A_5,A_6}$.~%
    %$\set{A_1,A_2,A_5,A_6,B_1,B_3}$.~%
    \footnote{
        The set of models for \eqref{eq:ex1:vicnf:line1}
        is
        \set{
            \set{\pos B_1,\pos A_1,\pos A_2},
            \set{\neg B_1,\pos A_1,\neg A_2},
            \set{\neg B_1,\neg A_1,\pos A_2},
            \set{\neg B_1,\neg A_1,\neg A_2}
        }, which can be covered only either by
        \set{
            \set{\pos B_1,\pos A_1,\pos A_2},
            \set{\neg B_1,\neg A_1},
            \set{\neg B_1,\pos A_1,\neg A_2}
        } or by
        \set{
            \set{\pos B_1,\pos A_1,\pos A_2},
            \set{\neg B_1,\neg A_2},
            \set{\neg B_1,\neg A_1,\pos A_2}
        }
        in the case of disjoint enumeration,
        and by
        \set{
            \set{\pos B_1,\pos A_1,\pos A_2},
            \set{\neg B_1,\neg A_1},
            \set{\neg B_1,\neg A_2}
        } in the case of non-disjoint enumeration.
        In all cases,   we need no less than 3 distinct partial assignments on
        \set{A_1,A_2}. Similar considerations hold for
        \eqref{eq:ex1:vicnf:line3}. Thus, we need no
        less than
        $3\times3=9$ partial assignments on $\set{A_1,A_2}\cup\set{A_5,A_6}$.
        %$3\times3=9$ partial assignments on $\set{B_1,A_1,A_2}\cup\set{B_3,A_5,A_6}$.
        %representing the $4\times4=16$ total ones.
        % \GMNOTE{
        % Bisogna coprire tutti i modelli totali projected su \set{A_1, A_2}
        % }
    }
    For instance, in the case of disjoint AllSAT, instead of the single partial
    assignment $\muA$~\eqref{eq:ex1:muA}, the solver may return the following list
    of 9 partial assignments satisfying $\exists\allB.\vicnfts$ which extend
    \muA{}:%\ignoreinshort{%}
    \begin{equation}%
        \label{eq:ex1:muA:all}
        \begin{array}{llllllll}
            \multicolumn{2}{c}{\overbrace{\rule{1.5cm}{0pt}}^{B_1}} &           &           & \multicolumn{2}{c}{\overbrace{\rule{1.5cm}{0pt}}^{B_3}} &           &                        \\
            \{\neg A_1,                                             &           & \neg A_3, & \neg A_4,                                               & \neg A_5, & \neg A_6, & \neg A_7\}
            \quad\eqcomment{\set{\neg B_1,\neg B_2,\neg B_3,\neg B_4,\pos B_5}}                                                                                                            \\
            \{\neg A_1,                                             &           & \neg A_3, & \neg A_4,                                               & \pos A_5, &           & \neg A_7\}
            \quad\eqcomment{\set{\neg B_1,\neg B_2,\pos B_3,\neg B_4,\pos B_5}}                                                                                                            \\
            \{\neg A_1,                                             &           & \neg A_3, & \neg A_4,                                               & \neg A_5, & \pos A_6, & \neg A_7\}
            \quad\eqcomment{\set{\neg B_1,\neg B_2,\pos B_3,\neg B_4,\pos B_5}}                                                                                                            \\
            \{\pos A_1,                                             & \neg A_2, & \neg A_3, & \neg A_4,                                               & \neg A_5, & \neg A_6, & \neg A_7\}
            \quad\eqcomment{\set{\neg B_1,\neg B_2,\neg B_3,\neg B_4,\pos B_5}}                                                                                                            \\
            \{\pos A_1,                                             & \neg A_2, & \neg A_3, & \neg A_4,                                               & \pos A_5, &           & \neg A_7\}
            \quad\eqcomment{\set{\neg B_1,\neg B_2,\pos B_3,\neg B_4,\pos B_5}}                                                                                                            \\
            \{\pos A_1,                                             & \neg A_2, & \neg A_3, & \neg A_4,                                               & \neg A_5, & \pos A_6, & \neg A_7\}
            \quad\eqcomment{\set{\neg B_1,\neg B_2,\pos B_3,\neg B_4,\pos B_5}}                                                                                                            \\
            \{\pos A_1,                                             & \pos A_2, & \neg A_3, & \neg A_4,                                               & \neg A_5, & \neg A_6, & \neg A_7\}
            \quad\eqcomment{\set{\pos B_1,\neg B_2,\neg B_3,\neg B_4,\pos B_5}}                                                                                                            \\
            \{\pos A_1,                                             & \pos A_2, & \neg A_3, & \neg A_4,                                               & \pos A_5, &           & \neg A_7\}
            \quad\eqcomment{\set{\pos B_1,\neg B_2,\pos B_3,\neg B_4,\pos B_5}}                                                                                                            \\
            \{\pos A_1,                                             & \pos A_2, & \neg A_3, & \neg A_4,                                               & \neg A_5, & \pos A_6, & \neg A_7\}
            \quad\eqcomment{\set{\pos B_1,\neg B_2,\pos B_3,\neg B_4,\pos B_5}}                                                                                                            \\
        \end{array}
    \end{equation}
    In the case of non-disjoint AllSAT, the solver may enumerate a similar
    set of partial assignments, with \set{\neg A_2} instead of \set{A_1,\neg A_2} and \set{A_6}
    instead of  \set{\neg A_5,A_6}.\exdone{}
\end{example}

\subsection{The impact of Plaisted and Greenbaum CNF transformation}%
\label{sec:problem:polarity}

We point out that also the \PlaistedCNF{}
transformation~\cite{plaistedStructurepreservingClauseForm1986} suffers for the
same weakness as \TseitinCNF{} ---that is, \cref{fact:tseitin} holds also for
\PlaistedCNF{}--- although its effects are mitigated.

%   \noindent
In fact, we recall that \PlaistedCNF{} works by applying recursively the
rewriting step (\sref{sec:bg:cnf}):
\begin{eqnarray}
    \label{eq:rewritingPlaisted}
    \vi\Longrightarrow\vi[\vi_i|B_i] \wedge \left \{
    \begin{array}{lll}
        (B_i\imp\vi_i)  & \mbox{if $\vi_i$ occurs only positively in $\vi$}                \\
        (B_i\limp\vi_i) & \mbox{if $\vi_i$ occurs only negatively in $\vi$}                \\
        (B_i\iff\vi_i)  & \mbox{if $\vi_i$ occurs both positively and negatively in $\vi$} \\
    \end{array}\right \},
\end{eqnarray}
\noindent and then by recursively CNF-izing the two conjuncts.
As with \TseitinCNF{}, consider~\eqref{eq:rewritingPlaisted} s.t.\ $\vi_i$ is
some sub-formula of \vi{} whose atoms are not assigned by \muA{}.

If $\vi_i$ occurs both positively and negatively in $\vi$,
then~\eqref{eq:rewritingPlaisted} reduces to~\eqref{eq:rewritingTseitin} and
\PlaistedCNF{} behaves like \TseitinCNF{}, so that \muA{} does not satisfy
$\exists B_i.( \vi[\vi_i|B_i]\wedge (B_i\iff\vi_i))$.

If instead $\vi_i$ occurs only positively \resp{negatively} in $\vi$, then it
is possible to extend \muA{} by assigning $B_i$ to $\bot$ \resp{$\top$} to
satisfy $(B_i\imp\vi_i)$ \resp{$(B_i\limp\vi_i)$} without assigning any atom in
$\vi_i$. Thus $\muA$ satisfies $\exists B_i.( \vi[\vi_i|B_i]\wedge
    (B_i\imp\vi_i))$ \resp{$\exists B_i.( \vi[\vi_i|B_i]\wedge (B_i\limp\vi_i))$}.

As with \TseitinCNF{}, given $\muA$ satisfying \vi{}, in order to produce an
assignment \muAprime{} satisfying $\exists \allB.\vicnfpg$ the enumerator is
most often forced to assign other atoms in \allA{}, so that
$\muAprime\supset\muA$. We notice, however, that the effect of this problem is
mitigated by the presence of single-polarity sub-formulas among those left
unassigned by \muA{}. As an extreme case, if all sub-formulas in $\vi$ occur
with single polarity, then no further assignment to atoms in \allA{} is needed. %as a consequence of \cref{fact:tseitin},%(unnecessarily)

\begin{example}%
    \label{ex2}
    Consider the formula $\vi$~\eqref{eq:ex1:vi} as in
    \cref{ex1}. Suppose that $\vi$ is converted into CNF using
    \PlaistedCNF{}. Then, we have:
    \begin{subequations}%
        \label{eq:ex2:vicnf}
        \begin{alignat}{2}
             & \vicnfpg\defas\nonumber                                                                                                         \\
             & \enspace(\neg B_1\vee\pos A_1)\wedge(\neg B_1\vee\pos A_2)                                         & \wedge
             & \quad\eqcomment{(B_1\imp (A_1\wedge A_2))}\label{eq:ex2:vicnf:line1}                                                            \\
             & \enspace(\pos B_2\vee\neg A_3)\wedge(\pos B_2\vee\neg A_4)\wedge(\neg B_2\vee\pos A_3\vee\pos A_4) & \wedge
             & \quad\eqcomment{(B_2\iff (A_3\vee A_4))}\label{eq:ex2:vicnf:line2}                                                              \\
             & \enspace(\pos B_3\vee\neg A_5)\wedge(\pos B_3\vee\neg A_6)\wedge(\neg B_3\vee\pos A_5\vee\pos A_6) & \wedge
             & \quad\eqcomment{(B_3\iff (A_5\vee A_6))}\label{eq:ex2:vicnf:line3}                                                              \\
             & \enspace(\neg B_4\vee\pos B_2)\wedge(\neg B_4\vee\pos B_3)\wedge(\pos B_4\vee\neg B_2\vee\neg B_3) & \wedge
             & \quad\eqcomment{(B_4\iff (B_2\wedge B_3))}\label{eq:ex2:vicnf:line4}                                                            \\
             & \enspace(\neg B_5\vee\pos B_4\vee\neg A_7)\wedge(\neg B_5\vee\neg B_4\vee\pos A_7)                 & \wedge
             & \quad\eqcomment{(B_5\imp (B_4\iff A_7))}\label{eq:ex2:vicnf:line5}                                                              \\
             & \enspace(\pos B_1\vee\pos B_5).                                                                    & \label{eq:ex2:vicnf:line6}
        \end{alignat}
    \end{subequations}
    We remark that~\eqref{eq:ex2:vicnf:line1}
    and~\eqref{eq:ex2:vicnf:line5} are shorter
    than~\eqref{eq:ex1:vicnf:line1} and~\eqref{eq:ex1:vicnf:line5}
    respectively, since the corresponding sub-formulas $(A_1\wedge
        A_2)$ and $((\dots)\iff A_7)$ occur only with
    positive polarity in \vi{}~\eqref{eq:ex1:vi}, so that only the one-way implication
    $(B_i\imp\vi_i)$ is needed.

    As in \cref{ex1}, consider the partial assignment $\muA\defas\set{\neg A_3,\neg
            A_4,\neg A_7}$~\eqref{eq:ex1:muA} which satisfies $\vi{}$~\eqref{eq:ex1:vi}. As
    before, $\muA$ does not satisfy $\exists\allB.\vicnfpg$. In fact, {there is no
            total truth assignment $\etaB$ on \allB{} such that
            $\muA\cup\etaB\pmodels\vicnfpg$}, because~\eqref{eq:ex2:vicnf:line3} cannot be
    satisfied by assigning only variables in $\allB$; rather, it is necessary to
    further assign at least one atom in
    %    $\set{A_1,A_2}$ to satisfy ~
    %    \eqref{eq:ex2:vicnf:line1}.
    %    and at least one in
    $\set{A_5,A_6}$ to satisfy~\eqref{eq:ex2:vicnf:line3}.
    Notice that, unlike with \cref{ex1}, in order to satisfy~\eqref{eq:ex2:vicnf:line1} %and  \eqref{eq:ex2:vicnf:line5}
    it is sufficient to set $B_1=\bot$ with no need to assign any atom in $\set{A_1,A_2}$.

    Suppose an enumerator assigns first the literals in \muA, which force it to
    assign also $\muB\defas\set{\neg B_2, \neg B_4}$ due
    to~\eqref{eq:ex2:vicnf:line2} and~\eqref{eq:ex2:vicnf:line4} respectively.
    Since $\muA\cup\muB$ satisfies all clauses except those in
    \eqref{eq:ex2:vicnf:line1}, \eqref{eq:ex2:vicnf:line3},
    %\eqref{eq:ex2:vicnf:line5},
    and \eqref{eq:ex2:vicnf:line6}, the enumerator needs extending $\muA\cup\muB$
    by enumerating partial assignments on the unassigned atoms
    \set{A_1,A_2,A_5,A_6,B_1,B_3,B_5} which satisfy them.
    Regardless of the search strategy adopted, {to
            satisfy~\eqref{eq:ex2:vicnf:line3} it is necessary to generate no less than $3$
            partial assignments} on $\set{A_5,A_6}$ (see \cref{ex1}); to satisfy
    \eqref{eq:ex2:vicnf:line1},
    %\eqref{eq:ex2:vicnf:line5},
    and \eqref{eq:ex2:vicnf:line6}, instead, the enumerator needs only assigning
    $B_1=\bot$, which forces it to assign $B_5=\top$ due to
    \eqref{eq:ex2:vicnf:line6}.

    % , and
    % \eqref{eq:ex2:vicnf:line5} respectively. 
    For instance, in the case of disjoint AllSAT, instead of the single partial
    assignment $\muA$, the solver may return the following list of 3 partial
    assignments satisfying $\exists\allB.\vicnfpg$ which extend \muA{}:%\ignoreinshort{%}
    \begin{equation}%
        \label{eq:ex2:muA:all}
        \begin{array}{llllllll}
            \multicolumn{2}{c}{} &  &           & \multicolumn{2}{c}{\overbrace{\rule{1.5cm}{0pt}}^{B_3}} &           &                        \\
            \{                   &  & \neg A_3, & \neg A_4,                                               & \neg A_5, & \neg A_6, & \neg A_7\}
            \quad\eqcomment{\set{\neg B_1,\neg B_2,\neg B_3,\neg B_4,\pos B_5}}                                                                \\
            \{                   &  & \neg A_3, & \neg A_4,                                               & \pos A_5, &           & \neg A_7\}
            \quad\eqcomment{\set{\neg B_1,\neg B_2,\pos B_3,\neg B_4,\pos B_5}}                                                                \\
            \{                   &  & \neg A_3, & \neg A_4,                                               & \neg A_5, & \pos A_6, & \neg A_7\}
            \quad\eqcomment{\set{\neg B_1,\neg B_2,\pos B_3,\neg
            B_4,\pos B_5}}                                                                                                                     \\
        \end{array}
    \end{equation}
    In the case of non-disjoint AllSAT, the solver may enumerate a similar
    set of partial assignments, with
    \set{A_6}
    instead of  \set{\neg A_5,A_6} in the third assignment.\exdone{}

    % Indeed, sub-formulas occurring with double polarity are labeled using double implications as for \TseitinCNF{}, raising the same problems as the latter. For instance, the sub-formula $(A_5\vee A_6)$ occurs with double polarity, since it is under the scope of an ``$\iff$''. Hence, the clauses in~\eqref{eq:ex2:vicnf:line3} must be satisfied by assigning a truth value also to $A_5$ or $A_6$, and so the partial truth assignment $\muA$ in~\eqref{eq:ex1:muA} does not suffice to satisfy $\exists\allB.\vicnfpg$. %because $\residual{\vicnfpg}{\muA\cup\etaB}=(\neg A_5)\wedge(\neg A_6)$.
    % \exdone{}
\end{example}

Notice that, to maximize the benefits of \PlaistedCNF{}, the sub-formulas
occurring with positive \resp{negative} polarity only must have their label
assigned to false \resp{true}. {In practice, this can be achieved in part by
instructing the solver to split on negative values in decision
branches.~\footnote{To exploit this heuristic also for sub-formulas occurring
    only negatively, the latter can be labeled with a negative label $\neg B_i$ as
    $(\neg B_i\limp\vi_i)$.} Even though the solver is not guaranteed to always
assign to false these labels, we empirically verified that this heuristic
provides a good approximation of this behavior.}

% \end{rschange}

\section{Enhancing enumeration via NNF preprocessing}%
\label{sec:solution}

In this section, we propose a solution to address the shortcomings of
\TseitinCNF{} and \PlaistedCNF{} %transformations in
%\ignoreinlong{model}\ignoreinshort{
for SAT and SMT enumeration described in~\sref{sec:problem}. %} 
% We show that a
% simple preprocessing can avoid this situation.
The idea is simple: \emph{we transform first the input formula \vi{} into NNF,
    {and then we apply \PlaistedCNF{} to \NNF{\vi}.}} In fact, NNF guarantees that
every non-atomic sub-formula of \NNF{\vi} occurs only positively, because every
sub-formula $\vi_i$ occurring with double polarity in \vi{} is converted into
two syntactically-different sub-formulas $\poslab{\vi_i}\defas\NNF{\vi_i}$ and
$\neglab{\vi_i}\defas\NNF{\neg\vi_i}$, each occurring only positively. Thus,
when computing \vinnfcnfpg{}, $\poslab{\vi_i}$ and $\neglab{\vi_i}$ are
labelled with two distinct atoms $\poslab{B_i}$ and $\neglab{B_i}$
respectively, adding the one-way implications
$(\poslab{B_i}\imp\poslab{\vi_i})$ and $(\neglab{B_i}\imp\neglab{\vi_i})$. (If
$\vi_i$ occurs only positively \resp{negatively} in \vi, then only
\poslab{\vi_i} \resp{\neglab{\vi_i}} occurs in \NNF{\vi}, so that only
\poslab{B_i} \resp{\neglab{B_i}} is introduced and only
$(\poslab{B_i}\imp\poslab{\vi_i})$ \resp{$(\neglab{B_i}\imp\neglab{\vi_i})$} is
added.)% an NNF DAG.\@

We remark that with this preprocessing we maintain the correctness and
completeness of the enumeration process, because $\vi$ is equivalent to
$\NNF{\vi}$, which is equivalent to $\exists\allB.\vinnfcnfpg$. We remark also
that we produce a linear-size CNF encoding, since {$\NNF{\vi}$} has linear size
w.r.t.\ $\vi$ and \vinnfcnfpg{} has linear size w.r.t.\ \NNF{\vi}
(see~\sref{sec:background:propositional-logic}). %the $\NNF{\vi}$ DAG 

We prove the following fact: {\em if a partial truth assignment $\muA$
satisfies $\vi$, then it satisfies also $\exists\allB.\vinnfcnfpg$}. %\ignoreinlong{every partial model $\mualpha$ for $\vi$ is also a model for}\ignoreinshort{
% , that is, if $\mualpha\pmodels\vi$, then there exists $\etaB$ s.t.\
% $\mualpha\cup\etaB\pmodels\vinnfcnfpg$.
(The vice versa holds trivially.)\@
{We remark that, as shown in~\sref{sec:problem}, this fact does not hold for
    \TseitinCNF{} and \PlaistedCNF{}.}\@
% \ignoreinlong{A complete formal proof of this fact is presented in
%     an extended version of this paper \cite{masina_cnf_2023}.
%     Intuitively, it is easy to see that the suitable $\etaB$ is defined  so that,
%     for each sub-formula $\vi_i$ of $\vi$,
%     if $\vi_i$  is made true, false or
%     is unassigned by  $\mualpha$,
%     then   \tuple{\etaB(B_i^+),\etaB(B_i^-)} is
%     %defined as
%     \tuple{\top,\bot}, \tuple{\bot,\top}, or \tuple{\bot,\bot}
%     respectively. }
%\end{rschange}
%
% \ignoreinshort{
%This is proved in the following theorem. %, which is proved in~\cref{sec:proofexistsetaB}.
%
\begin{theorem}%
    \label{th:existsetaB}
    Consider a formula $\vi$, and a \emph{partial} truth assignment $\muA$ such that $\muA\pmodels\vi$.
    Then $\muA\pmodels\exists\allB.\vinnfcnfpg$.
    % Indeed, there exists a total truth assignment $\etaB$ such that $\mualpha\cup\etaB$ satisfies $\vinnfcnfpg$.
    % For every partial truth assignment $\mualpha$ s.t.\ $\mualpha\pmodels\vi(\allalpha)$, there exists a total truth assignment \etaB s.t.\ $\mualpha\cup\etaB\pmodels\vinnfcnfpg$.
\end{theorem}
% \TODO{Riscrivere come $\mualpha\pmodels\vi(\allalpha)\implies\mualpha\pmodels\vinnfcnfpg(\allalpha)$}

\begin{proof}%
    We show that, for every $\muA$ such that $\muA\pmodels\vi$,
    there exists a total truth assignment $\etaB$ such that $\muA\cup\etaB\pmodels\vinnfcnfpg$.
    We first show how such a $\etaB$ can be built, then we prove that it satisfies \vinnfcnfpg{} if conjoined with $\muA$.
    In the following, the symbol ``\any'' denotes any formula that is not in
    $\set{\top,\bot}$, so that ``$\residual{\vi_i}{\muA}=\any$'' means that the
    residual of $\vi_i$ under $\muA$ is not a constant. %is not assigned a truth value by $\muA$.

    %In order to simplify the proof, we assume that all sub-formulas occur with double polarity in $\vi$.

    For each sub-formula $\vi_i$ of $\vi$, whose positive and negative occurrences
    in $\vinnfcnfpg$ are associated with the variables $\poslab{B_i}$ and
    $\neglab{B_i}$ respectively, do:
    \begin{enumerate}[(a)]
        \item\label{item:existsetaB:true} if $\residual{\vi_i}{\muA} = \top$, and hence $\residual{\NNF{\vi_i}}{\muA}=\top$ and $\residual{\NNF{\neg\vi_i}}{\muA}=\bot$ by \cref{th:munnf}, then
              %set $\etaB(\poslab{B_i})=\top$ and $\etaB(\neglab{B_i})=\bot$;
              add \set{\poslab{B_i},\neg\neglab{B_i}} to \etaB;
        \item\label{item:existsetaB:false} if $\residual{\vi_i}{\muA} = \bot$, and
              hence $\residual{\NNF{\vi_i}}{\muA}=\bot$ and
              $\residual{\NNF{\neg\vi_i}}{\muA}=\top$ by \cref{th:munnf}, then add
              \set{\neg\poslab{B_i},\neglab{B_i}} to \etaB; \item\label{item:existsetaB:star}
              otherwise if $\residual{\vi_i}{\muA}=\any$, then add
              \set{\neg\poslab{B_i},\neg\neglab{B_i}} to \etaB;%set $\etaB(\poslab{B_i})=\bot$ and $\etaB(\neglab{B_i})=\top$;%set $\etaB(\poslab{B_i})=\etaB(\neglab{B_i})=\bot$.
    \end{enumerate}
    (If $\vi_i$ occurs only positively or negatively, then only assign $\poslab{B_i}$ or $\neglab{B_i}$ respectively.)

    We prove that $\muA\cup\etaB\pmodels\vinnfcnfpg$ by induction on the structure
    of the formula $\vinnfcnfpg$. Consider a sub-formula $\vi_i$ of $\vi$:
    \begin{enumerate}[(i)]
        \item if $\vi_i$ occurs positively in $\vi$, then
              %\GMSIDENOTE{Scrivere con $\imp$ invece di $\vee$ come in \eqref{eq:rewritingTseitin} e \eqref{eq:rewritingPlaisted}?}
              \begin{equation}
                  \vinnfcnfpg\defas\PlaistedCNF(
                  \overbrace{\NNF{\vi}[\NNF{\vi_i}|\poslab{B_i}]}^{\subtermeqlabel{eq:proofmuB:pos:vi}} \wedge
                  \overbrace{\left(\poslab{B_i}\imp\NNF{\vi_i}\right)}^{\subtermeqlabel{eq:proofmuB:pos:lab}})
              \end{equation}
        \item if $\vi_i$ occurs negatively in $\vi$, then
              \begin{equation}
                  \vinnfcnfpg\defas\PlaistedCNF(
                  \overbrace{\NNF{\vi}[\NNF{\neg\vi_i}|\neglab{B_i}]}^{\subtermeqlabel{eq:proofmuB:neg:vi}} \wedge
                  \overbrace{\left(\neglab{B_i}\imp\NNF{\neg\vi_i}\right)}^{\subtermeqlabel{eq:proofmuB:neg:lab}})
              \end{equation}
    \end{enumerate}
    % \begin{subequations}
    % \begin{align}
    %     \vinnfcnfpg\defas\PlaistedCNF\Big(
    %         \label{eq:proof:a}&\NNF{\vi}[\NNF{\vi_i},\NNF{\neg\vi_i}/\poslab{B_i},\neglab{B_i}]
    %      \wedge\\
    %      \label{eq:proof:b}& \left(\neg \poslab{B_i}\vee\NNF{\vi_i}\right) \wedge \\
    %      \label{eq:proof:c}& \left(\neg \neglab{B_i}\vee\NNF{\neg\vi_i}\right)\Big).
    % \end{align}
    % \end{subequations}
    For each pair of cases we have:
    \begin{enumerate}[(a)]
        \item\begin{enumerate}[(i)]
            \item $\muA\cup\etaB\pmodels\eqref{eq:proofmuB:pos:vi}$ since we substitute $\NNF{\vi_i}$ with $\poslab{B_i}$ and $\residual{\NNF{\vi_i}}{\muA}=\residual{\poslab{B_i}}{\etaB}=\top$;\\
                  %we substitute $\top$ for $\top$; 
                  $\muA\cup\etaB\pmodels\eqref{eq:proofmuB:pos:lab}$ since $\muA\pmodels\NNF{\vi_i}$;
            \item $\muA\cup\etaB\pmodels\eqref{eq:proofmuB:neg:vi}$ since we substitute $\NNF{\neg\vi_i}$ with $\neglab{B_i}$ and $\residual{\NNF{\neg\vi_i}}{\muA}=\residual{\neglab{B_i}}{\etaB}=\bot$; \\
                  $\muA\cup\etaB\pmodels\eqref{eq:proofmuB:neg:lab}$ since $\etaB\pmodels\neg\neglab{B_i}$;
        \end{enumerate}
        \item \begin{enumerate}[(i)]
                  \item $\muA\cup\etaB\pmodels\eqref{eq:proofmuB:pos:vi}$  since we substitute $\NNF{\vi_i}$ with $\poslab{B_i}$ and $\residual{\NNF{\vi_i}}{\muA}=\residual{\poslab{B_i}}{\etaB}=\bot$;\\ $\muA\cup\etaB\pmodels\eqref{eq:proofmuB:pos:lab}$ since $\etaB\pmodels\neg\poslab{B_i}$;
                  \item $\muA\cup\etaB\pmodels\eqref{eq:proofmuB:neg:vi}$ since we substitute $\NNF{\neg\vi_i}$ with $\neglab{B_i}$ and $\residual{\NNF{\neg\vi_i}}{\muA}=\residual{\neglab{B_i}}{\etaB}=\topt$; \\ $\muA\cup\etaB\pmodels\eqref{eq:proofmuB:neg:lab}$ since $\muA\pmodels\NNF{\neg\vi_i}$;
              \end{enumerate}
        \item \begin{enumerate}[(i)]
                  \item $\muA\cup\etaB\pmodels\eqref{eq:proofmuB:pos:vi}$ since we substitute $\NNF{\vi_i}$ with $\poslab{B_i}$, and substituting $\residual{\NNF{\vi_i}}{\muA}=*$ with $\residual{\poslab{B_i}}{\etaB}=\bot$ preserves the satisfiability;
                        $\muA\cup\etaB\pmodels\eqref{eq:proofmuB:pos:lab}$ since $\etaB\pmodels\neg\poslab{B_i}$;
                  \item $\muA\cup\etaB\pmodels\eqref{eq:proofmuB:neg:vi}$ since we substitute $\NNF{\neg\vi_i}$ with $\neglab{B_i}$, and substituting $\residual{\NNF{\neg\vi_i}}{\muA}=*$ with $\residual{\neglab{B_i}}{\etaB}=\bot$ preserves the satisfiability;
                        $\muA\cup\etaB\pmodels\eqref{eq:proofmuB:neg:lab}$ since $\etaB\pmodels\neg\neglab{B_i}$;
              \end{enumerate}
    \end{enumerate}
    Therefore, $\muA\cup\etaB\pmodels\vinnfcnfpg$.
\end{proof}

% }

As a consequence of \cref{th:existsetaB}, given $\muA$ satisfying \vi{}, the
enumerator is no more forced to assign any more atom in \allA to satisfy
$\exists \allB.\vinnfcnfpg$. This prevents the enumerator from producing
multiple assignments extending \muA{}, avoiding thus the blow-up in the number
of assignments for \TseitinCNF{(\vi)} and \PlaistedCNF{(\vi)} described in
\sref{sec:problem}.
We stress the fact that~{\Cref{th:existsetaB}} is agnostic of the AllSAT (or
AllSMT) algorithm adopted, and that it holds for both disjoint and non-disjoint
enumeration. %\ignoreinlong{disjoint-}%\ignoreinshort{\ %} %\ignoreinshort{%}.

\begin{remark}%
    \label{rem:ex3:preconv}
    Unlike with AllSAT
    or AllSMT, the pre-conversion into NNF is typically never used in plain SAT %\ignoreinshort{
    or SMT
    %}
    \emph{solving}, because it causes the unnecessary duplication of labels $\poslab{B_i}$ and $\neglab{B_i}$, with extra overhead and no benefit for the solver.
\end{remark}

We notice that the proof of \cref{th:existsetaB} is {\em constructive}, that
is, it not only says that an assignment $\etaB$ s.t.\
$\muA\cup\etaB\pmodels\vinnfcnfpg$ exists, but also it shows how to construct
it. In particular,
points~\ref{item:existsetaB:true},~\ref{item:existsetaB:false},
and~\ref{item:existsetaB:star} implicitly suggest a strategy for assigning the
right values to the \allB{} atoms given \muA{}, which is based on the iterative
applications of the following steps, interleaved with the assignment of values
which are forced by residual constraints:%, bottom-up:
%  assign all truth values which are forced by constraints;
%  for each unassigned  \poslab{B_i} \resp{\neglab{B_i}}, 
%  if some constraint forces assigning it some truth value, do it;
\begin{enumerate}[(a)]
    \item\label{item:rule:true} if \poslab{B_i} occurs negatively in already-satisfied
          clauses only, then add \set{\poslab{B_i},\neg\neglab{B_i}} to \etaB;
    \item\label{item:rule:false} if \neglab{B_i} occurs negatively in already-satisfied
          clauses only, then  add \set{\neglab{B_i},\neg\poslab{B_i}} to \etaB;
    \item\label{item:rule:star} if \poslab{B_i} \resp{\neglab{B_i}} occurs negatively in a not-yet-satisfied clause whose
          other unassigned literals are all \allA{}-literals, then add \set{\neg\poslab{B_i}} \resp{\set{\neg\neglab{B_i}}} to \etaB.
\end{enumerate}
%   (a) 
%   \\(b) 
%   \\(c) \\%
%   \\
%\GMSIDENOTE{Usare notazione uniforme qui e nella proof: ``set $\etaB(\poslab{B_i})=\top$ and $\etaB(\neglab{B_i})=\bot$'' oppure
%``add \set{\poslab{B_i},\neg\neglab{B_i}} to \etaB'' }
\noindent This strategy
mimics  the application of points~\ref{item:existsetaB:true},~\ref{item:existsetaB:false}, and~\ref{item:existsetaB:star} in the
proof to the sub-formulas $\poslab{\vi_i}$ and $\neglab{\vi_i}$ of \NNF{\vi}
in a bottom-up fashion, from the leaves to the root.~%
\footnote{In fact, we notice that \poslab{B_i} \resp{\neglab{B_i}} occurs
    negatively only in the clauses encoding the
    $(\poslab{B_i}\imp\poslab{\vi_i})$
    \resp{$(\neglab{B_i}\imp\neglab{\vi_i})$ } constraints, because by
    construction it always occurs positively elsewhere, since it
    substitutes some sub-formula \poslab{\vi_i} \resp{\neglab{\vi_i}}
    which occurs only positively in \NNF{\vi}. Also, by construction,
    we can have at most one negative \allB-literal per clause.}

We also notice that, due to the constraints $(\poslab{B_i}\imp\poslab{\vi_i})$
and $(\neglab{B_i}\imp\neglab{\vi_i})$ and to the fact that \poslab{\vi_i} and
\neglab{\vi_i} are mutually inconsistent by construction, no assignment
satisfying $\vinnfcnfpg$ may assign both \poslab{B_i} and \neglab{B_i} to
$\top$. Thus, to further improve the efficiency of the enumeration procedure
without affecting its outcome, we also add to \vinnfcnfpg{} the
mutual-exclusion clauses $(\neg \poslab{B_i}\vee\neg\neglab{B_i})$ when both
$\poslab{B_i}$ and $\neglab{B_i}$ are introduced, which prevent the solver from
assigning both $\poslab{B_i}$ and $\neglab{B_i}$ to $\top$, and thus from
wasting time in exploring inconsistent search sub-trees for residual formulas
like $(\ldots\wedge\poslab{\vi_i}\wedge\neglab{\vi_i}\wedge\ldots)$.

We illustrate the benefit of our proposed technique with the following example.

%\newpage
\begin{example}%
    \label{ex3}
    Consider the formula $\vi$ of~\cref{ex1}. By converting it into NNF, we obtain:
    \begin{align*}
         & \NNF{\vi}\defas                                \\
         & \overbrace{(A_1\wedge A_2)}^{\poslab{B_1}}\vee
        \overbrace{(
        \overbrace{(
        \overbrace{(
        \overbrace{(\neg A_3\wedge\neg A_4)}^{\neglab{B_2}}\vee
        \overbrace{(\neg A_5\wedge\neg A_6)}^{\neglab{B_3}}
        )}^{\neglab{B_4}}\vee A_7
        )}^{\poslab{B_5}} \wedge
        \overbrace{(
        \overbrace{(
        \overbrace{(     A_3\vee       A_4)}^{\poslab{B_2}}\wedge
        \overbrace{(     A_5\vee       A_6)}^{\poslab{B_3}}
        )}^{\poslab{B_4}}\vee\neg A_7
        )}^{\poslab{B_6}}
        )}^{\poslab{B_7}}
    \end{align*}
    % We remark that, by definition, each sub-formula of $\vinnf$ occurs only with positive polarity. For instance, consider the sub-formula $(A_3\vee A_4)$ that occurs in $\vi$ with double polarity. In $\vinnf$, the positive occurrence remains the same, while the negative occurrence is converted into $(\neg A_3\wedge\neg A_4)$. This implies that the two occurrences correspond to two different sub-formula, each occurring only positively, and thus they will be labeled with two different atoms.
    \noindent
    Applying \PlaistedCNF{} and adding the mutual-exclusion clauses we obtain the CNF formula:
    %Suppose, then, that the formula is converted into CNF using \PlaistedCNF{}. Then, the following CNF formula is obtained:
    \begin{subequations}%
        \label{eq:ex3:vicnf}
        \begin{alignat}{2}
             & \vinnfcnfpg\defas                                                                                                                                                                                                                                                                    \\
             & \qquad(\neg \poslab{B_1}\vee\pos A_1)\wedge(\neg \poslab{B_1}\vee\pos A_2)                   & \wedge
             & \quad\eqcomment{(\poslab{B_1}\imp (\pos A_1\wedge\pos A_2))}\label{eq:ex3:vicnf:line1}                                                                                                                                                                                               \\
             & \qquad(\neg \neglab{B_2}\vee\neg A_3)\wedge(\neg \neglab{B_2}\vee\neg A_4)                   & \wedge                                                                    & \quad\eqcomment{(\neglab{B_2}\imp (\neg A_3\wedge\neg A_4))}\label{eq:ex3:vicnf:line2}                    \\
             & \qquad(\neg \neglab{B_3}\vee\neg A_5)\wedge(\neg \neglab{B_3}\vee\neg A_6)                   & \wedge                                                                    & \quad\eqcomment{(\neglab{B_3}\imp (\neg A_5\wedge\neg A_6))}\label{eq:ex3:vicnf:line3}                    \\
             & \qquad(\neg \neglab{B_4}\vee\pos \neglab{B_2}\vee\pos \neglab{B_3})                          & \wedge                                                                    & \quad\eqcomment{(\neglab{B_4}\imp (\pos \neglab{B_2}\vee\pos \neglab{B_3}))}\label{eq:ex3:vicnf:line4}    \\
             & \qquad(\neg \poslab{B_5}\vee\pos \neglab{B_4}\vee\pos A_7)                                   & \wedge                                                                    & \quad\eqcomment{(\poslab{B_5}\imp (\pos \neglab{B_4}\vee\pos A_7))}\label{eq:ex3:vicnf:line5}             \\
             & \qquad(\neg \poslab{B_2}\vee\pos A_3\vee\pos A_4)                                            & \wedge                                                                    & \quad\eqcomment{(\poslab{B_2}\imp (\pos A_3\vee\pos A_4))}\label{eq:ex3:vicnf:line6}                      \\
             & \qquad(\neg \poslab{B_3}\vee\pos A_5\vee\pos A_6)                                            & \wedge                                                                    & \quad\eqcomment{(\poslab{B_3}\imp (\pos A_5\vee\pos A_6))}\label{eq:ex3:vicnf:line7}                      \\
             & \qquad(\neg \poslab{B_4}\vee\pos \poslab{B_2})\wedge(\neg \poslab{B_4}\vee\pos \poslab{B_3}) & \wedge                                                                    & \quad\eqcomment{(\poslab{B_4}\imp (\pos \poslab{B_2}\wedge\pos \poslab{B_3}))}\label{eq:ex3:vicnf:line8}  \\
             & \qquad(\neg \poslab{B_6}\vee\pos \poslab{B_4}\vee\neg A_7)                                   & \wedge                                                                    & \quad\eqcomment{(\poslab{B_6}\imp (\pos \poslab{B_4}\vee\neg A_7))}\label{eq:ex3:vicnf:line9}             \\
             & \qquad(\neg \poslab{B_7}\vee\pos \poslab{B_5})\wedge(\neg \poslab{B_7}\vee\pos \poslab{B_6}) & \wedge                                                                    & \quad\eqcomment{(\poslab{B_7}\imp (\pos \poslab{B_5}\wedge\pos \poslab{B_6}))}\label{eq:ex3:vicnf:line10} \\
             & \qquad(\pos \poslab{B_1}\vee\pos \poslab{B_7})                                               & \wedge                                                                    & \label{eq:ex3:vicnf:line11}                                                                               \\
             & \qquad(\neg \poslab{B_2}\vee \neg \neglab{B_2})
             & \wedge                                                                                       & \quad\eqcomment{\textit{mutual exclusion for \poslab{B_2}, \neglab{B_2}}}                                                                                                             \\
             & \qquad(\neg \poslab{B_3}\vee \neg \neglab{B_3})                                              & \wedge                                                                    & \quad\eqcomment{\textit{mutual exclusion for \poslab{B_3}, \neglab{B_3}}}                                 \\
             & \qquad(\neg \poslab{B_4}\vee \neg \neglab{B_4})                                              &                                                                           & \quad\eqcomment{\textit{mutual exclusion for \poslab{B_4}, \neglab{B_4}}}
        \end{alignat}
    \end{subequations}
    \noindent
    Notice that the labels
    \neglab{B_1}, \neglab{B_5}, \neglab{B_6}, \neglab{B_7} and their
    respective one-way constraints are not
    introduced, because there is no negative occurrence of the respective
    sub-formulas in \vi{}.

    As in \cref{ex1,ex2}, consider the partial assignment $\muA\defas\set{\neg
            A_3,\neg A_4,\neg A_7}$~\eqref{eq:ex1:muA} which satisfies $\vi{}$. First, it
    is easy to see that \muA{} satisfies also \NNF{\vi}, in compliance with
    \cref{th:munnf}. Also, \muA{} satisfies $\exists\allB.\vinnfcnfpg$, because
    $\muA\cup\etaB\pmodels\vinnfcnfpg$ where $\etaB\defas\set{\neg \poslab{B_1},
            \neg \poslab{B_2},\neglab{B_2},\neg \poslab{B_3},\neg \neglab{B_3},\neg
            \poslab{B_4},\neglab{B_4},\poslab{B_5},\poslab{B_6},\poslab{B_7}}$.

    The above assignment can be produced by adopting the strategy described above.
    First, we assign the literals in \muA, which force to add also $\set{\neg
            \poslab{B_2}, \neg \poslab{B_4}}$ due to \eqref{eq:ex3:vicnf:line6} and
    \eqref{eq:ex3:vicnf:line8} respectively. Then we add \set{\poslab{B_6}} by
    step~\ref{item:rule:true} on~\eqref{eq:ex3:vicnf:line9}, \set{\neglab{B_2}} by
    step~\ref{item:rule:false} on~\eqref{eq:ex3:vicnf:line2}, and
    \set{\neg\poslab{B_1},\neg\neglab{B_3},\neg\poslab{B_3}} by
    step~\ref{item:rule:star} on \eqref{eq:ex3:vicnf:line1},
    \eqref{eq:ex3:vicnf:line3}, \eqref{eq:ex3:vicnf:line7} respectively. These
    force to add \set{\poslab{B_7}} and \set{\poslab{B_5}} due to
    \eqref{eq:ex3:vicnf:line11} and \eqref{eq:ex3:vicnf:line10} respectively.
    Finally, we add \set{\neglab{B_4}} by step~\ref{item:rule:false} on
    \eqref{eq:ex3:vicnf:line4}. Overall, this corresponds to enumerate only the
    partial assignment $\muA{}$ satisfying $\exists\allB.\vinnfcnfpg$, with both
    disjoint and non-disjoint enumeration:
    \begin{equation*}
        \label{eq:ex3:final}
        \set{\neg A_3, \neg A_4, \neg A_5}  \quad\eqcomment{\set{\neg \poslab{B_1}, \neg
                \poslab{B_2},\neglab{B_2},\neg \poslab{B_3},\neg
                \neglab{B_3},\neg
                \poslab{B_4},\neglab{B_4},\poslab{B_5},\poslab{B_6},\poslab{B_7}}}.
        \hfill\exdone{}
    \end{equation*}

\end{example}

Notice that, the shorter is $\muA$ w.r.t.\ a total assignment, the higher is
the chance that \TseitinCNF{} and \PlaistedCNF{} force the production of a high
number of extra assignments, the more beneficial is the usage of \NNFPlaisted{}
which avoids it. This said, if the enumerator is able to produce short partial
assignments \muA{} satisfying the formula, we expect a high benefit from using
\NNFPlaisted{} instead of \TseitinCNF{} and \PlaistedCNF{}; vice versa, if the
enumerator produces only total or nearly-total assignments, we expect no or
very-limited benefit respectively.

Implementation-wise, the strategy to assign the values of \etaB{} described
above is difficult to implement inside the current enumerators. Therefore,
there is no formal guarantee that a generic enumeration procedure always finds
exactly the $\etaB$ which prevents the generation of longer assignments. For
example, the enumeration procedure of \sref{sec:background:allsat} finds a
total truth assignment $\etaA\cup\etaB$ that satisfies the formula, and then
finds $\muA\subseteq\etaA$ that is minimal w.r.t.\ that specific $\etaB$ such
that $\muA\cup\etaB\pmodels\vinnfcnfpg$, so that the $\etaB$ found is not
guaranteed to be the one that allows for the most effective minimization of
$\muA$. Ad-hoc enumeration heuristics could be adopted. %}%\ignoreinlong{investigated}\ignoreinshort{
Nevertheless, in \sref{sec:experiments} we show that a very simple heuristic
---i.e., force the assignments of false values first to decision atoms---
guarantees dramatic improvements w.r.t.\ previous approaches using two
state-of-the-art AllSAT/AllSMT enumerators.
%\RSTODO{RISCRITTO FIN QUI} %%%%%%%%%%%%%%%%%%%%%%%%%%%%%%5

\section{Experimental evaluation}%
\label{sec:experiments}
In this section, we experimentally evaluate the impact of different
CNF-izations on the disjoint and non-disjoint AllSAT %\ignoreinshort{%} 
%\ignoreinshort{
and AllSMT tasks. In order to compare them on fair ground, we have implemented
a base version of each from scratch in
PySMT~\cite{garioPySMTSolveragnosticLibrary2015}, avoiding specific
optimizations done by the solvers. %}
%\ignoreinshort{

We have conducted a very extensive collection of experiments on AllSAT and
AllSMT. We tested the encodings on \mathsat{}, which supports disjoint and
non-disjoint AllSAT and AllSMT, and on \tabularallsat{} and \tabularallsmt{},
which support disjoint AllSAT and AllSMT, respectively.
%  Unfortunately, to the best of our knowledge, there is no other available solver that allows us to test different CNF encodings for AllSAT and AllSMT, i.e., that allows performing projected enumeration of partial truth assignments for CNF formulas.\@ In~\sref{sec:survey-allsat-solvers} we present a survey of other candidate solvers, and the reasons why they are not suitable for our experiments.\@ %}
The choice of the solvers is motivated in~\sref{sec:survey-allsat-solvers}; the
benchmarks are described in~\sref{sec:experiments:benchmarks}, the information
for the reproducibility of the experiments is given in
\sref{sec:experimentsdata}, and the results are presented
in~\sref{sec:experiments:results}; finally, in \sref{sec:experiments:others} we
analyze and discuss the application of the encodings to related fields.
% \ignoreinlong{We used \mathsat{}~\cite{mathsat5_tacas13} as a SAT \ignoreinshort{\GMCHANGE{and \smt{}}} enumerators, because it implements the enumeration strategy by~\citeA{lahiriSMTTechniquesFast2006} described in~\sref{sec:background:allsat}.\@}
%To compute a set \TA{}, we invoked \mathsat{} with the following options: \textsf{-dpll.allsat\_allow\_duplicates=false} (assignments must be pairwise disjoint), \textsf{-dpll.allsat\_minimize\_model=true} (assignments are minimal), \textsf{-preprocessor.simplification=0} and \textsf{-preprocessor.toplevel\_propagation=false} (disable several non-validity-preserving preprocessing steps).
% \ignoreinlong{
%     We set the options \textsf{-dpll.branching\_initial\_phase=0} to split on the false branch first and \textsf{-dpll.branching\_cache\_phase=2} to enable phase caching.}

%     \ignoreinshort{\GMCHANGE{

% \begin{gmchangep}
\subsection{An analysis of available solvers}%
\label{sec:survey-allsat-solvers}
% table reporting: solver name + citation | input cnf | projected | partial assignments | disjoint | non-disjoint | code available | notes
%
%%%% moved to appendix
%\input{tools-table}
%
In order to evaluate the different CNF encodings, we need an AllSAT/AllSMT
solver that
\begin{enumerate*}[(i)]
    \item\label{item:tool:avail} is publicly available,
    \item\label{item:tool:cnf} takes as input a CNF formula,
    \item\label{item:tool:proj} allows performing {\em projected} enumeration,
    \item\label{item:tool:part} allows enumerating (disjoint or
          non-disjoint) {\em partial} truth assignments;
    \item\label{item:tool:minimal} also, the ability of producing \emph{minimal} assignments is valuable although not necessary.
          %    even minimal ones \item\label{item:tool:minimal} if possible. 
          % \item\label{item:tool:minimal} produces {\em minimal} partial assignments (when possible) .
          % \RSCHANGE{(even {\em minimal}
          % ones if possible)}.
\end{enumerate*}
% If  
%     \item\label{item:tool:minimal} produces {\em minimal} partial assignments.

%In order to test different CNF encodings for SAT and SMT enumeration, 
In the literature, we found the following candidate solvers: for AllSAT, we
found
% 1997
\textsc{RELSAT}~\cite{bayardoUsingCSPLookback1997},
% 2004
\textsc{Grumberg}~\cite{grumbergMemoryEfficientAllSolutions2004}, \textsc{SOLALL}~\cite{liNovelSATAllsolutions2004},
% 2005
\textsc{Jin}~\cite{jinEfficientConflictAnalysis2005,jin2005prime},
% 2007
\textsc{clasp}~\cite{gebserConflictDrivenAnswerSet2007},
% 2008
\textsc{PicoSAT}~\cite{bierePicoSATEssentials2008},
% 2014
\textsc{Yu}~\cite{yuAllSATUsingMinimal2014},
% 2016
\textsc{BC}, \textsc{NBC}, and \textsc{BDD}~\cite{todaImplementingEfficientAll2016},
% 2017
\textsc{depbdd}~\cite{todaExploitingFunctionalDependencies2017},
% 2018
\textsc{Dualiza}~\cite{mohleDualizingProjectedModel2018},
% 2020
\textsc{BASolver}~\cite{zhangAcceleratingAllSATComputation2020},
% 2022
\textsc{AllSATCC}~\cite{liangAllSATCCBoostingAllSAT2022},
% 2023
\textsc{HALL}~\cite{friedAllSATCombinationalCircuits2023,friedEntailingGeneralizationBoosts2024},
\tabularallsat~\cite{spallittaDisjointPartialEnumeration2024,spallittaDisjointProjectedEnumeration2025},
and \decdnnf~\cite{lagniezLeveragingDecisionDNNFCompilation2024}.
For AllSMT, the only candidates are
% 2013
\mathsat{}~\cite{mathsat5_tacas13},
% 2015
\textsc{aZ3}~\cite{phanAllSolutionSatisfiabilityModulo2015}, and
\tabularallsmt{}~\cite{spallittaDisjointProjectedEnumeration2025}.

In Appendix~\ref{appendix:tooltable} \Cref{tab:allsat-solvers} we report an
analysis of the above features \ref{item:tool:avail}-\ref{item:tool:minimal}
for each of the above-mentioned solvers. Overall, we observe that only four
solvers match all our needs: \mathsat{}, \tabularallsat{}, \tabularallsmt{},
and \textsc{Dualiza}. \mathsat{} supports disjoint and non-disjoint AllSAT and
AllSMT enumeration of minimal partial assignments, using the
blocking-clause-based enumeration strategy
of~\cite{lahiriSMTTechniquesFast2006} described
in~\sref{sec:background:allsat}. \tabularallsat{} supports disjoint AllSAT
enumeration of non-minimal partial assignments, based on chronological
backtracking without blocking
clauses~\cite{spallittaDisjointPartialEnumeration2024,spallittaDisjointProjectedEnumeration2025}.
\tabularallsmt{}, its counterpart for AllSMT, supports disjoint AllSMT.
\textsc{Dualiza} supports disjoint and non-disjoint AllSAT enumeration of
non-minimal partial assignments, based on dual reasoning. Since
\textsc{Dualiza} has been shown empirically to perform
poorly~\cite{friedEntailingGeneralizationBoosts2024}, we decided to leave it
out of our experiments, and to focus on \mathsat{}, \tabularallsat{}, and
\tabularallsmt{} only.
%Overall, we observe that the only solver that matches all our needs is
%\mathsat{}, both for AllSAT and AllSMT.\@
% \ignore{For AllSAT, a possible alternative would be \textsc{Dualiza},
%     % but its enumeration strategy %does not match properties~\ref{item:mumodelsvi}-\ref{item:muAminimal} 
%     % is different from the one described in~\sref{sec:background:allsat} since it is based on dual reasoning. %The \textsc{Grumberg} solver would match all the requirements, but it is not publicly available.
%     \GMCHANGE{but its enumerations strategy based on dual reasoning has
%         been shown empirically to perform poorly, since it requires studying both the satisfiability of the original formula and the unsatisfiability of its negation to retrieve partial assignments.}
%     For AllSMT, the only alternative is \textsc{aZ3}, but it only
%     enumerates total truth assignments over the relevant atoms, which,
%     besides,  can only be Boolean atoms\footnote{Notice that even though
%         it is always possible to label a \T-atom $\alpha$ with a Boolean
%         atom $A$ as $(A\iff{}\alpha)$ and use $A$ as relevant atom, this
%         implicitly forces every partial assignment to assign a truth value
%         to $A$, for the same reasons that we analyzed
%         in~\sref{sec:problem:label} for the labels used by the Tseitin
%         encoding.}.
% }
% Notice that 7 solvers were not available, even after contacting the authors.
% \end{gmchangep}

\subsection{Description of the problem sets}%
\label{sec:experiments:benchmarks}
We consider %\ignoreinshort{\GMCHANGE{
five %}}\ignoreinlong{three} 
sets of benchmarks of non-CNF formulas coming from different sources, both synthetic and real-world.
% \ignoreinlong{
%     In the first set of benchmarks, we generate random Boolean formulas by nesting Boolean operators up to a fixed depth. The second problem set consists of Boolean formulas encoding properties of ISCAS'85 circuits~\cite{brglezNeutralNetlist101985,hansenUnveilingISCAS85Benchmarks1999,tibebuAugmentingAllSolution2018}. As a third set of problems, we consider formulas encoding Booleanized Weighted Model Integration (WMI) problems~\cite{morettin-wmi-ijcar17,morettin-wmi-aij19,spallittaSMTbasedWeightedModel2022}.
% }
% \begin{ignoreinshortenv}
%     \begin{gmchangep}
For AllSAT, we evaluate the different CNF encodings on three sets of
benchmarks. The first one consists of synthetic Boolean formulas, which were
randomly generated by nesting Boolean operators up to a fixed depth. The second
one consists of formulas encoding properties of ISCAS'85
circuits~\cite{brglezNeutralNetlist101985,hansenUnveilingISCAS85Benchmarks1999}
as done in~\cite{tibebuAugmentingAllSolution2018}. The third one is a set of
benchmarks on combinatorial circuits encoded as And-Inverter Graphs (AIGs) used
in~\cite{friedAllSATCombinationalCircuits2023}.

For AllSMT, we consider two sets of benchmarks. The first one consists of
synthetic \smtlarat{} formulas, which were randomly generated by nesting
Boolean and \smtlarat{} operators up to a fixed depth. The second one consists
of formulas encoding Weighted Model Integration (WMI)
problems~\cite{morettin-wmi-ijcar17,morettin-wmi-aij19,spallittaSMTbasedWeightedModel2022,spallittaEnhancingSMTbasedWeighted2024a}.

With respect to the conference version of the
paper~\cite{masinaCNFConversionDisjoint2023}, we have extended the set of
Boolean synthetic benchmarks and added the AIG benchmarks. Moreover, we have
added the \smtlarat{} benchmarks, and modified the WMI benchmarks so that they
contain also \larat{} atoms.

%     \end{gmchangep}
%   \end{ignoreinshortenv}

%\subsubsection*{The  Boolean synthetic benchmarks}
\paragraph{The  Boolean synthetic benchmarks.}
The Boolean synthetic benchmarks were generated by nesting Boolean operators
$\wedge, \vee, \iff$ until some fixed depth $d$. Internal and leaf nodes are
negated with $0.5$ probability. Operators in internal nodes are chosen
randomly, giving less probability to the $\iff$ operator. In particular, $\iff$
is chosen with $0.1$ probability, whereas the other two are chosen with an
equal probability of $0.45$. We generated 100 instances over a set of 20
Boolean atoms and depth $d=8$, 100 instances over a set of 25 Boolean atoms and
depth $d=8$, and 100 instances over a set of 30 Boolean atoms and depth $d=6$,
for a total of 300 instances. %\ignoreinshort{\GMCHANGE{, %}}

% \ignoreinshort{
% \GMCHANGEp{
% We have extended this set of benchmarks with respect to the conference version of the paper~\cite{masinaCNFConversionDisjoint2023}, where we used only the first 100 instances described above. %}%}.

\paragraph{The %\ignoreinlong{circuits}\ignoreinshort{\GMCHANGEp{
    ISCAS'85 %}} 
    benchmarks.}%
\label{sec:benchmarks:iscas}
The ISCAS'85 benchmarks are a set of 10 combinatorial circuits used in test generation, timing analysis, and technology mapping~\cite{brglezNeutralNetlist101985}. They have well-defined, high-level structures and functions based on multiplexers, ALUs, decoders, and other common building blocks~\cite{hansenUnveilingISCAS85Benchmarks1999}.
We generated random instances as described in~\cite{tibebuAugmentingAllSolution2018}. In particular, for each circuit, we constrained 60\%, 70\%, 80\%, 90\%, and 100\% of the outputs to either 0 or 1, for a total of 250 instances.

% \begin{ignoreinshortenv}
%     \begin{gmchangep}
\paragraph{The AIG benchmarks.}
The AIG benchmarks are a set of formulas encoded as And-Inverter Graphs used
in~\cite{friedAllSATCombinationalCircuits2023}. They consist of a total of 89
instances, subdivided into 3 groups: 29 instances encoding industrial Static
Timing Analysis (STA) problems containing up to $13000$ input variables, 40
instances from the ``arithmetic'' and ``random\_control'' benchmarks of the
EPFL suite~\cite{amaruEPFLCombinationalBenchmark2015a} ---combining the
multiple outputs with an \emph{or} or an \emph{xor} operator--- containing up
to 1200 input variables, and 20 instances consisting of large
randomly-generated AIGs containing up to 2800 input variables.

We notice that the discussed CNF encodings can be applied to AIGs as well, as
an AIG can be seen as a non-CNF formula involving only $\wedge$ and $\neg$
operators. Moreover, we store formulas as DAGs, so that the same sub-formula is
not duplicated multiple times, and the CNF encodings use the same Boolean atom
to label the different occurrences of the same sub-formula. Hence, all the
discussed CNF encodings result in a CNF formula of linear size w.r.t.\ the size
of the AIG.
% This is particularly important for AIGs, since the sharing of sub-formulas is at the core of their compactness.
%     \end{gmchangep}
% \end{ignoreinshortenv}

% \begin{ignoreinshortenv}
%     \begin{gmchange}
\paragraph{The \smtlarat{} synthetic benchmarks.} These problems were generated with the same procedure as
the Boolean ones, with the difference that atoms are randomly-generated
\smtlarat{} atoms over a set of $R$ real variables $\set{x_1, \dots, x_R}$, in
the form $(\sum_{i=1}^R a_i x_i \leq b)$, where each $a_i$ and $b$ are random
real coefficients. We generated 100 instances with depth $d=5$, 100 instances
with depth $d=6$, and 100 instances with depth $d=7$, all of them involving
$R=5$ real variables, for a total of 300 instances.
%     \end{gmchange}
% \end{ignoreinshortenv}

\paragraph{The WMI benchmarks.}
WMI problems were generated using the procedure described
in~\cite{spallittaSMTbasedWeightedModel2022}. Specifically, the paper
addresses the problem of enumerating all the different paths of the weight
function by encoding it into a skeleton formula. Each instance consists of a
skeleton formula of a randomly-generated weight function, where the conditions
are random formulas over Boolean and \larat{}-atoms. Since the conditions are
typically non-atomic, the resulting formula is not in CNF, and thus we
preprocessed it with the different CNF-izations before enumerating its
\TA{\dots}. %\ignoreinlong{only over Boolean atoms}\ignoreinshort{\GMCHANGE{%}}. %\ignoreinlong{models}\ignoreinshort{\GMCHANGE{%}}.\@ \ignoreinlong{We generate 10 instances for each depth value 3, 5, 7, 9, each instance involving 10 Boolean atoms and no real variable, for a total of 40 problems.}
%\ignoreinshort{\GMCHANGE{
As done in~\cite{spallittaSMTbasedWeightedModel2022}, we fixed the number of
Boolean atoms to 3, and the number of real variables to 3, and we generated 10
instances for each depth value 4, 5, 6, and 7, for a total of 40 problems.
%}}
% \ignoreinlong{We remark on two aspects of these benchmarks. First, we have chosen to have Boolean-only weight conditions in order to better analyze the capacity of Boolean reasoning of the solver with the different transformations, without additional factors brought by the SMT component. Nevertheless, we expect to have similar outcomes also for formulas involving both Boolean and $\smtlarat$ atoms. Notice that these can still be meaningful WMI instances, as the $\larat$ component may be constrained by the rest of the formula.}
% (i.e.\ the \emph{support} formula~\cite{morettin-wmi-ijcar17,morettin-wmi-aij19,spallittaSMTbasedWeightedModel2022})}. 
% \ignoreinlong{Second,}
% \GMNOTE{Togliere la frase qui sotto che usa notazione non spiegata?}
% \ignoreinshort{\GMCHANGE{We remark that}}
% these formulas contain existentially quantified $\euf$-atoms, so that we enumerate $\exists\ally.\vi(\allA,\ignoreinshort{\allx,}\ally)$ by projecting the \ignoreinlong{models  of $\vi$}\ignoreinshort{\GMCHANGE{assignments}} over the \ignoreinlong{relevant}\ignoreinshort{\GMCHANGEp{Boolean }}atoms $\allA$\ignoreinshort{\GMCHANGEp{\ and the \larat{}-atoms containing only variables in \allx}}~\cite{spallittaSMTbasedWeightedModel2022}.

% \ignoreinshort{
%     \GMCHANGEp{
% We have extended this set of benchmarks with respect to the conference version of the paper~\cite{masinaCNFConversionDisjoint2023}, where we restricted to Booleanized WMI problems.
%         }
%         }

\subsection{Information for the reproducibility of the experiments}
\label{sec:experimentsdata}
We ran \mathsat{} with the option \textsf{-dpll.allsat\_minimize\_model=true}
to enumerate minimal partial assignments. For disjoint and non-disjoint
enumeration, we set the options \textsf{-dpll.allsat\_allow\_duplicates=false}
and \textsf{-dpll.allsat\_allow\_duplicates=true}, respectively. We also set the options
\textsf{-preprocessor.simplification=0} and
\textsf{-preprocessor.toplevel\_propagation=false} to disable several
non-validity-preserving preprocessing steps.
%To use the heuristic discussed in~\sref{sec:problem:polarity}
%and~\sref{sec: solution}  to best exploit the properties of both
%Plaisted and Greenbaum encoding and our encoding,  
As discussed in~\sref{sec:problem:polarity} and~\sref{sec:solution}, we also
set the options \textsf{-dpll.branching\_initial\_phase=0} and
\textsf{-dpll.branching\_cache\_phase=2} to split on the false branch first but
enabling phase caching. \tabularallsat{} was run with default options, which
include branching on the false branch first. \tabularallsmt{} was run with the
same options as \mathsat{} to disable preprocessing steps, since it uses
\mathsat{} as a backend for theory reasoning.
% }}
%\GMCHANGE{
%    We checked the results by ensuring that (a) $\muA\pmodels\vi$ for every $\muA\in\TA{\exists\allB.\vicnfgen}$; (b) for every $\etaA\in\TTA{\vi}$ there exists a $\muA\in\TA{\exists\allB.\vicnfgen}$ s.t. $\muA\subseteq\etaA$. This check was done only on the WMI benchmarks, because the problems in the other two groups of benchmarks had too many partial and total truth assignments to be checked in a reasonable time.
%}

% The code and the experiments are available at: \url{TODO}\\
All the experiments were run on an Intel Xeon Gold 6238R @ 2.20GHz 28 Core
machine with 128 GB of RAM, running Ubuntu Linux 20.04. For each problem set,
we set a timeout of 3600s. %}}.
%\ignoreinshort{\GMCHANGE{
Benchmarks, results, and source code are made available online on a Zenodo
repository~\cite{masina_2024_cnf_results,masina_2024_cnf_code}.\@ An updated
version of the source code is available at
\url{https://github.com/masinag/allsat-cnf}.

\subsection{Results}%
\label{sec:experiments:results}

\begin{figure}
    \centering
    \begin{subfigure}[t]{\textwidth}
        \centering
        \includegraphics[height=1.8em]{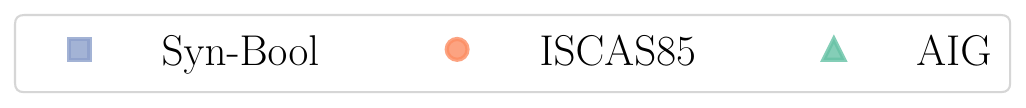}
    \end{subfigure}
    \begin{subfigure}[t]{\textwidth}
        \begin{subfigure}[t]{0.26\textwidth}
            \centering
            \includegraphics[width=.85\textwidth]{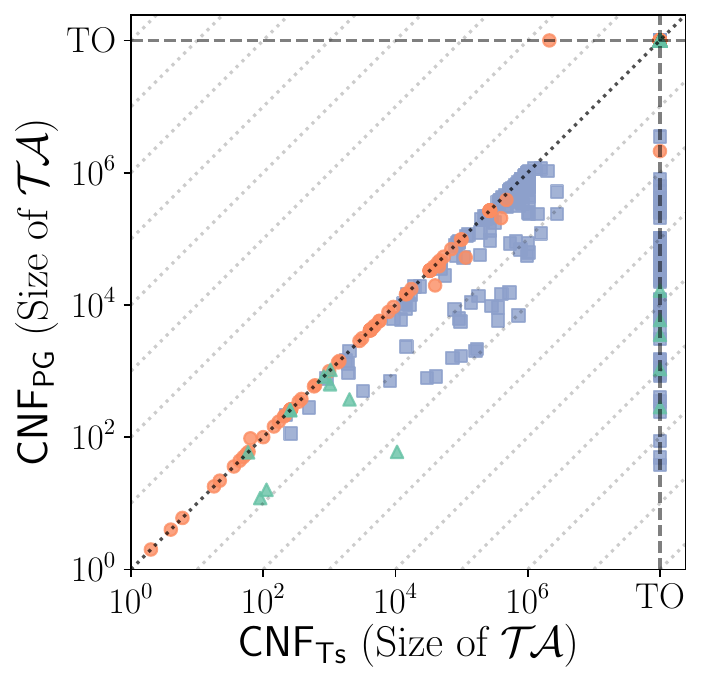}%
            \label{fig:plt:all:bool:norep:models:lab_vs_pol}
        \end{subfigure}\hfill
        \begin{subfigure}[t]{0.26\textwidth}
            \centering
            \includegraphics[width=.85\textwidth]{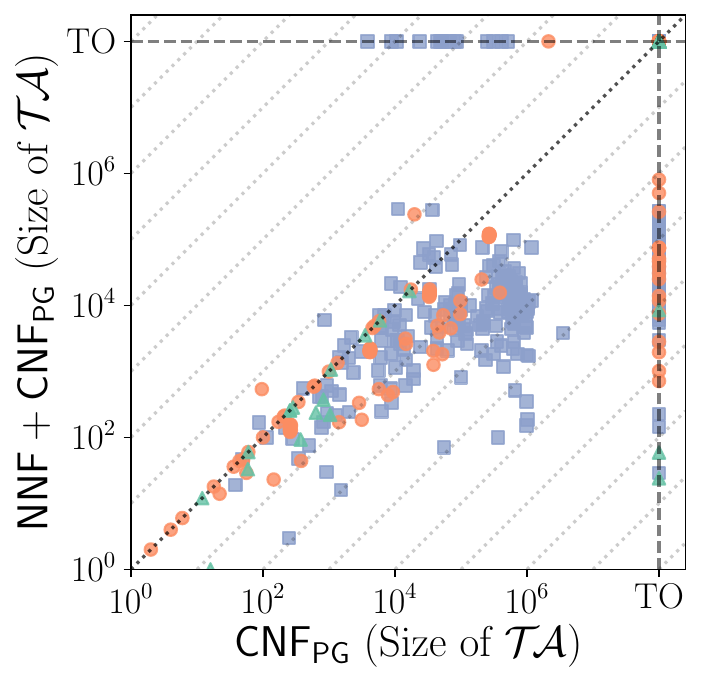}%
            \label{fig:plt:all:bool:norep:models:pol_vs_nnfpol}
        \end{subfigure}\hfill
        \begin{subfigure}[t]{0.26\textwidth}
            \centering
            \includegraphics[width=.85\textwidth]{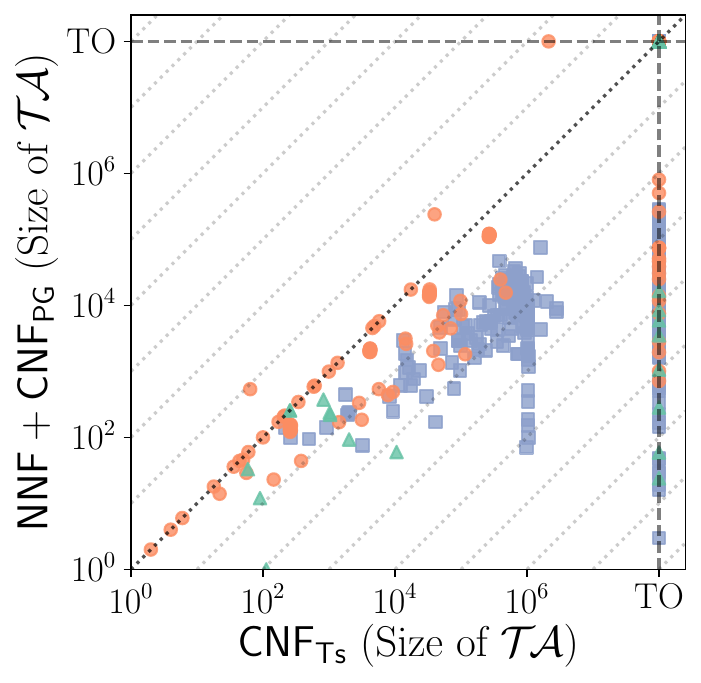}%
            \label{fig:plt:all:bool:norep:models:lab_vs_nnfpol}
        \end{subfigure}\hfill
        \begin{subfigure}[t]{0.26\textwidth}
            \centering
            \includegraphics[width=.85\textwidth]{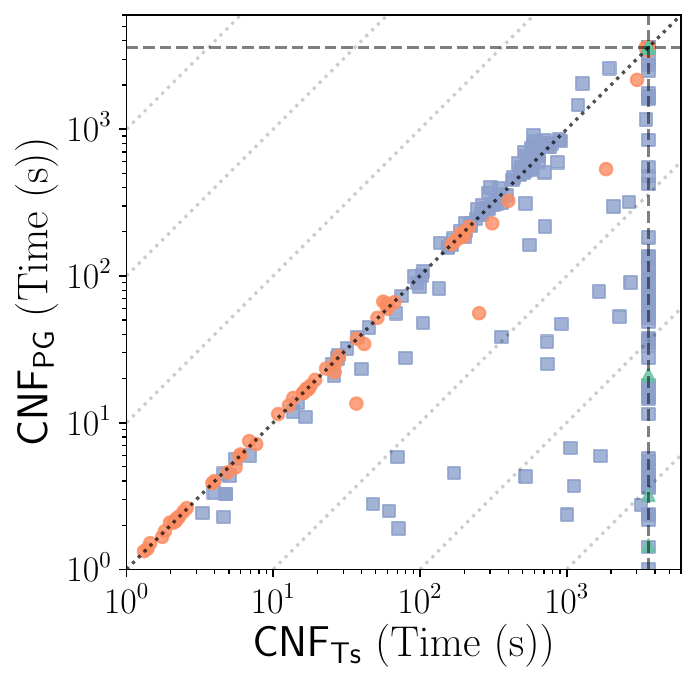}%
            \label{fig:plt:all:bool:norep:time:lab_vs_pol}
        \end{subfigure}\hfill
        \begin{subfigure}[t]{0.26\textwidth}
            \centering
            \includegraphics[width=.85\textwidth]{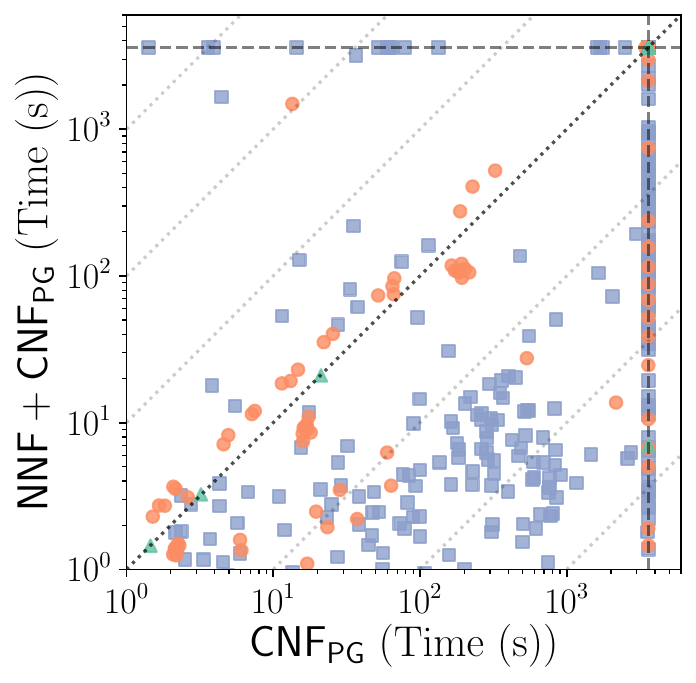}%
            \label{fig:plt:all:bool:norep:time:pol_vs_nnfpol}
        \end{subfigure}\hfill
        \begin{subfigure}[t]{0.26\textwidth}
            \centering
            \includegraphics[width=.85\textwidth]{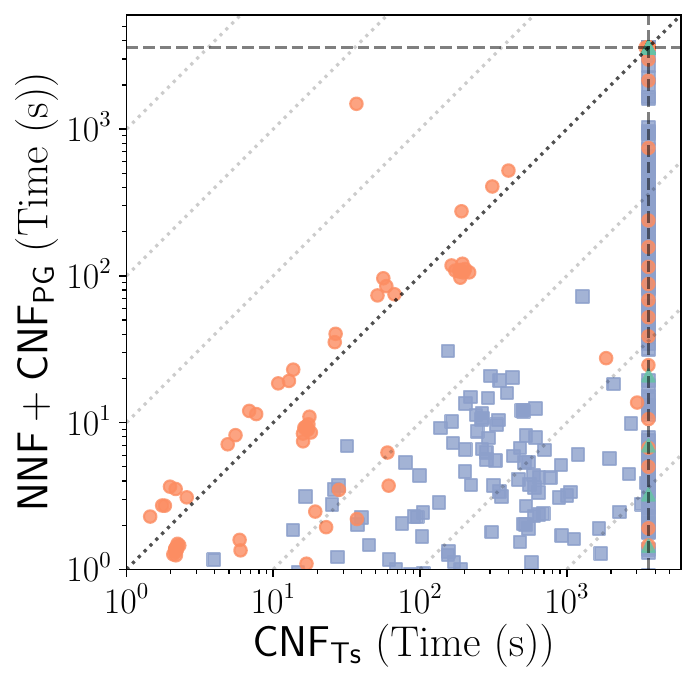}%
            \label{fig:plt:all:bool:norep:time:lab_vs_nnfpol}
            % \end{subfigure}
        \end{subfigure}
        \caption{Results for disjoint enumeration.}%
        \label{fig:plt:all:bool:norep:scatter}
    \end{subfigure}
    %%%%%%%%%%%% REP %%%%%%%%%%%%%
    \begin{subfigure}[t]{\textwidth}
        \begin{subfigure}[t]{0.26\textwidth}
            \centering
            \includegraphics[width=.85\textwidth]{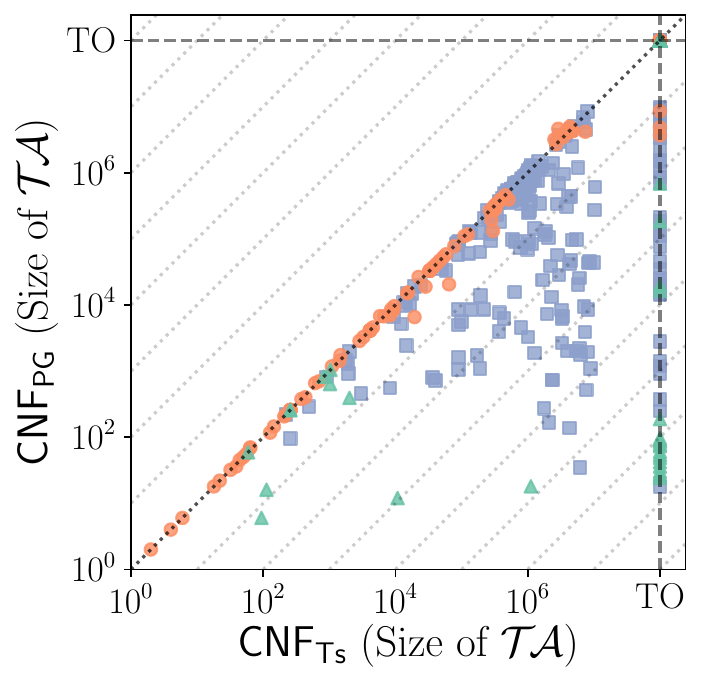}%
            \label{fig:plt:all:bool:rep:models:lab_vs_pol}
        \end{subfigure}\hfill
        \begin{subfigure}[t]{0.26\textwidth}
            \centering
            \includegraphics[width=.85\textwidth]{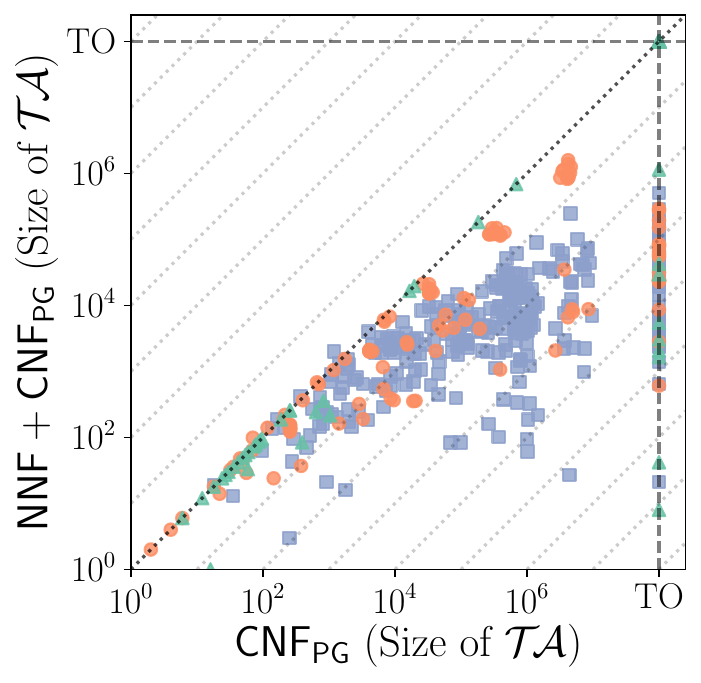}%
            \label{fig:plt:all:bool:rep:models:pol_vs_nnfpol}
        \end{subfigure}\hfill
        \begin{subfigure}[t]{0.26\textwidth}
            \centering
            \includegraphics[width=.85\textwidth]{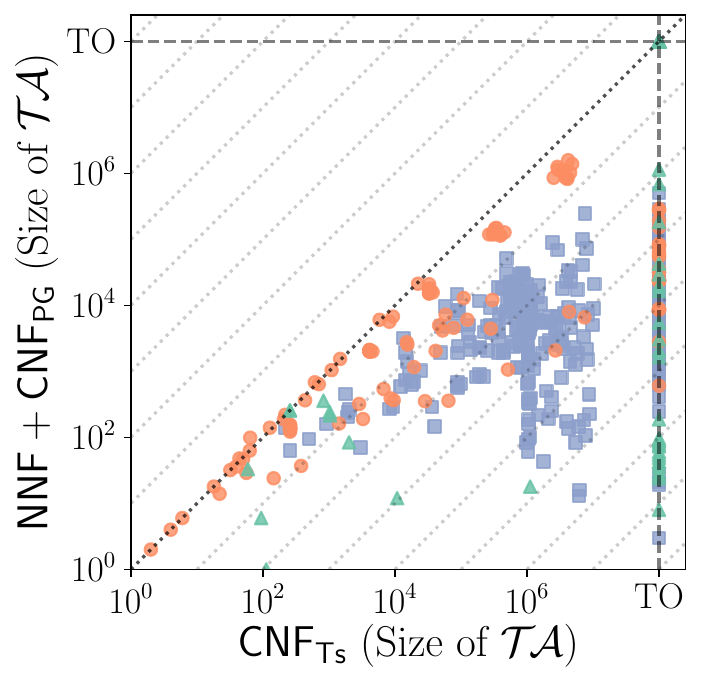}%
            \label{fig:plt:all:bool:rep:models:lab_vs_nnfpol}
        \end{subfigure}\hfill
        \begin{subfigure}[t]{0.26\textwidth}
            \centering
            \includegraphics[width=.85\textwidth]{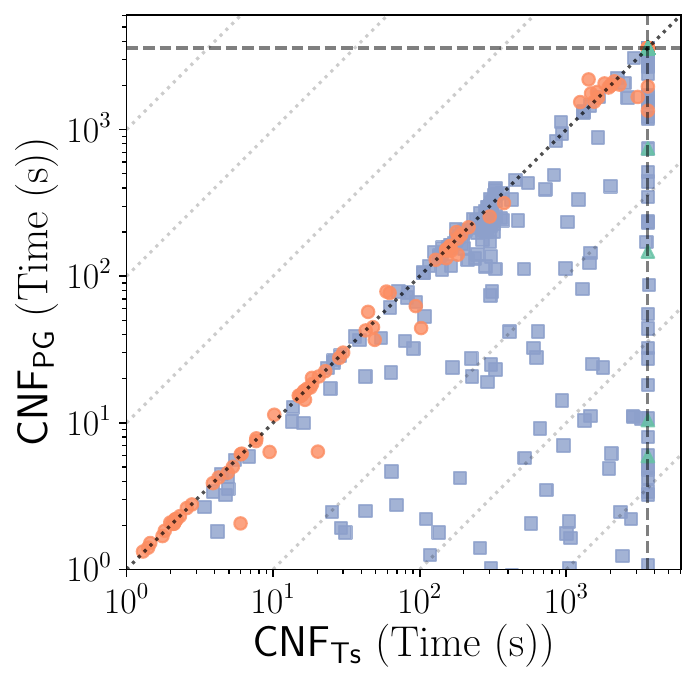}%
            \label{fig:plt:all:bool:rep:time:lab_vs_pol}
        \end{subfigure}\hfill
        \begin{subfigure}[t]{0.26\textwidth}
            \centering
            \includegraphics[width=.85\textwidth]{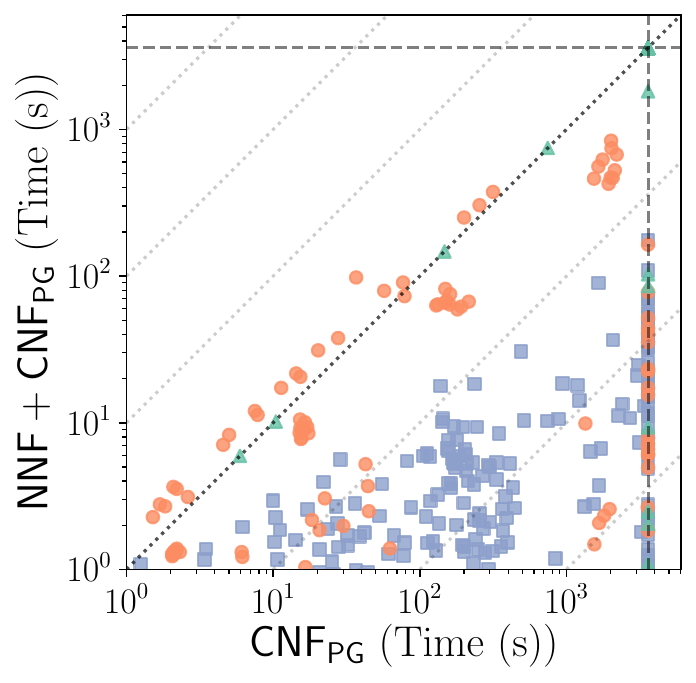}%
            \label{fig:plt:all:bool:rep:time:pol_vs_nnfpol}
        \end{subfigure}\hfill
        \begin{subfigure}[t]{0.26\textwidth}
            \centering
            \includegraphics[width=.85\textwidth]{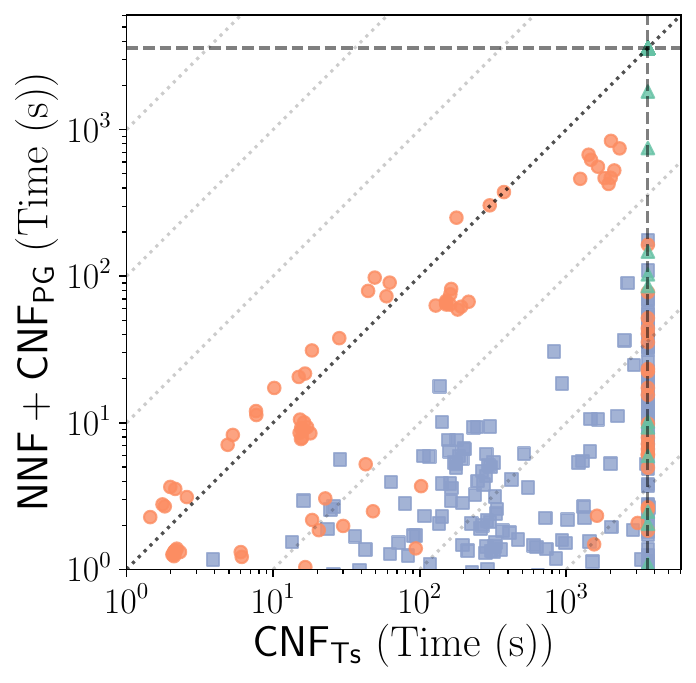}%
            \label{fig:plt:all:bool:rep:time:lab_vs_nnfpol}
            % \end{subfigure}
        \end{subfigure}
        \caption{Results for non-disjoint enumeration.}%
        \label{fig:plt:all:bool:rep:scatter}
    \end{subfigure}
    \begin{subfigure}[t]{\textwidth}
        \vspace{.2em}
        \centering
        {\footnotesize
            
% - Total problems:
% mode         LAB  LABELNEG_POL  NNF_MUTEX_POL
% problem_set                                  
% aig           89            89             89
% iscas85      250           250            250
% syn-bool     300           300            300
%
% - disjoint timeouts:
% mode         LAB LABELNEG_POL NNF_MUTEX_POL
% problem_set                                
% aig           79           73            71
% iscas85       47           43            30
% syn-bool     151           84            20
%
%
% non-disjoint timeouts:
% mode        LAB LABELNEG_POL NNF_MUTEX_POL
% problem_set                               
% aig          78           60            50
% iscas85      27           22             1
% syn-bool     88           48             0

% table like: Problem Set | \# problems |           T.O.            |
%                         |             |     disj    |   nondisj   |
%                         |             | LAB | LABELNEGPOL | ... | LAB | LABELNEGPOL | ... |
% --------------------------------------------------------------------------------------------

\newcommand{\best}[1]{\textbf{#1}}
% \begin{figure}[th]
%     \centering
\begin{tabularx}{.7\textwidth}{l|c|ccc|ccc}
    % \toprule
    \multirow{3}{*}{Bench.} & \multirow{3}{*}{Instances} & \multicolumn{3}{c|}{T.O.\ for disjoint AllSAT} & \multicolumn{3}{c}{T.O.\ for non-disjoint AllSAT}                                          \\
                            &                            & \TseitinCNF{}                             & \PlaistedCNF{}                               & \NNFPlaisted{}
                            & \TseitinCNF{}              & \PlaistedCNF{}                            & \NNFPlaisted{}                                                               \\[0.2em]
    \hline
    Syn-Bool                & 300                        & 151                                       & 84                                           & \best{20}                      & 88 & 48 & \best{0}  \\
    ISCAS85                 & 250                        & 47                                        & 43                                           & \best{30}                      & 27 & 22 & \best{1}  \\
    AIG                     & 89                         & 79                                        & 73                                           & \best{71}                      & 78 & 60 & \best{50} \\
    % \bottomrule
\end{tabularx}
%     \caption{Number of timeouts for the plots in \cref{fig:plt:all:bool:scatter}.}%
%     \label{tab:timeouts:bool}

% \end{figure}

        }
        \caption{Number of timeouts.}%
        \label{tab:timeouts:bool}
    \end{subfigure}
    \caption{Results on the Boolean benchmarks using \mathsat{}.
        Plots in~\ref{fig:plt:all:bool:norep:scatter} and~\ref{fig:plt:all:bool:rep:scatter} compare CNF-izations by \TAna{} size (first row) and execution time (second row).
        Points on dashed lines represent timeouts, shown in~\ref{tab:timeouts:bool}.
        All axes use a logarithmic scale.}%
    \label{fig:plt:all:bool:scatter}
\end{figure}

\begin{figure}
    \centering
    \begin{subfigure}[t]{\textwidth}
        \centering
        \includegraphics[height=1.8em]{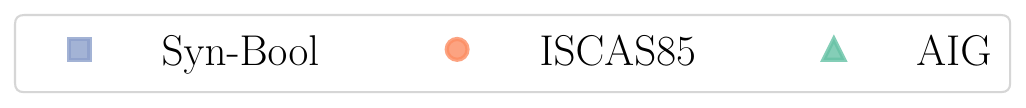}
    \end{subfigure}
    \begin{subfigure}[t]{\textwidth}
        \begin{subfigure}[t]{0.26\textwidth}
            \centering
            \includegraphics[width=.85\textwidth]{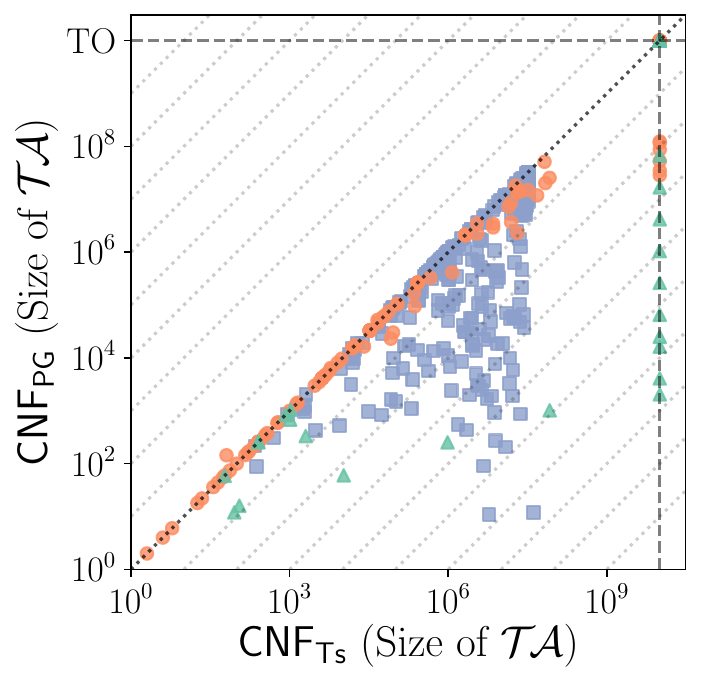}%
            \label{fig:plt:tabula:all:bool:norep:models:lab_vs_pol}
        \end{subfigure}\hfill
        \begin{subfigure}[t]{0.26\textwidth}
            \centering
            \includegraphics[width=.85\textwidth]{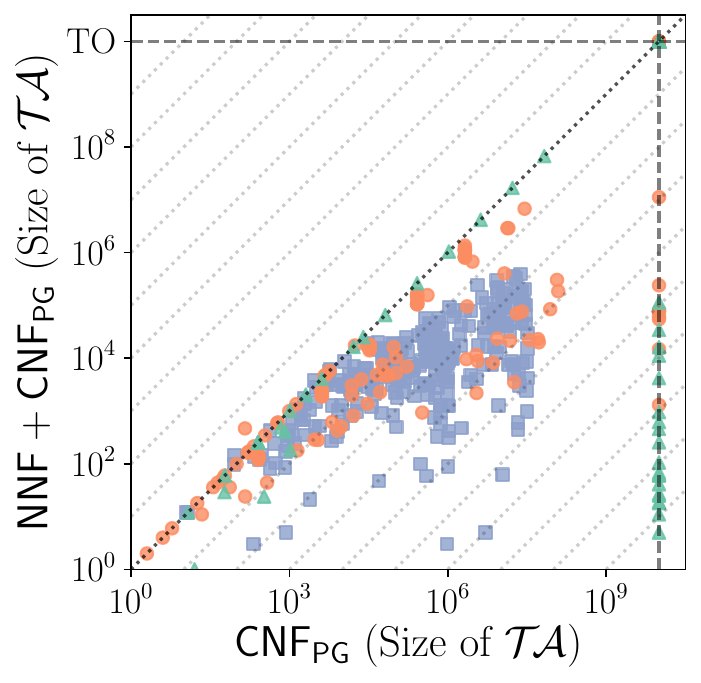}%
            \label{fig:plt:tabula:all:bool:norep:models:pol_vs_nnfpol}
        \end{subfigure}\hfill
        \begin{subfigure}[t]{0.26\textwidth}
            \centering
            \includegraphics[width=.85\textwidth]{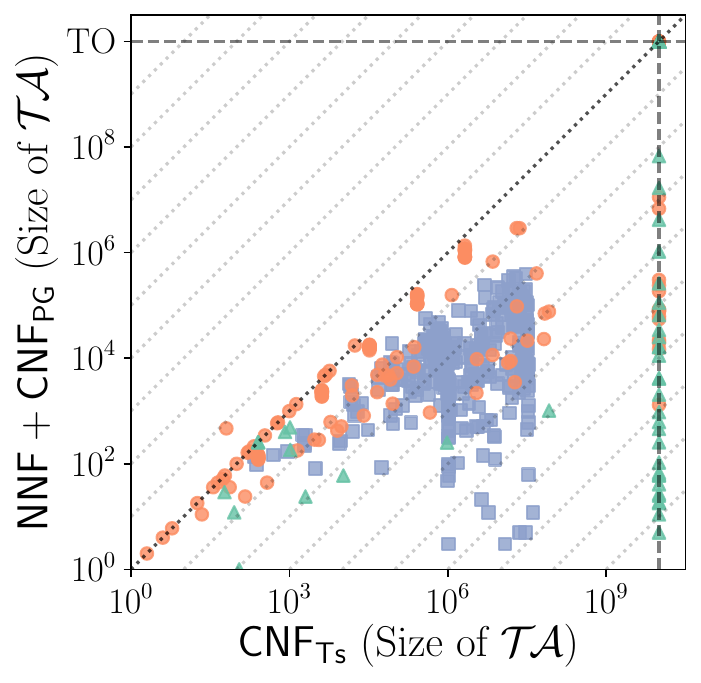}%
            \label{fig:plt:tabula:all:bool:norep:models:lab_vs_nnfpol}
        \end{subfigure}\hfill
        \begin{subfigure}[t]{0.26\textwidth}
            \centering
            \includegraphics[width=.85\textwidth]{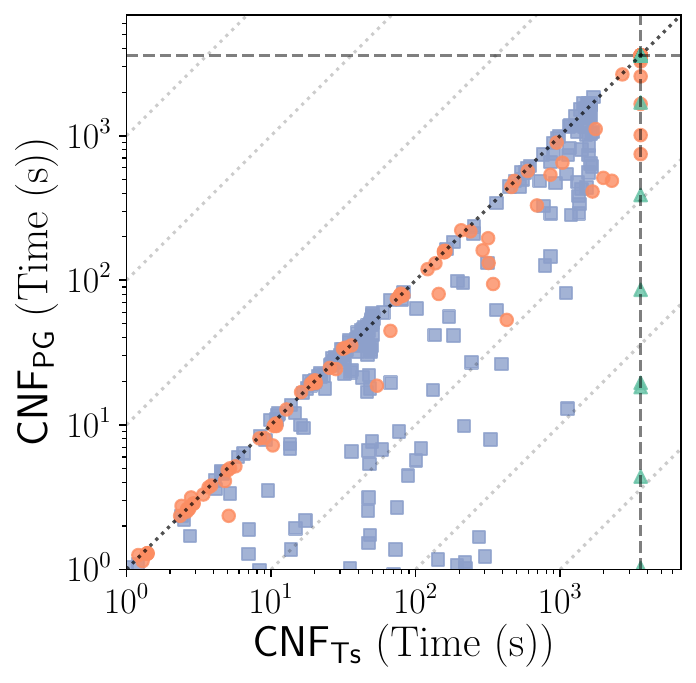}%
            \label{fig:plt:tabula:all:bool:norep:time:lab_vs_pol}
        \end{subfigure}\hfill
        \begin{subfigure}[t]{0.26\textwidth}
            \centering
            \includegraphics[width=.85\textwidth]{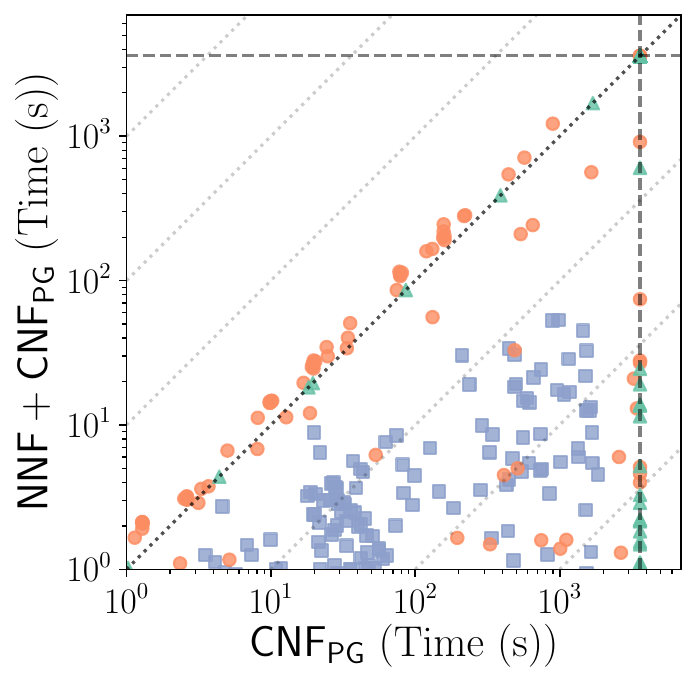}%
            \label{fig:plt:tabula:all:bool:norep:time:pol_vs_nnfpol}
        \end{subfigure}\hfill
        \begin{subfigure}[t]{0.26\textwidth}
            \centering
            \includegraphics[width=.85\textwidth]{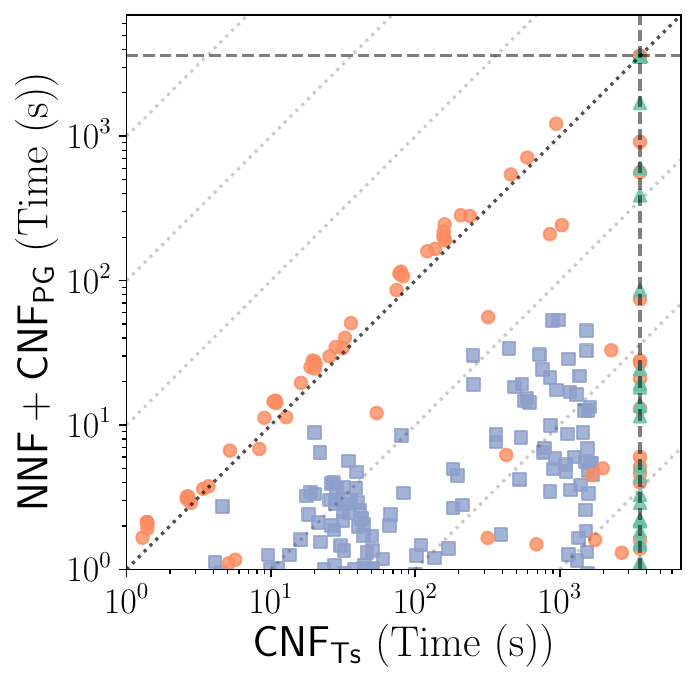}%
            \label{fig:plt:tabula:all:bool:norep:time:lab_vs_nnfpol}
            % \end{subfigure}
        \end{subfigure}
        \caption{Results for disjoint enumeration.}%
        \label{fig:plt:tabula:all:bool:norep:scatter}
    \end{subfigure}
    \begin{subfigure}[t]{\textwidth}
        \vspace{.2em}
        \centering
        {\footnotesize
            \newcommand{\best}[1]{\textbf{#1}}
% \begin{figure}[th]
%     \centering
\begin{tabularx}{.44\textwidth}{l|c|ccc}
    % \toprule
    \multirow{3}{*}{Bench.} & \multirow{3}{*}{Instances} & \multicolumn{3}{c}{T.O.\ for disjoint AllSAT}           \\
                            &                            & \TseitinCNF{}                             & \PlaistedCNF{}                               & \NNFPlaisted{}
                               \\[0.2em]
    \hline
    Syn-Bool                & 300                        & \best{0}                                       & \best{0}                                           & \best{0}  \\
    ISCAS85                 & 250                        & 17                                        & 11                                           & \best{3}  \\
    AIG                     & 89                         & 77                                        & 67                                           & \best{47} \\
    % \bottomrule
\end{tabularx}
%     \caption{Number of timeouts for the plots in \cref{fig:plt:all:bool:scatter}.}%
%     \label{tab:timeouts:bool}

% \end{figure}

        }
        \caption{Number of timeouts.}%
        \label{tab:timeouts:tabula:bool}
    \end{subfigure}
    \caption{Results on the Boolean benchmarks using \tabularallsat{}.
        Plots in~\ref{fig:plt:tabula:all:bool:norep:scatter} compare CNF-izations by \TAna{} size (first row) and execution time (second row).
        Points on dashed lines represent timeouts, shown in~\ref{tab:timeouts:tabula:bool}.
        All axes use a logarithmic scale.}%
    \label{fig:plt:tabula:all:bool:scatter}
\end{figure}

% \ignoreinlong{\cref{fig:plt:syn:bool:scatter,,fig:plt:circ:scatter,,fig:plt:wmi:scatter} show the results of the experiments on the synthetic, circuits and WMI benchmarks, respectively.}
% \ignoreinshort{\GMCHANGEp{
\cref{fig:plt:all:bool:scatter,fig:plt:tabula:all:bool:scatter} show the results for AllSAT with \mathsat{} and \tabularallsat{}, respectively.\@ \cref{fig:plt:all:lra:scatter,fig:plt:tabula:all:lra:scatter} show the results for AllSMT with \mathsat{} and \tabularallsmt{}, respectively.
%}}
%     }%
% }%
%
%
For each figure, we report a pair of subfigures comparing the CNF-izations for
disjoint and non-disjoint enumeration (when both are supported).\@ Each
subfigure reports a set of scatter plots to compare \TseitinCNF{},
\PlaistedCNF{} and \NNFPlaisted{} in terms of number of partial truth
assignments (size of \TAna), in the first row, and execution time, in the
second row. %\ignoreinshort{\GMCHANGE{%}}%\ignoreinlong{models}\ignoreinshort{\GMCHANGE{%}}, 
Each problem set is represented by a different color and marker.
(In~\sref{appendix:experiments}, we report the scatter plots for each group of
benchmarks separately.) Timeouts are represented by the points on the dashed
line, and summarized in the tables below the plots.

Notice the logarithmic scale of the axes (!).
%
%\cref{tab:timeouts:bool,tab:timeouts:lra} report the total number of timeouts for the Boolean and \smtlarat{} benchmarks, respectively.

% Each problem set is represented by a different color, 
%\ignoreinshort{\GMCHANGE{
%}}
% Moreover, we show with CDF plots (\cref{fig:plt:syn:time:ecdf,fig:plt:circ:time:ecdf,fig:plt:wmi:time:ecdf}) the cumulative execution time taken to enumerate the models of the formulas CNF-ized with the different transformations.
\begin{figure}
    \centering
    \begin{subfigure}[t]{\textwidth}
        \centering
        \includegraphics[height=1.8em]{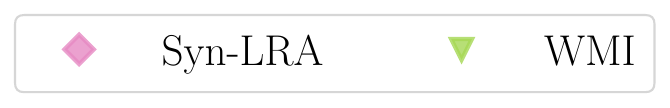}
    \end{subfigure}
    \begin{subfigure}[t]{\textwidth}
        \begin{subfigure}[t]{0.26\textwidth}
            \centering
            \includegraphics[width=.85\textwidth]{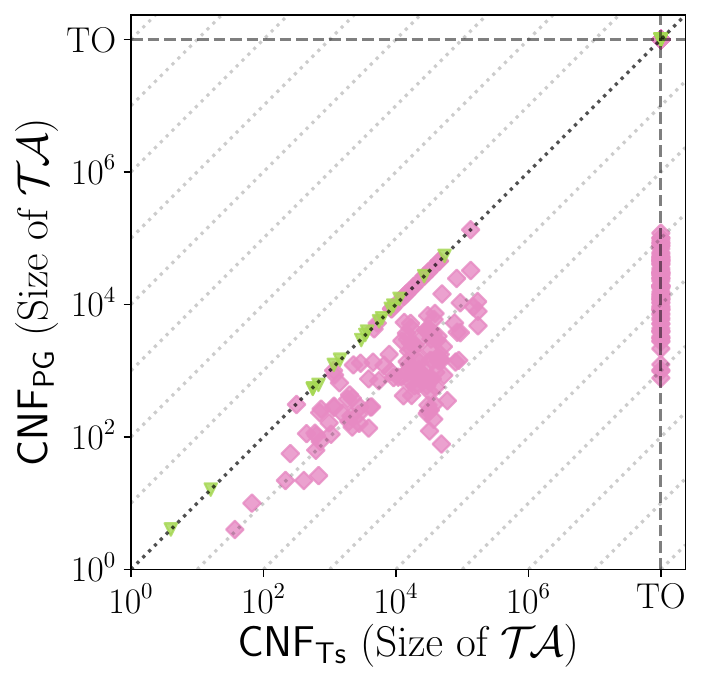}%
            \label{fig:plt:all:lra:norep:models:lab_vs_pol}
        \end{subfigure}\hfill
        \begin{subfigure}[t]{0.26\textwidth}
            \centering
            \includegraphics[width=.85\textwidth]{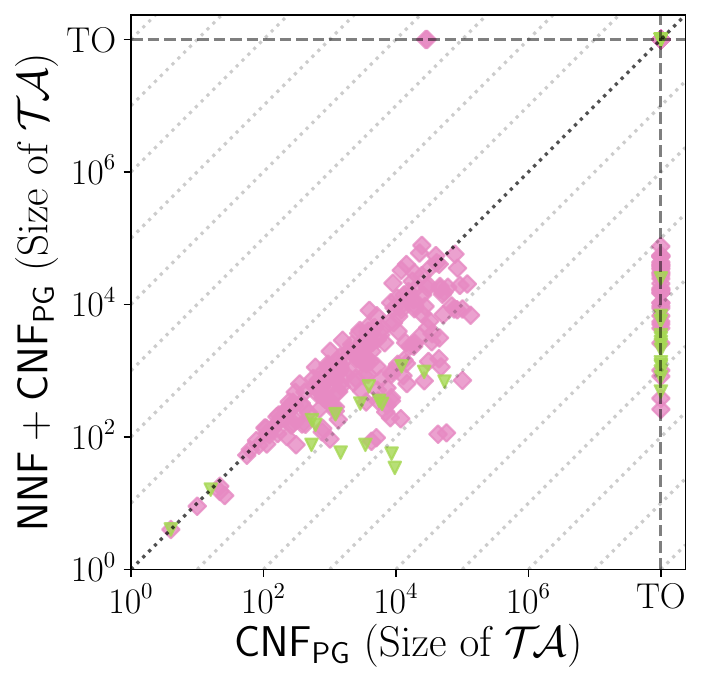}%
            \label{fig:plt:all:lra:norep:models:pol_vs_nnfpol}
        \end{subfigure}\hfill
        \begin{subfigure}[t]{0.26\textwidth}
            \centering
            \includegraphics[width=.85\textwidth]{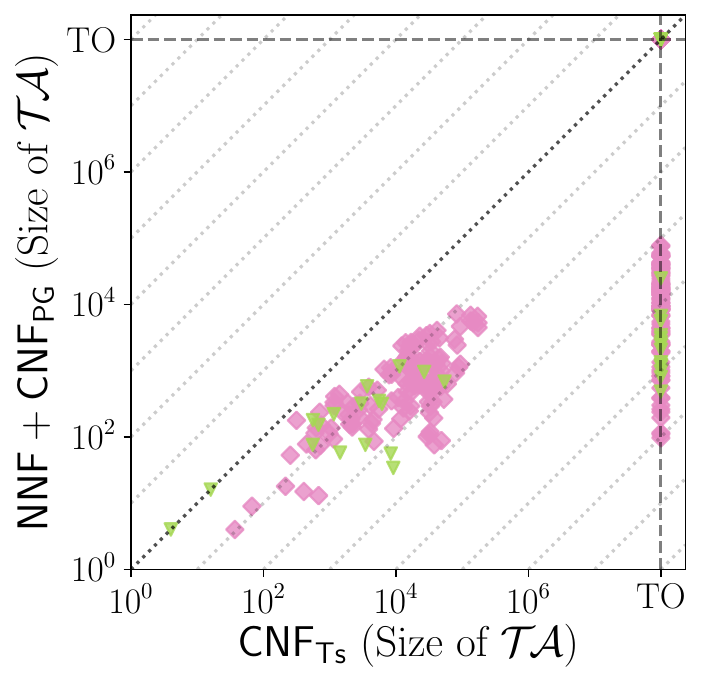}%
            \label{fig:plt:all:lra:norep:models:lab_vs_nnfpol}
        \end{subfigure}\hfill
        \begin{subfigure}[t]{0.26\textwidth}
            \centering
            \includegraphics[width=.85\textwidth]{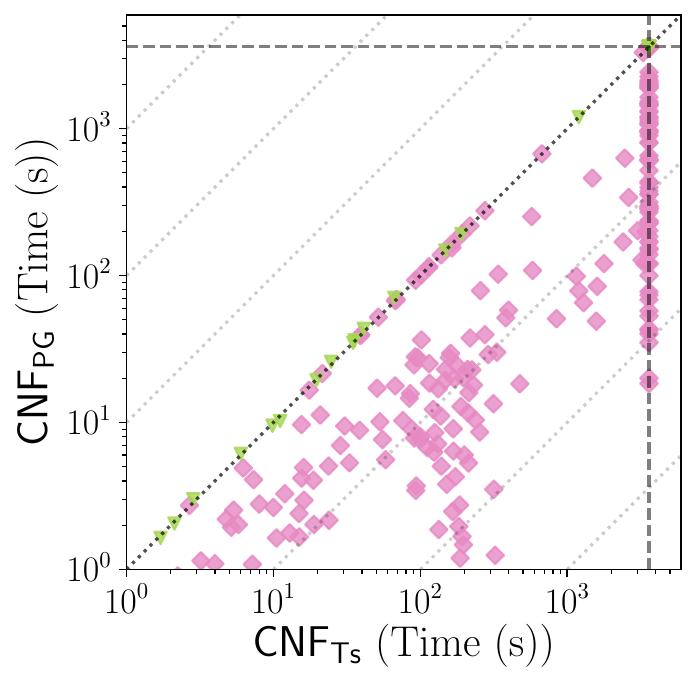}%
            \label{fig:plt:all:lra:norep:time:lab_vs_pol}
        \end{subfigure}\hfill
        \begin{subfigure}[t]{0.26\textwidth}
            \centering
            \includegraphics[width=.85\textwidth]{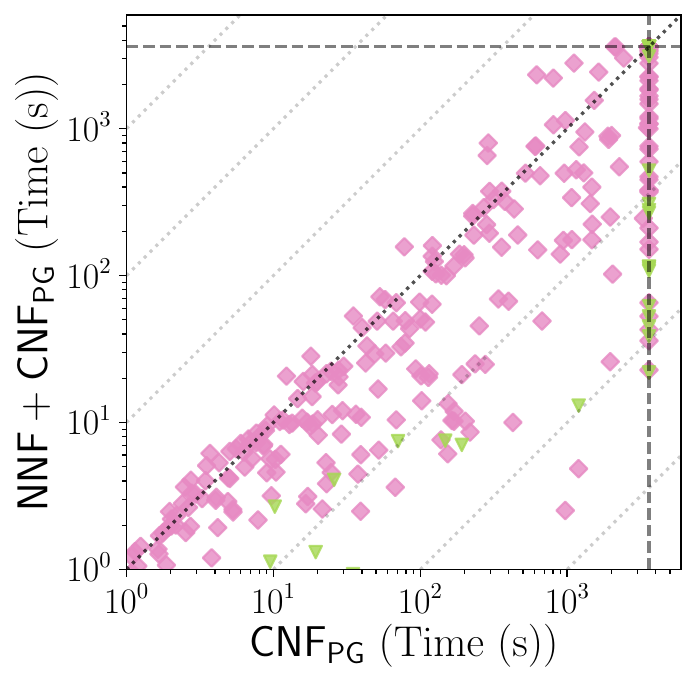}%
            \label{fig:plt:all:lra:norep:time:pol_vs_nnfpol}
        \end{subfigure}\hfill
        \begin{subfigure}[t]{0.26\textwidth}
            \centering
            \includegraphics[width=.85\textwidth]{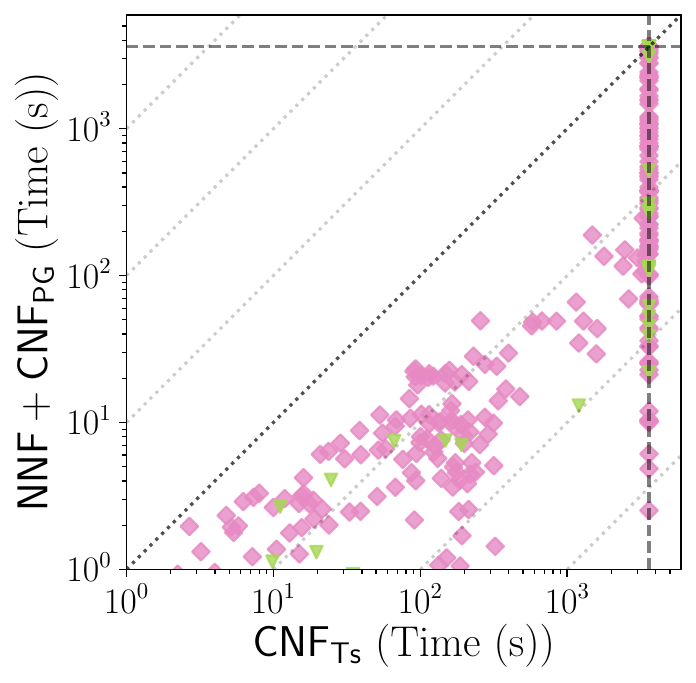}%
            \label{fig:plt:all:lra:norep:time:lab_vs_nnfpol}
            % \end{subfigure}
        \end{subfigure}
        \caption{Results for disjoint enumeration.}%
        \label{fig:plt:all:lra:norep:scatter}
    \end{subfigure}
    %%%%%%%%%%%% REP %%%%%%%%%%%%%
    \begin{subfigure}[t]{\textwidth}
        \begin{subfigure}[t]{0.26\textwidth}
            \centering
            \includegraphics[width=.85\textwidth]{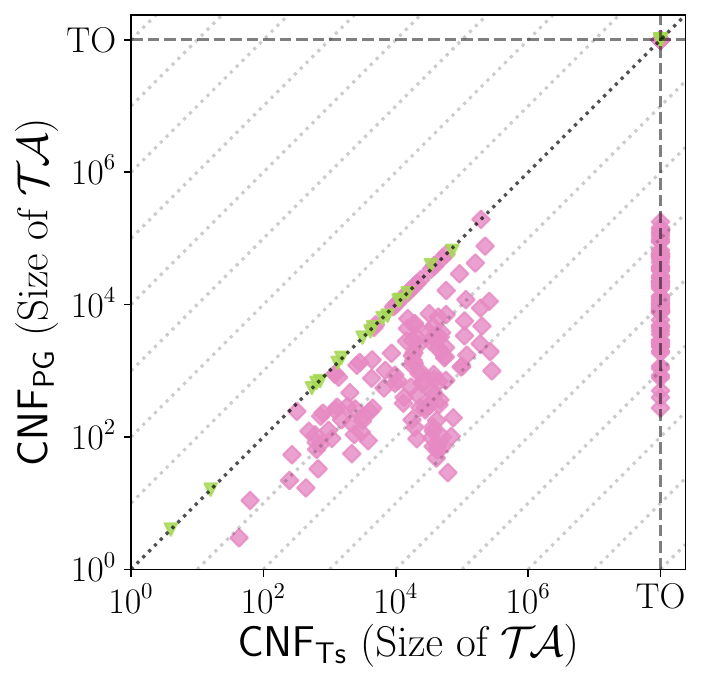}%
            \label{fig:plt:all:lra:rep:models:lab_vs_pol}
        \end{subfigure}\hfill
        \begin{subfigure}[t]{0.26\textwidth}
            \centering
            \includegraphics[width=.85\textwidth]{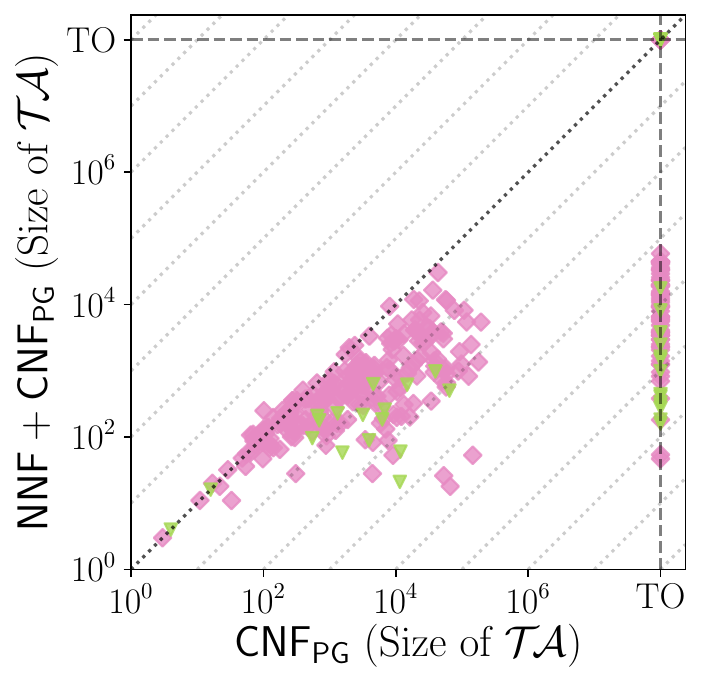}%
            \label{fig:plt:all:lra:rep:models:pol_vs_nnfpol}
        \end{subfigure}\hfill
        \begin{subfigure}[t]{0.26\textwidth}
            \centering
            \includegraphics[width=.85\textwidth]{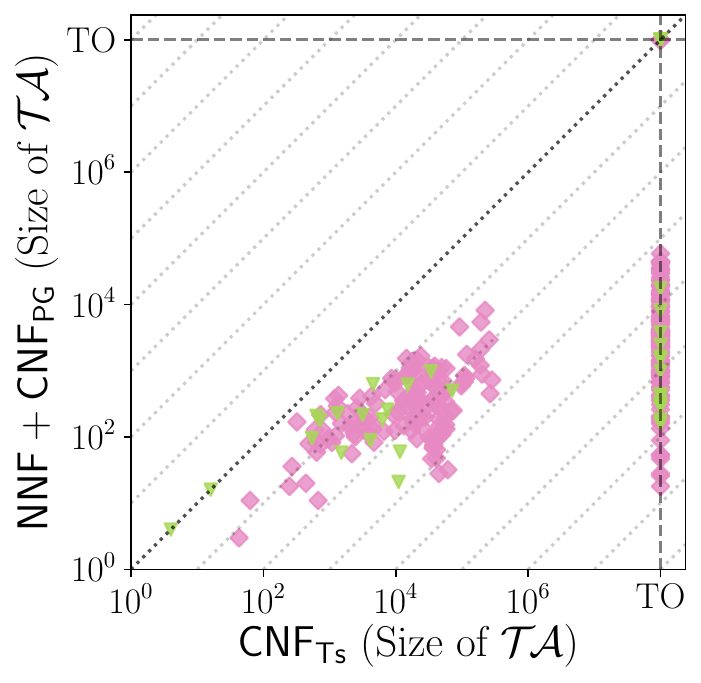}%
            \label{fig:plt:all:lra:rep:models:lab_vs_nnfpol}
        \end{subfigure}\hfill
        \begin{subfigure}[t]{0.26\textwidth}
            \centering
            \includegraphics[width=.85\textwidth]{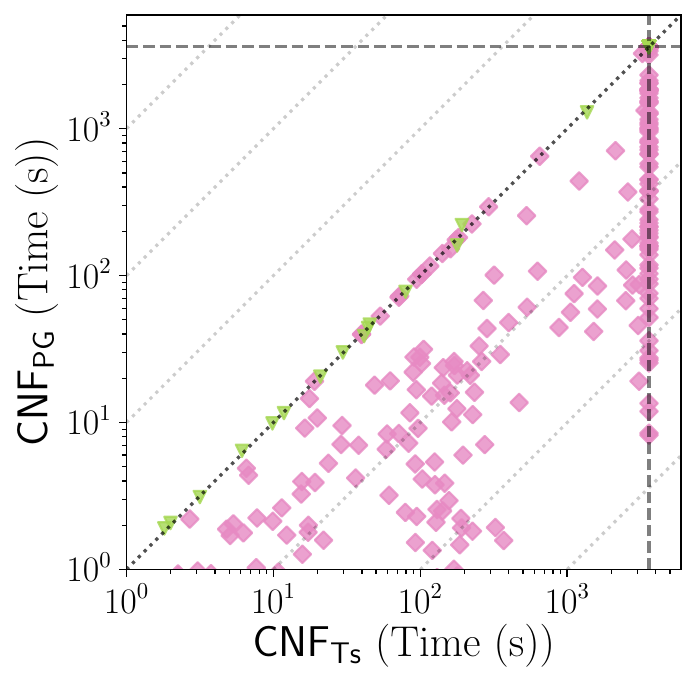}%
            \label{fig:plt:all:lra:rep:time:lab_vs_pol}
        \end{subfigure}\hfill
        \begin{subfigure}[t]{0.26\textwidth}
            \centering
            \includegraphics[width=.85\textwidth]{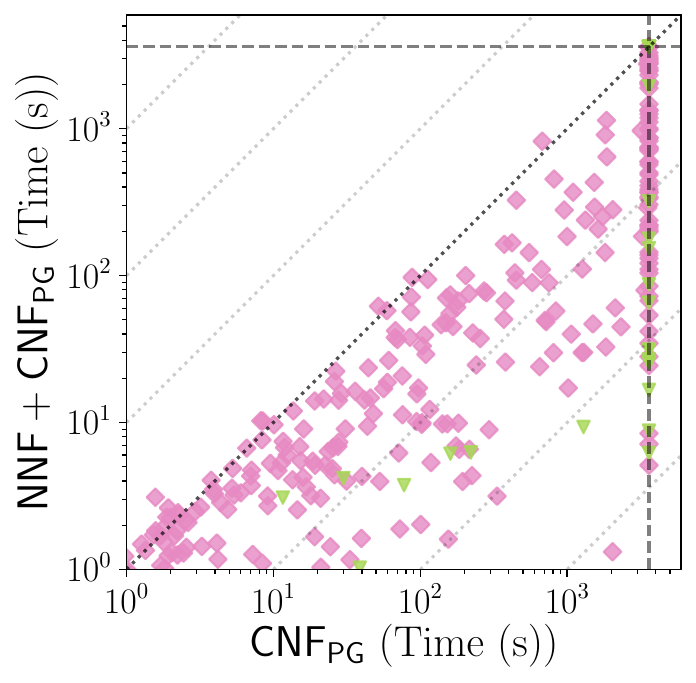}%
            \label{fig:plt:all:lra:rep:time:pol_vs_nnfpol}
        \end{subfigure}\hfill
        \begin{subfigure}[t]{0.26\textwidth}
            \centering
            \includegraphics[width=.85\textwidth]{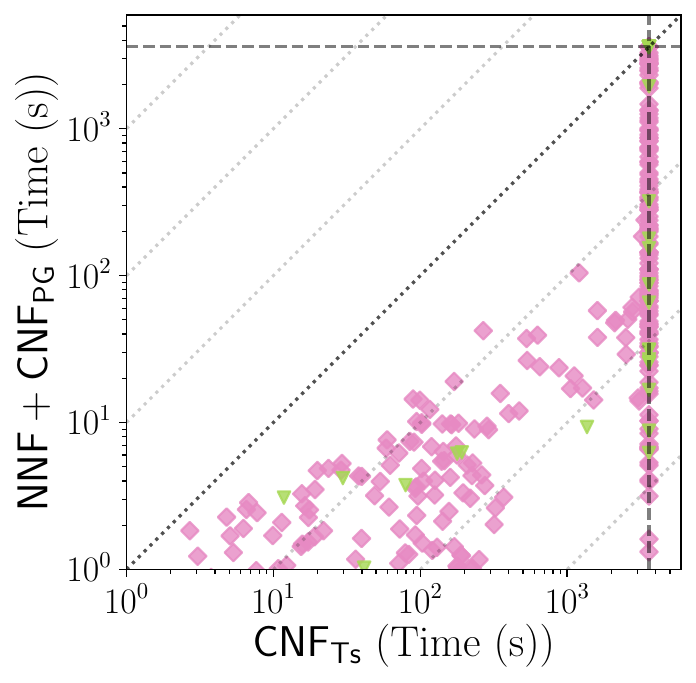}%
            \label{fig:plt:all:lra:rep:time:lab_vs_nnfpol}
            % \end{subfigure}
        \end{subfigure}
        \caption{Results for non-disjoint enumeration.}%
        \label{fig:plt:all:lra:rep:scatter}
    \end{subfigure}
    \begin{subfigure}[t]{\textwidth}
        \vspace{.2em}
        \centering
        {\footnotesize
            \newcommand{\best}[1]{\textbf{#1}}
% \begin{figure}[th]
%     \centering
\begin{tabularx}{.7\textwidth}{l|c|ccc|ccc}
    % \toprule
    \multirow{3}{*}{Bench.} & \multirow{3}{*}{Instances} & \multicolumn{3}{c|}{T.O.\ for disjoint AllSMT} & \multicolumn{3}{c}{T.O.\ for non-disjoint AllSMT}                                        \\
                            &                            & \TseitinCNF{}                                  & \PlaistedCNF{}                                    & \NNFPlaisted{}
                            & \TseitinCNF{}              & \PlaistedCNF{}                                 & \NNFPlaisted{}                                                                           \\[0.2em]
    \hline
    Syn-LRA                 & 300                        & 155                                            & 88                                                & \best{48}      & 152 & 74 & \best{7} \\
    WMI                     & 40                         & 23                                             & 23                                                & \best{9}       & 23  & 23 & \best{9} \\
    % \bottomrule
\end{tabularx}
%     \caption{Number of timeouts for the plots in \cref{fig:plt:all:lra:scatter}.}%
%     \label{tab:timeouts:lra}
% \end{figure}
        }
        \caption{Number of timeouts.}%
        \label{tab:timeouts:lra}
    \end{subfigure}
    \caption{Results on the \smtlarat{} benchmarks using \mathsat{}.
        Plots in~\ref{fig:plt:all:lra:norep:scatter},~\ref{fig:plt:all:lra:rep:scatter} compare CNF-izations by \TAna{} size (first row) and execution time (second row).
        Points on dashed lines represent timeouts, shown in~\ref{tab:timeouts:lra}.
        All axes use a logarithmic scale.}%
    \label{fig:plt:all:lra:scatter}
\end{figure}
\begin{figure}
    \centering
    \begin{subfigure}[t]{\textwidth}
        \centering
        \includegraphics[height=1.8em]{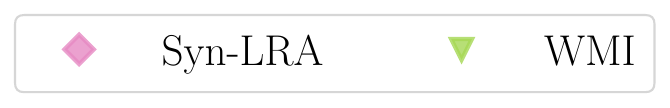}
    \end{subfigure}
    \begin{subfigure}[t]{\textwidth}
        \begin{subfigure}[t]{0.26\textwidth}
            \centering
            \includegraphics[width=.85\textwidth]{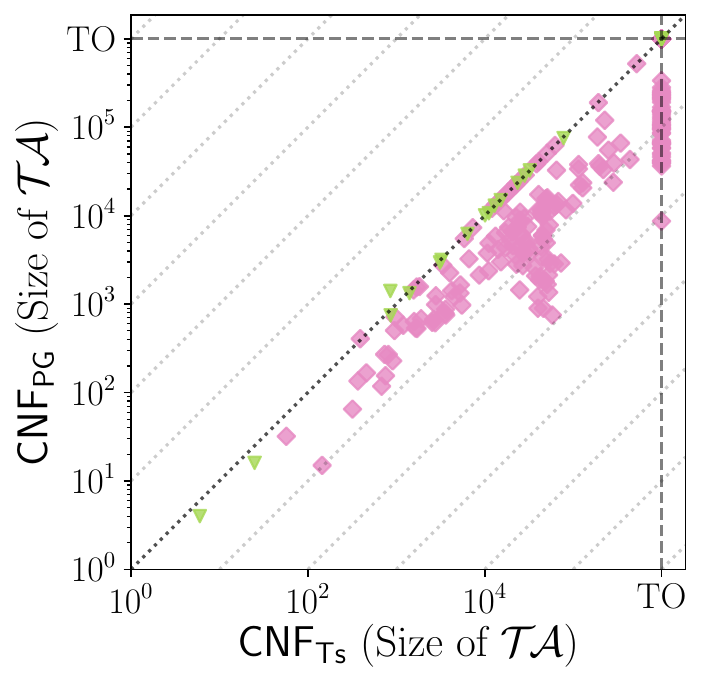}%
            \label{fig:plt:tabula:all:lra:norep:models:lab_vs_pol}
        \end{subfigure}\hfill
        \begin{subfigure}[t]{0.26\textwidth}
            \centering
            \includegraphics[width=.85\textwidth]{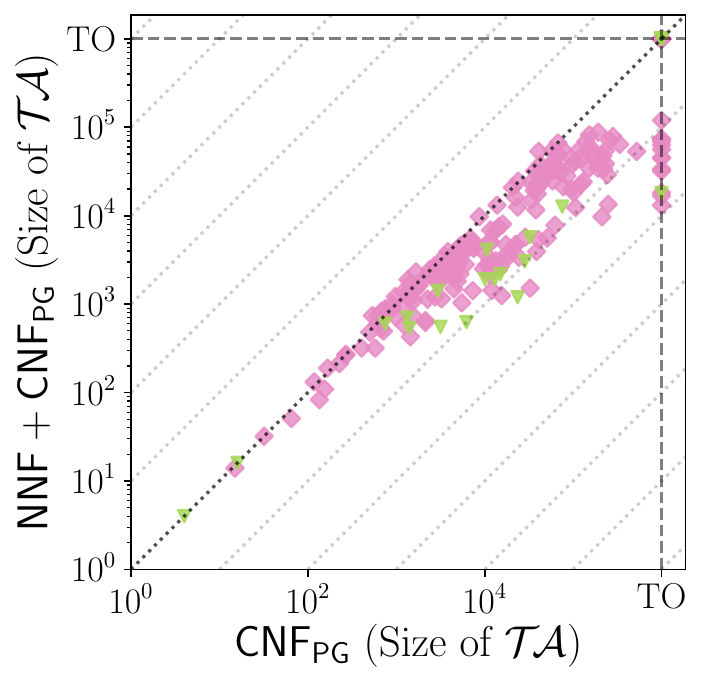}%
            \label{fig:plt:tabula:all:lra:norep:models:pol_vs_nnfpol}
        \end{subfigure}\hfill
        \begin{subfigure}[t]{0.26\textwidth}
            \centering
            \includegraphics[width=.85\textwidth]{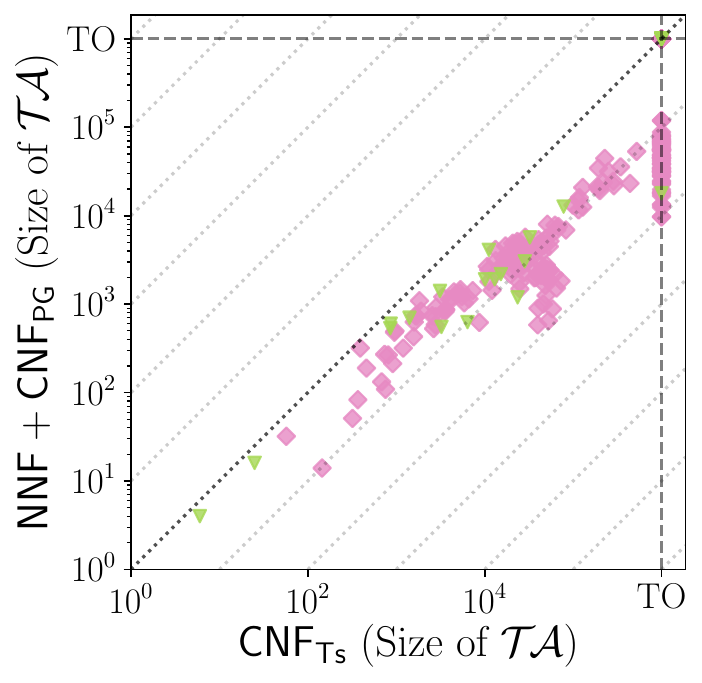}%
            \label{fig:plt:tabula:all:lra:norep:models:lab_vs_nnfpol}
        \end{subfigure}\hfill
        \begin{subfigure}[t]{0.26\textwidth}
            \centering
            \includegraphics[width=.85\textwidth]{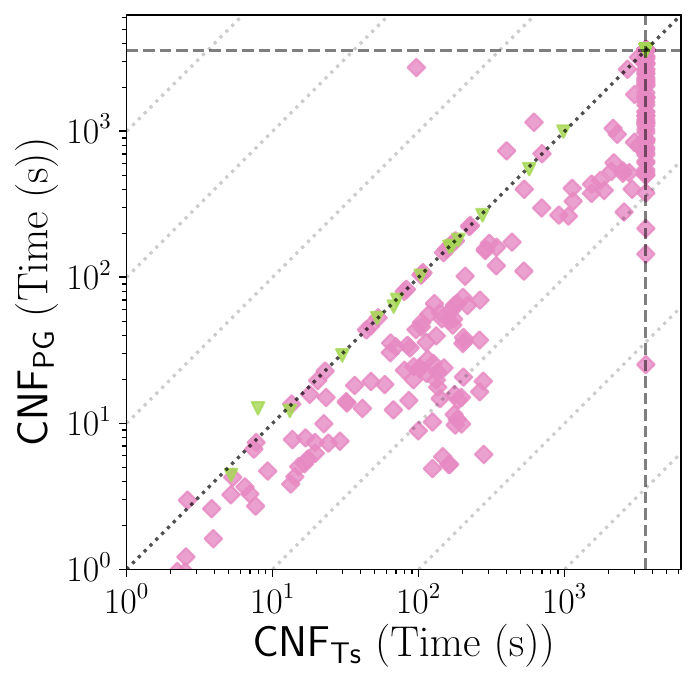}%
            \label{fig:plt:tabula:all:lra:norep:time:lab_vs_pol}
        \end{subfigure}\hfill
        \begin{subfigure}[t]{0.26\textwidth}
            \centering
            \includegraphics[width=.85\textwidth]{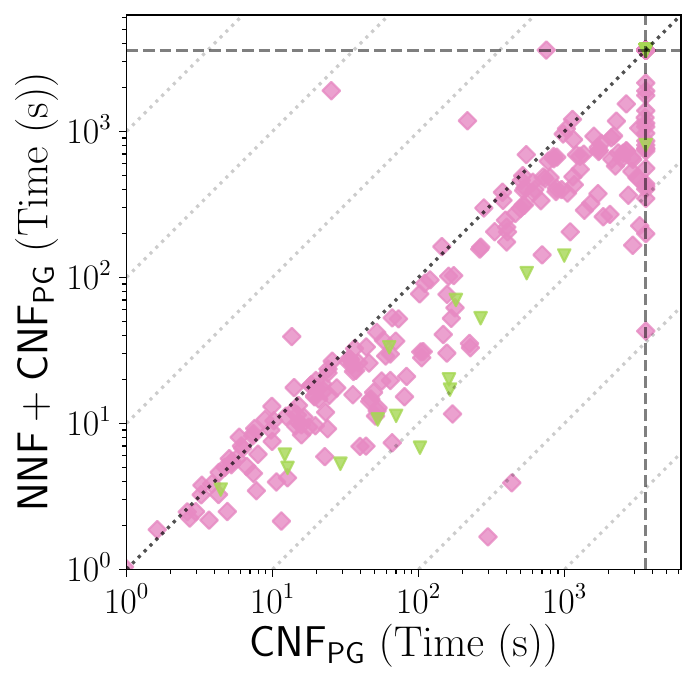}%
            \label{fig:plt:tabula:all:lra:norep:time:pol_vs_nnfpol}
        \end{subfigure}\hfill
        \begin{subfigure}[t]{0.26\textwidth}
            \centering
            \includegraphics[width=.85\textwidth]{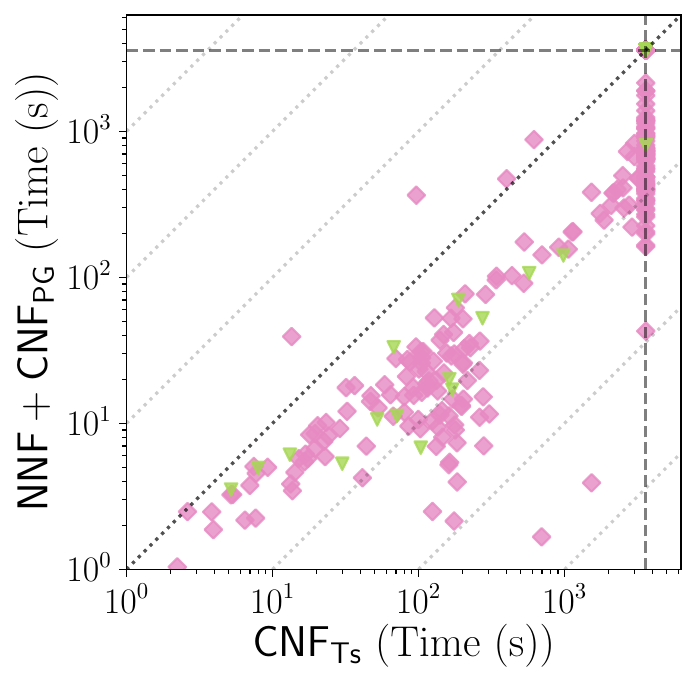}%
            \label{fig:plt:tabula:all:lra:norep:time:lab_vs_nnfpol}
            % \end{subfigure}
        \end{subfigure}
        \caption{Results for disjoint enumeration.}%
        \label{fig:plt:tabula:all:lra:norep:scatter}
    \end{subfigure}
    \begin{subfigure}[t]{\textwidth}
        \vspace{.2em}
        \centering
        {\footnotesize
            \newcommand{\best}[1]{\textbf{#1}}
% \begin{figure}[th]
%     \centering
\begin{tabularx}{.44\textwidth}{l|c|ccc}
    % \toprule
    \multirow{3}{*}{Bench.} & \multirow{3}{*}{Instances} & \multicolumn{3}{c}{T.O.\ for disjoint AllSMT}                                            \\
                            &                            & \TseitinCNF{}                                  & \PlaistedCNF{} & \NNFPlaisted{}
    \\[0.2em]
    \hline
    Syn-LRA                 & 300                        & 147                                            & 92             & \best{73}               \\
    WMI                     & 40                         & 24                                             & 24             & \best{23}               \\
    % \bottomrule
\end{tabularx}
%     \caption{Number of timeouts for the plots in \cref{fig:plt:all:lra:scatter}.}%
%     \label{tab:timeouts:lra}
% \end{figure}
        }
        \caption{Number of timeouts.}%
        \label{tab:timeouts:tabula:lra}
    \end{subfigure}
    \caption{Results on the \smtlarat{} benchmarks using \tabularallsmt{}.
        Plots in~\ref{fig:plt:tabula:all:lra:norep:scatter} compare CNF-izations by \TAna{} size (first row) and execution time (second row).
        Points on dashed lines represent timeouts, shown in~\ref{tab:timeouts:tabula:lra}.
        All axes use a logarithmic scale.}%
    \label{fig:plt:tabula:all:lra:scatter}
\end{figure}

\paragraph{The %\ignoreinshort{\GMCHANGE{
    Boolean %}}
    synthetic benchmarks.} (\cref{fig:plt:all:bool:scatter,fig:plt:tabula:all:bool:scatter})
%
%
% The results on the %\ignoreinshort{\GMCHANGE{
% Boolean %}} 
% synthetic benchmarks are shown in~\cref{fig:plt:syn:bool:scatter}. %\ignoreinlong{All the problems were solved for all the encodings within the timeout.} 
% The plots %\ignoreinshort{\GMCHANGE{\ 
% in~\cref{fig:plt:syn:bool:norep:scatter} %}} 
% show that %\ignoreinshort{\GMCHANGE{
The plots show that in the disjoint case \PlaistedCNF{} performs better than
\TseitinCNF{}, since both \mathsat{}~(\cref{fig:plt:all:bool:norep:scatter})
and \tabularallsat~(\cref{fig:plt:tabula:all:bool:norep:scatter}) enumerate a
smaller \TAna (first row) in less time (second row) on every instance.
Furthermore, the combination of \NNFna{} and \PlaistedCNF{} yields by far the
best results for both solvers, drastically reducing the size of \TAna{} and the
execution time by orders of magnitude w.r.t.\ both \TseitinCNF{} and
\PlaistedCNF{}. % in~\cref{fig:plt:all:bool:scatter} %\ }}%The analysis on the execution time is confirmed by the CDF in~\cref{fig:plt:syn:time:ecdf}.}
%\ignoreinshort{\GMCHANGE{
The advantage of \NNFPlaisted{} over both \TseitinCNF{} and \PlaistedCNF{} is
even more evident for the non-disjoint case
(\cref{fig:plt:all:bool:rep:scatter}). %, as shown in~\cref{fig:plt:all:bool:rep:scatter}.
% }}

\paragraph{The ISCAS'85 benchmarks.}
(\cref{fig:plt:all:bool:scatter,fig:plt:tabula:all:bool:scatter})
%
%\cref{fig:plt:circ:scatter} shows the performance of the different CNF-izations in the circuits benchmarks. %\ignoreinlong{The timeouts are represented by the points on the dashed lines. }
%
First, we notice that \TseitinCNF{} and \PlaistedCNF{} have very similar
behavior, both in terms of execution time and size of \TAna{}. The reason is
that in circuits most of the sub-formulas occur with double polarity, so that
the two encodings are very similar, if not identical. %, it is typical to have a lot of sharing of sub-formulas. Since we constrain the outputs to be 0 or 1 at random~\cite{tibebuAugmentingAllSolution2018}, 

Second, we notice that by converting the formula into NNF before applying
\PlaistedCNF{} the enumeration, both disjoint and non-disjoint, %\ignoreinshort{\GMCHANGE{, 
% }} 
is much more effective, as a much smaller \TAna is enumerated, with only a few
outliers.
%This can be due to the fact that the enumerated \TA{}  is \emph{minimal} and not \emph{minimum}, and its size can be affected also by the order in which the assignments are found.
%We also notice that a smaller number of models does not always correspond to a shorter execution time. 
The fact that for some instances \NNFPlaisted{} takes a little more time can be
due to the fact that it can produce a formula that is up to twice as large and
contains up to twice as many label atoms as the other two encodings, increasing
the time to find the assignments. Notice also that, even enumerating a smaller
\TAna{} at a price of a small time overhead can be beneficial in many
applications. %, for instance in
% WMI~\cite{morettin-wmi-ijcar17,morettin-wmi-aij19,spallittaSMTbasedWeightedModel2022,spallittaEnhancingSMTbasedWeighted2024a}.

% \begin{ignoreinshortenv}
%     \begin{gmchangep}

\paragraph{The AIG benchmarks.} (\cref{fig:plt:all:bool:scatter,fig:plt:tabula:all:bool:scatter})
%
%In~\Cref{fig:plt:aig:scatter} we show the results on the AIG benchmarks. 
This set of benchmarks is by far the most challenging one, as they contain many
Boolean atoms and a very complex structure. For this reason, even the
best-performing encoding, \NNFPlaisted{}, reports many timeouts. %even though we use a higher timeout of 3600 seconds, 

Nevertheless, we can still observe that the combination of \NNFna{} and
\PlaistedCNF{} is the best-performing encoding for both solvers, both in terms
of size of \TAna{} and execution time. %, for both disjoint and non-disjoint enumeration.

%     \end{gmchangep}
% \end{ignoreinshortenv}

% \begin{ignoreinshortenv}
%     \begin{gmchange}
\paragraph{The \smtlarat{} synthetic benchmarks.}
(\cref{fig:plt:all:lra:scatter,fig:plt:tabula:all:lra:scatter})
%
%The results on the \smtlarat{} benchmarks are shown in \cref{fig:plt:all:lra:scatter}.
The plots confirm that the analysis holds also for the All\smt{} case. The
results are in line with those obtained on the Boolean benchmarks, for both the
disjoint and non-disjoint cases. %(\cref{fig:plt:syn:lra:norep:scatter}) %(\cref{fig:plt:syn:lra:rep:scatter}) 
%     \end{gmchange}
% \end{ignoreinshortenv}

\paragraph{The WMI benchmarks.} (\cref{fig:plt:all:lra:scatter,fig:plt:tabula:all:lra:scatter})
%
%
%
%The plots in~\cref{fig:plt:wmi:scatter} compare the different CNF-izations in the WMI benchmarks in terms of size of \TAna{} and time. 
%All the problems were solved for all the encodings within the timeout. NNant,
In these benchmarks, most of the sub-formulas occur with double polarity, so
that \TseitinCNF{} and \PlaistedCNF{} encodings are almost identical, and they
obtain very similar results in both metrics. The advantage is significant,
instead, if the formula is converted into NNF upfront, since by using
\NNFPlaisted{} the solvers enumerate a smaller \TAna. In this application, it
is crucial to enumerate as few partial truth assignments as possible, since for
each truth assignment an integral must be computed, which is a very expensive
operation~\cite{morettin-wmi-ijcar17,morettin-wmi-aij19,spallittaSMTbasedWeightedModel2022,spallittaEnhancingSMTbasedWeighted2024a}.
% \ignoreinshort{\GMCHANGE{
Notice also that WMI requires disjoint enumeration to decompose the whole
integral into a sum of directly-computable integrals, one for each satisfying
truth assignment. Nevertheless, for completeness, we also report the results
for the non-disjoint case.
% }}

\paragraph{Overall Results.} Overall, from the plots and tables in
\cref{fig:plt:all:bool:scatter,fig:plt:tabula:all:bool:scatter,fig:plt:all:lra:scatter,fig:plt:tabula:all:lra:scatter},
it is eye-catching that the usage of \NNFPlaisted{} instead of \TseitinCNF{} or
\PlaistedCNF{} causes a dramatic improvement in the performances of \mathsat{},
\tabularallsat{}, and \tabularallsmt{}, in terms of both CPU time and size of
the assignment sets generated, for both disjoint and non-disjoint enumeration,
on all AllSAT and AllSMT benchmark sets under test.

%\newpage
%\subsection{Comparing the CNF encodings for  related problems}%
\subsection{Discussion: Using \NNFPlaisted{} with related problems}\label{sec:experiments:others}%
%We wish to draw the borderlines for the applicability of the proposed \NNFPlaisted{} encoding.
We stress the fact that the proposed \NNFPlaisted{} encoding is conceived and
expected to be effective for \emph{enumeration} techniques which generate
\emph{partial} assignments (see \sref{sec:solution}). Thus, we do not expect
benefits in using it with plain SAT or SMT \emph{solving} (see also
\cref{rem:ex3:preconv}), nor with enumeration techniques which generate {\em
        total} or {\em nearly-total} assignments only.
To this extent, we investigate empirically and discuss some related problems
for which we do not expect benefits in using \NNFPlaisted{} for the above
reason. (See also \sref{sec:related-work}.)
%These investigate this fact, we have performed some extra 

%\subsubsection{Using \vinnfcnfpg{} with  plain  SAT  and SMT solving}%
\paragraph{Plain  SAT  and SMT solving.}%
\label{sec:experiments:sat}
%
% \ignore{%
%     In~\sref{sec:experiments:sat} we also report the CDF of the execution
%     time for {\em plain SAT and SMT solving} on the above groups of benchmarks\ADDED{, using \mathsat}. We see from the results that, unlike with enumeration, for plain solving the pre-conversion into NNF has no benefit and even slightly worsens performance, as we observed in~\cref{rem:ex3:preconv}.
% }
To support \cref{rem:ex3:preconv}, in Appendix~\ref{appendix:sat} we compare
the different CNF encodings for plain SAT and SMT solving on all the benchmarks
in~\sref{sec:experiments:benchmarks}. Even though these problems are very small
for plain solving, and SAT and SMT%\ignoreinshort{\GMCHANGEp{%}} %\ignoreinshort{\GMCHANGEp{
% }} 
solvers deal with them efficiently, we can see that the usage of \NNFPlaisted{}
brings no advantage, and solving is uniformly slower than with \PlaistedCNF{}
or \TseitinCNF{}. This shows that our novel technique works specifically for
enumeration but not for solving, as expected.

\paragraph{Answer Set Programming.}
For the reason discussed above, we do not expect our new CNF encoding to bring
benefits to {Answer Set Programming (ASP)},
%\RSSIDENOTE{Riportato anche qui}
because answer sets represent \emph{total} truth assignments, since they
include atoms assigned to true, and other atoms are implicitly assumed to be
assigned to false. (We received confirmation from the clasp and clingo
developers that these tools do not support partial assignments. See the
discussion at
\url{https://github.com/potassco/clingo/issues/512#issuecomment-2223162682}.)
Notice that the concept of ``minimal'' in ASP refers to the minimality of the
set of atoms that are true in the answer set, and not to the minimality of the
set of atoms that are assigned a truth value, as in this paper.

\paragraph{(Weighted) Model Counting.}
For {(weighted) model counting} the \TseitinCNF{} encoding is currently the
most suitable, because it preserves the model count, whereas \PlaistedCNF{}
---and hence \NNFPlaisted{}--- does not.
%With \PlaistedCNF{()} or \PlaistedCNF{(\NNF{}) one should project away 
In order to preserve the model count with these encodings and with current
model counters, one should make them project away the \allB{} labels,
drastically worsening their performances.

Although investigating ad-hoc model counting techniques for
\NNFPlaisted{}-encoded formulas (e.g., by exploiting the polarity of label
variables) could be an interesting topic for future research, it would
definitely exceed the scope of this paper.

% \subsubsection{Comparing the CNF encodings for d-DNNF-based
% enumeration and counting}%
%\paragraph{d-DNNF-based (projected)       enumeration and counting.}  
\paragraph{d-DNNF compilation, (projected)       enumeration and counting.}%
% \subsubsection{Using \vinnfcnfpg{} with d-DNNF-based
%      enumeration and counting} 
\label{sec:experiments:d4}
We investigated the effect of the different CNF-izations on d-DNNF compilation, (projected) enumeration, and model counting.
We tested the tools \dfdecdnnf{}~\cite{lagniezLeveragingDecisionDNNFCompilation2024} for {projected} d-DNNF compilation and enumeration,
and \df{}~\cite{lagniezImprovedDecisionDNNFCompiler2017} for
projected d-DNNF model counting, on the Boolean benchmarks in
\sref{sec:experiments:benchmarks}.
Each formula was converted into CNF, and then compiled into d-DNNF projected on
the original atoms \allA{}.

%
%    We ran experiments on the Boolean benchmarks proposed in \sref{sec:experiments:benchmarks}.
The results are reported in \sref{appendix:d4:enum} and \sref{appendix:d4:counting} for
enumeration and model counting, respectively. We notice that, by using
\NNFPlaisted{} instead of \TseitinCNF{} or \PlaistedCNF{}, these tools
\begin{enumerate*}[(a)]
    \item\label{item:d4:analysis:models} show no substantial difference in terms of number and size of the assignments generated, and
    \item\label{item:d4:analysis:time} even worsen their time performance.
\end{enumerate*}
The former fact~\ref{item:d4:analysis:models} is due to the fact that \df{}$_p$ algorithm~\cite{lagniezRecursiveAlgorithmProjected2019} ---implemented by \df{} for projected compilation and counting--- branches first on
important atoms \allA, and only at the
end of the branch it tries to assign also the non-important atoms \allB{}  to be
projected away. Consequently,
the projected assignments include all or almost-all important
atoms, thus producing total or nearly-total projected assignments on \allA.
%\RSSIDENOTE{Spostata qui da \S6}
This is not surprising, since d-DNNF-based tools typically do not rely their
efficiency on the enumeration of short partial assignments, but rather on the
effective decomposition and caching of subproblems.
% (Other projected d-DNNF compilation algorithms behave similarly.)
% \GMSIDENOTE{Forse toglierei la frase tra parentesi, non ne sono
% sicurissimo.} % RS: DONE
%
The latter fact~\ref{item:d4:analysis:time} may be due to several reasons:
unlike with \TseitinCNF, with \PlaistedCNF{} the assignment of the \allA{}
values does not force the deterministic assignment of the \allB values by BCP,
causing extra useless search;
also, the NNF transformation up to doubles the number of the non-important
atoms \allB{} to be projected away;
furthermore, we conjecture that the duplication of labels \poslab{B_i} and
\neglab{B_i}, along with distinct encodings \poslab{\vi_i} and \neglab{\vi_i}
for \NNF{\vi_i} and \NNF{\neg\vi_i}, may affect the effectiveness of the
caching and partitioning mechanisms inside the \df{} d-DNNF compiler.

% Overall, these tools do not rely their
% efficiency on the enumeration of short partial assignments, rather on decomposition
% and caching of subproblems. 
Although in principle ad-hoc d-DNNF compilation strategies for d-DNNF-based
model counting and enumeration could be devised to exploit the properties of
our \NNFPlaisted{} encoding, investigating such techniques would exceed the
scope of this paper.

\section{Related work}
\label{sec:related-work}
% \begin{gmchange}
% \GMTODO{Other AllSAT applications (from~\cite{morgadoGoodLearningImplicit2005,todaImplementingEfficientAll2016,zhangAcceleratingAllSATComputation2020}): in data mining (e.g., frequent itemsets), network verification, image and preimage computation in unbounded Model Checking. AllSMT: predicate abstraction, in probabilistic reasoning (e.g., \#SMT, WMI).}
\paragraph{Applications of AllSAT and AllSMT}
SAT and \smt{} enumeration has an important role in a variety of applications,
ranging from artificial intelligence to formal verification. AllSAT and AllSMT,
mainly in their disjoint version, play a foundational role in several
frameworks for \emph{probabilistic reasoning}, such as model counting in
\smt{}~(\emph{\#\smt})~\cite{chistikovApproximateCountingSMT2017} and
\emph{Weighted Model Integration
    (WMI)}~\cite{belleProbabilisticInferenceHybrid2015,morettin-wmi-ijcar17,morettin-wmi-aij19,spallittaSMTbasedWeightedModel2022,spallittaEnhancingSMTbasedWeighted2024a}.
Specifically,
\#\smtlarat{}~\cite{maVolumeComputationBoolean2009,zhouEstimatingVolumeSolution2015,geComputingEstimatingVolume2018}
consists in summing up the volumes of the convex polytopes defined by each of
the \larat{}-satisfiable truth assignments propositionally satisfying a
\smtlarat{} formula, and has been employed for value estimation of
probabilistic programs~\cite{chistikovApproximateCountingSMT2017} and for
quantitative program analysis~\cite{liuProgramAnalysisQualitative2011}.\@ WMI
can be seen as a generalization of \#\smtlarat{} that additionally considers a
weight function $w$ that has to be integrated over each of such polytopes, and
has been used to perform inference in hybrid probabilistic models such as
Bayesian and Markov networks~\cite{belleProbabilisticInferenceHybrid2015} and
Density Estimation Trees~\cite{spallittaSMTbasedWeightedModel2022}.
%Hence, in these cases, it is essential to enumerate disjoint partial truth assignments that are as small and as few as possible.
%
%
AllSAT, both disjoint and non-disjoint, has applications also in \emph{data
    mining}, where the problem of frequent itemsets can be encoded into a
propositional formula whose satisfying assignments are the frequent
itemsets~\cite{boudaneSATbasedApproachMining2016,dlalaComparativeStudySATBased2016}.
%Specifically, in \#SMT(\larat) consists in summing up the volumes of the convex polytopes defined by each of the assignments, whereas in WMI some function $w$ has to be integrated over each of such polytopes. %Hence, in these cases, it is essential to enumerate disjoint partial \ignoreinlong{models}\ignoreinshort{\GMCHANGE{truth assignments}} that are as small and as few as possible.
% In the context of knowledge compilation (cite), AllSAT can be used to compile a formula into deterministic Decomposable Normal Form (d-DNNF) (cite), that has found applications, e.g., in planning (cite). 
It has also been used in the context of \emph{software testing} to generate a
suite of test inputs that should match a given
output~\cite{khurshidCaseEfficientSolution2004}, and in \emph{circuit design},
to convert a formula from CNF to
DNF~\cite{minatoFindingAllSimple1998,miltersenConvertingCNFDNF2005,bernasconiCompactDSOPPartial2013},
and for Static Timing Analysis to determine the inputs that can trigger a
timing violation in a circuit~\cite{friedAllSATCombinationalCircuits2023}.
AllSAT and AllSMT have also been applied in \emph{network verification} for
checking network reachability and for analyzing the correctness and consistency
of network connectivity
policies~\cite{lopes2013network,jayaraman2014automated,lopesCheckingBeliefsDynamic2015}.
Moreover, they have been used for computing the image and preimage of a given
set of states in \emph{unbounded model
    checking}~\cite{mcmillanApplyingSATMethods2002,grumbergMemoryEfficientAllSolutions2004,liNovelSATAllsolutions2004},
and they are also at the core of algorithms for computing \emph{predicate
    abstraction}, a concept widely used in formal verification for automatically
computing finite-state abstractions for systems with potentially infinite state
space~\cite{lahiriSymbolicApproachPredicate2003,clarkePredicateAbstractionANSIC2004,lahiriSMTTechniquesFast2006}.

\paragraph{AllSAT}
Most of the works on AllSAT have focused on the enumeration of satisfying
assignments for CNF formulas
(e.g.,~\cite{morgadoGoodLearningImplicit2005,huangUsingDPLLEfficient2004,yuAllSATUsingMinimal2014,todaBDDConstructionAll2015,todaImplementingEfficientAll2016,liangAllSATCCBoostingAllSAT2022,spallittaDisjointPartialEnumeration2024}),
with several efforts in developing efficient and effective techniques for
minimizing partial assignments
(e.g.,~\cite{raviMinimalAssignmentsBounded2004,morgadoGoodLearningImplicit2005,todaImplementingEfficientAll2016}).
The problem of minimizing truth assignments for Tseitin-encoded problems was
addressed in~\cite{iserMinimizingModelsTseitinEncoded2013a}. They propose to
%first simplify the formula by considering its original structure and the current model; then they 
make iterative calls to a SAT solver imposing increasingly tighter cardinality
constraints to obtain a minimal assignment. Whereas this approach can be used
to find a single short truth assignment, it can be very expensive, and thus it
is unsuitable for enumeration.

Other works have concentrated on the enumeration of satisfying assignments for
combinatorial circuits, exploiting the structural information of the circuits
to minimize the partial assignments over the input variables
(e.g.,~\cite{jinEfficientConflictAnalysis2005,jin2005prime,tibebuAugmentingAllSolution2018,friedAllSATCombinationalCircuits2023,friedEntailingGeneralizationBoosts2024}).
%
% \begin{gmchange}
In \cite{lagniezLeveragingDecisionDNNFCompilation2024}, the authors proposed a tool to enumerate disjoint partial satisfying assignments of formulas in decomposable, deterministic NNF (d-DNNF).
The tool takes as input a d-DNNF formula compiled with \df{}~\cite{lagniezImprovedDecisionDNNFCompiler2017}, and then efficiently traverses it to enumerate the partial assignments.
However, projected d-DNNFs typically do not allow for finding short partial assignments~\cite{lagniezRecursiveAlgorithmProjected2019}, and hence preprocessing the formula with our encoding does not bring any benefit, as confirmed by the experimental evaluation in~\sref{sec:experiments:others}.
% Most d-DNNF compilers, such as \textsc{D4}~\cite{lagniezImprovedDecisionDNNFCompiler2017} used in their evaluation, compile formulas from CNF to d-DNNF, possibly projected over a set of relevant atoms. Hence, in principle, this tool can be used to perform AllSAT on any formula, by first converting it into CNF, and then compiling it into a d-DNNF projected over the original atoms, and finally enumerating the assignments. 
% This three-step approach, however, would not be very effective, due to the compilation strategy of \textsc{D4} for projected d-DNNFs~\cite{lagniezRecursiveAlgorithmProjected2019}.
% \textsc{D4} uses a DPLL-like compilation strategy, where decisions are made on relevant atoms first. Whenever no relevant atom occurs in the current formula, a SAT solver is invoked to see if the assignment can be extended to non-relevant atoms. This search strategy prevents finding short, and even more so, minimal partial assignments~\cite{lagniezRecursiveAlgorithmProjected2019}. 
% \GMCHANGE{This analysis was confirmed also by experimental evidence, which we report in Appendix~\ref{appendix:d4:enum}.
% }
%Moreover, using different CNF encodings would not help, since the SAT solver is invoked only when no important atom occurs in the formula. % Other d-DNNF compilers, e.g., \textsc{c2d}~\cite{darwicheCompilerDeterministicDecomposable2002} and \textsc{DSharp}~\cite{muiseDsharpFastDDNNF2012}, also work on CNF input and implement similar strategies, and would likely suffer from the same limitations. 
% \end{gmchange}

A problem closely related to AllSAT is that of finding all the prime implicants
of a formula
(e.g.,~\cite{previtiPrimeCompilationNonClausal,jabbourEnumeratingPrimeImplicants2014,luoEfficientTwophaseMethod2021}).
AllSAT is a much simpler problem, and can be viewed as the problem of finding a
subset of implicants (i.e., partial satisfying assignments), not necessarily
prime (i.e., minimal), which covers all the total satisfying assignments of the
formula.
%not-necessarily-prime implicant cover for the formula.

\paragraph{Projected AllSAT}
Projected enumeration, i.e., enumeration of satisfying assignments over a set
of relevant atoms, has been studied mainly for CNF formulas, e.g.,
in~\cite{grumbergMemoryEfficientAllSolutions2004,liNovelSATAllsolutions2004,lahiriSMTTechniquesFast2006,todaExploitingFunctionalDependencies2017,spallittaDisjointProjectedEnumeration2025}.
The ambiguity of the notion of ``satisfiability by partial assignment'' for
non-CNF and existentially-quantified formulas has been raised
in~\cite{sebastianiAreYouSatisfied2020,mohleFourFlavorsEntailment2020,sebastianiEntailmentVsVerification2025},
highlighting the difference between ``evaluation to true'', which is simpler to
check and typically used by SAT solvers, and ``logical entailment'', which
allows for producing shorter assignments. The approach based on dual reasoning
in~\cite{mohleDualizingProjectedModel2018,mohleEnumeratingShortProjected2025},
although able to detect logical entailment and thus to produce shorter partial
assignments, is very inefficient even for small formulas. %\GMCHANGEp{, %}. 

\paragraph{AllSMT}
The literature on AllSMT is very limited, and AllSMT algorithms are highly
based on AllSAT techniques and tools. For instance,
\mathsatfive{}~\cite{mathsat5_tacas13} implements an AllSMT functionality based
on the procedure in~\cite{lahiriSMTTechniquesFast2006}. A similar procedure
has been described in~\cite{phanAllSolutionSatisfiabilityModulo2015}. An
AllSMT algorithm based on chronological backtracking has been recently proposed
in~\cite{spallittaDisjointProjectedEnumeration2025}.

\paragraph{Enumeration in related areas}
%    Enumeration of satisfying assignments is also relevant in neighboring areas.   
In Answer Set Programming (ASP), answer sets can be seen as truth assignments
satisfying a logic program. ASP programs can be translated into CNF formulas
and vice versa, so as to have a one-to-one correspondence between satisfying
assignments and answer
sets~\cite{clarkNegationFailure1978,linASSATComputingAnswer2004}. Answer sets,
however, correspond to \emph{total} truth assignments. %Modern ASP solvers are indeed based on AllSAT solvers (e.g.,~\cite{giunchigliaAnswerSetProgramming2006,gebserConflictDrivenAnswerSet2007,gebserSolutionEnumerationProjected2009}). %\footnote{
%     See also \url{https://github.com/potassco/clingo/issues/512\#issuecomment-2223162682}.
% }
Thus, our work is not directly applicable to ASP. %, as no minimization of partial assignments is performed. 
%In principle, if the notion of partial assignment were to be extended to ASP, our approach could be used to enumerate partial answer sets.

Propositional Model Counting (\#SAT) is the task of \emph{counting} the number
of satisfying assignments of a propositional formula. Counting is simpler than
enumeration, since only the number of truth assignments is relevant, but a
complete exploration of the space of solutions is still needed. (This is
different from \#SMT, where each truth assignment is ``consumed'' to compute a
volume.) Efficient algorithms have been developed for \#SAT on CNF formulas,
mostly based either on DPLL-like exhaustive search, or on Knowledge Compilation
(see~\cite{gomesModelCounting2021}).
\ignore{% RS: spostato in sez 5
    Typically, they do not rely their efficiency on the enumeration of \emph{short} partial assignments, but rather on decomposition and caching of subproblems, so that the usage of our encoding is not expected to bring any benefit, as confirmed by the experimental evaluation in~\sref{sec:experiments:others}.}
%Hence, we do not expect \#SAT algorithms to benefit from our encoding, as confirmed by the experimental evaluation in~\sref{sec:experiments:d4}.
% We report in Appendix~\ref{appendix:d4:counting} an experimental evaluation of our encoding for \#SAT with \df~\cite{lagniezImprovedDecisionDNNFCompiler2017}.
Notably,~\cite{azizExistsSATProjected2015} proposed the model counter
\textsc{\#clasp} based on enumeration of minimal partial assignment, which
would likely benefit from our encoding. Unfortunately, we did not receive an
answer from the authors about the availability of the tool.
%
% The role of CNF encodings in \#SAT has been analyze, e.g., by~\citeA{kuiterTseitinNotTseitin2022}.
%This said, we remark that the Tseitin CNF preserves the model count, and so it is the most obvious choice for \#SAT. The Plaisted and Greenbaum CNF, instead, would require performing \emph{projected} model counting (\#$\exists$SAT)~\cite{azizExistsSATProjected2015} on the original atoms, which is typically harder than \#SAT. 

\paragraph{The role of CNF encodings}
Although most AllSAT solvers assume the formulas to be in CNF, little or no
work has been done to investigate the impact of the different CNF encodings on
their effectiveness and efficiency. The role of the CNF-ization has been widely
studied for SAT solving
(e.g.,~\cite{boydelatourOptimalityResultClause1992,jacksonClauseFormConversions2005,bjorkSuccessfulSATEncoding2009,kuiterTseitinNotTseitin2022})
and in a recent work also for propositional model counting in the context of
feature model analysis~\cite{kuiterTseitinNotTseitin2022}.

% \subparagraph{Dual-rail encoding.}%
% \label{sec:dualrail}
The idea of using two different variables to label the positive and negative occurrences of a sub-formula shares some similarities with the so-called \emph{dual-rail} encoding~\cite{bryantCOSMOSCompiledSimulator1987,palopoliAlgorithmsSelectiveEnumeration1999}, which has been shown to be successful for prime implicant enumeration~\cite{%palopoliAlgorithmsSelectiveEnumeration1999,
    previtiPrimeCompilationNonClausal,luoEfficientTwophaseMethod2021}, and recently also for AllSAT on combinatorial circuits~\cite{friedAllSATCombinationalCircuits2023}.

Given a CNF formula, the dual-rail encoding maps each atom $A_i$ into a pair of
dual-rail atoms $\pair{\poslab{A_i}}{\neglab{A_i}}$, s.t.\ $A_i,\neg{}A_i$ are
substituted with $\poslab{A_i},\neglab{A_i}$, respectively, and
$(\neg\poslab{A_i}\vee\neg\neglab{A_i})$ is conjoined with the formula. The
resulting CNF formula is equisatisfiable to the original one, and every total
assignment $\eta$ satisfying the dual-rail encoding corresponds to a partial
assignment $\mu$ satisfying the original formula. Such $\mu$ assigns $\mu(A_i)$
to $\top$ or $\bot$ if the pair $\pair{\eta(\poslab{A_i})}{\eta(\neglab{A_i})}$
is $\pair{\top}{\bot}$ or $\pair{\bot}{\top}$, respectively, and leaves it
unassigned if
$\pair{\eta(\poslab{A_i})}{\eta(\neglab{A_i})}=\pair{\bot}{\bot}$. Short
partial assignments can be obtained by maximizing the number of pairs assigned
to $\pair{\bot}{\bot}$, which can be done either exactly by solving a MaxSAT
problem, or heuristically by assigning negative value first to decision
atoms~\cite{friedAllSATCombinationalCircuits2023}.

% \begin{gmchangep}

Using the dual-rail encoding for enumeration, however, requires ad-hoc
enumeration algorithms, since the solver must take into account the
three-valued semantics of truth assignments over dual-rail atoms when
enumerating the assignments~\cite{friedAllSATCombinationalCircuits2023}. Our
contribution, instead, focuses on CNF-ization approaches that can be used in
combination with any enumeration algorithm matching the properties described
in~\sref{sec:background:allsat}.
Also, comparing the $\NNFPlaisted{}$ encoding with the dual-rail encoding, we
observe that the former duplicates only the label atoms in $\allB{}$,
introducing fewer atoms than the dual-rail encoding. Moreover, because of this,
the clauses in the form $(\neg\poslab{B_i}\vee\neg\neglab{B_i})$ are not needed
for correctness but only for efficiency, and can in principle be omitted. In
the dual-rail encoding, instead, their presence is essential to prevent the
illegal assignment $\eta(\poslab{B_i})=\eta(\neglab{B_i})=\top$.

\section{Conclusions and future work}
\label{sec:conclusions}
We have presented a theoretical and empirical analysis of the impact of
different CNF-ization approaches on SAT and SMT enumeration, both disjoint and
non-disjoint. We have shown how the most popular transformations conceived for
SAT and SMT solving, namely the Tseitin and the Plaisted and Greenbaum
CNF-izations, prevent the solver from producing short partial assignments, thus
seriously affecting the effectiveness of the enumeration. To overcome this
limitation, we have proposed to preprocess the formula by converting it into
NNF before applying the Plaisted and Greenbaum transformation. We have shown,
both theoretically and empirically, that the latter approach can fully overcome
the problem and can drastically reduce both the number of partial assignments
and the execution time. %\ignoreinshort{\GMCHANGE{%}}%\ignoreinshort{\GMCHANGE{, %}}. %\ignoreinshort{\GMCHANGEp{%}} 

% we plan to further investigate the
% impact of CNF conversion also on disjoint SMT enumeration. We expect that in
% this domain the impact can be even more relevant, since in SMT
% multiple instances of the same theory atoms are typically rarer than for atoms in
% the Boolean case. Also, disjoint SMT enumeration has a fundamental role in Weighted Model Integration~\cite{morettin-wmi-ijcar17,morettin-wmi-aij19,spallittaSMTbasedWeightedModel2022}, an important framework for probabilistic inference in hybrid domains. Hence, we believe that our contribution can have a great impact on this application, where non-CNF formulas occur frequently.
% Finally, we think that work should be done to understand the impact on enumeration with repetitions, i.e.\ where models may not be disjoint, for instance in Predicate Abstraction~\cite{lahiriSMTTechniquesFast2006}. %Moreover, there is an alternative definition of partial assignment satisfiability that is based on the notion of \emph{entailment}~\cite{sebastianiAreYouSatisfied2020,mohleFourFlavorsEntailment2020}. Understanding the impact of the CNF conversion on solvers that use this definition is an interesting direction.

This work opens an interesting research avenue: investigate the role of
CNF-ization in neighbor fields as d-DNNF compilation and model counting,
possibly adapting d-DNNF compilers and model counters so that to exploit
different forms of CNF-izations.

% \acks{The authors wish to thank Hans-Martin Adorf, Don Rosenthal, 
% Richard Franier, Peter Cheeseman and Monte Zweben for their assistance
% and advice.  We also thank Ron Musick and our anonymous reviewers for
% their comments.  The Space Telescope Science Institute is operated by
% the Association of Universities for Research in Astronomy for NASA.
% }

% \acks{We acknowledge the support of the MUR PNRR project FAIR --
%   Future AI Research (PE00000013), under the NRRP MUR program funded
%   by the NextGenerationEU. The work was partially supported by the
%   project ``AI@TN'' funded by the Autonomous Province of Trento.
%   This research was partially supported by TAILOR, a project funded
%   by the EU Horizon 2020 research and innovation program under GA No 952215.
% }
\FloatBarrier

% \ignoreinshort{
\newpage
\renewcommand{\theHsection}{A\arabic{section}}
\appendix
% \begin{gmchange}
\section{Proofs}
\subsection{Proof for~\cref{th:nnfdaglinear}  in~\cref{sec:background:propositional-logic}}\label{sec:proofnnfdaglinear}%
\begin{figure}[th]
    \centering
    \begin{subfigure}[t]{0.48\textwidth}
        \begin{tikzpicture}[remember picture,->,auto,node distance=1.8cm,semithick]
  \begin{footnotesize}
    % First layer
    \node[circle, draw] (A) {$\wedge$};
    \node[circle, draw] (B) [right of=A] {$\vee$};

    % Second layer
    \node[circle, draw] (D) [below=.4cm of A] {\phantom{$\vee$}};
    \node[circle, draw] (C) [left of=D] {\phantom{$\vee$}};
    \node[circle, draw] (E) [right of=D] {\phantom{$\vee$}};
    \node[circle, draw] (F) [right of=E] {\phantom{$\vee$}};

    % % Edges
    \path (A) edge node {} (C)
    (A) edge node {} (E)
    (B) edge node {} (D)
    (B) edge node {} (F);

    % Text labels
    \node[above left=.4cm and -.8cm of A] at (A) {$\NNF{\vi_1\wedge\vi_2}$};
    \node[above right=.4cm and -.8cm of B] at (B) {$\NNF{\neg(\vi_1\wedge\vi_2)}$};
    \node[below=.4cm of C]                at (C) {$\NNF{\vi_1}$};
    \node[below=.4cm of D]                at (D) {$\NNF{\neg\vi_1}$};
    \node[below=.4cm of E]                at (E) {$\NNF{\vi_2}$};
    \node[below=.4cm of F]                at (F) {$\NNF{\neg\vi_2}$};

    \node[draw,densely dotted,red,fit=(A) (B)] {};
    \node[draw,densely dotted,red,fit=(C) (D)] {};
    \node[draw,densely dotted,red,fit=(E) (F)] {};
  \end{footnotesize}
\end{tikzpicture}
        \caption{Case $\vi\defas\vi_1\wedge\vi_2$.}%
        \label{fig:2root:and}
    \end{subfigure}
    \begin{subfigure}[t]{0.48\textwidth}
        \begin{tikzpicture}[remember picture,->,auto,node distance=1.8cm,semithick]
  \begin{footnotesize}
    % First layer
    \node[circle, draw] (A) {$\wedge$};
    \node[circle, draw] (B) [right of=A] {$\wedge$};

    % Second layer
    \node[circle, draw] (D) [below=.4cm of A] {$\vee$};
    \node[circle, draw] (C) [left of=D] {$\vee$};
    \node[circle, draw] (E) [right of=D] {$\vee$};
    \node[circle, draw] (F) [right of=E] {$\vee$};

    % % Third layer
    \node[circle, draw] (G) [below=.6cm of C] {\phantom{$\vee$}};
    \node[circle, draw] (H) [right of=G] {\phantom{$\vee$}};
    \node[circle, draw] (I) [right of=H] {\phantom{$\vee$}};
    \node[circle, draw] (J) [right of=I] {\phantom{$\vee$}};

    % % Edges
    \path (A) edge node {} (C)
    (A) edge node {} (D)
    (B) edge node {} (E)
    (B) edge node {} (F)
    (C) edge node {} (G)
    (C) edge node {} (J)
    (D) edge node {} (H)
    (D) edge node {} (I)
    (E) edge node {} (G)
    (E) edge node {} (I)
    (F) edge node {} (H)
    (F) edge node {} (J);

    % Text labels
    \node[above left=.4cm and -.8cm of A] at (A) {$\NNF{\vi_1\iff\vi_2}$};
    \node[above right=.4cm and -.8cm of B] at (B) {$\NNF{\neg(\vi_1\iff\vi_2)}$};
    \node[below=.4cm of G]                at (G) {$\NNF{\vi_1}$};
    \node[below=.4cm of H]                at (H) {$\NNF{\neg\vi_1}$};
    \node[below=.4cm of I]                at (I) {$\NNF{\vi_2}$};
    \node[below=.4cm of J]                at (J) {$\NNF{\neg\vi_2}$};

    \node[draw,densely dotted,red,fit=(A) (B)] {};
    \node[draw,densely dotted,red,fit=(G) (H)] {};
    \node[draw,densely dotted,red,fit=(I) (J)] {};
  \end{footnotesize}
\end{tikzpicture}
        \caption{Case $\vi\defas\vi_1\iff\vi_2$.}%
        \label{fig:2root:iff}
    \end{subfigure}
    %    \vspace{-.3cm}
    \caption{2-root DAGs for the pair $\langle{\NNF{\vi}},\NNF{\neg\vi}\rangle$ for $\vi\defas\vi_1\wedge\vi_2$ and $\vi\defas\vi_1\iff\vi_2$ (those for $\vi_1\vee\vi_2$ and $\vi_1\imp\vi_2$ are similar to that of $\vi_1\wedge\vi_2$).}%
    \label{fig:2root}
\end{figure}
%We present the proof for~\cref{th:nnfdaglinear} in~\cref{sec:background:propositional-logic}.
\begin{proof}
    %The NNF DAG that represents 
    \NNF{\vi} is a sub-graph of the 2-root DAG for the
    pair $\langle{\NNF{\vi}},\NNF{\neg\vi}\rangle$. %\GMCHANGE{
    %}. 
    We prove that the latter grows linearly in size w.r.t.\ $\vi$ by reasoning
    inductively on the structure of \vi. (The size of a DAG is denoted with
    ``$|\dots|$''.)
    %i.e.\ ``|\dots|'' is \#nodes(\dots)+\#arcs(\dots).
    %    We prove it by induction on the structure of \vi.
    \begin{description}
        \item[if $\vi$ is an atom:] $\NNF{\vi}=\vi$ and $\NNF{\neg\vi}=\neg\vi$, so that:
              % \begin{equation}\label{eq:recursizebase}
              $|\langle\NNF{\vi},\NNF{\neg\vi}\rangle|=2$.
              % \end{equation}
        \item[if $\vi\defas\neg\vi_1$:] we assume by induction that
              %we have computed
              $|\tuple{\NNF{\vi_1},\NNF{\neg\vi_1}}|$ is linear in $|\vi_1|$. \\
              Then
              $\tuple{\NNF{\vi},\NNF{\neg\vi}}=\tuple{\NNF{\neg\vi_1},\NNF{\vi_1}}$
              (i.e., we just invert the order), s.t.
              \begin{equation}\label{eq:recursizenot}
                  |\tuple{\NNF{\vi},\NNF{\neg\vi}}|=|\tuple{\NNF{\neg\vi_1},\NNF{\vi_1}}|=|\tuple{\NNF{\vi_1},\NNF{\neg\vi_1}}|
              \end{equation}
              % $|\langle\NNF{\neg\vi_1},\NNF{\vi_1}\rangle|=|\langle\NNF{\vi_1},\NNF{\neg\vi_1}\rangle|$.
              %        \item[if $\vi\defas(\vi_1\bowtie\vi_2)$ s.t.\ $\bowtie\ \in\set{\wedge,\vee,\imp,\iff}$:]
        \item[if $\vi\defas(\vi_1\bowtie\vi_2)$ s.t.\ $\bowtie\ \in\set{\wedge,\iff}$:] we
              assume by induction that
              %we have computed
              $|\tuple{\NNF{\vi_1},\NNF{\neg\vi_1}}|$ and
              $|\tuple{\NNF{\vi_2},\NNF{\neg\vi_2}}|$ are linear in $|\vi_1|$ and $|\vi_2|$, respectively.
              (See \Cref{fig:2root}):
              % Specifically:
              \begin{description}
                  \item[if $\bowtie$ is $\wedge$:] the DAGs for $\NNF{\vi}, \NNF{\neg\vi}$ add 2
                        ``$\wedge/\vee$'' nodes and $2+2$ arcs:\\ $
                                \NNF{\vi_1\wedge\vi_2} \defas\NNF{\vi_1}\wedge\NNF{\vi_2}$ and
                                $\NNF{\neg(\vi_1\wedge\vi_2)}                     \defas\NNF{\neg\vi_1}\vee\NNF{\neg\vi_2}$ thus
                        {\setlength{\mathindent}{-.4cm}%
                                \begin{equation}\label{eq:recursizeand}
                                    |\tuple{\NNF{\vi},\NNF{\neg\vi}}| =
                                    6+|\tuple{\NNF{\vi_1},\NNF{\neg\vi_1}}|+
                                    |\tuple{\NNF{\vi_2},\NNF{\neg\vi_2}}|
                                \end{equation}%
                            }
                  \item[if $\bowtie$ is $\iff$:] the DAGs for $\NNF{\vi}, \NNF{\neg\vi}$ %share the sub-DAGs for
                        %$\NNF{\vi_1}$, $\NNF{\neg\vi_1}$, $\NNF{\vi_2}$, $\NNF{\neg\vi_2}$, adding
                        add 3+3 ``$\wedge$''/``$\vee$'' nodes and $6+6$ arcs:\\ $
                            \begin{aligned}
                                \NNF{\pos\phantom{(}\vi_1\iff\vi_2\phantom{)}} & \defas(\NNF{\neg\vi_1}\vee\NNF{\vi_2})\wedge(\NNF{\pos\vi_1}\vee \NNF{\neg\vi_2})\text{ and}  \\
                                \NNF{\neg(\vi_1\iff\vi_2)}                     & \defas(\NNF{\pos\vi_1}\vee\NNF{\vi_2})\wedge(\NNF{\neg\vi_1}\vee\NNF{\neg\vi_2})\text{, thus}
                            \end{aligned}
                        $
                        {\setlength{\mathindent}{-.4cm}%
                                \begin{equation}\label{eq:recursizeiff}
                                    |\tuple{\NNF{\vi},\NNF{\neg\vi}}| =
                                    18+|\tuple{\NNF{\vi_1},\NNF{\neg\vi_1}}|+
                                    |\tuple{\NNF{\vi_2},\NNF{\neg\vi_2}}|
                                \end{equation}%
                            }
              \end{description}
              % Then:
              % \begin{equation}\label{eq:recursize}
              %     |\langle \NNF{\vi},\NNF{\neg\vi}\rangle| \leq
              %     18+|\langle \NNF{\vi_1},\NNF{\neg\vi_1}\rangle|+|\langle
              %     \NNF{\vi_2},\NNF{\neg\vi_2}\rangle|:
              % \end{equation}
        \item[if $\vi\defas(\vi_1\bowtie\vi_2)$ s.t.\ $\bowtie\ \in\set{\vee,\imp}$:]
              %          \item[if $\bowtie\ \in \set{\vee, \imp}$:]
              these cases can be reduced to the previous cases, since
              \begin{displaymath}
                  \NNF{\vi_1\vee\vi_2}=
                  \NNF{\neg(\neg\vi_1\wedge\neg\vi_2)}\text{ and }\NNF{\vi_1\imp\vi_2}=
                  \NNF{\neg(\vi_1\wedge\neg\vi_2)}.
              \end{displaymath}
    \end{description}
    Therefore, from~\eqref{eq:recursizenot},~\eqref{eq:recursizeand} and~\eqref{eq:recursizeiff} we have
    that $|\langle \NNF{\vi},\NNF{\neg\vi}\rangle|$ is $O(|\vi|)$.
    %18\cdot|\vi|$.}
    %
    % RS; gia' detto nel main part del paper
    % Intuitively, we only need at most 2 nodes for each sub-formula $\vi_i$ of $\vi$,
    % representing $\NNF{\vi_i}$ and $\NNF{\neg\vi_i}$ for positive and negative
    % occurrences of $\vi_i$ respectively. These nodes are shared among up
    % to exponentially-many branches generated by expanding the nested iffs.
\end{proof}
%\end{gmchange}

%%%%%%%%%%%%%%%%%%%%%%%%%%%%%%%%%%%%%%%%%%%%%%%%%%%%%%%%%%%%%
%%%
%%%%%%%%%%%%%%%%%%%%%%%%%%%%%%%%%%%%%%%%%%%%%%%%%%%%%%%%%%%%%
\newpage
%\begin{gmchange}
\newcommand{\muof}[1]{\residual{#1}{\mu}}
\subsection{Proof for~\cref{th:munnf} in~\cref{sec:background:propositional-logic}}\label{sec:proofmunnf}%

\begin{figure}[th]
  % \footnotesize
  \small
  % \centering
  %  \begin{tabular}{cc}
  %  \begin{minipage}[t]{0.49\textwidth}
  \centering
  $
    \begin{array}{||c||l|l|l|l|l|l|l|l|l||}
      \hline
      \muof{\vi_1}                   & \top & \top & \top & \any & \any & \any & \bot & \bot & \bot \\
      \muof{\vi_2}                   & \top & \any & \bot & \top & \any & \bot & \top & \any & \bot \\
      \hline
      \neg(\muof{\vi_1})             & \bot & \bot & \bot & \any & \any & \any & \top & \top & \top \\
      \muof{\vi_1}\wedge\muof{\vi_2} & \top & \any & \bot & \any & \any & \bot & \bot & \bot & \bot \\
      \muof{\vi_1}\vee\muof{\vi_2}   & \top & \top & \top & \top & \any & \any & \top & \any & \bot \\
      \muof{\vi_1}\imp\muof{\vi_2}   & \top & \any & \bot & \top & \any & \any & \top & \top & \top \\
      \muof{\vi_1}\iff\muof{\vi_2}   & \top & \any & \bot & \any & \any & \any & \bot & \any & \top \\
      \hline
    \end{array}
  $
  \caption{
    Three-value-semantics of $\muof{\vi}$ in terms of \set{\top ,\bot ,\any }
    (``true'', ``false'', ``unknown''). \newline
    % \ignoreinlong{The definition of
    %   $\muof{\vi_1\bowtie\vi_2}$ s.t. $\bowtie\ \in\set{\vee,\imp,\limp,\iff}$
    %   follows straighforwardly.}
  }\label{fig:threeval}
  %  \end{minipage}
  \ignore{
    &
    \begin{minipage}[t]{0.4\textwidth}
      %\centering
      $
        \begin{array}{|cl|cl|}
          \hline
          \neg\top                     & \Rightarrow \bot & \neg \bot                     & \Rightarrow \top    \\
          \top\wedge\vi, \vi\wedge\top & \Rightarrow \vi  & \bot\wedge\vi,  \vi\wedge\bot & \Rightarrow \bot    \\
          \top\vee\vi, \vi\vee\top     & \Rightarrow \top & \bot\vee\vi, \vi\vee\bot      & \Rightarrow \vi     \\
          \top\imp\vi                  & \Rightarrow \vi  & \bot\imp\vi                   & \Rightarrow \top    \\
          \vi\imp\top                  & \Rightarrow \top & \vi\imp\bot                   & \Rightarrow \neg\vi \\
          \top\iff\vi, \vi\iff\top     & \Rightarrow \vi  & \bot\iff\vi, \vi\iff\bot      & \Rightarrow \neg\vi \\
                                       &                  &                               &                     \\
          \hline
        \end{array}
      $\\
      \caption{\label{fig:boolprop}
        Propagation of truth values
        through the Boolean connectives.}
    \end{minipage}%
  }
  %\end{tabular}
\end{figure}

In the following the symbol ``\any'' denotes any formula which is not
in $\set{\top,\bot}$. Following~\cite{sebastianiAreYouSatisfied2020,sebastianiEntailmentVsVerification2025}, we adopt a 3-value
semantics for residuals $\residual{\vi}{\mu}\in\set{\top,\bot,\any}$,
so that ``$\residual{\vi}{\mu}=\any$'' means
``$\residual{\vi}{\mu}\not\in\set{\top,\bot}$'' and
``$\residual{\vi_1}{\mu}=\residual{\vi_2}{\mu}$'' means that the two
residuals $\residual{\vi_1}{\mu}$ and $\residual{\vi_2}{\mu}$ are either both $\top$, or both $\bot$, or neither is in
\set{\top,\bot}. (Notice that, in the latter case,
$\residual{\vi_1}{\mu}=\residual{\vi_2}{\mu}$ even if
$\residual{\vi_1}{\mu}$ and $\residual{\vi_2}{\mu}$ are different
formulas.)
We extend the definition to tuples in an obvious way:
$\tuple{\residual{\vi_1}{\mu},\dots,\residual{\vi_n}{\mu}}=
    \tuple{\residual{\psi_1}{\mu},\dots,\residual{\psi_n}{\mu}}$ iff
$\residual{\vi_i}{\mu}=\residual{\psi_i}{\mu}$ for each $i\in[1\dots n]$.

The 3-value semantics of the Boolean operators is reported for convenience in
Figure~\ref{fig:threeval}. As a straightforward consequence of the above
semantics, we have that:
\begin{eqnarray*}
    \residual{(\neg\vi)}{\mu}&=&\neg(\residual{\vi}{\mu})\ \  \mbox{(hereafter
        simply ``$\residual{\neg\vi}{\mu}$'')}\\
    \residual{(\vi_1\bowtie\vi_2)}{\mu}&=&
    \residual{\vi_1}{\mu}\bowtie\residual{\vi_2}{\mu}\ \  \mbox{for $\bowtie\ \in \set{\wedge,\vee,\imp,\iff}$}.
\end{eqnarray*}
Also, the usual transformations apply:
$\muof{\neg(\vi_1\wedge\vi_2)}=\muof{\neg\vi_1}\vee\muof{\neg\vi_2}$,
$\muof{\neg(\vi_1\vee\vi_2)}=\muof{\neg\vi_1}\wedge\muof{\neg\vi_2}$,
$\muof{(\vi_1\iff\vi_2)}=
    (\neg\muof{\vi_1}\vee\muof{\vi_2})\wedge(\muof{\vi_1}\vee\neg\muof{\vi_2})$
and
$\muof{\neg(\vi_1\iff\vi_2)}=
    (\muof{\vi_1}\vee\muof{\vi_2})\wedge(\neg\muof{\vi_1}\vee\neg\muof{\vi_2})$.
For convenience, sometimes we denote as $\bar{v}$ the complement of
$v\in\set{\top,\bot,\any}$, i.e., $\bar{v}\defas\neg v$, so that
$\bar{\top}=\bot, \bar{\bot}=\top, \bar{\any}=\any$.

We prove the following lemma, from which \cref{th:munnf}
in~\sref{sec:background:propositional-logic} follows directly.

% \ignore{\begin{property}%
%     \label{th:munnf2}
%     Consider a Boolean formula $\vi(\allA)$, and let $\NNF{\vi}$ be its NNF DAG. 
%     Consider a partial assignment $\muA$ on $\allA$. Then
%     $\residual{\vi}{\muA}=v$ iff $\residual{\NNF{\vi}}{\muA}=v$ for $v\in\set{\top,\bot,\any}$.
% \end{property}}

\begin{lemma}
    Consider a %\ignoreinlong{Boolean }
    formula $\vi$, %and let $\NNF{\vi}$ be its NNF DAG.\@ 
    and a partial assignment $\mu$. % on $\allA$.
    % Let $\langle\NNF{\vi},\NNF{\neg\vi}\rangle$ be the 2-root DAG as
    % in~\cref{sec:proofnnfdaglinear}.
    Then:
    \begin{equation}
        \label{eq:nnfeq}
        \pair{\residual{\vi}{\mu}}{\residual{\neg\vi}{\mu}}=\pair{\residual{\NNF{\vi}}{\mu}}{\residual{\NNF{\neg\vi}}{\mu}}.
    \end{equation}

\end{lemma}

\begin{proof}%
    % RS: spostato sopra
    % We first notice that $\residual{\vi}{\mu}=v$ iff
    % $\residual{\neg\vi}{\mu}=\bar{v}$, where $\bar{v}$ is the
    % complement of $v$, i.e.\ $\bar{\top}=\bot, \bar{\bot}=\top,
    % \bar{\any}=\any$.
    As in~\sref{sec:proofnnfdaglinear}, we prove this fact by reasoning on the
    2-root DAG for the pair $\langle\NNF{\vi},\NNF{\neg\vi}\rangle$. Specifically,
    we prove~\eqref{eq:nnfeq} by induction on the structure of $\vi$.

    %% RS: spostato a sopra
    % (With a little abuse of notation, we say that $\residual{\vi}{\mu}=\residual{\psi}{\mu}$ even in the case when $\residual{\vi}{\mu}$ and $\residual{\psi}{\mu}$ are different formulas s.t.\ $\residual{\vi}{\mu},\residual{\psi}{\mu}\notin\set{\top,\bot}$, meaning that they both fall in the case $\any$.)
    \begin{description}
        \item[if $\vi$ is an atom:] then $\NNF{\vi}=\vi$ and $\NNF{\neg\vi}=\neg\vi$, so that
              $\residual{\vi}{\mu}=\residual{\NNF{\vi}}{\mu}$ and
              $\residual{\neg\vi}{\mu}=\residual{\NNF{\neg\vi}}{\mu}$.
              %------------------- NOT -------------------
        \item[if $\vi\defas\neg\vi_1$:] we assume by induction that
              $\pair{\residual{\vi_1}{\mu}}{\residual{\neg\vi_1}{\mu}}=\pair{\residual{\NNF{\vi_1}}{\mu}}{\residual{\NNF{\neg\vi_1}}{\mu}}$.
                  {Let $v$ be s.t.
                      $\pair{\residual{\vi}{\mu}}{\residual{\neg\vi}{\mu}}=\pair{v}{\bar{v}}$. Then:}
              \[
                  \begin{aligned}
                      \pair{\residual{\neg\vi_1}{\mu}}{\residual{\pos\vi_1}{\mu}}=\pair{v}{\bar{v}}             & \Iff    \\
                      \pair{\residual{\pos\vi_1}{\mu}}{\residual{\neg\vi_1}{\mu}}=\pair{\bar{v}}{v}             & \Iffind \\
                      \pair{\residual{\NNF{\pos\vi_1}}{\mu}}{\residual{\NNF{\neg\vi_1}}{\mu}}=\pair{\bar{v}}{v} & \Iff    \\
                      \pair{\residual{\NNF{\neg\vi_1}}{\mu}}{\residual{\NNF{\pos\vi_1}}{\mu}}=\pair{v}{\bar{v}} &
                  \end{aligned}
              \]

              % then $\langle\residual{\neg\vi_1}{\mu},\residual{\vi_1}{\mu}\rangle=\langle v,\bar{v}\rangle$ iff $\langle\residual{\vi_1}{\mu},\residual{\neg\vi_1}{\mu}\rangle=\langle\bar{v},v\rangle$. By induction hypothesis, this holds iff $\langle\residual{\NNF{\vi_1}}{\mu},\residual{\NNF{\neg\vi_1}}{\mu}\rangle=\langle\bar{v},v\rangle$, and so iff $\langle\residual{\NNF{\vi}}{\mu},\residual{\NNF{\neg\vi}}{\mu}\rangle=\langle v,\bar{v}\rangle$;%, hence $\residual{\NNF{\vi}}{\mu}=v$ and $\residual{\NNF{\neg\vi}}{\mu}=\bar{v}$. 
              %        \item[if $\vi\defas(\vi_1\bowtie\vi_2)$:] s.t.\ $\bowtie~\in\set{\wedge,\vee,\imp,\iff}$. We assume by induction that
        \item[if $\vi\defas(\vi_1\bowtie\vi_2)$:] s.t.\ $\bowtie~\in\set{\wedge,\iff}$. We
              assume by induction that
              \[
                  \begin{aligned}
                      \pair{\residual{\vi_1}{\mu}}{\residual{\neg\vi_1}{\mu}}=\pair{\residual{\NNF{\vi_1}}{\mu}}{\residual{\NNF{\neg\vi_1}}{\mu}}, \\
                      \pair{\residual{\vi_2}{\mu}}{\residual{\neg\vi_2}{\mu}}=\pair{\residual{\NNF{\vi_2}}{\mu}}{\residual{\NNF{\neg\vi_2}}{\mu}}.
                  \end{aligned}
              \]
              Then,
              \begin{description}
                  %--------------- AND ------------------------
                  \item[if $\bowtie$ is $\wedge$:]
                        \[
                            \begin{aligned}[t]
                                \pair{\residual{(\vi_1\wedge\vi_2)}{\mu}}{\residual{\neg(\vi_1\wedge\vi_2)}{\mu}}                                                       & =      \\
                                %\pair{\residual{\vi_1}{\mu}\wedge\residual{\vi_2}{\mu}}{\neg(\residual{\vi_1}{\mu}\wedge\residual{\vi_2}{\mu})}=
                                \pair{\residual{\vi_1}{\mu}\wedge\residual{\vi_2}{\mu}}{\residual{\neg\vi_1}{\mu}\vee\residual{\neg\vi_2}{\mu}}                         & \eqind \\
                                \pair{\residual{\NNF{\vi_1}}{\mu}\wedge\residual{\NNF{\vi_2}}{\mu}}{\residual{\NNF{\neg\vi_1}}{\mu}\vee\residual{\NNF{\neg\vi_2}}{\mu}} & =      \\
                                \pair{\residual{(\NNF{\vi_1}\wedge\NNF{\vi_2})}{\mu}}{\residual{(\NNF{\neg\vi_1}\vee\NNF{\neg\vi_2})}{\mu}}                             & =      \\
                                \pair{\residual{\NNF{\vi_1\wedge\vi_2}}{\mu}}{\residual{\NNF{\neg(\vi_1\wedge\vi_2)}}{\mu}}                                             &
                            \end{aligned}
                        \]
                  \item[if $\bowtie$ is $\iff$:]
                        \[
                            \begin{aligned}
                                \langle\residual{(\vi_1\iff\vi_2)}{\mu},\residual{\neg(\vi_1\iff\vi_2)}{\mu}\rangle                                                                                                                                                         & =      \\
                                \langle\residual{((\neg\vi_1\vee\vi_2)\wedge(\vi_1\vee\neg\vi_2))}{\mu},\residual{((\vi_1\vee\vi_2)\wedge(\neg\vi_1\vee\neg\vi_2))}{\mu}\rangle                                                                                     & =      \\
                                % push the residual inside the conjunctions and disjunctions
                                \langle(\residual{\neg\vi_1}{\mu}\vee\residual{\vi_2}{\mu})\wedge(\residual{\vi_1}{\mu}\vee\residual{\neg\vi_2}{\mu}),(\residual{\vi_1}{\mu}\vee\residual{\vi_2}{\mu})\wedge(\residual{\neg\vi_1}{\mu}\vee\residual{\neg\vi_2}{\mu})\rangle & \eqind \\
                                % apply induction hypothesis by changing each psi into NNF(psi)
                                \langle(\residual{\NNF{\neg\vi_1}}{\mu}\vee\residual{\NNF{\vi_2}}{\mu})\wedge(\residual{\NNF{\vi_1}}{\mu}\vee\residual{\NNF{\neg\vi_2}}{\mu}),                                                                                                       \\
                                (\residual{\NNF{\vi_1}}{\mu}\vee\residual{\NNF{\vi_2}}{\mu})\wedge(\residual{\NNF{\neg\vi_1}}{\mu}\vee\residual{\NNF{\neg\vi_2}}{\mu})\rangle                                                                                               & =      \\
                                % extract only the residual (and not NNF) from the conjunctions and disjunctions
                                \langle\residual{((\NNF{\neg\vi_1}\vee\NNF{\vi_2})\wedge(\NNF{\vi_1}\vee\NNF{\neg\vi_2}))}{\mu},                                                                                                                                            &        \\\residual{((\NNF{\vi_1}\vee\NNF{\vi_2})\wedge(\NNF{\neg\vi_1}\vee\NNF{\neg\vi_2}))}{\mu}\rangle&=\\
                                % say that NNF(...) is NNF(\vi_1 \iff \vi_2)
                                \langle\residual{\NNF{\vi_1\iff\vi_2}}{\mu},\residual{\NNF{\neg(\vi_1\iff\vi_2)}}{\mu}\rangle                                                                                                                                               & .
                            \end{aligned}
                        \]

              \end{description}
              %--------------- OR and IMPLIES ------------------------
        \item[if $\vi\defas(\vi_1\bowtie\vi_2)$:] s.t.\ $\bowtie~\in\set{\vee,\imp}$. These
              cases can be reduced to the previous cases, since
              $\NNF{\vi_1\vee\vi_2}=\NNF{\neg(\neg\vi_1\wedge\neg\vi_2)}$ and
              $\NNF{\vi_1\imp\vi_2}= \NNF{\neg(\vi_1\wedge\neg\vi_2)}$.
    \end{description}
\end{proof}
{Thus,
$\pair{\residual{\vi}{\mu}}{\residual{\neg\vi}{\mu}}=\pair{\residual{\NNF{\vi}}{\mu}}{\residual{\NNF{\neg\vi}}{\mu}}$
so that $\residual{\vi}{\mu}=v$ iff
$\residual{\NNF{\vi}}{\mu}=v$ for every
$v\in\set{\top,\bot}$, so that \cref{th:munnf}
in~\sref{sec:background:propositional-logic} holds.}
%\end{gmchange}

%%%%%%%%%%%%%%%%%%%%%%%%%%%%%%%%%%%%%%%%%%%%%%%%%%%%%%%%%%%%%
%%%
%%%%%%%%%%%%%%%%%%%%%%%%%%%%%%%%%%%%%%%%%%%%%%%%%%%%%%%%%%%%%

% \newpage

% % \begin{gmchange}
% \subsection{Proof for~\cref{th:existsetaB} in~\cref{sec:solution}}%
% \label{sec:proofexistsetaB}
% %We present the proof for~\cref{th:existsetaB} in~\cref{sec:solution}.

% \end{gmchange}
\newpage
\section{An analysis of candidate solvers}%
\label{appendix:tooltable}

In Table~\ref{tab:allsat-solvers} we report an analysis of the features of
every candidate AllSAT and AllSMT solver, as discussed in
\sref{sec:survey-allsat-solvers}.

\begin{table}[th]
    \newcommand{\cmark}{\textcolor{green4}{\ding{51}}}%
    \newcommand{\xmark}{\textcolor{red}{\ding{55}}}%
    \newcommand{\nomark}{}
    \newcommand*\rot{\rotatebox{90}}
    \newcolumntype{C}{>{\centering\arraybackslash}m{0.4cm}}
    \centering
    \begin{tabularx}{.8\textwidth}{c|l|CCCCC|CC|X}
        % \hline
                                               &                              &
        \multicolumn{5}{c|}{Required features} & \multicolumn{2}{c|}{Options} &                                                                                        \\
        \hline
                                               & Solver                       &
        \rot{Available}                        &
        \rot{CNF input\ }                      &
        \rot{Projected}                        &
        \rot{Partial}                          &
        \rot{Minimal}                          &
        \rot{Disjoint}                         &
        \rot{Non-disjoint}                     &
        {Notes}                                                                                                                                                        \\
        \hline
        \parbox[t]{4mm}{\multirow{15}{*}{\rotatebox[origin=c]{90}{AllSAT}}}
                                               & \textsc{RELSAT}              & \cmark & \cmark & \xmark & \xmark & \xmark & \cmark & \xmark &                         \\
                                               & \textsc{Grumberg}            & \xmark & \cmark & \cmark & \cmark & \cmark & \cmark & \cmark & Ex-post minimization    \\
                                               & \textsc{SOLALL}              & ?      & \cmark & \cmark & \cmark & \xmark & \cmark & \xmark & Result stored in (O)BDD \\
                                               & \textsc{Jin}                 & \xmark & \xmark & \xmark & \cmark & \cmark & \xmark & \cmark & AIG input               \\
                                               & \textsc{clasp}               & \cmark & \cmark & \cmark & \xmark & \xmark & \cmark & \xmark & ASP solver              \\
                                               & \textsc{PicoSAT}             & \cmark & \cmark & \xmark & \xmark & \xmark & \cmark & \xmark &                         \\
                                               & \textsc{Yu}                  & ?      & \cmark & \xmark & \cmark & \cmark & \cmark & \cmark &                         \\
                                               & \textsc{BC}                  & \cmark & \cmark & \xmark & \cmark & \cmark & \cmark & \cmark &                         \\
                                               & \textsc{NBC}                 & \cmark & \cmark & \xmark & \xmark & \xmark & \cmark & \xmark &                         \\
                                               & \textsc{BDD}                 & \cmark & \cmark & \xmark & \cmark & \xmark & \xmark & \cmark & Result stored in (O)BDD \\
                                               & \textsc{depbdd}              & \xmark & \cmark & \cmark & \cmark & \xmark & \cmark & \xmark & Result stored in (O)BDD \\
                                               & \textsc{Dualiza}             & \cmark & \cmark & \cmark & \cmark & \xmark & \cmark & \cmark & Dual-reasoning-based    \\
                                               & \textsc{BASolver}            & \xmark & \cmark & \xmark & \cmark & \cmark & \xmark & \cmark &                         \\
                                               & \textsc{AllSATCC}            & \xmark & \cmark & \xmark & \cmark & \cmark & \cmark & \xmark &                         \\
                                               & \textsc{HALL}                & \cmark & \xmark & \xmark & \cmark & \cmark & \cmark & \cmark & AIG input               \\
                                               & \tabularallsat{}             & \cmark & \cmark & \cmark & \cmark & \xmark & \cmark & \xmark &                         \\
                                               & \decdnnf{}                   & \cmark & \xmark & \xmark & \cmark & \xmark & \cmark & \xmark & d-DNNF input            \\
        \hline
        \parbox[t]{4mm}{\multirow{5}{*}{\rotatebox[origin=c]{90}{AllSMT}}}
                                               &                              &        &        &        &        &        &        &                                  \\
                                               & \mathsat{}                   & \cmark & \cmark & \cmark & \cmark & \cmark & \cmark & \cmark &                         \\
                                               & \textsc{aZ3}                 & \cmark & \cmark & \cmark & \xmark & \xmark & \cmark & \xmark &                         \\
                                               & \tabularallsmt{}             & \cmark & \cmark & \cmark & \cmark & \xmark & \cmark & \xmark &                         \\
                                               &                              &        &        &        &        &        &        &
        \\
        % \textsc{TALE} & \xmark & \xmark & \cmark & \xmark & \cmark & \cmark & AIG input \\
        % \textsc{MARS} & \xmark & \xmark & \cmark & \cmark & \cmark & \cmark & AIG input \\
        % \textsc{DUTY} & \xmark & \xmark & \cmark & \xmark & \cmark & \cmark & AIG input \\
    \end{tabularx}
    \caption{List of AllSAT and AllSMT solvers (rows) and supported
        features (columns).
        % Required features
        % (\ref{item:tool:avail}-\ref{item:tool:minimal}): the tool is
        % publicly available, input CNF formula, enumerate projected
        % assignments, enumerate partial assignments. Optional features:
        % partial assignments are minimal, enumerate disjoint assignments,
        % enumerate non-disjoint assignment.
        % \newline
        % In the ``Available'' column, ``\xmark'' and ``?'' indicate respectively that the
        % authors confirmed to us the non-availability of the tool and
        % that they did not reply to our request.
        %
        (In the ``Available'' column, ``\xmark'' indicates that the
        authors confirmed to us the unavailability of their tool, whereas
        ``?'' indicates that they did not reply to our enquiry.)
    }
    \label{tab:allsat-solvers}
\end{table}

\newpage
\section{Experimental results on plain SAT and SMT solving}%
\label{appendix:sat}
In this section, we report the results of the experiments on plain SAT and SMT solving with \mathsat{} using the different CNF-izations.
The plots in~\cref{fig:sat:ecdf} show that \NNFPlaisted{} does not bring any advantage for plain SAT solving, supporting the analysis in~\sref{sec:experiments:sat}.
\begin{figure}[h!]
    \centering
    \captionsetup[subfigure]{justification=centering}
    \begin{subfigure}[t]{\textwidth}
        \centering
        \includegraphics[height=2em]{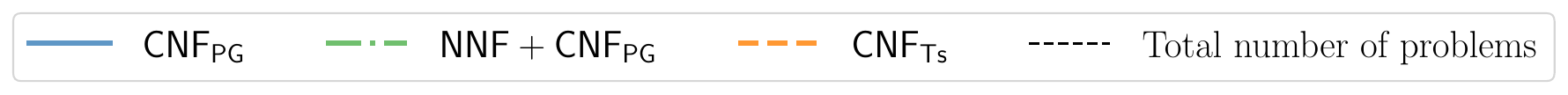}
    \end{subfigure}
    \begin{subfigure}[b]{0.33\textwidth}
        \centering
        \includegraphics[width=\textwidth]{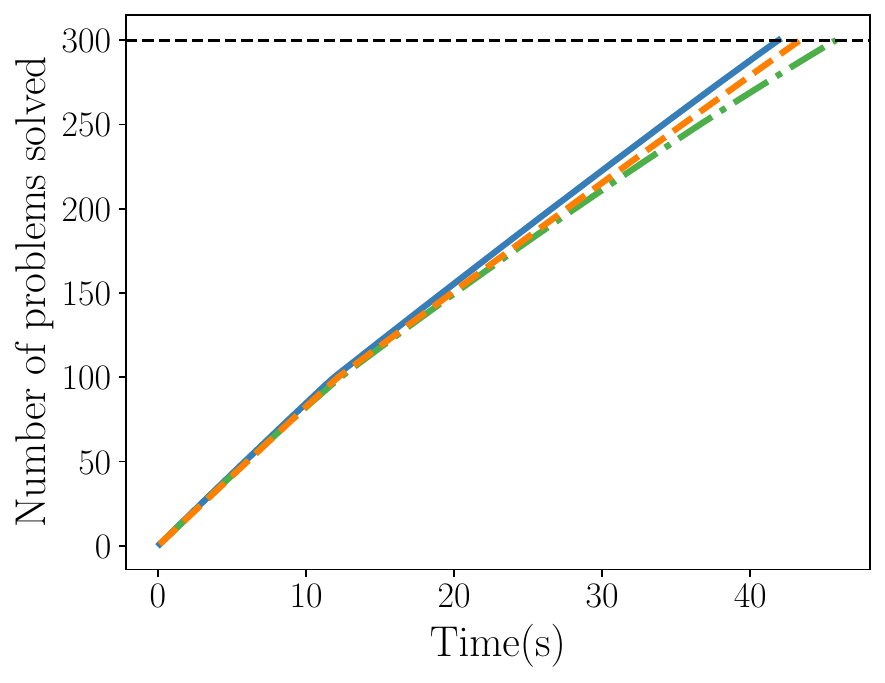}
        \caption{Boolean synthetic.}%
        \label{fig:sat:syn:bool:ecdf}
    \end{subfigure}%
    \begin{subfigure}[b]{0.33\textwidth}
        \centering
        \includegraphics[width=\textwidth]{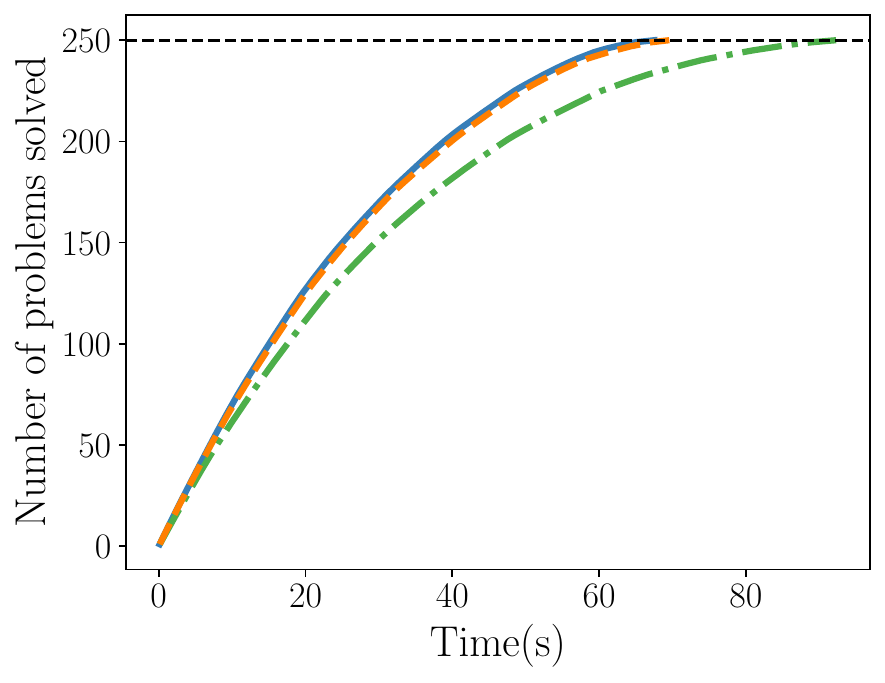}
        \caption{ISCAS'85.}%
        \label{fig:sat:circ:ecdf}
    \end{subfigure}%
    \begin{subfigure}[b]{0.33\textwidth}
        \centering
        \includegraphics[width=\textwidth]{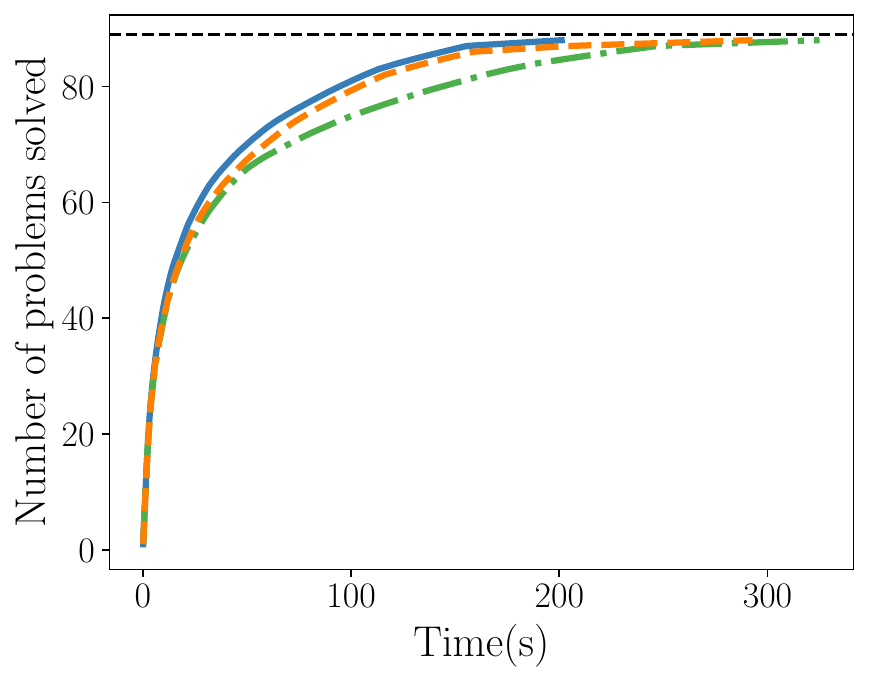}
        \caption{AIG.}%
        \label{fig:sat:aig:ecdf}
    \end{subfigure}
    \begin{subfigure}[b]{0.33\textwidth}
        \centering
        \includegraphics[width=\textwidth]{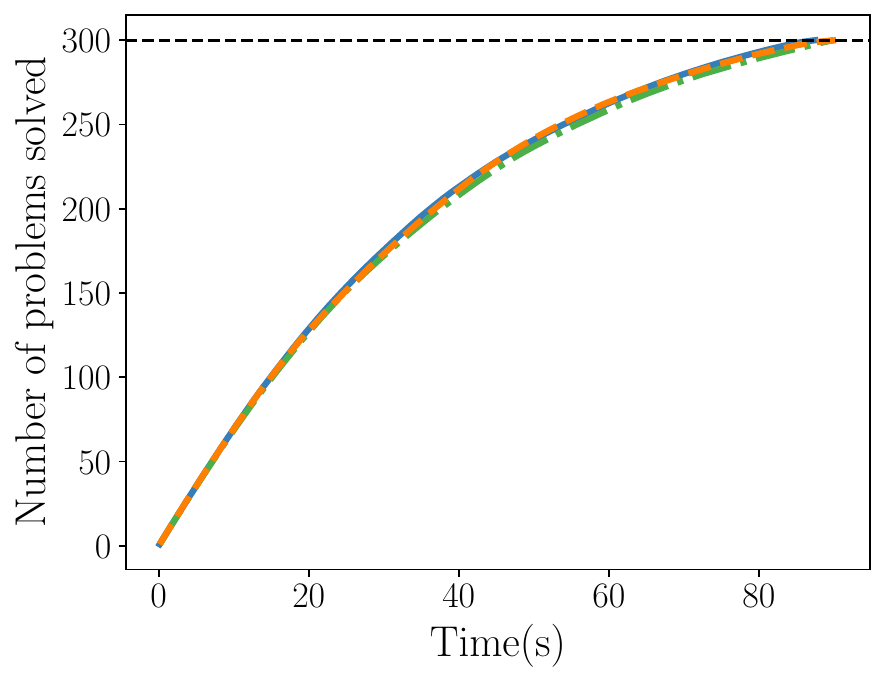}
        \caption{\smtlarat{} synthetic.}%
        \label{fig:sat:syn:lra:ecdf}
    \end{subfigure}%
    \begin{subfigure}[b]{0.33\textwidth}
        \centering
        \includegraphics[width=\textwidth]{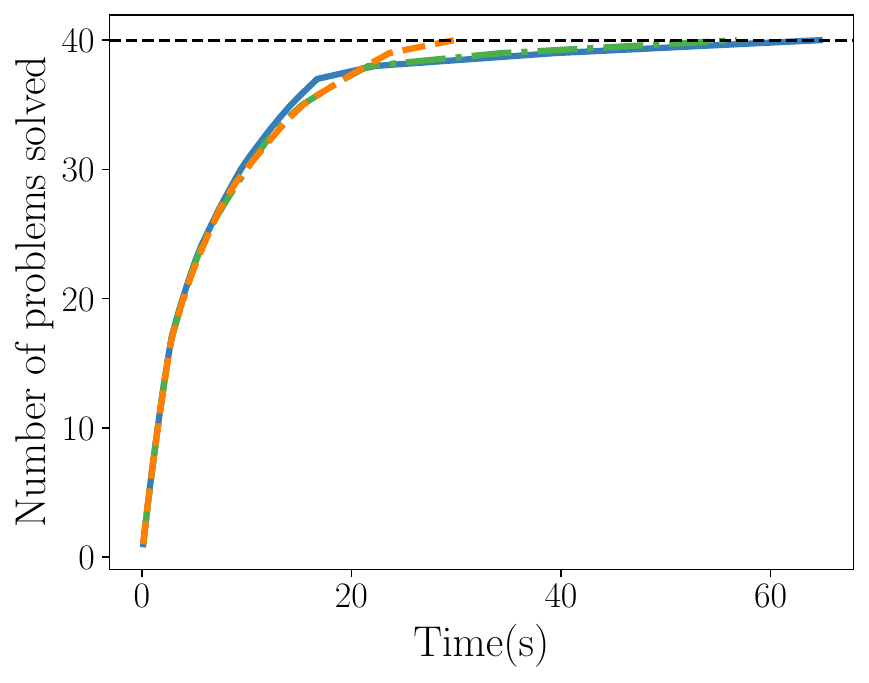}
        \caption{WMI.}%
        \label{fig:sat:wmi:ecdf}
    \end{subfigure}
    \caption{Time taken for plain SAT %\ignoreinshort{\GMCHANGEp{
        and SMT %}}
        solving using the different CNF transformations. The $y$-axis reports the number of instances for which the solver finished within the cumulative time on the $x$-axis.
    }%
    \label{fig:sat:ecdf}
\end{figure}

% ---- d-DNNF enum ----
\newpage
% \begin{gmchange}
\section{Experimental results on d-DNNF-based enumeration}%
\label{appendix:d4:enum}
Recently, in \cite{lagniezLeveragingDecisionDNNFCompilation2024}, the authors presented a tool, named \decdnnf{}, to enumerate all models of a d-DNNF formula~\cite{darwicheKnowledgeCompilationMap2002}.

Following the \dfdecdnnf{}
approach~\cite{lagniezLeveragingDecisionDNNFCompilation2024}, we first compile
the CNF-ization of the input formula into a d-DNNF using the \textsc{d4}
compiler~\cite{lagniezImprovedDecisionDNNFCompiler2017}, and then enumerate its
models using \decdnnf{}. To avoid enumerating on CNF labels, we project the
d-DNNF onto the original atoms \allA{} of the input formula, and then use
\decdnnf{} to enumerate the models of the projected d-DNNF.

The results of the experiments on the Boolean benchmarks are shown in
\cref{fig:plt:d4:all:bool:scatter}. We see that \TseitinCNF{} is uniformly the
best-performing CNF-ization, supporting the analysis
in~\sref{sec:experiments:d4}
% \PlaistedCNF{} and
% \NNFPlaisted{} only introduce time-overhead, without any benefit in terms of the number of models enumerated. This can be explained by the fact that d-DNNF compilation does not aim at representing short partial truth assignments, but rather at effectively decomposing the formula into atom-disjoint components, so to allow for caching and sharing of sub-formulas.
% In fact, compilation to d-DNNF is typically performed following a DPLL-like procedure, where the formula is decomposed into a DAG of decision nodes, each representing an assignment to one or more atoms. In the case of projected d-DNNF, decisions are made on the relevant atoms first, and whenever no relevant atom is left, a SAT solver is used to complete the assignment. This search strategy prevents finding short partial assignments~\cite{lagniezRecursiveAlgorithmProjected2019}, and thus the
% benefit of our encoding for enumeration is lost.
% \end{gmchange}
% --------------- d4 enum ----------------
\begin{figure}[h!]
    \centering
    \begin{subfigure}[t]{\textwidth}
        \centering
        \includegraphics[height=2em]{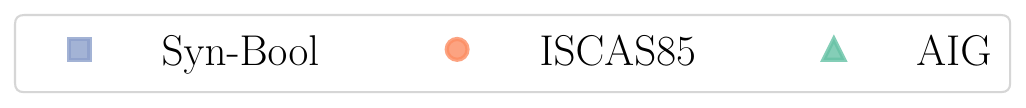}
    \end{subfigure}
    \begin{subfigure}[t]{\textwidth}
        \begin{subfigure}[t]{0.29\textwidth}
            \centering
            \includegraphics[width=.85\textwidth]{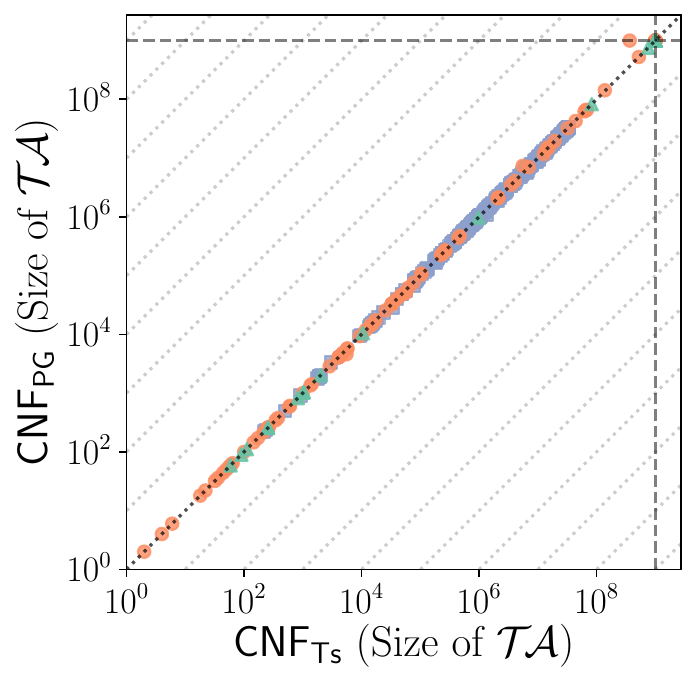}%
            \label{fig:plt:d4:all:bool:norep:models:lab_vs_pol}
        \end{subfigure}\hfill
        \begin{subfigure}[t]{0.29\textwidth}
            \centering
            \includegraphics[width=.85\textwidth]{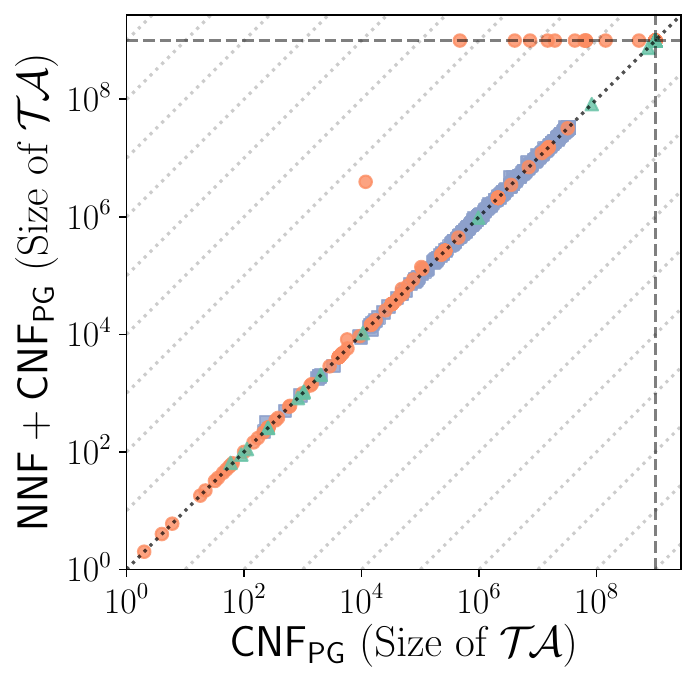}%
            \label{fig:plt:d4:all:bool:norep:models:pol_vs_nnfpol}
        \end{subfigure}\hfill
        \begin{subfigure}[t]{0.29\textwidth}
            \centering
            \includegraphics[width=.85\textwidth]{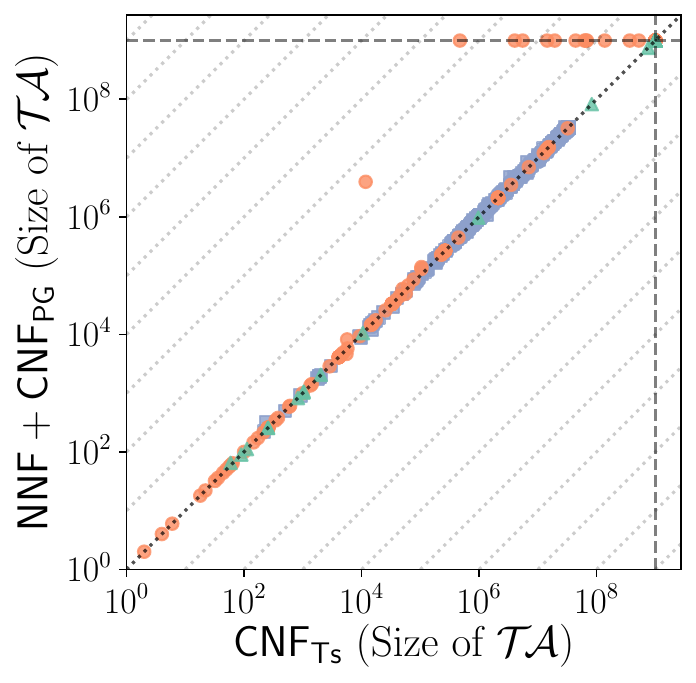}%
            \label{fig:plt:d4:all:bool:norep:models:lab_vs_nnfpol}
        \end{subfigure}\hfill
        \begin{subfigure}[t]{0.29\textwidth}
            \centering
            \includegraphics[width=.85\textwidth]{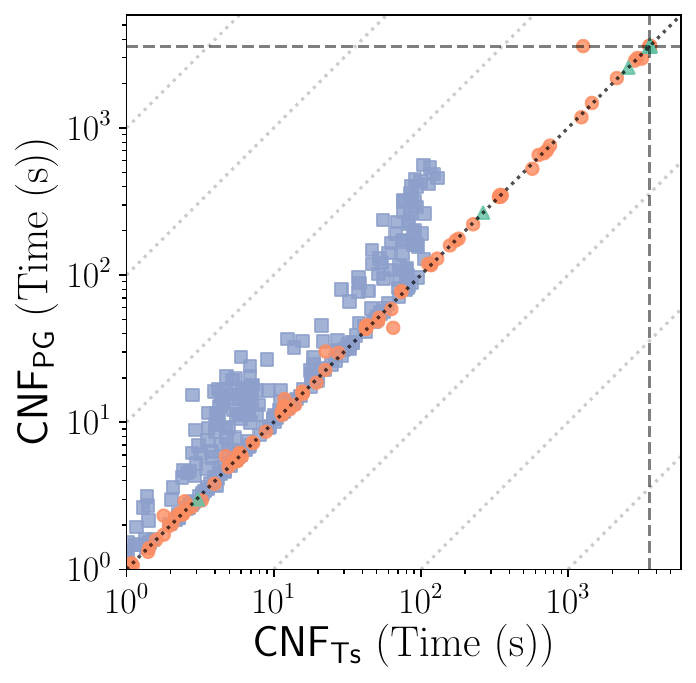}%
            \label{fig:plt:d4:all:bool:norep:time:lab_vs_pol}
        \end{subfigure}\hfill
        \begin{subfigure}[t]{0.29\textwidth}
            \centering
            \includegraphics[width=.85\textwidth]{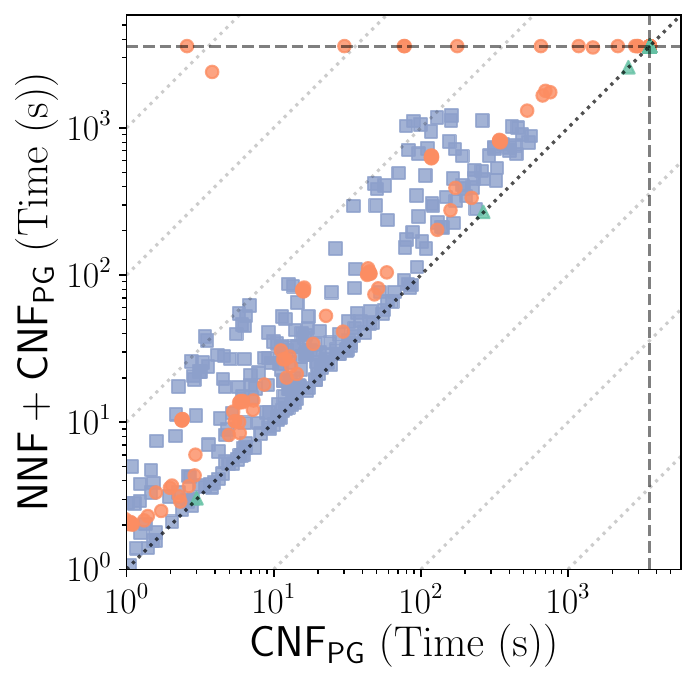}%
            \label{fig:plt:d4:all:bool:norep:time:pol_vs_nnfpol}
        \end{subfigure}\hfill
        \begin{subfigure}[t]{0.29\textwidth}
            \centering
            \includegraphics[width=.85\textwidth]{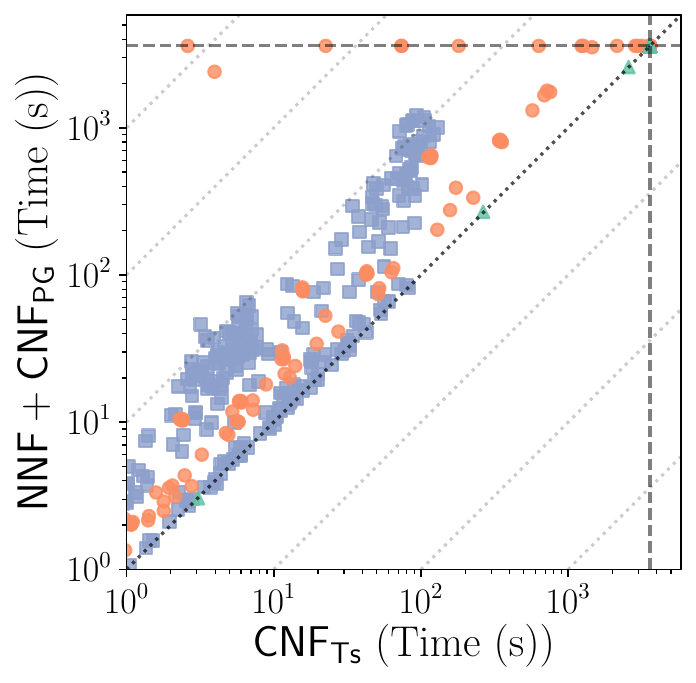}%
            \label{fig:plt:d4:all:bool:norep:time:lab_vs_nnfpol}
            % \end{subfigure}
        \end{subfigure}
        \caption{Results for disjoint enumeration.}%
        \label{fig:plt:d4:all:bool:norep:scatter}
    \end{subfigure}
    \begin{subfigure}[t]{\textwidth}
        \vspace{.5em}
        \centering
        {\small
            % Timeouts:
% mode        LAB LABELNEG_POL NNF_MUTEX_POL
% problem_set                               
% syn-bool      0            0             0
% iscas85      15           16            27
% aig          64           64            76

% Total problems:
% mode         LAB  LABELNEG_POL  NNF_MUTEX_POL
% problem_set                                  
% syn-bool     300           300            300
% iscas85      250           250            250
% aig           77            77             89

\newcommand{\best}[1]{\textbf{#1}}
% \begin{figure}[th]
%     \centering
\begin{tabularx}{.44\textwidth}{l|c|ccc}
    % \toprule
    \multirow{3}{*}{Bench.} & \multirow{3}{*}{Instances} & \multicolumn{3}{c}{T.O.\ for disjoint AllSAT}                                            \\
                            &                            & \TseitinCNF{}                                 & \PlaistedCNF{} & \NNFPlaisted{} \\[0.2em]
    \hline
    Syn-Bool                & 300                        & \best{0}                                      & \best{0}       & \best{0}                \\
    ISCAS85                 & 250                        & \best{15}                                     & 16             & 27                      \\
    AIG                     & 89                         & \best{76}                                     & \best{76}      & \best{76}               \\
    % \bottomrule
\end{tabularx}
        }
        \caption{Number of timeouts.}%
        \label{tab:timeouts:d4:bool}
    \end{subfigure}
    \caption{Results on the Boolean benchmarks using \dfdecdnnf{}.
        Plots in~\ref{fig:plt:d4:all:bool:norep:scatter} compare CNF-izations by \TAna{} size (first row) and execution time (second row).
        Points on dashed lines represent timeouts, shown in~\ref{tab:timeouts:d4:bool}.
        All axes use a logarithmic scale.}%
    \label{fig:plt:d4:all:bool:scatter}
\end{figure}

% ---- d-DNNF counting ----
\newpage
% \begin{gmchange}
\section{Experimental results on d-DNNF-based model counting}%
\label{appendix:d4:counting}
We tested the different CNF-izations also for model counting using \df{}. Since \PlaistedCNF{} and \NNFPlaisted{} do not preserve the model count, we perform projected model counting on the original atoms of the input formula. Even though this would not be necessary for \TseitinCNF{}, we apply the same procedure to ensure a fair comparison.

The results of the experiments on the Boolean benchmarks are shown in
\cref{fig:plt:d4:counting:all:bool:scatter}. We can observe that also for model
counting, \TseitinCNF{} is uniformly the best-performing CNF-ization,
supporting the analysis in~\sref{sec:experiments:d4}.
% \PlaistedCNF{} and \NNFPlaisted{} only introduce time-overhead.
% Similar considerations as for enumeration apply here, since model counting strategies are typically based on effective decompositions of the formula into atom-disjoint components, so to allow for caching and sharing of sub-formulas, rather than on finding short partial truth assignments. Hence, our encoding does not provide any benefit for model counting using state-of-the-art counters.
% \end{gmchange}
% --------------- d4 enum ----------------
\begin{figure}[h!]
    \centering
    \begin{subfigure}[t]{\textwidth}
        \centering
        \includegraphics[height=2em]{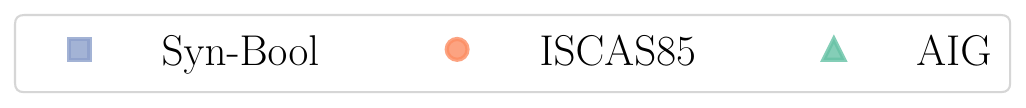}
    \end{subfigure}
    \begin{subfigure}[t]{\textwidth}
        \begin{subfigure}[t]{0.29\textwidth}
            \centering
            \includegraphics[width=.85\textwidth]{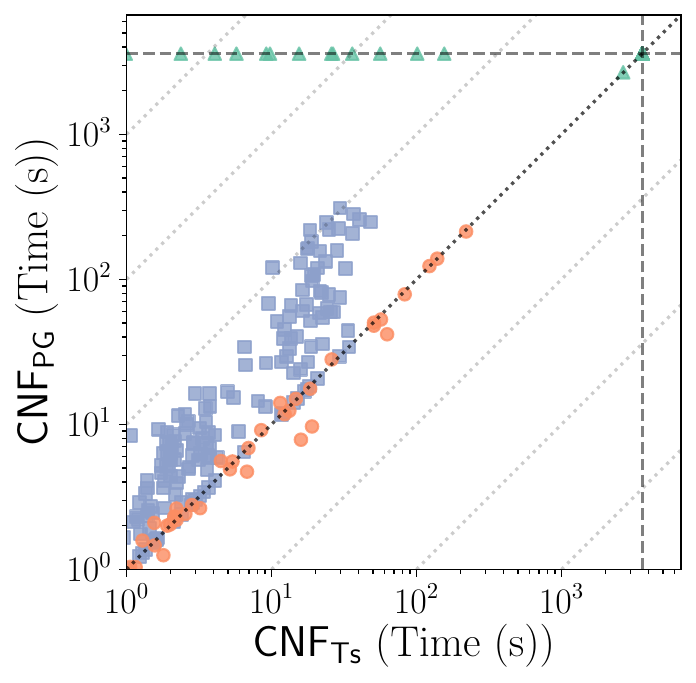}%
            \label{fig:plt:d4:counting:all:bool:norep:time:lab_vs_pol}
        \end{subfigure}\hfill
        \begin{subfigure}[t]{0.29\textwidth}
            \centering
            \includegraphics[width=.85\textwidth]{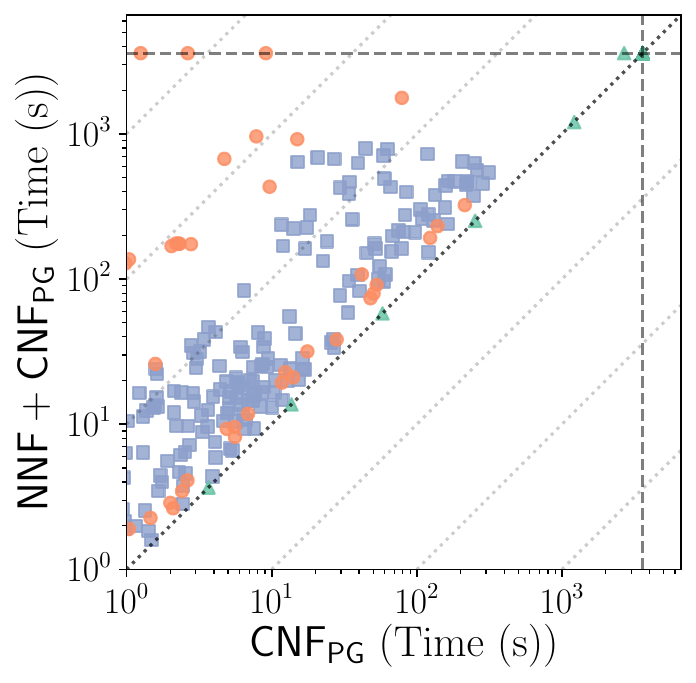}%
            \label{fig:plt:d4:counting:all:bool:norep:time:pol_vs_nnfpol}
        \end{subfigure}\hfill
        \begin{subfigure}[t]{0.29\textwidth}
            \centering
            \includegraphics[width=.85\textwidth]{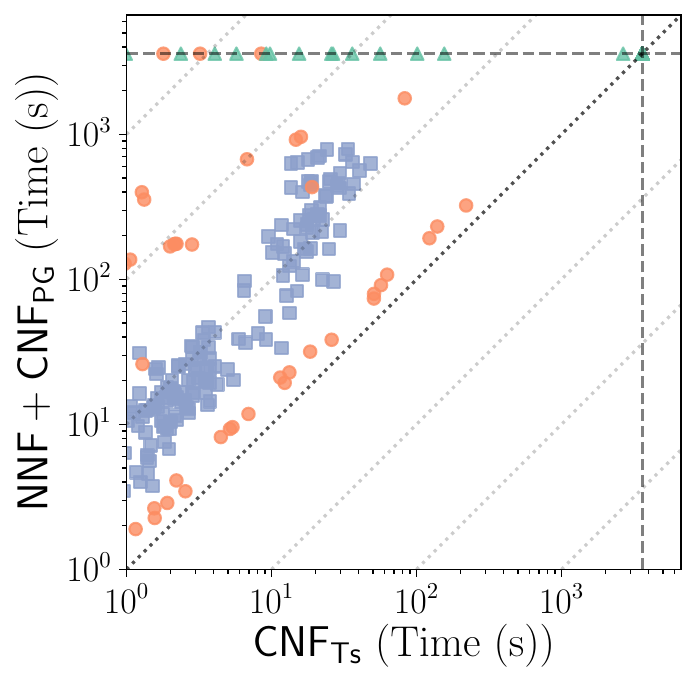}%
            \label{fig:plt:d4:counting:all:bool:norep:time:lab_vs_nnfpol}
            % \end{subfigure}
        \end{subfigure}
        \caption{Results for model counting.}%
        \label{fig:plt:d4:counting:all:bool:norep:scatter}
    \end{subfigure}
    \begin{subfigure}[t]{\textwidth}
        \centering
        {\small
            % Timeouts:
% mode        LAB LABELNEG_POL NNF_MUTEX_POL
% problem_set                               
% syn-bool      0            0             0
% iscas85       0            0             4
% aig          25           26            29

% Total problems:
% mode         LAB  LABELNEG_POL  NNF_MUTEX_POL
% problem_set                                  
% syn-bool     300           300            300
% iscas85      250           250            250
% aig           30            30             37

\newcommand{\best}[1]{\textbf{#1}}
% \begin{figure}[th]
%     \centering
\begin{tabularx}{.44\textwidth}{l|c|ccc}
    % \toprule
    \multirow{3}{*}{Bench.} & \multirow{3}{*}{Instances} & \multicolumn{3}{c}{T.O.\ for disjoint AllSAT}                                            \\
                            &                            & \TseitinCNF{}                                 & \PlaistedCNF{} & \NNFPlaisted{} \\[0.2em]
    \hline
    Syn-Bool                & 300                        & \best{0}                                      & \best{0}       & \best{0}                \\
    ISCAS85                 & 250                        & \best{0}                                      & \best{0}       & 4                       \\
    AIG                     & 89                         & \best{50}                                     & 67             & 68                      \\
    % \bottomrule
\end{tabularx}
        }
        \caption{Number of timeouts.}%
        \label{tab:timeouts:d4:bool:counting}
    \end{subfigure}
    \caption{Results on the Boolean benchmarks using \df{} for model counting.
        Plots in~\ref{fig:plt:d4:counting:all:bool:norep:scatter} compare CNF-izations by execution time.
        Points on dashed lines represent timeouts, shown in~\ref{tab:timeouts:d4:bool:counting}.
        All axes use a logarithmic scale.}%
    \label{fig:plt:d4:counting:all:bool:scatter}
\end{figure}

\newpage

\section{Details on experimental results in the paper}%
\label{appendix:experiments}
In this section, we report the scatter plots on individual
benchmarks for the experiments presented
in~\sref{sec:experiments:results}, \cref{fig:plt:all:bool:scatter,,fig:plt:tabula:all:bool:scatter,,fig:plt:all:lra:scatter}.

For AllSAT,
\cref{fig:plt:syn:bool:scatter,,fig:plt:circ:scatter,,fig:plt:aig:scatter} show
the results for \mathsat{} on the Boolean synthetic, ISCAS'85 and AIG
benchmarks, respectively.
\cref{fig:plt:tabula:syn:bool:scatter,,fig:plt:tabula:circ:scatter,,fig:plt:tabula:aig:scatter}
show the results for \tabularallsat{} on the same benchmarks.

For AllSMT, \cref{fig:plt:syn:lra:scatter,,fig:plt:wmi:scatter} show the
results for \mathsat{} on the \smtlarat{} synthetic and WMI benchmarks,
respectively. \cref{fig:plt:tabula:syn:lra:scatter,,fig:plt:tabula:wmi:scatter}
show the results for \tabularallsmt{} on the same benchmarks.
% --------------- AllSAT ----------------
% --------------- AllSAT MathSAT ----------------
\begin{figure}
    \centering
    % \begin{subfigure}[t]{\textwidth}
    \begin{subfigure}[t]{\textwidth}
        \begin{subfigure}[t]{0.29\textwidth}
            \includegraphics[width=.85\textwidth]{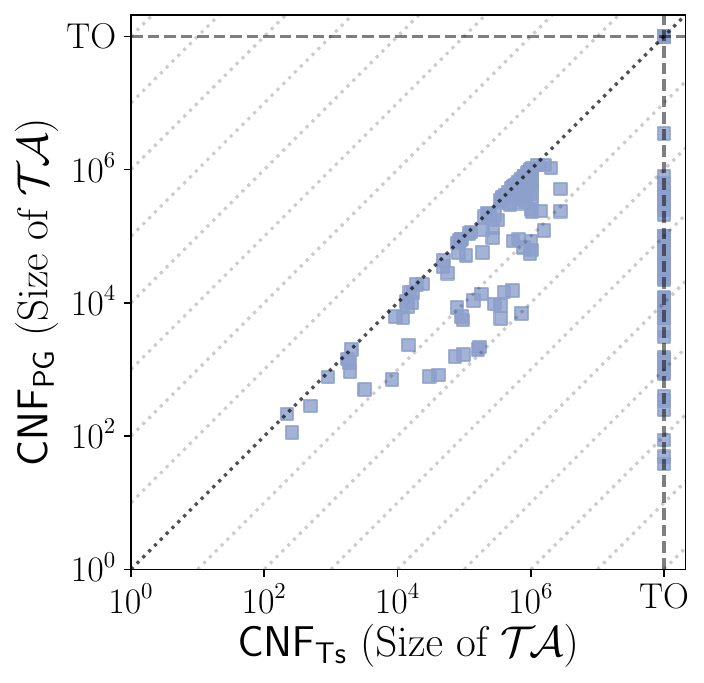}%
            \label{fig:plt:syn:bool:norep:models:lab_vs_pol}
        \end{subfigure}\hfill
        \begin{subfigure}[t]{0.29\textwidth}
            \includegraphics[width=.85\textwidth]{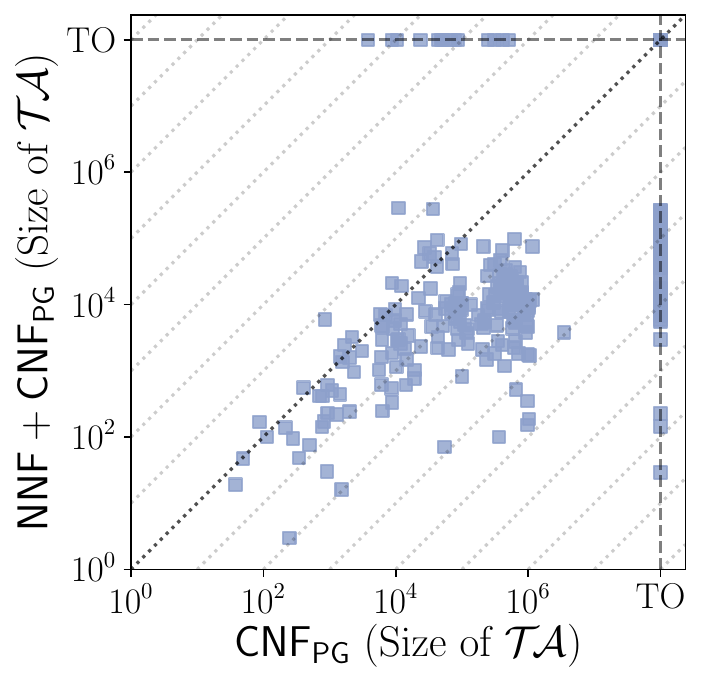}%
            \label{fig:plt:syn:bool:norep:models:pol_vs_nnfpol}
        \end{subfigure}\hfill
        \begin{subfigure}[t]{0.29\textwidth}
            \includegraphics[width=.85\textwidth]{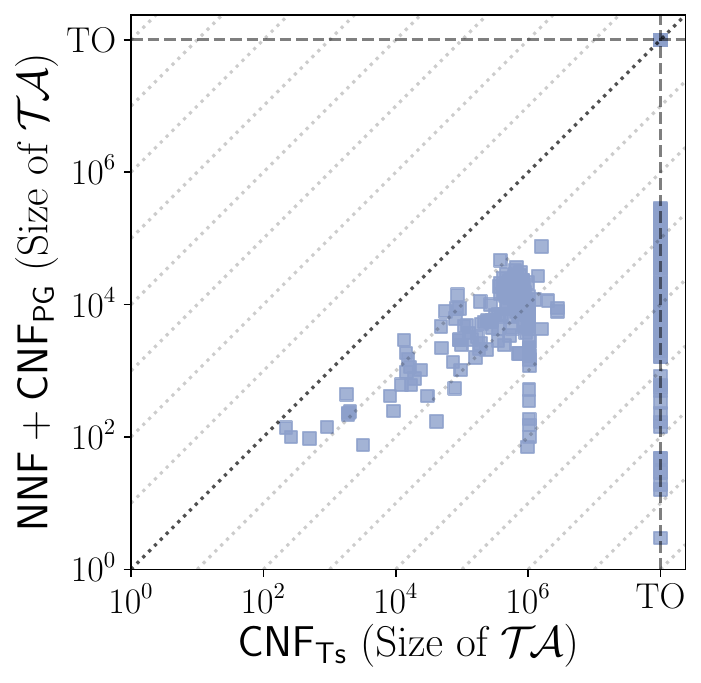}%
            \label{fig:plt:syn:bool:norep:models:lab_vs_nnfpol}
        \end{subfigure}\hfill
        \begin{subfigure}[t]{0.29\textwidth}
            \includegraphics[width=.85\textwidth]{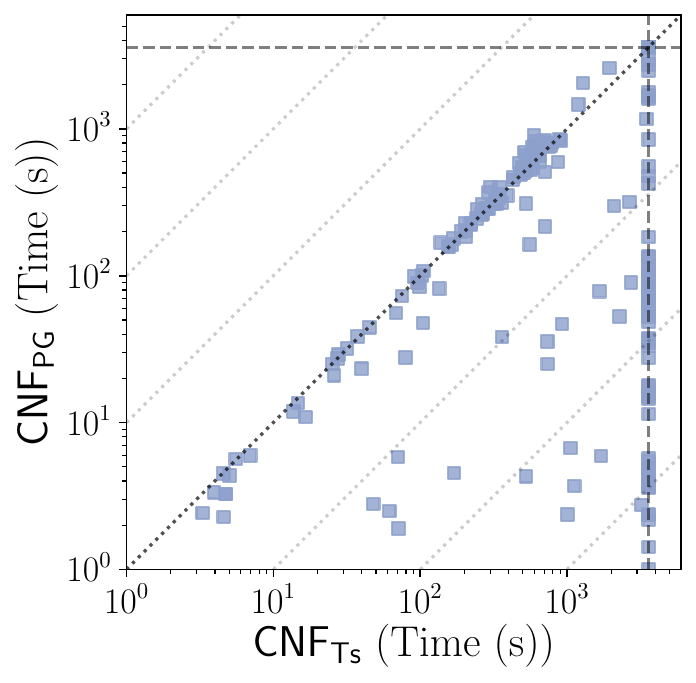}%
            \label{fig:plt:syn:bool:norep:time:lab_vs_pol}
        \end{subfigure}\hfill
        \begin{subfigure}[t]{0.29\textwidth}
            \includegraphics[width=.85\textwidth]{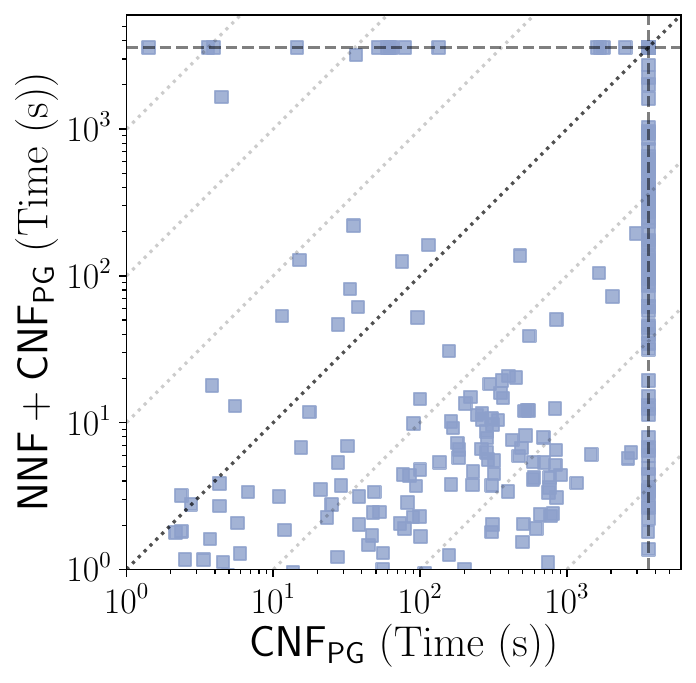}%
            \label{fig:plt:syn:bool:norep:time:pol_vs_nnfpol}
        \end{subfigure}\hfill
        \begin{subfigure}[t]{0.29\textwidth}
            \includegraphics[width=.85\textwidth]{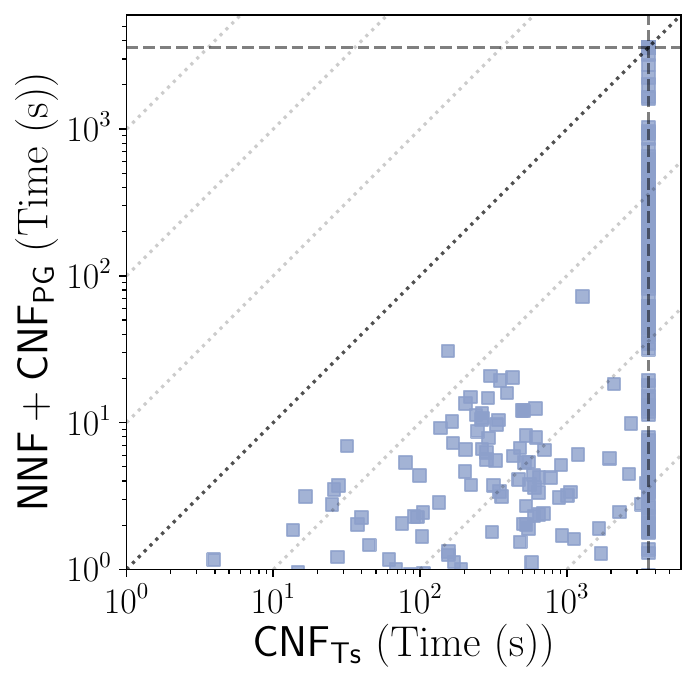}%
            \label{fig:plt:syn:bool:norep:time:lab_vs_nnfpol}
            % \end{subfigure}
        \end{subfigure}
        \caption{Results for disjoint enumeration. %\TseitinCNF{}, \PlaistedCNF{} and $\NNFPlaisted{}$ reported 163, 94 and 20 timeouts, respectively (points on the dashed lines).
        }%
        \label{fig:plt:syn:bool:norep:scatter}
    \end{subfigure}
    %%%%%%%%%%%% REP %%%%%%%%%%%%%
    \begin{subfigure}[t]{\textwidth}
        \begin{subfigure}[t]{0.29\textwidth}
            \includegraphics[width=.85\textwidth]{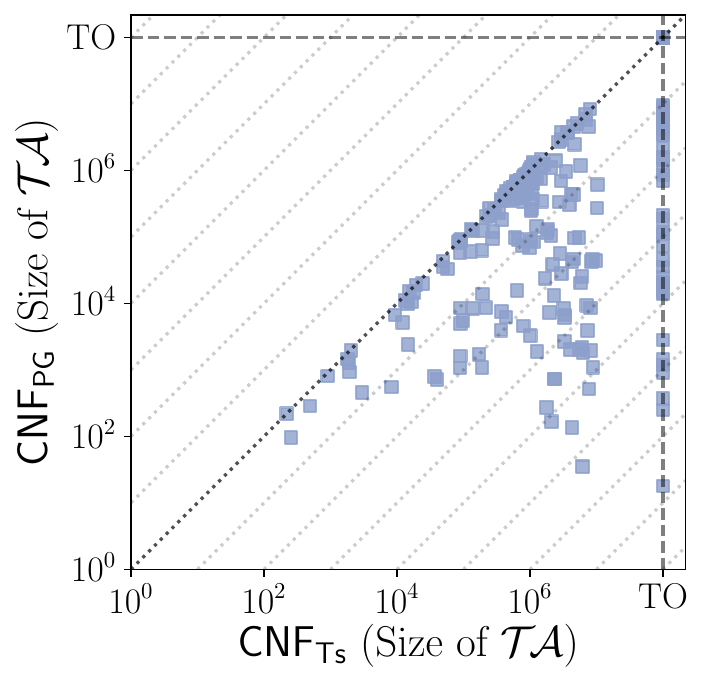}%
            \label{fig:plt:syn:bool:rep:models:lab_vs_pol}
        \end{subfigure}\hfill
        \begin{subfigure}[t]{0.29\textwidth}
            \includegraphics[width=.85\textwidth]{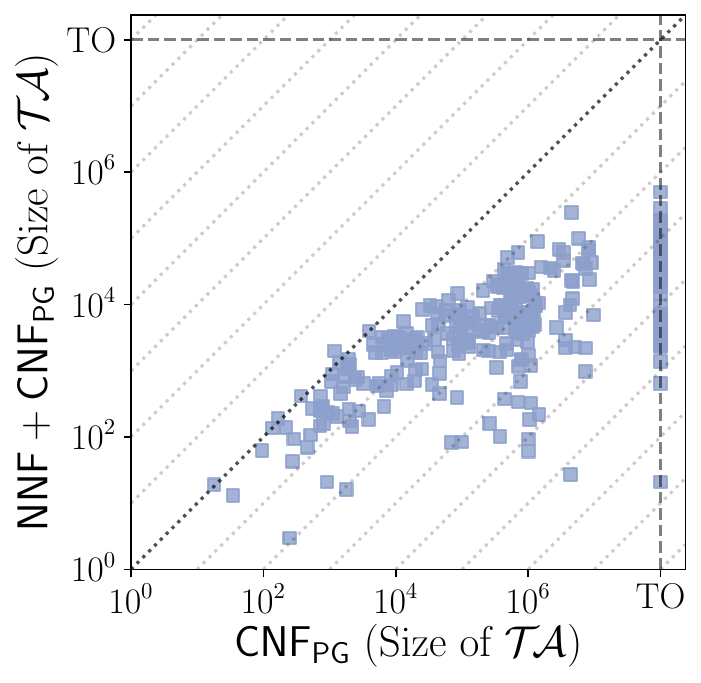}%
            \label{fig:plt:syn:bool:rep:models:pol_vs_nnfpol}
        \end{subfigure}\hfill
        \begin{subfigure}[t]{0.29\textwidth}
            \includegraphics[width=.85\textwidth]{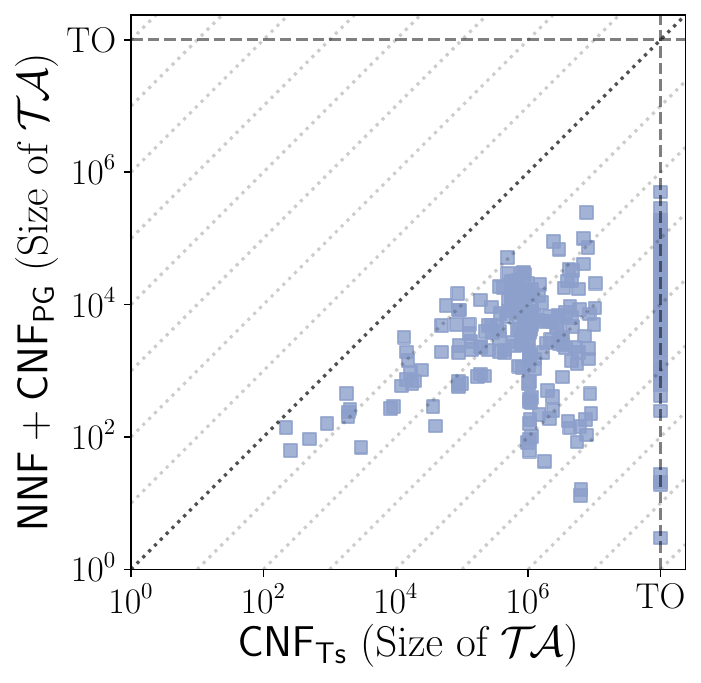}%
            \label{fig:plt:syn:bool:rep:models:lab_vs_nnfpol}
        \end{subfigure}\hfill
        \begin{subfigure}[t]{0.29\textwidth}
            \includegraphics[width=.85\textwidth]{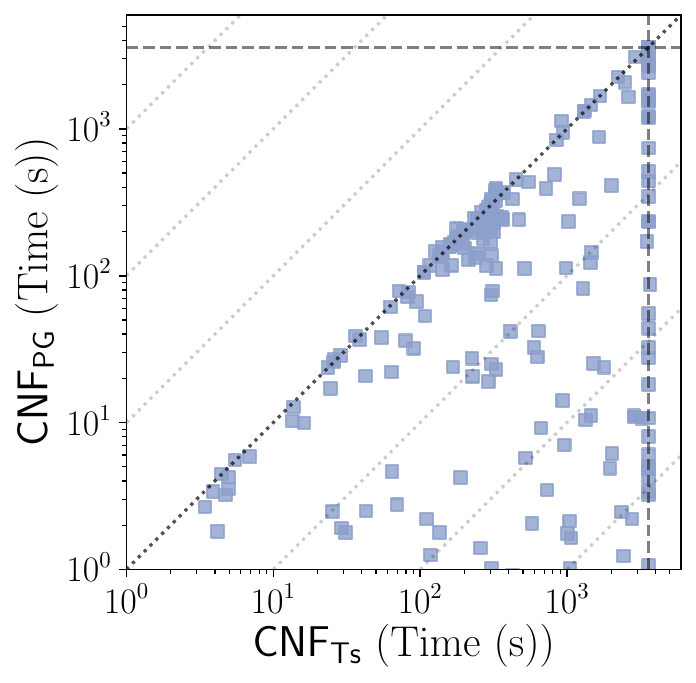}%
            \label{fig:plt:syn:bool:rep:time:lab_vs_pol}
        \end{subfigure}\hfill
        \begin{subfigure}[t]{0.29\textwidth}
            \includegraphics[width=.85\textwidth]{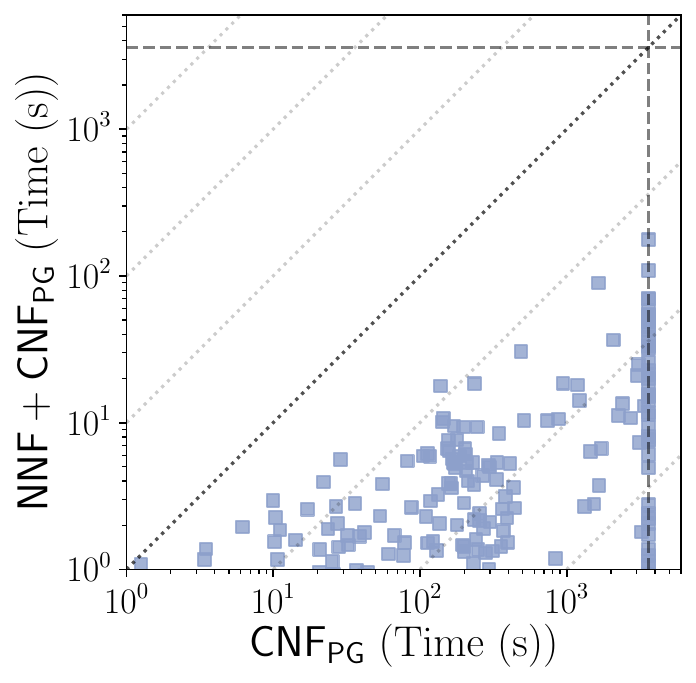}%
            \label{fig:plt:syn:bool:rep:time:pol_vs_nnfpol}
        \end{subfigure}\hfill
        \begin{subfigure}[t]{0.29\textwidth}
            \includegraphics[width=.85\textwidth]{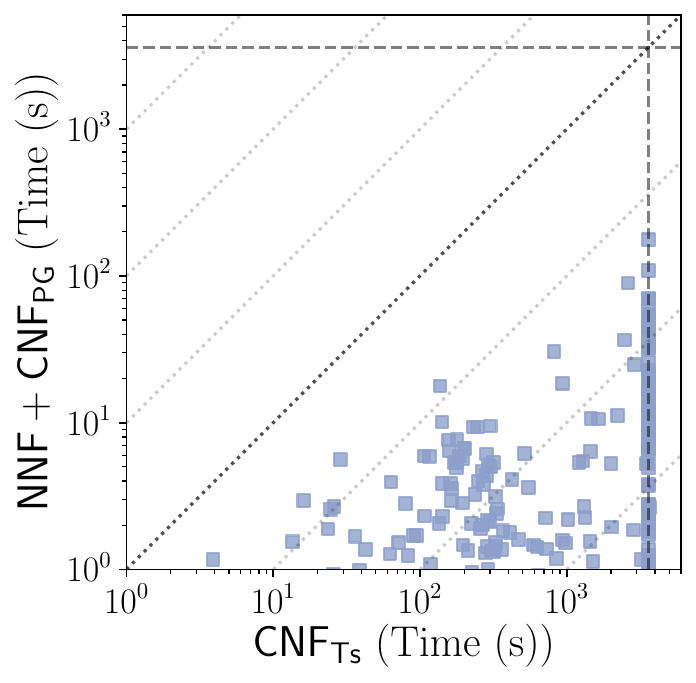}%
            \label{fig:plt:syn:bool:rep:time:lab_vs_nnfpol}
            % \end{subfigure}
        \end{subfigure}
        \caption{Results for non-disjoint enumeration. %\TseitinCNF{}, \PlaistedCNF{} and $\NNFPlaisted{}$ reported 122, 69 and 0 timeouts, respectively (points on the dashed lines).
        }%
        \label{fig:plt:syn:bool:rep:scatter}
    \end{subfigure}
    \caption{Results on the Boolean synthetic benchmarks using \mathsat{}.
        Plots in~\ref{fig:plt:syn:bool:norep:scatter} and~\ref{fig:plt:syn:bool:rep:scatter} compare CNF-izations by \TAna{} size (first row) and execution time (second row).
        Points on dashed lines represent timeouts, shown in~\ref{tab:timeouts:bool}.
        All axes use a logarithmic scale.}%
    \label{fig:plt:syn:bool:scatter}
\end{figure}

\begin{figure}
    \centering
    % \begin{subfigure}[t]{\textwidth}
    \begin{subfigure}[t]{\textwidth}
        \begin{subfigure}[t]{0.29\textwidth}
            \includegraphics[width=.85\textwidth]{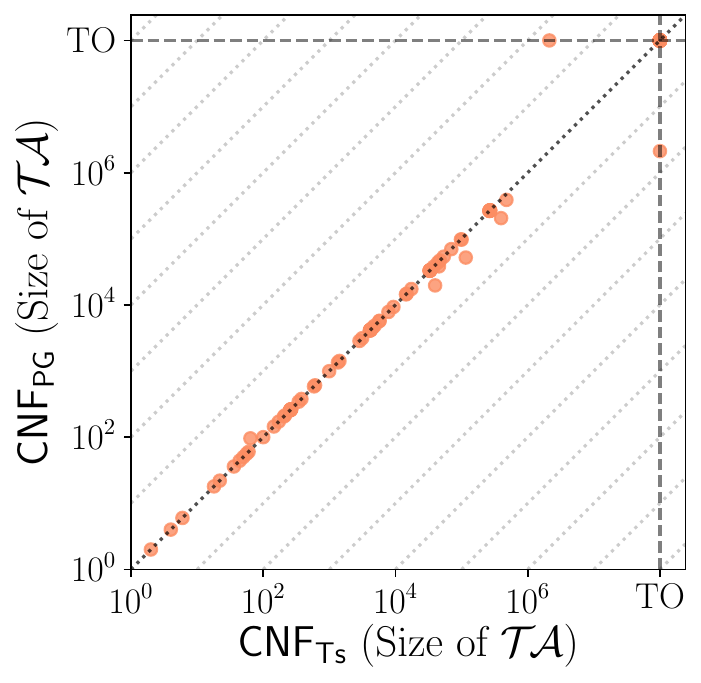}%
            \label{fig:plt:circ:norep:models:lab_vs_pol}
        \end{subfigure}\hfill
        \begin{subfigure}[t]{0.29\textwidth}
            \includegraphics[width=.85\textwidth]{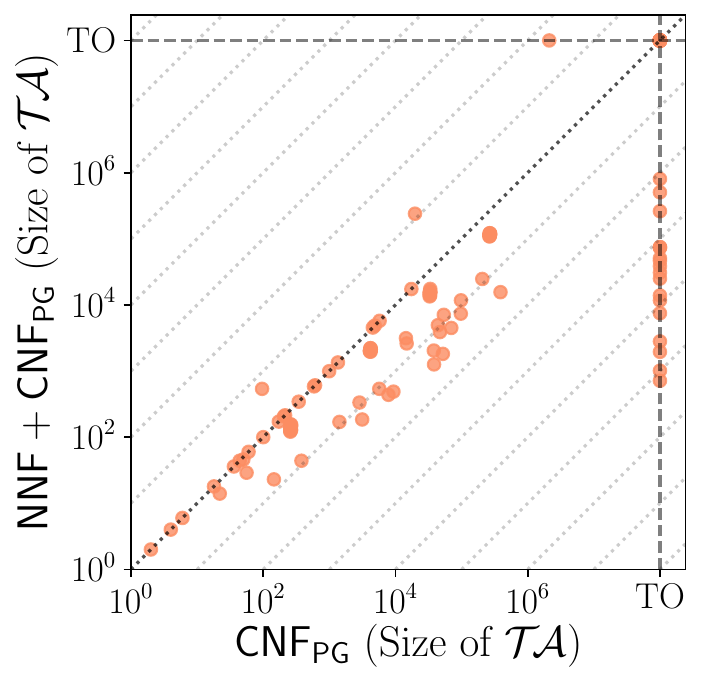}%
            \label{fig:plt:circ:norep:models:pol_vs_nnfpol}
        \end{subfigure}\hfill
        \begin{subfigure}[t]{0.29\textwidth}
            \includegraphics[width=.85\textwidth]{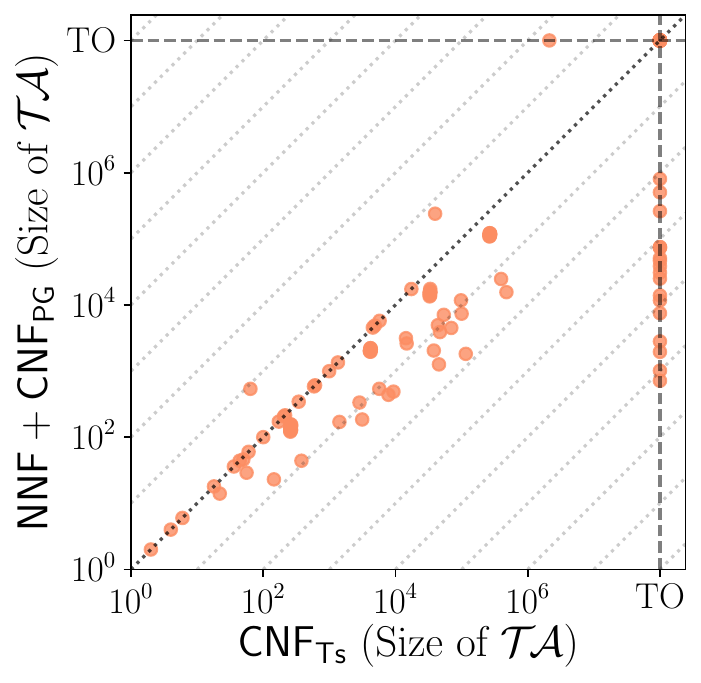}%
            \label{fig:plt:circ:norep:models:lab_vs_nnfpol}
        \end{subfigure}\hfill
        \begin{subfigure}[t]{0.29\textwidth}
            \includegraphics[width=.85\textwidth]{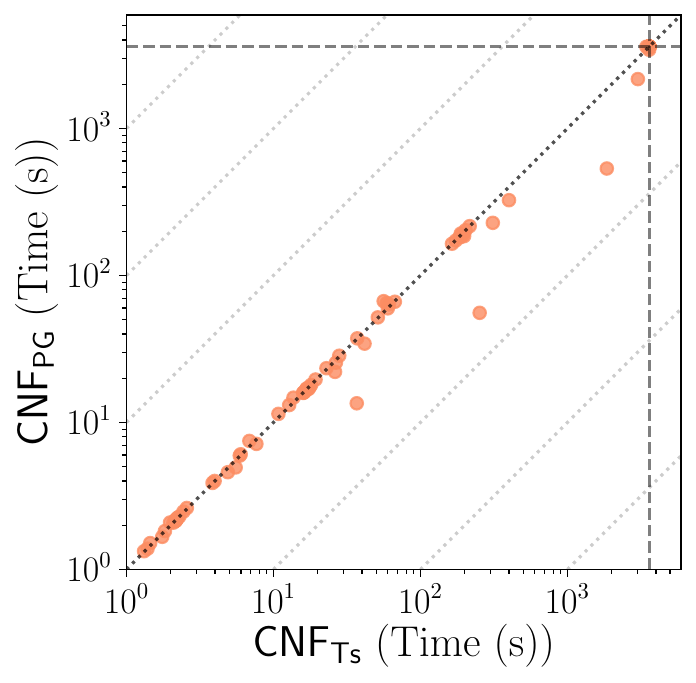}%
            \label{fig:plt:circ:norep:time:lab_vs_pol}
        \end{subfigure}\hfill
        \begin{subfigure}[t]{0.29\textwidth}
            \includegraphics[width=.85\textwidth]{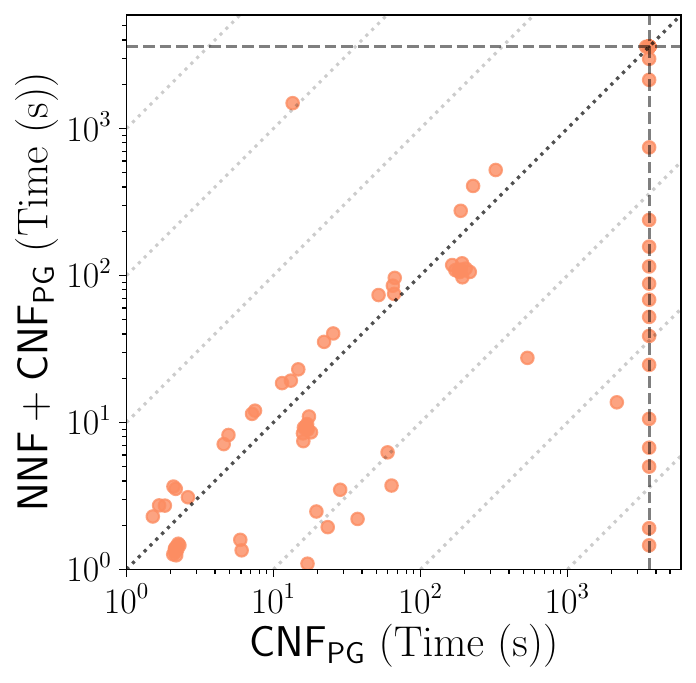}%
            \label{fig:plt:circ:norep:time:pol_vs_nnfpol}
        \end{subfigure}\hfill
        \begin{subfigure}[t]{0.29\textwidth}
            \includegraphics[width=.85\textwidth]{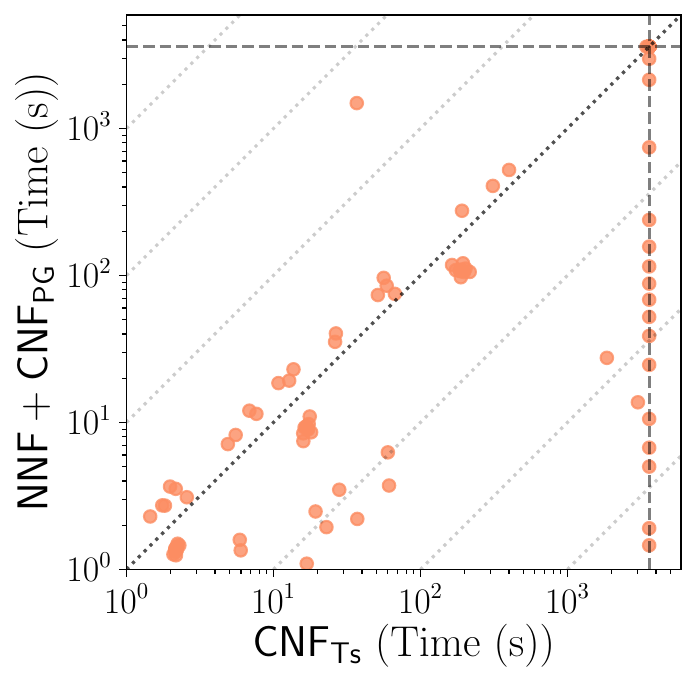}%
            \label{fig:plt:circ:norep:time:lab_vs_nnfpol}
            % \end{subfigure}
        \end{subfigure}
        \caption{Results for disjoint enumeration. %\TseitinCNF{}, \PlaistedCNF{} and $\NNFPlaisted{}$ reported 49, 44 and 27 timeouts, respectively (points on the dashed lines).
        }%
        \label{fig:plt:circ:norep:scatter}
    \end{subfigure}
    %%%%%%%%%%%% REP %%%%%%%%%%%%%
    \begin{subfigure}[t]{\textwidth}
        \begin{subfigure}[t]{0.29\textwidth}
            \includegraphics[width=.85\textwidth]{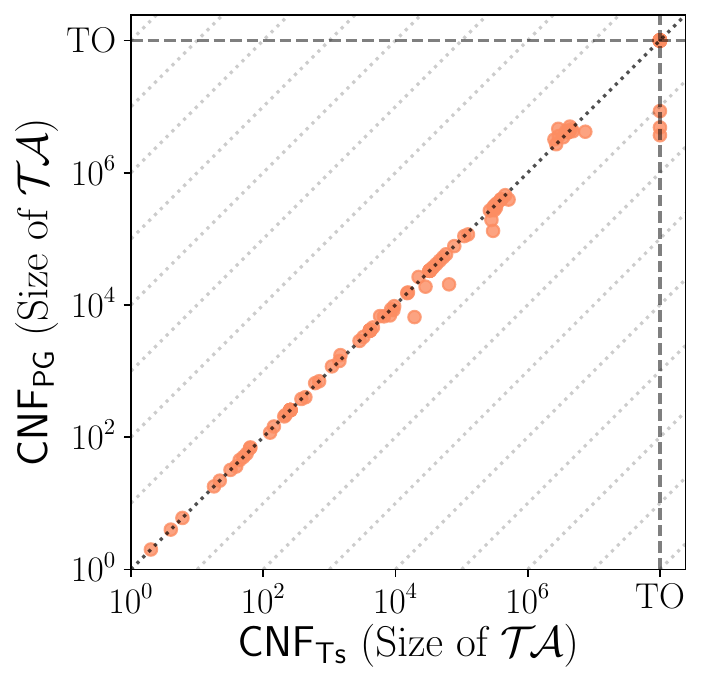}%
            \label{fig:plt:circ:rep:models:lab_vs_pol}
        \end{subfigure}\hfill
        \begin{subfigure}[t]{0.29\textwidth}
            \includegraphics[width=.85\textwidth]{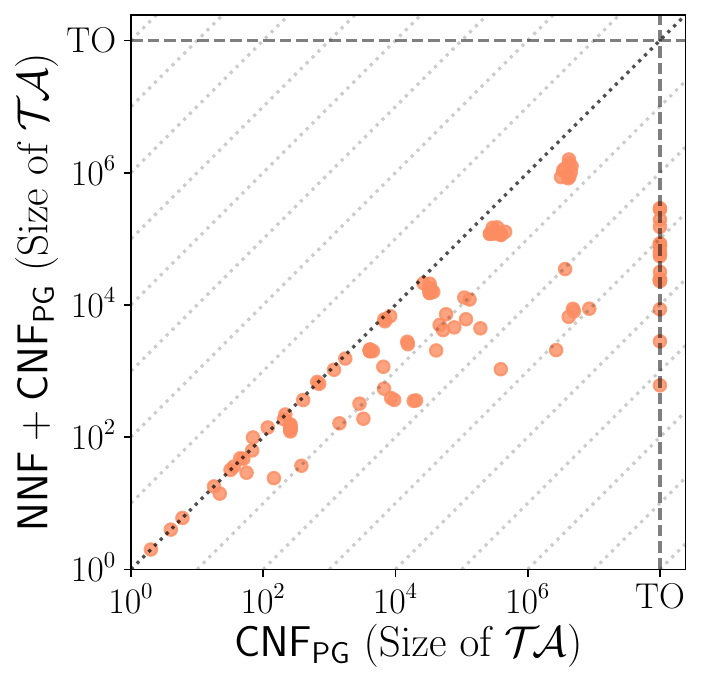}%
            \label{fig:plt:circ:rep:models:pol_vs_nnfpol}
        \end{subfigure}\hfill
        \begin{subfigure}[t]{0.29\textwidth}
            \includegraphics[width=.85\textwidth]{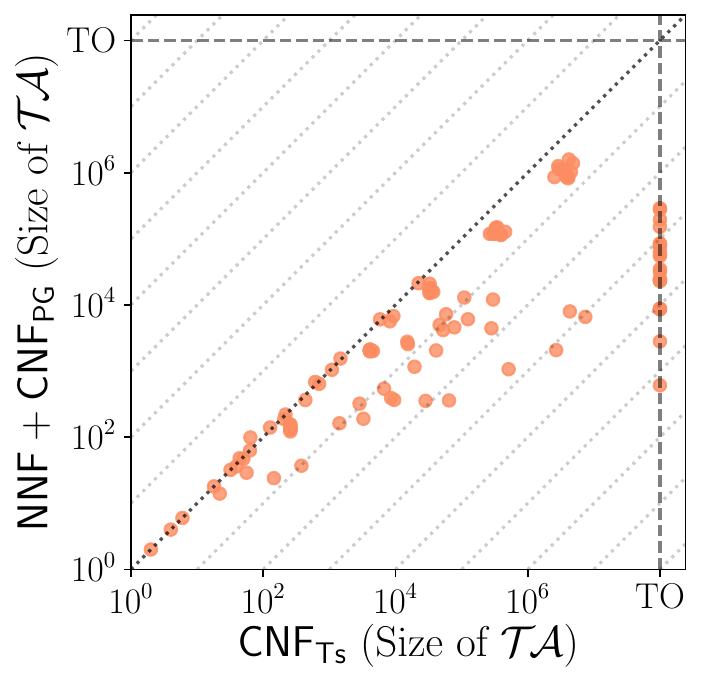}%
            \label{fig:plt:circ:rep:models:lab_vs_nnfpol}
        \end{subfigure}\hfill
        \begin{subfigure}[t]{0.29\textwidth}
            \includegraphics[width=.85\textwidth]{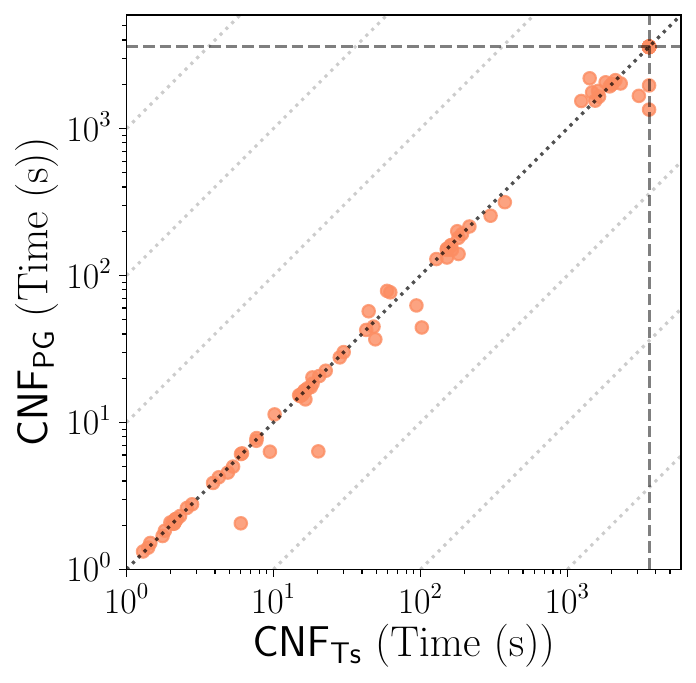}%
            \label{fig:plt:circ:rep:time:lab_vs_pol}
        \end{subfigure}\hfill
        \begin{subfigure}[t]{0.29\textwidth}
            \includegraphics[width=.85\textwidth]{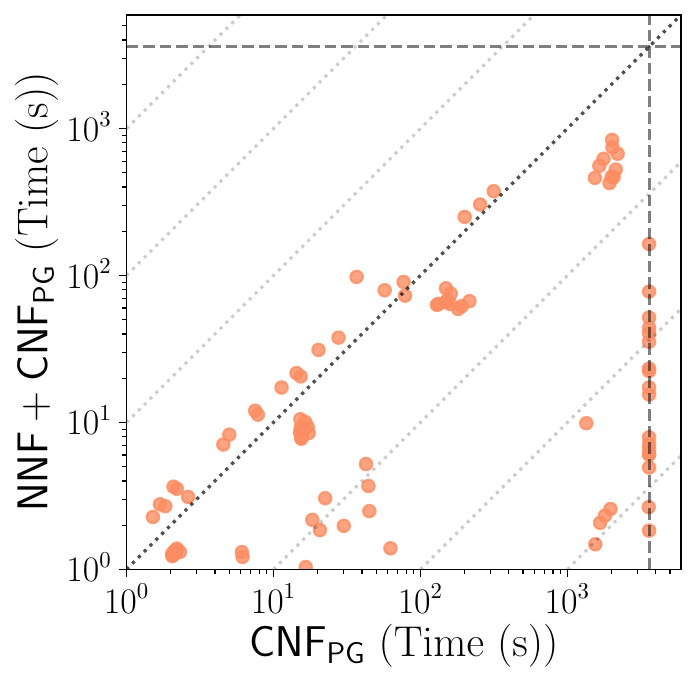}%
            \label{fig:plt:circ:rep:time:pol_vs_nnfpol}
        \end{subfigure}\hfill
        \begin{subfigure}[t]{0.29\textwidth}
            \includegraphics[width=.85\textwidth]{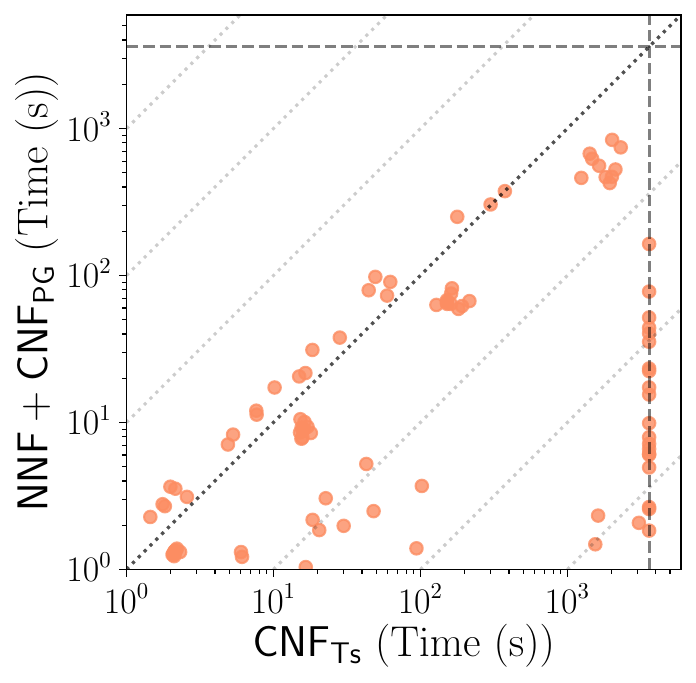}%
            \label{fig:plt:circ:rep:time:lab_vs_nnfpol}
            % \end{subfigure}
        \end{subfigure}
        \caption{Results for non-disjoint enumeration. %\TseitinCNF{}, \PlaistedCNF{} and $\NNFPlaisted{}$ reported 41, 38 and 3 timeouts, respectively (points on the dashed lines).
        }%
        \label{fig:plt:circ:rep:scatter}
    \end{subfigure}
    \caption{Results on the ISCAS'85 benchmarks using \mathsat{}.
        Plots in~\ref{fig:plt:circ:norep:scatter} and~\ref{fig:plt:circ:rep:scatter} compare CNF-izations by \TAna{} size (first row) and execution time (second row).
        Points on dashed lines represent timeouts, shown in~\ref{tab:timeouts:bool}.
        All axes use a logarithmic scale.}%
    \label{fig:plt:circ:scatter}
\end{figure}
\begin{figure}
    \centering
    % \begin{subfigure}[t]{\textwidth}
    \begin{subfigure}[t]{\textwidth}
        \begin{subfigure}[t]{0.29\textwidth}
            \includegraphics[width=.85\textwidth]{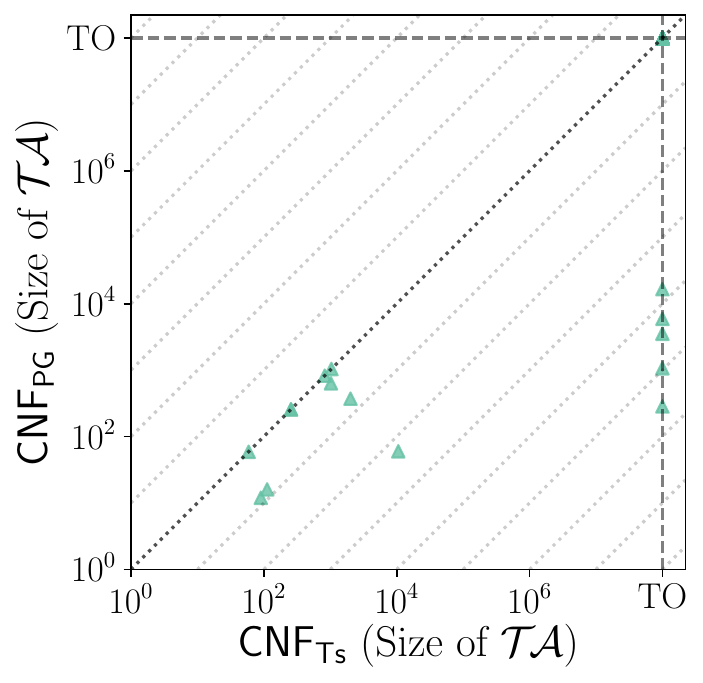}%
            \label{fig:plt:aig:norep:models:lab_vs_pol}
        \end{subfigure}\hfill
        \begin{subfigure}[t]{0.29\textwidth}
            \includegraphics[width=.85\textwidth]{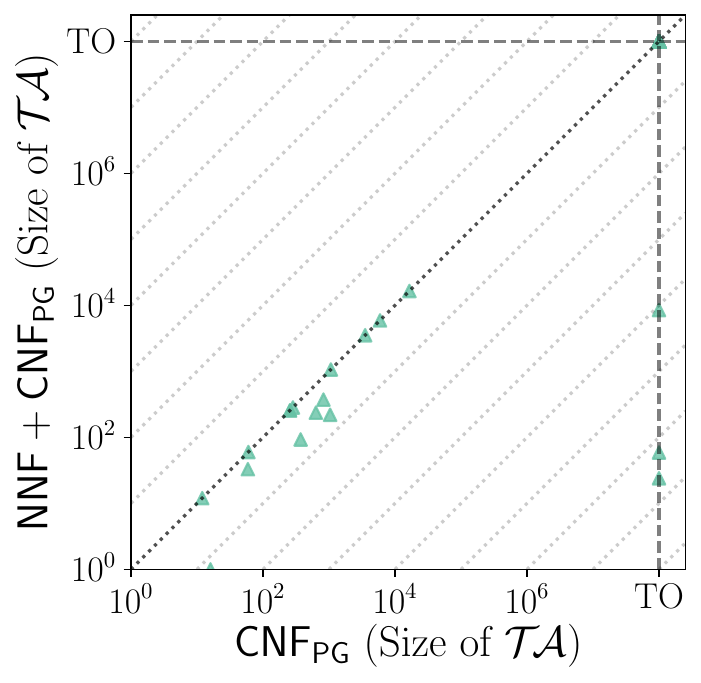}%
            \label{fig:plt:aig:norep:models:pol_vs_nnfpol}
        \end{subfigure}\hfill
        \begin{subfigure}[t]{0.29\textwidth}
            \includegraphics[width=.85\textwidth]{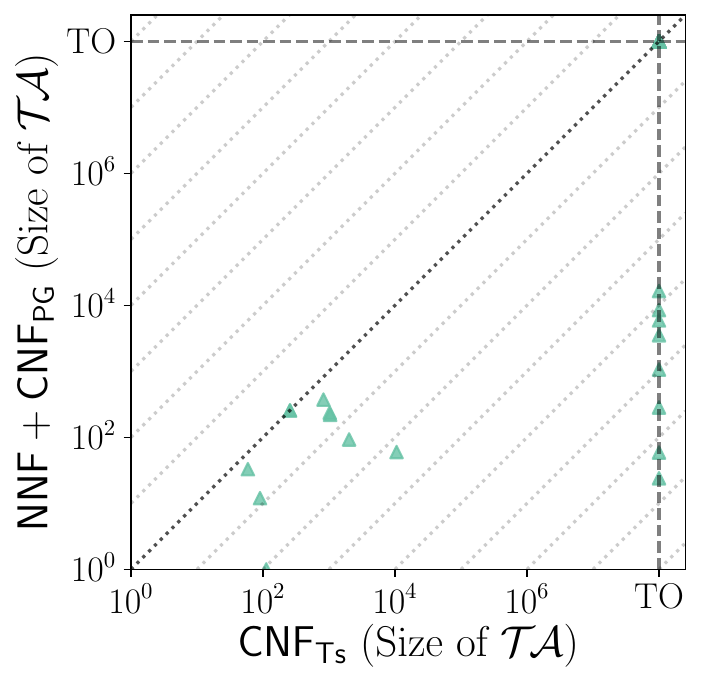}%
            \label{fig:plt:aig:norep:models:lab_vs_nnfpol}
        \end{subfigure}\hfill
        \begin{subfigure}[t]{0.29\textwidth}
            \includegraphics[width=.85\textwidth]{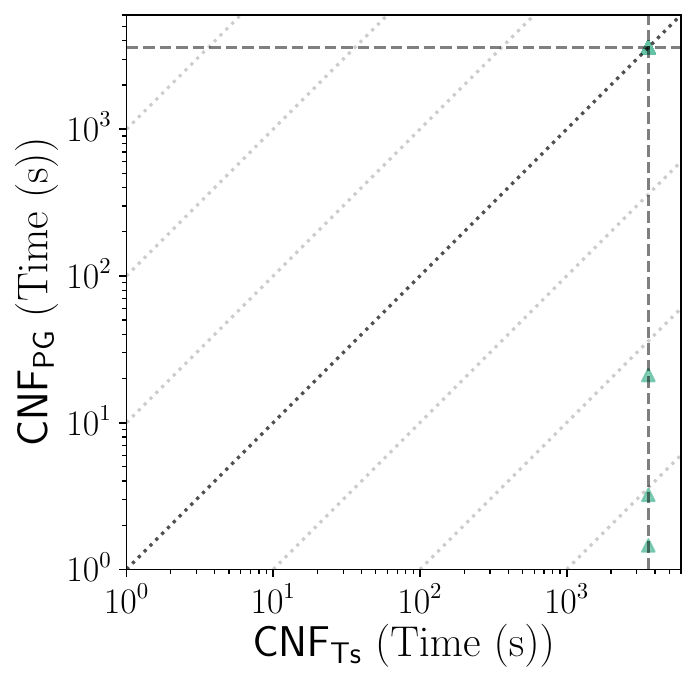}%
            \label{fig:plt:aig:norep:time:lab_vs_pol}
        \end{subfigure}\hfill
        \begin{subfigure}[t]{0.29\textwidth}
            \includegraphics[width=.85\textwidth]{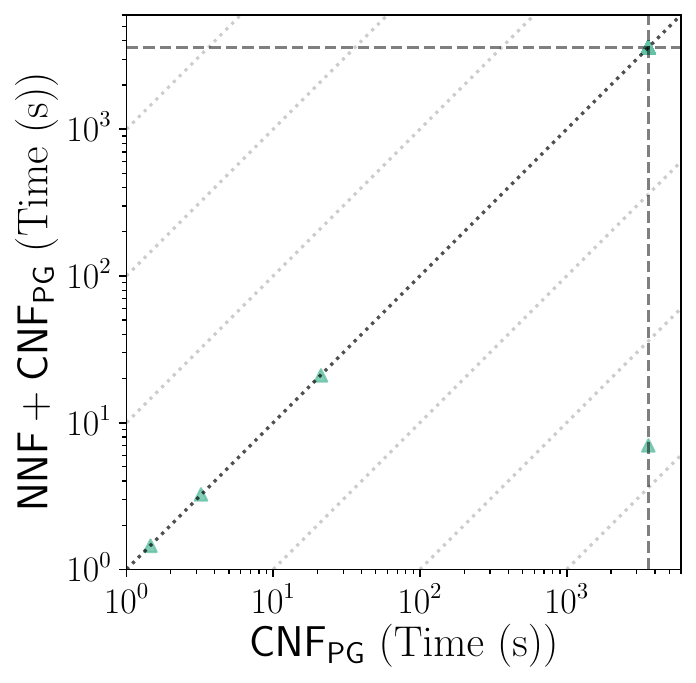}%
            \label{fig:plt:aig:norep:time:pol_vs_nnfpol}
        \end{subfigure}\hfill
        \begin{subfigure}[t]{0.29\textwidth}
            \includegraphics[width=.85\textwidth]{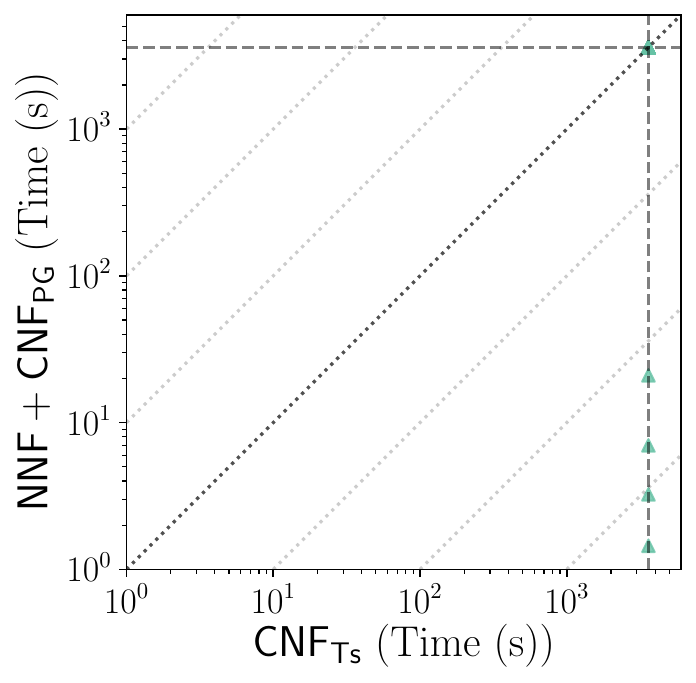}%
            \label{fig:plt:aig:norep:time:lab_vs_nnfpol}
            % \end{subfigure}
        \end{subfigure}
        \caption{Results for disjoint enumeration.
            %\TseitinCNF{},\ \PlaistedCNF{} and $\NNFPlaisted{}$ reported 79, 73 and 72 timeouts, respectively.
        }%
        \label{fig:plt:aig:norep:scatter}
    \end{subfigure}
    %%%%%%%%%%%% REP %%%%%%%%%%%%%
    \begin{subfigure}[t]{\textwidth}
        \begin{subfigure}[t]{0.29\textwidth}
            \includegraphics[width=.85\textwidth]{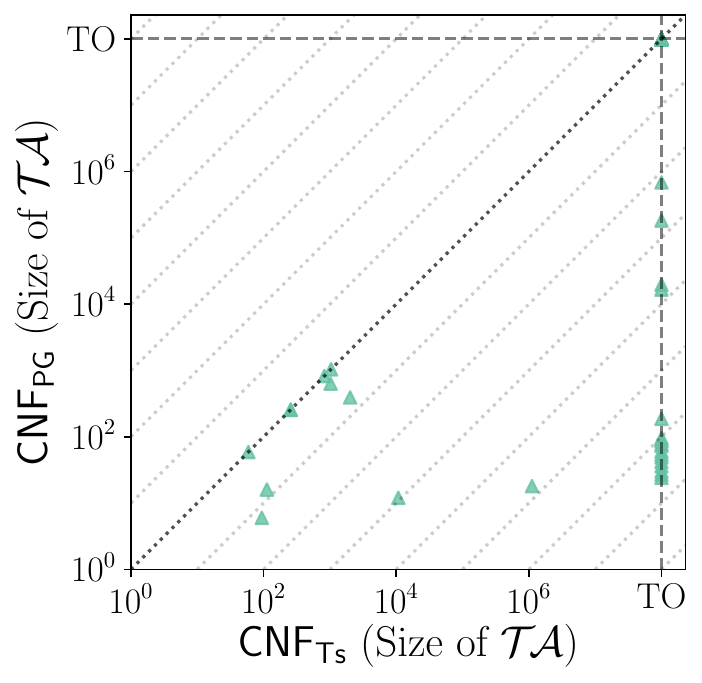}%
            \label{fig:plt:aig:rep:models:lab_vs_pol}
        \end{subfigure}\hfill
        \begin{subfigure}[t]{0.29\textwidth}
            \includegraphics[width=.85\textwidth]{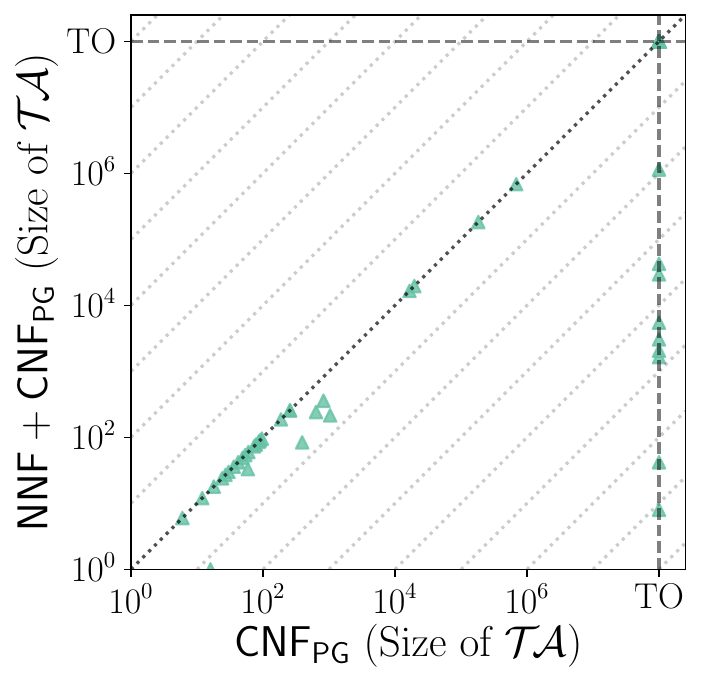}%
            \label{fig:plt:aig:rep:models:pol_vs_nnfpol}
        \end{subfigure}\hfill
        \begin{subfigure}[t]{0.29\textwidth}
            \includegraphics[width=.85\textwidth]{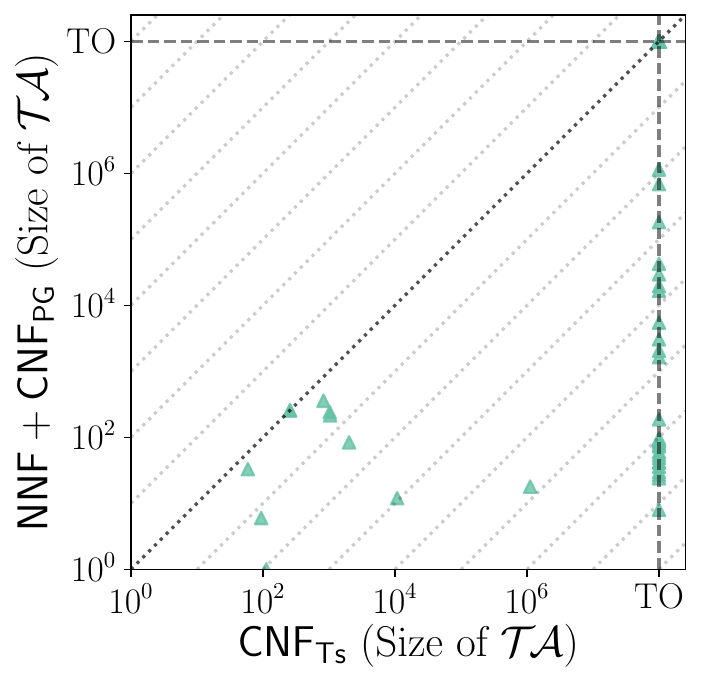}%
            \label{fig:plt:aig:rep:models:lab_vs_nnfpol}
        \end{subfigure}\hfill
        \begin{subfigure}[t]{0.29\textwidth}
            \includegraphics[width=.85\textwidth]{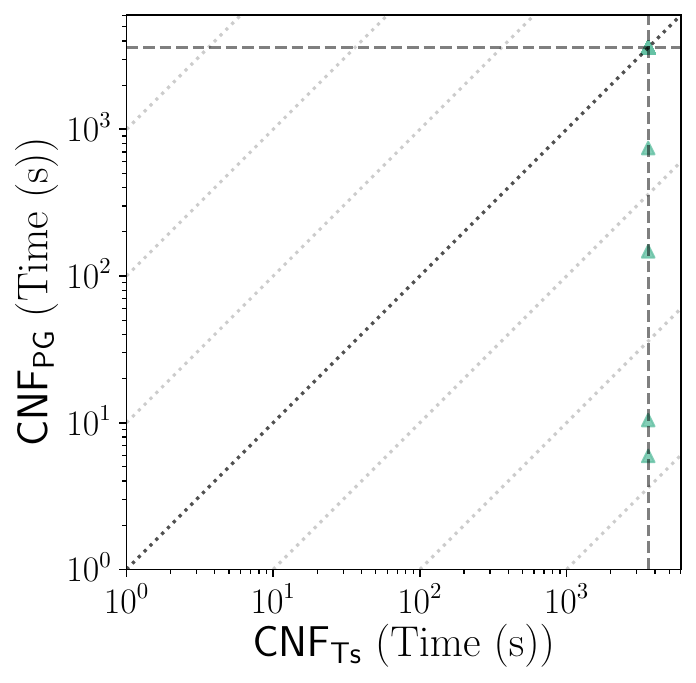}%
            \label{fig:plt:aig:rep:time:lab_vs_pol}
        \end{subfigure}\hfill
        \begin{subfigure}[t]{0.29\textwidth}
            \includegraphics[width=.85\textwidth]{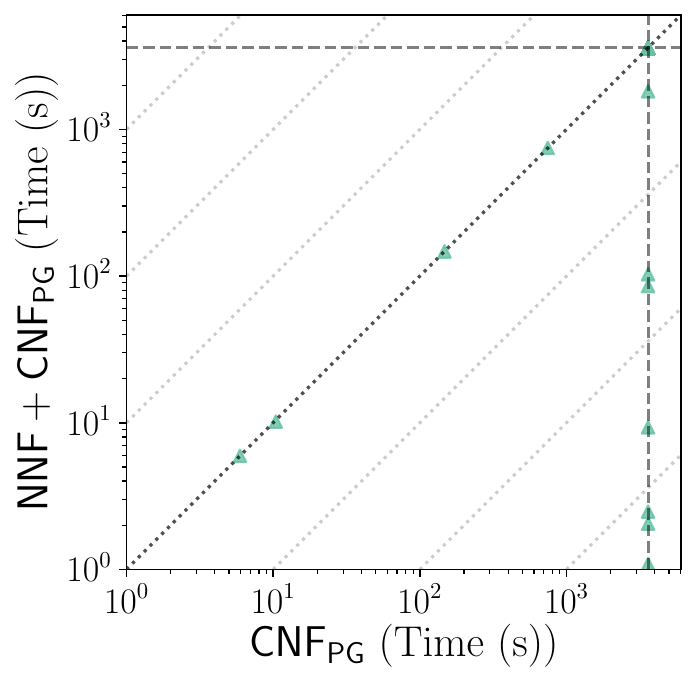}%
            \label{fig:plt:aig:rep:time:pol_vs_nnfpol}
        \end{subfigure}\hfill
        \begin{subfigure}[t]{0.29\textwidth}
            \includegraphics[width=.85\textwidth]{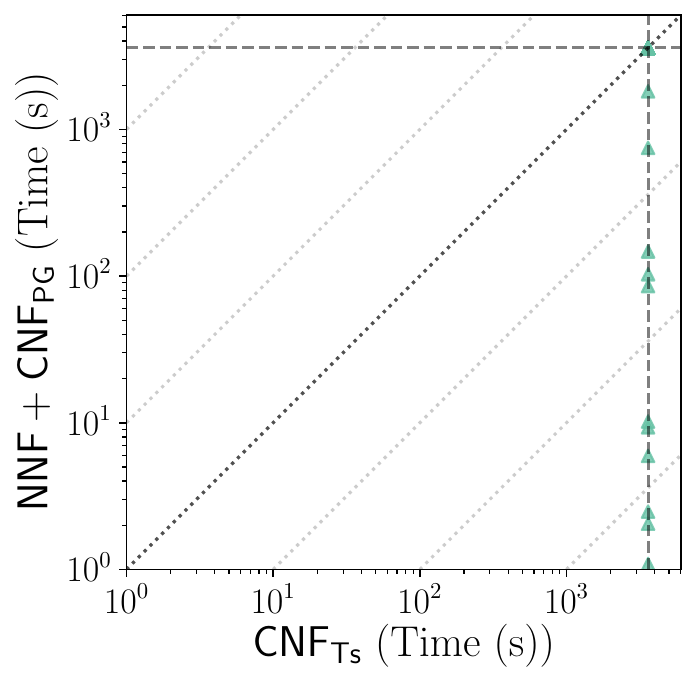}%
            \label{fig:plt:aig:rep:time:lab_vs_nnfpol}
            % \end{subfigure}
        \end{subfigure}
        \caption{Results for non-disjoint enumeration.
            %\TseitinCNF{},\ \PlaistedCNF{} and $\NNFPlaisted{}$ reported 78, 60 and 51 timeouts, respectively.
        }%
        \label{fig:plt:aig:rep:scatter}
    \end{subfigure}
    \caption{Results on the AIG benchmarks using \mathsat{}.
        Plots in~\ref{fig:plt:aig:norep:scatter} and~\ref{fig:plt:aig:rep:scatter} compare CNF-izations by \TAna{} size (first row) and execution time (second row).
        Points on dashed lines represent timeouts, shown in~\ref{tab:timeouts:bool}.
        All axes use a logarithmic scale.}%
    \label{fig:plt:aig:scatter}
\end{figure}
% ---- AllSAT Tabula ----

\begin{figure}
    \centering
    % \begin{subfigure}[t]{\textwidth}
    \begin{subfigure}[t]{\textwidth}
        \begin{subfigure}[t]{0.29\textwidth}
            \includegraphics[width=.85\textwidth]{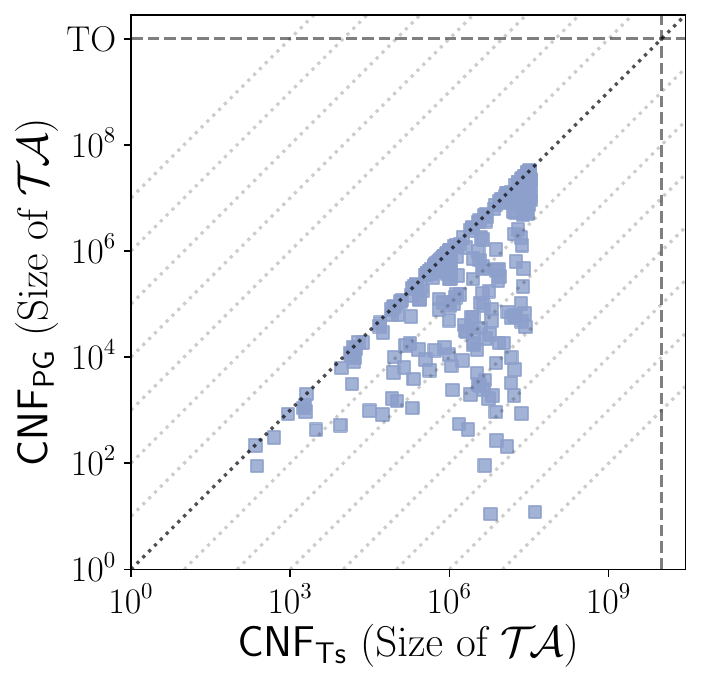}%
            \label{fig:plt:tabula:syn:bool:norep:models:lab_vs_pol}
        \end{subfigure}\hfill
        \begin{subfigure}[t]{0.29\textwidth}
            \includegraphics[width=.85\textwidth]{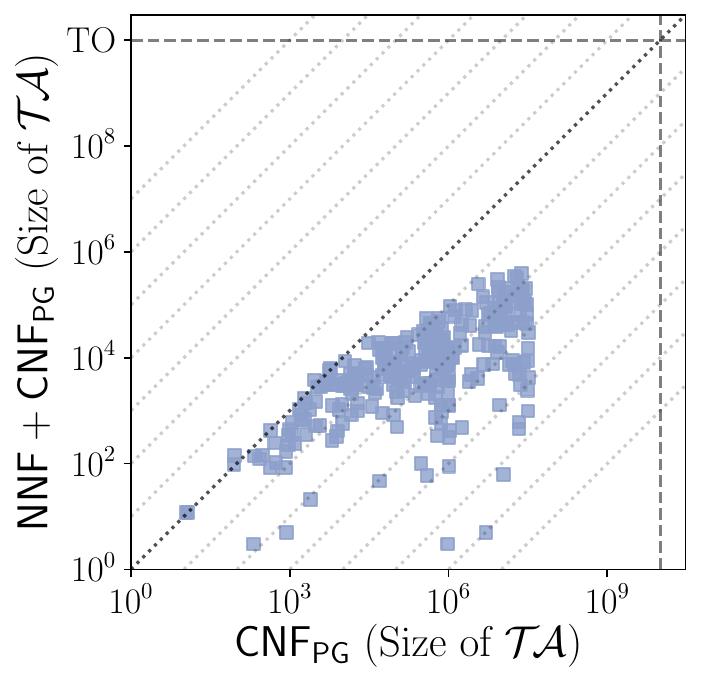}%
            \label{fig:plt:tabula:syn:bool:norep:models:pol_vs_nnfpol}
        \end{subfigure}\hfill
        \begin{subfigure}[t]{0.29\textwidth}
            \includegraphics[width=.85\textwidth]{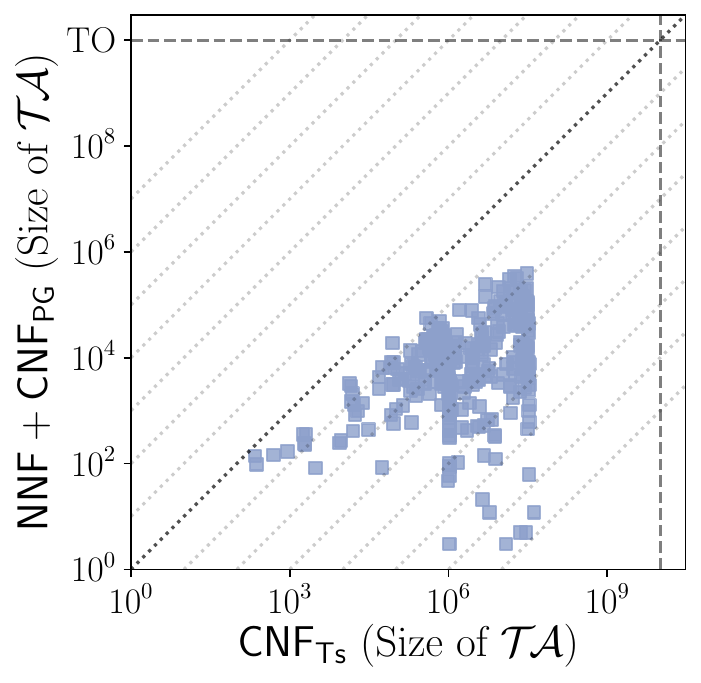}%
            \label{fig:plt:tabula:syn:bool:norep:models:lab_vs_nnfpol}
        \end{subfigure}\hfill
        \begin{subfigure}[t]{0.29\textwidth}
            \includegraphics[width=.85\textwidth]{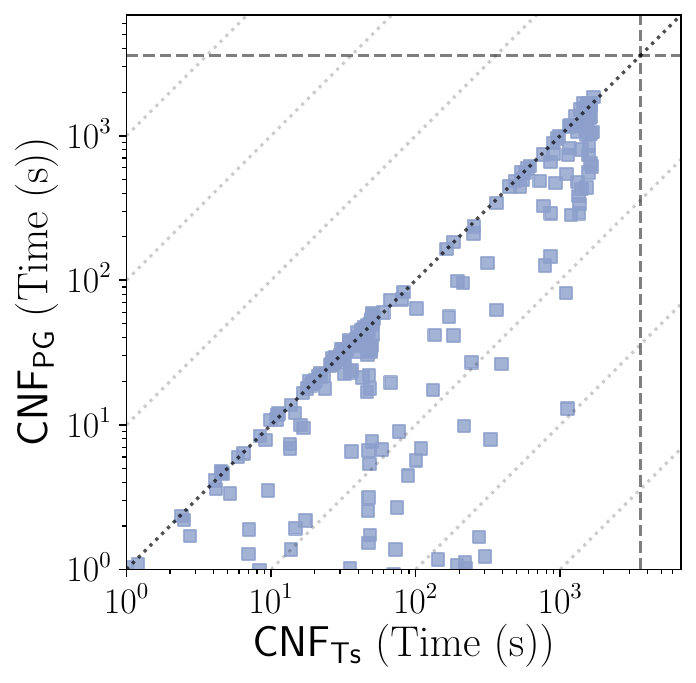}%
            \label{fig:plt:tabula:syn:bool:norep:time:lab_vs_pol}
        \end{subfigure}\hfill
        \begin{subfigure}[t]{0.29\textwidth}
            \includegraphics[width=.85\textwidth]{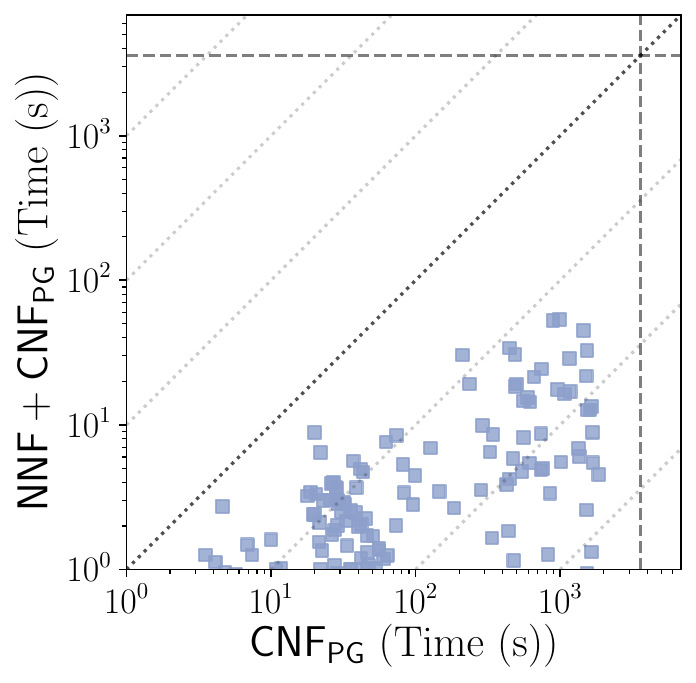}%
            \label{fig:plt:tabula:syn:bool:norep:time:pol_vs_nnfpol}
        \end{subfigure}\hfill
        \begin{subfigure}[t]{0.29\textwidth}
            \includegraphics[width=.85\textwidth]{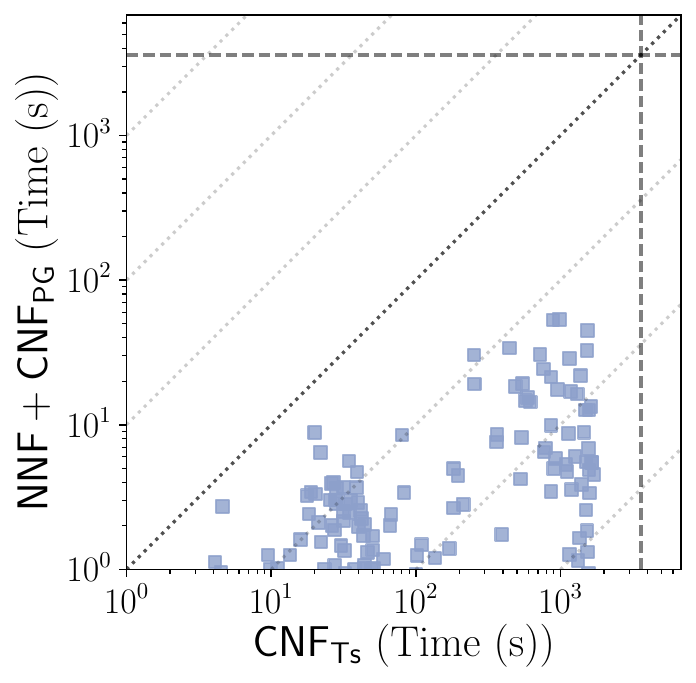}%
            \label{fig:plt:tabula:syn:bool:norep:time:lab_vs_nnfpol}
            % \end{subfigure}
        \end{subfigure}
        \caption{Results for disjoint enumeration. %\TseitinCNF{}, \PlaistedCNF{} and $\NNFPlaisted{}$ reported 163, 94 and 20 timeouts, respectively (points on the dashed lines).
        }%
        \label{fig:plt:tabula:syn:bool:norep:scatter}
    \end{subfigure}
    \caption{Results on the Boolean synthetic benchmarks using \tabularallsat{}.
        Plots in~\ref{fig:plt:tabula:syn:bool:norep:scatter} compare CNF-izations by \TAna{} size (first row) and execution time (second row).
        Points on dashed lines represent timeouts, shown in~\ref{tab:timeouts:tabula:bool}.
        All axes use a logarithmic scale.}%
    \label{fig:plt:tabula:syn:bool:scatter}
\end{figure}

\begin{figure}
    \centering
    % \begin{subfigure}[t]{\textwidth}
    \begin{subfigure}[t]{\textwidth}
        \begin{subfigure}[t]{0.29\textwidth}
            \includegraphics[width=.85\textwidth]{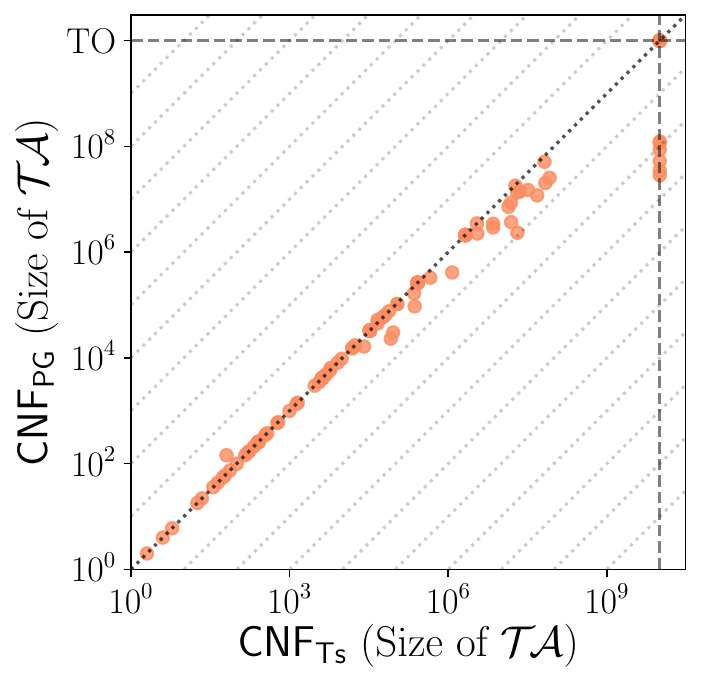}%
            \label{fig:plt:tabula:circ:norep:models:lab_vs_pol}
        \end{subfigure}\hfill
        \begin{subfigure}[t]{0.29\textwidth}
            \includegraphics[width=.85\textwidth]{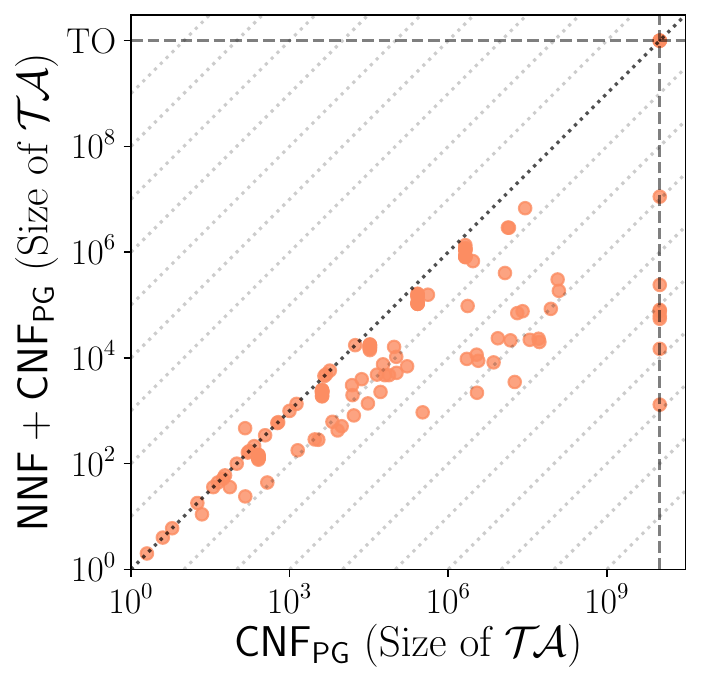}%
            \label{fig:plt:tabula:circ:norep:models:pol_vs_nnfpol}
        \end{subfigure}\hfill
        \begin{subfigure}[t]{0.29\textwidth}
            \includegraphics[width=.85\textwidth]{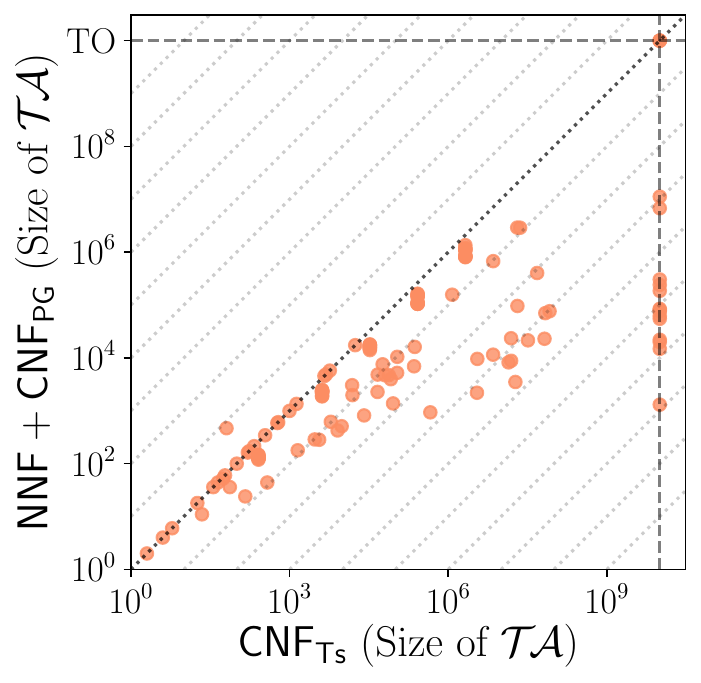}%
            \label{fig:plt:tabula:circ:norep:models:lab_vs_nnfpol}
        \end{subfigure}\hfill
        \begin{subfigure}[t]{0.29\textwidth}
            \includegraphics[width=.85\textwidth]{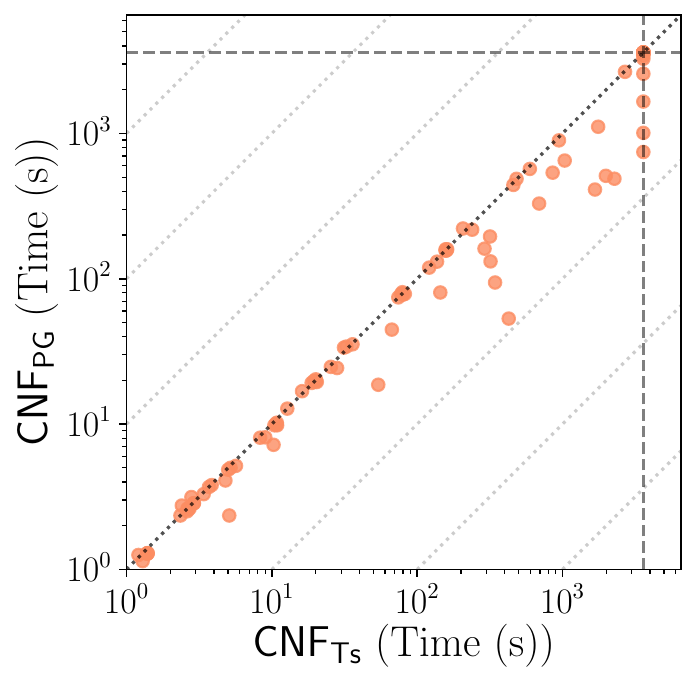}%
            \label{fig:plt:tabula:circ:norep:time:lab_vs_pol}
        \end{subfigure}\hfill
        \begin{subfigure}[t]{0.29\textwidth}
            \includegraphics[width=.85\textwidth]{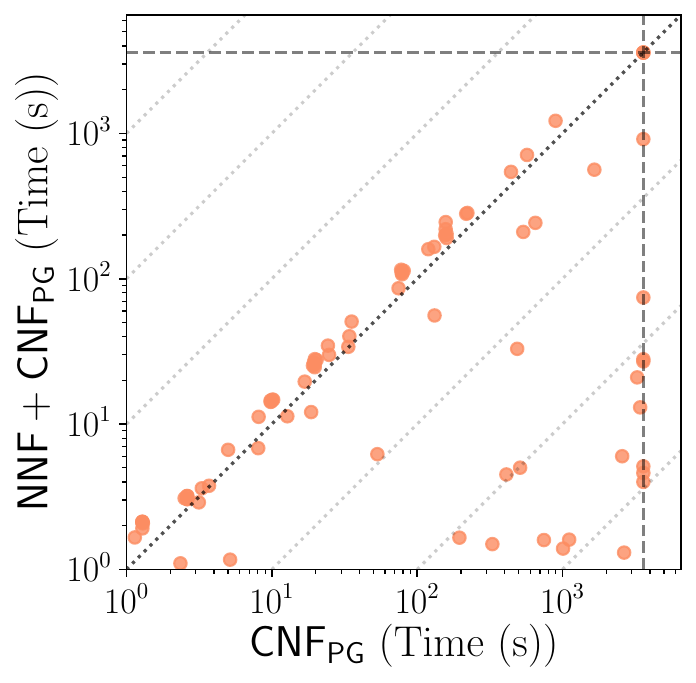}%
            \label{fig:plt:tabula:circ:norep:time:pol_vs_nnfpol}
        \end{subfigure}\hfill
        \begin{subfigure}[t]{0.29\textwidth}
            \includegraphics[width=.85\textwidth]{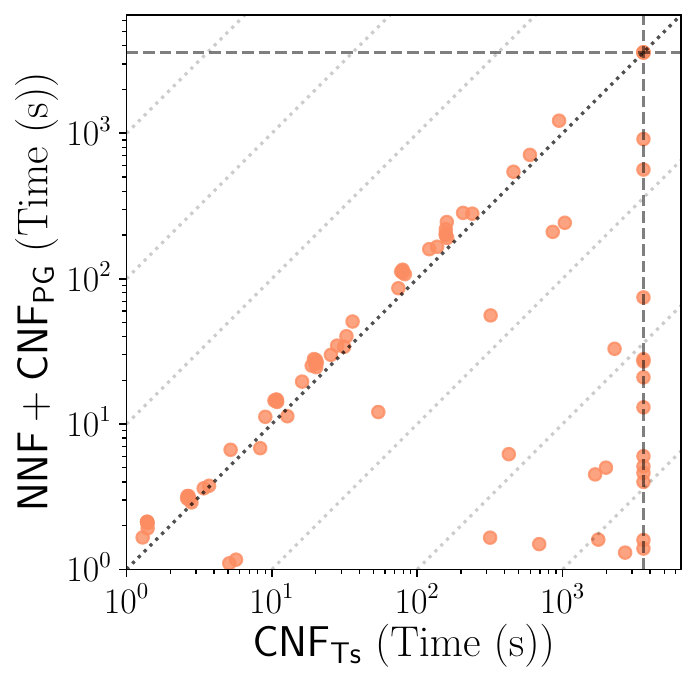}%
            \label{fig:plt:tabula:circ:norep:time:lab_vs_nnfpol}
            % \end{subfigure}
        \end{subfigure}
        \caption{Results for disjoint enumeration. %\TseitinCNF{}, \PlaistedCNF{} and $\NNFPlaisted{}$ reported 49, 44 and 27 timeouts, respectively (points on the dashed lines).
        }%
        \label{fig:plt:tabula:circ:norep:scatter}
    \end{subfigure}
    \caption{Results on the ISCAS'85 benchmarks using \tabularallsat{}.
        Plots in~\ref{fig:plt:tabula:circ:norep:scatter} compare CNF-izations by \TAna{} size (first row) and execution time (second row).
        Points on dashed lines represent timeouts, shown in~\ref{tab:timeouts:tabula:bool}.
        All axes use a logarithmic scale.}%
    \label{fig:plt:tabula:circ:scatter}
\end{figure}
\begin{figure}
    \centering
    % \begin{subfigure}[t]{\textwidth}
    \begin{subfigure}[t]{\textwidth}
        \begin{subfigure}[t]{0.29\textwidth}
            \includegraphics[width=.85\textwidth]{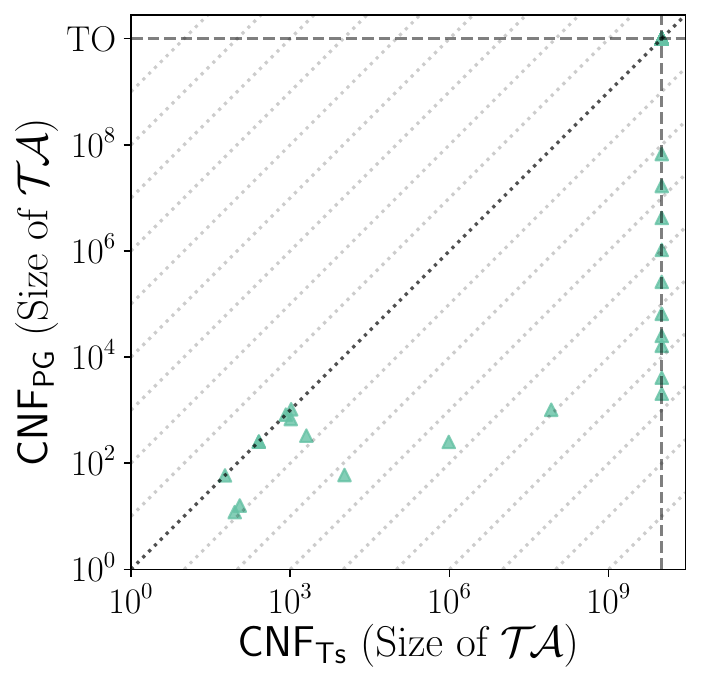}%
            \label{fig:plt:tabula:aig:norep:models:lab_vs_pol}
        \end{subfigure}\hfill
        \begin{subfigure}[t]{0.29\textwidth}
            \includegraphics[width=.85\textwidth]{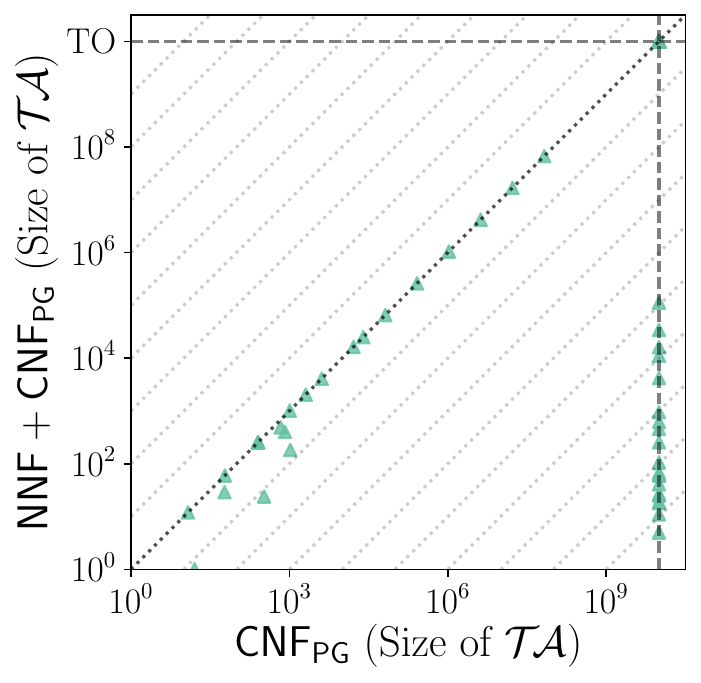}%
            \label{fig:plt:tabula:aig:norep:models:pol_vs_nnfpol}
        \end{subfigure}\hfill
        \begin{subfigure}[t]{0.29\textwidth}
            \includegraphics[width=.85\textwidth]{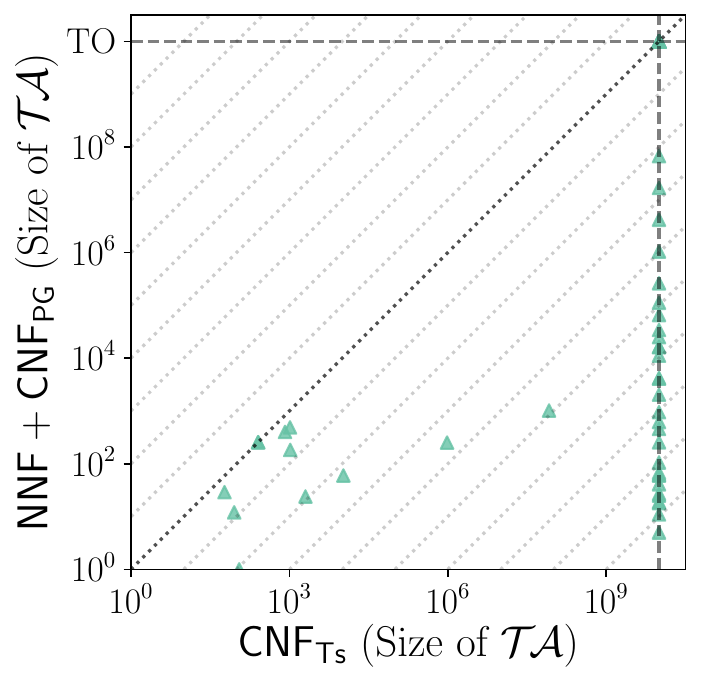}%
            \label{fig:plt:tabula:aig:norep:models:lab_vs_nnfpol}
        \end{subfigure}\hfill
        \begin{subfigure}[t]{0.29\textwidth}
            \includegraphics[width=.85\textwidth]{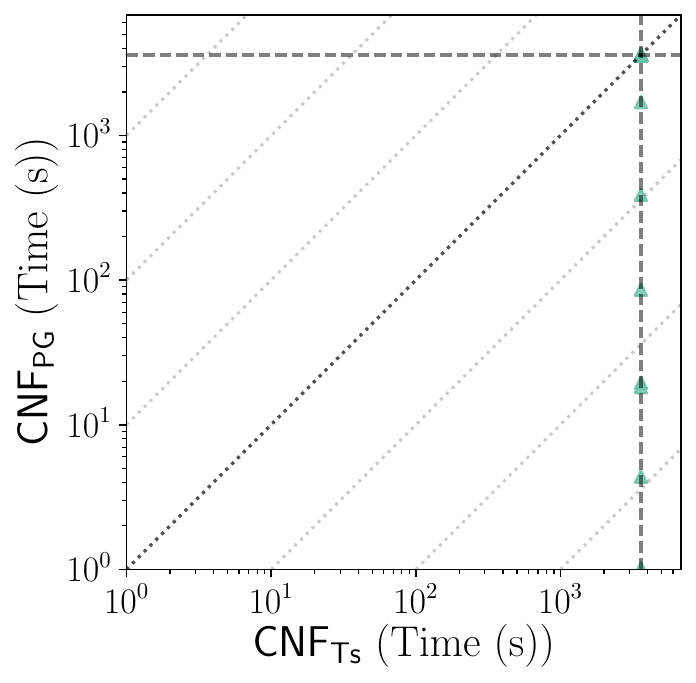}%
            \label{fig:plt:tabula:aig:norep:time:lab_vs_pol}
        \end{subfigure}\hfill
        \begin{subfigure}[t]{0.29\textwidth}
            \includegraphics[width=.85\textwidth]{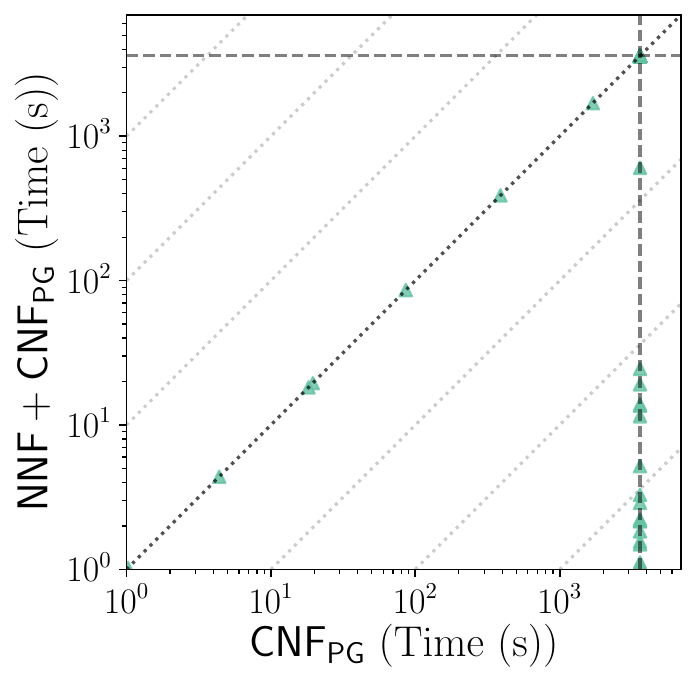}%
            \label{fig:plt:tabula:aig:norep:time:pol_vs_nnfpol}
        \end{subfigure}\hfill
        \begin{subfigure}[t]{0.29\textwidth}
            \includegraphics[width=.85\textwidth]{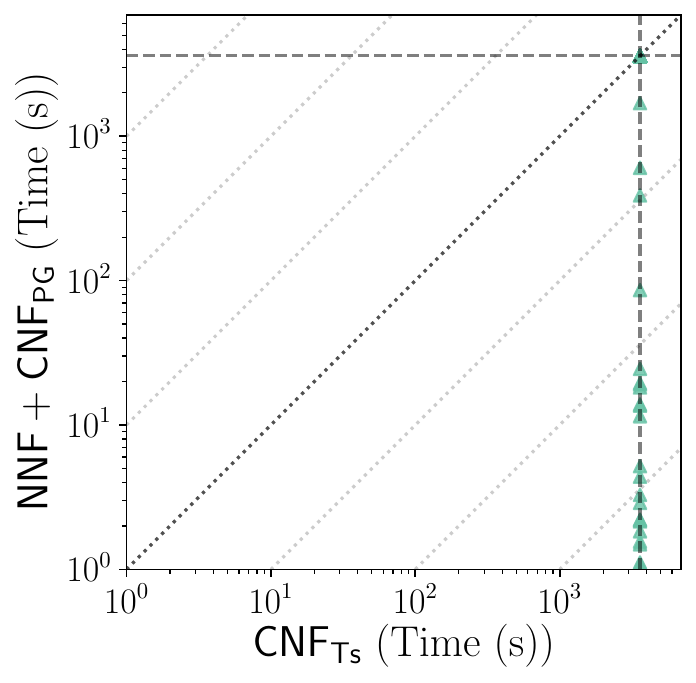}%
            \label{fig:plt:tabula:aig:norep:time:lab_vs_nnfpol}
            % \end{subfigure}
        \end{subfigure}
        \caption{Results for disjoint enumeration.
            %\TseitinCNF{},\ \PlaistedCNF{} and $\NNFPlaisted{}$ reported 79, 73 and 72 timeouts, respectively.
        }%
        \label{fig:plt:tabula:aig:norep:scatter}
    \end{subfigure}
    \caption{Results on the AIG benchmarks using \tabularallsat{}.
        Plots in~\ref{fig:plt:tabula:aig:norep:scatter} compare CNF-izations by \TAna{} size (first row) and execution time (second row).
        Points on dashed lines represent timeouts, shown in~\ref{tab:timeouts:tabula:bool}.
        All axes use a logarithmic scale.}%
    \label{fig:plt:tabula:aig:scatter}
\end{figure}

% ---- ALLSMT Mathsat ----
\begin{figure}
    \centering
    % \begin{subfigure}[t]{\textwidth}
    \begin{subfigure}[t]{\textwidth}
        \begin{subfigure}[t]{0.29\textwidth}
            \includegraphics[width=.85\textwidth]{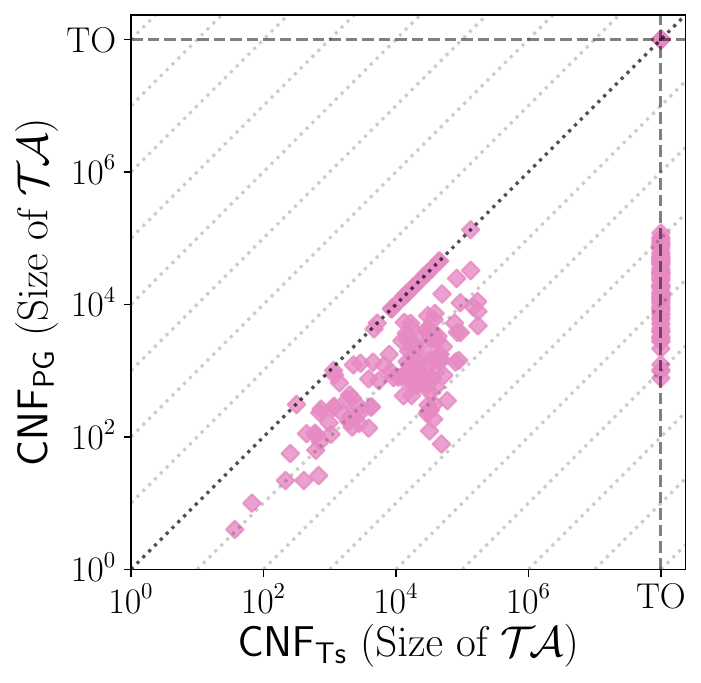}%
            \label{fig:plt:syn:lra:norep:models:lab_vs_pol}
        \end{subfigure}\hfill
        \begin{subfigure}[t]{0.29\textwidth}
            \includegraphics[width=.85\textwidth]{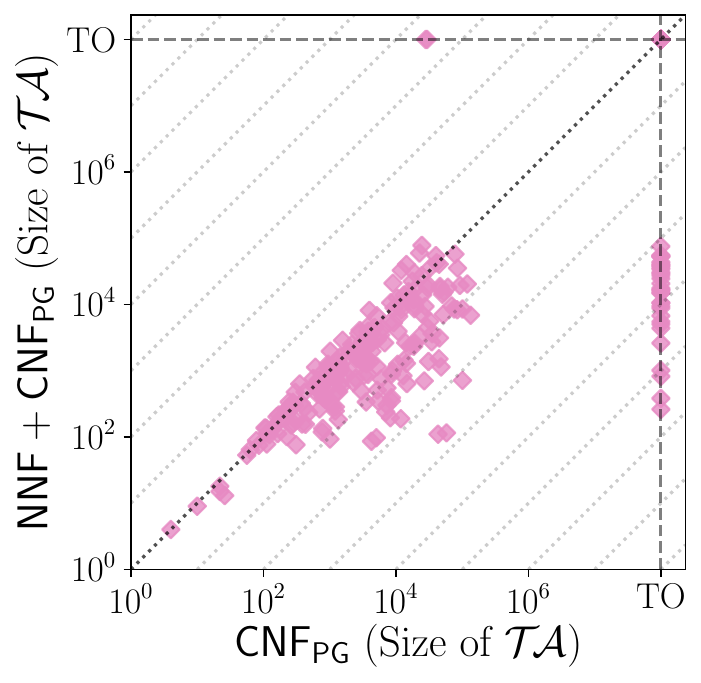}%
            \label{fig:plt:syn:lra:norep:models:pol_vs_nnfpol}
        \end{subfigure}\hfill
        \begin{subfigure}[t]{0.29\textwidth}
            \includegraphics[width=.85\textwidth]{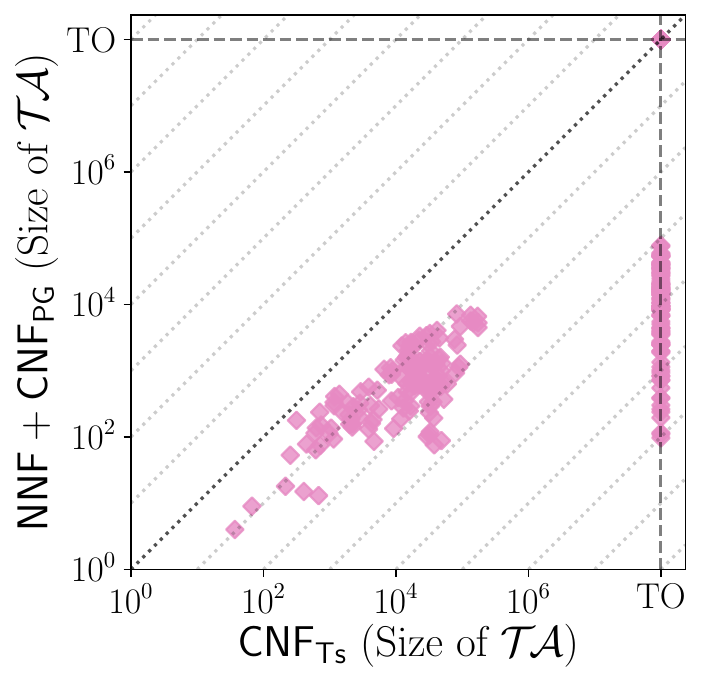}%
            \label{fig:plt:syn:lra:norep:models:lab_vs_nnfpol}
        \end{subfigure}\hfill
        \begin{subfigure}[t]{0.29\textwidth}
            \includegraphics[width=.85\textwidth]{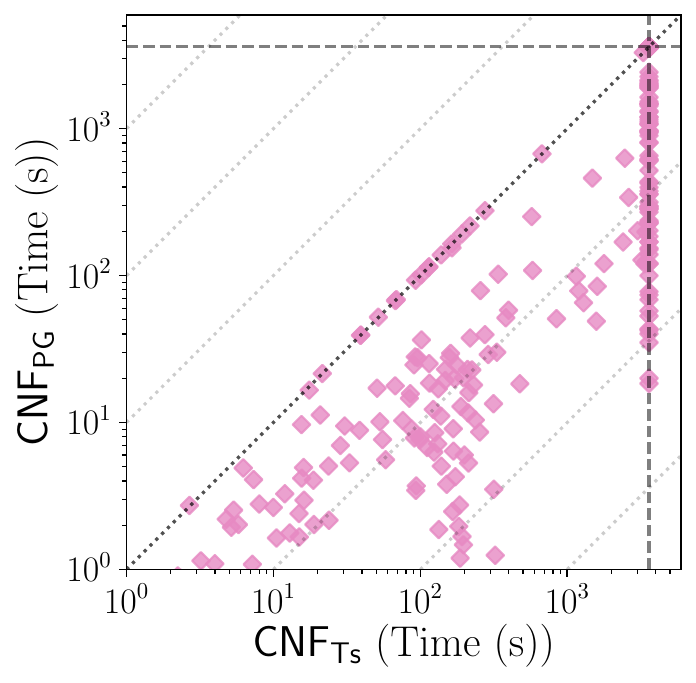}%
            \label{fig:plt:syn:lra:norep:time:lab_vs_pol}
        \end{subfigure}\hfill
        \begin{subfigure}[t]{0.29\textwidth}
            \includegraphics[width=.85\textwidth]{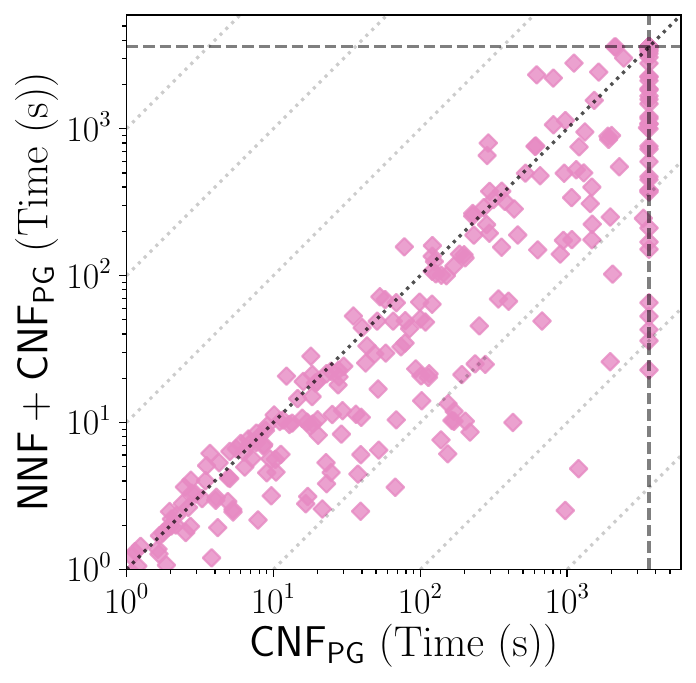}%
            \label{fig:plt:syn:lra:norep:time:pol_vs_nnfpol}
        \end{subfigure}\hfill
        \begin{subfigure}[t]{0.29\textwidth}
            \includegraphics[width=.85\textwidth]{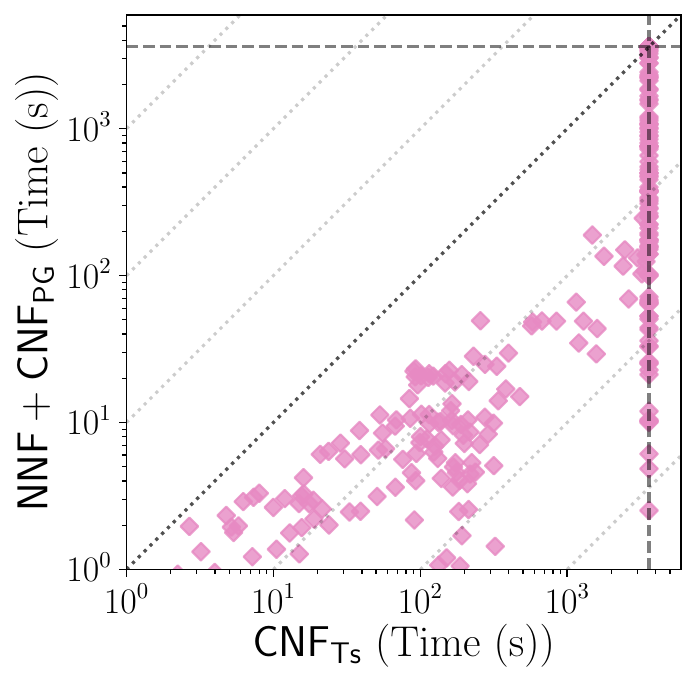}%
            \label{fig:plt:syn:lra:norep:time:lab_vs_nnfpol}
            % \end{subfigure}
        \end{subfigure}
        \caption{Results for disjoint enumeration. %\TseitinCNF{}, \PlaistedCNF{} and $\NNFPlaisted{}$ reported 171, 98 and 58 timeouts, respectively (points on the dashed lines).
        }%
        \label{fig:plt:syn:lra:norep:scatter}
    \end{subfigure}
    %%%%%%%%%%%% REP %%%%%%%%%%%%%
    \begin{subfigure}[t]{\textwidth}
        \begin{subfigure}[t]{0.29\textwidth}
            \includegraphics[width=.85\textwidth]{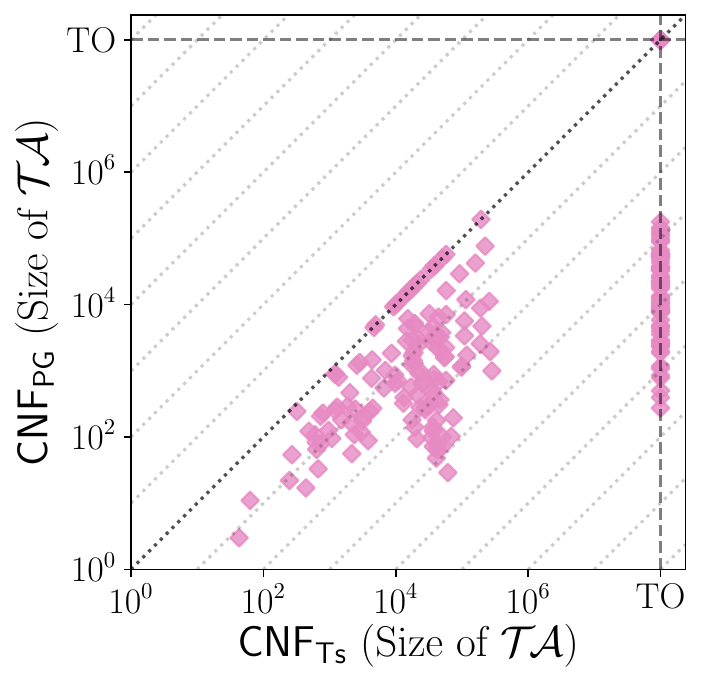}%
            \label{fig:plt:syn:lra:rep:models:lab_vs_pol}
        \end{subfigure}\hfill
        \begin{subfigure}[t]{0.29\textwidth}
            \includegraphics[width=.85\textwidth]{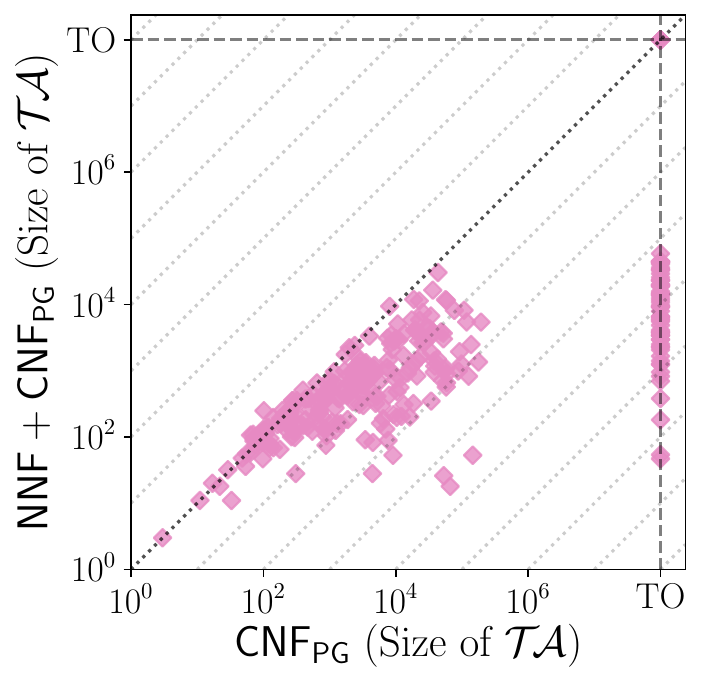}%
            \label{fig:plt:syn:lra:rep:models:pol_vs_nnfpol}
        \end{subfigure}\hfill
        \begin{subfigure}[t]{0.29\textwidth}
            \includegraphics[width=.85\textwidth]{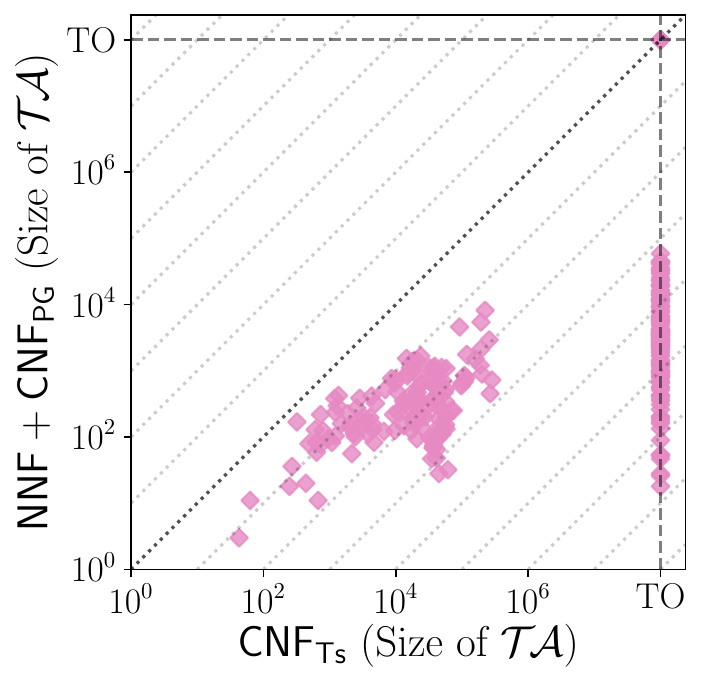}%
            \label{fig:plt:syn:lra:rep:models:lab_vs_nnfpol}
        \end{subfigure}\hfill
        \begin{subfigure}[t]{0.29\textwidth}
            \includegraphics[width=.85\textwidth]{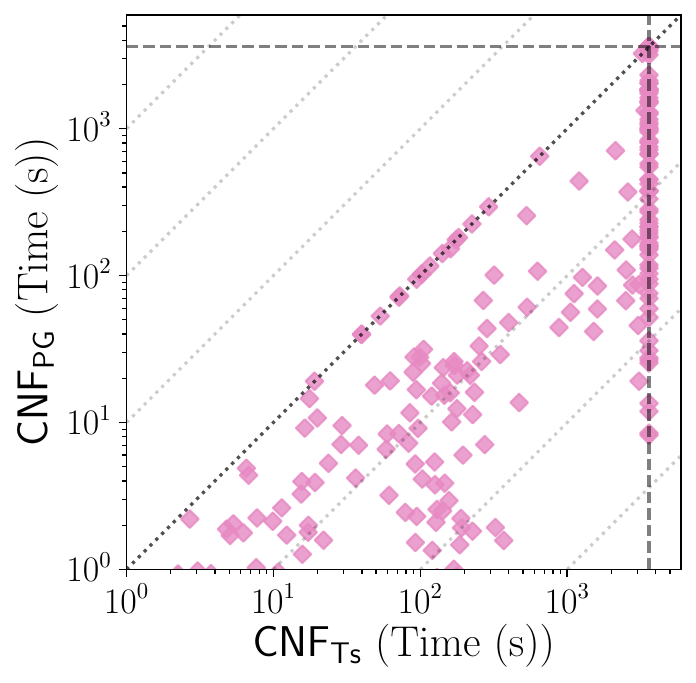}%
            \label{fig:plt:syn:lra:rep:time:lab_vs_pol}
        \end{subfigure}\hfill
        \begin{subfigure}[t]{0.29\textwidth}
            \includegraphics[width=.85\textwidth]{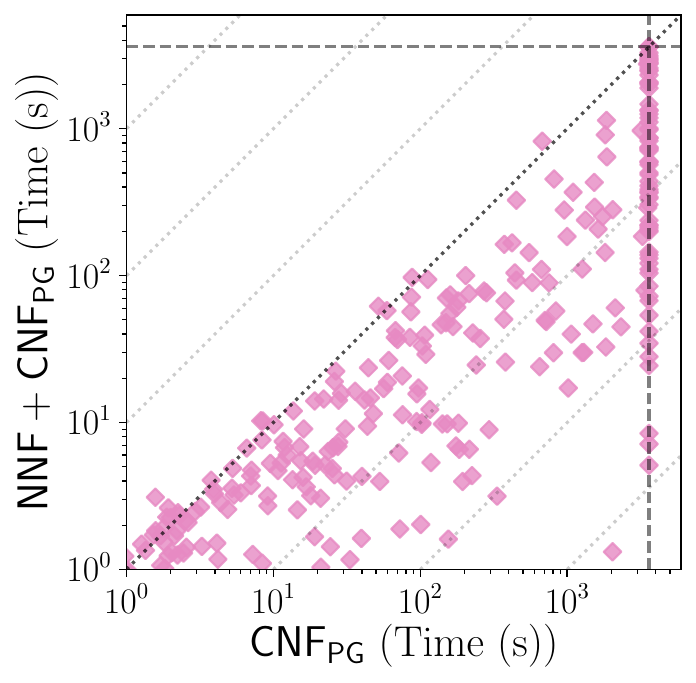}%
            \label{fig:plt:syn:lra:rep:time:pol_vs_nnfpol}
        \end{subfigure}\hfill
        \begin{subfigure}[t]{0.29\textwidth}
            \includegraphics[width=.85\textwidth]{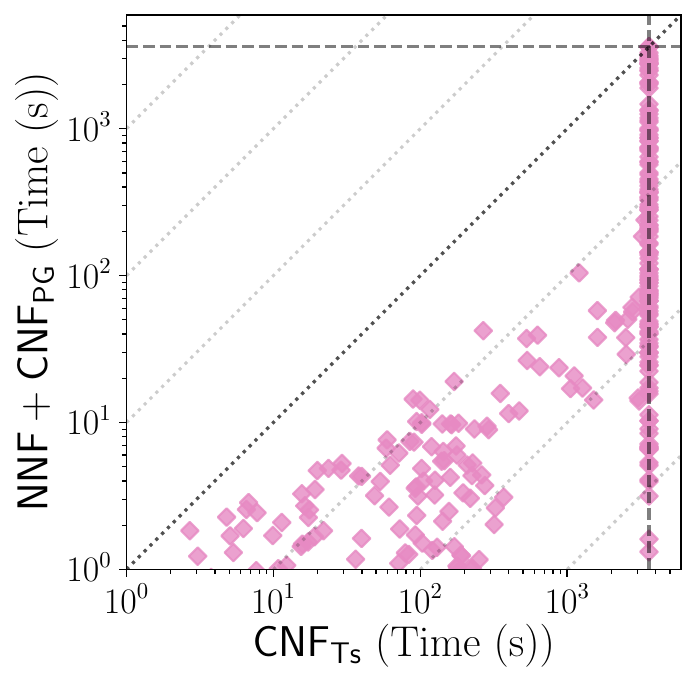}%
            \label{fig:plt:syn:lra:rep:time:lab_vs_nnfpol}
            % \end{subfigure}
        \end{subfigure}
        \caption{Results for non-disjoint enumeration. %\TseitinCNF{}, \PlaistedCNF{} and $\NNFPlaisted{}$ reported 168, 95 and 26 timeouts, respectively (points on the dashed lines).
        }%
        \label{fig:plt:syn:lra:rep:scatter}
    \end{subfigure}
    \caption{Results on the \smtlarat{} synthetic benchmarks using \mathsat{}.
        Plots in~\ref{fig:plt:syn:lra:norep:scatter} and \ref{fig:plt:syn:lra:rep:scatter} compare CNF-izations by \TAna{} size (first row) and execution time (second row).
        Points on dashed lines represent timeouts, shown in~\ref{tab:timeouts:lra}.
        All axes use a logarithmic scale.}%
    \label{fig:plt:syn:lra:scatter}
\end{figure}
\begin{figure}
    \centering
    % \begin{subfigure}[t]{\textwidth}
    \begin{subfigure}[t]{\textwidth}
        \begin{subfigure}[t]{0.29\textwidth}
            \includegraphics[width=.85\textwidth]{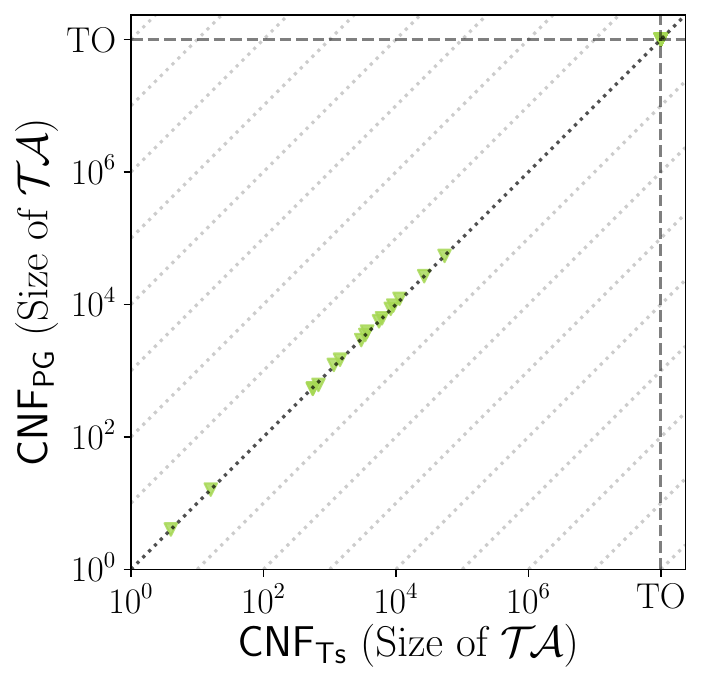}%
            \label{fig:plt:wmi:norep:models:lab_vs_pol}
        \end{subfigure}\hfill
        \begin{subfigure}[t]{0.29\textwidth}
            \includegraphics[width=.85\textwidth]{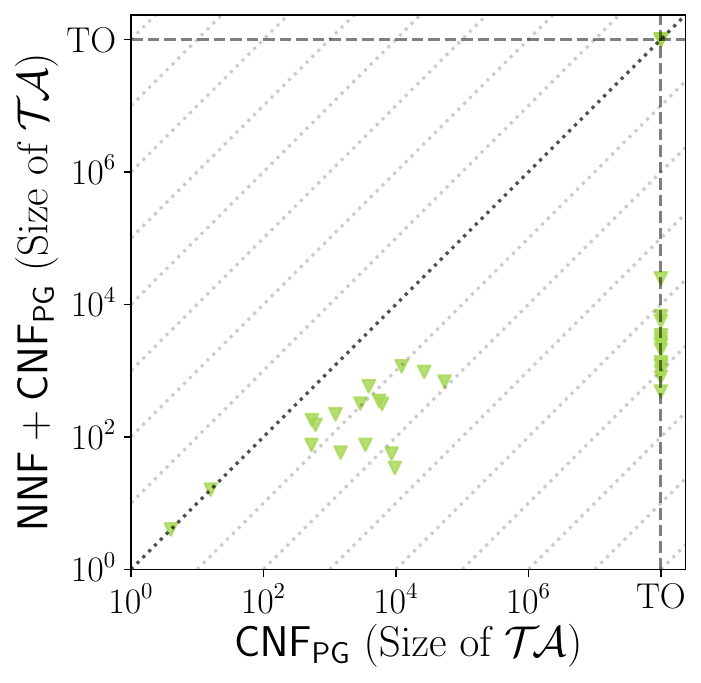}%
            \label{fig:plt:wmi:norep:models:pol_vs_nnfpol}
        \end{subfigure}\hfill
        \begin{subfigure}[t]{0.29\textwidth}
            \includegraphics[width=.85\textwidth]{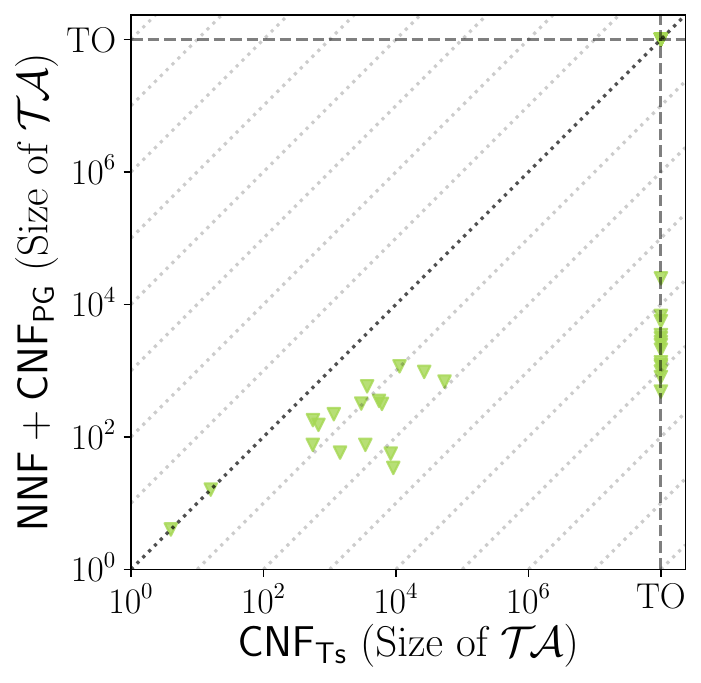}%
            \label{fig:plt:wmi:norep:models:lab_vs_nnfpol}
        \end{subfigure}\hfill
        \begin{subfigure}[t]{0.29\textwidth}
            \includegraphics[width=.85\textwidth]{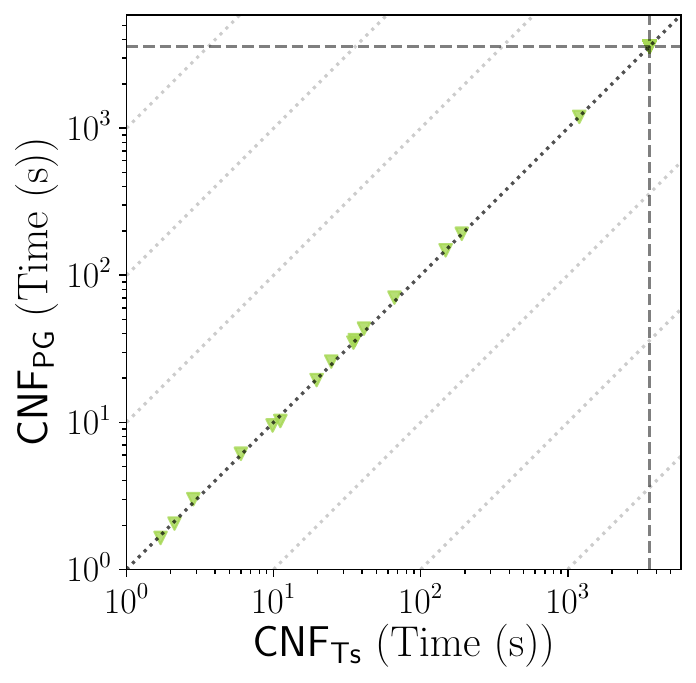}%
            \label{fig:plt:wmi:norep:time:lab_vs_pol}
        \end{subfigure}\hfill
        \begin{subfigure}[t]{0.29\textwidth}
            \includegraphics[width=.85\textwidth]{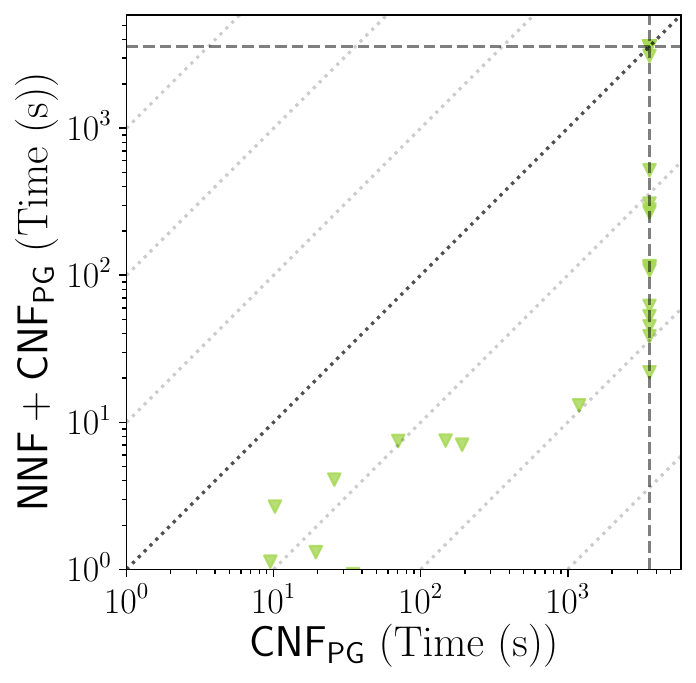}%
            \label{fig:plt:wmi:norep:time:pol_vs_nnfpol}
        \end{subfigure}\hfill
        \begin{subfigure}[t]{0.29\textwidth}
            \includegraphics[width=.85\textwidth]{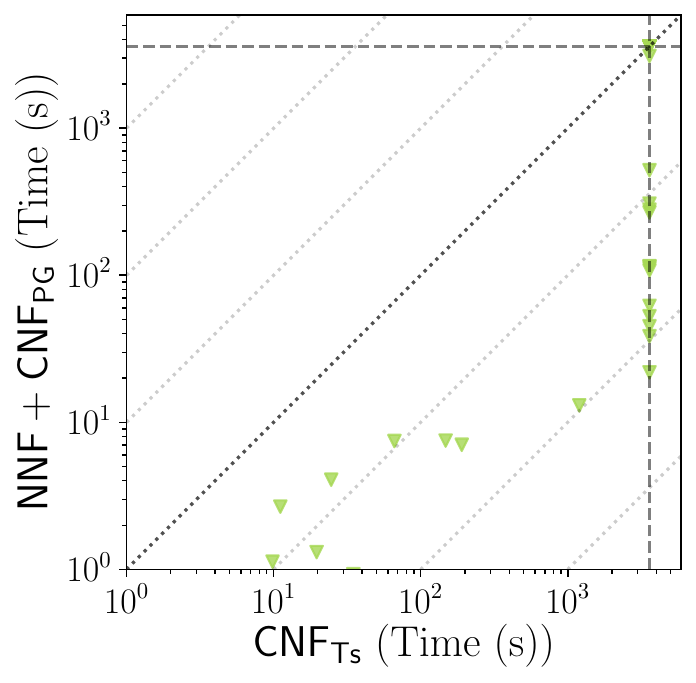}%
            \label{fig:plt:wmi:norep:time:lab_vs_nnfpol}
            % \end{subfigure}
        \end{subfigure}
        \caption{Results for disjoint enumeration. %\TseitinCNF{}, \PlaistedCNF{} and $\NNFPlaisted{}$ reported 24, 23 and 10 timeouts, respectively (points on the dashed lines).
        }%
        \label{fig:plt:wmi:norep:scatter}
    \end{subfigure}
    %%%%%%%%%%%% REP %%%%%%%%%%%%%
    \begin{subfigure}[t]{\textwidth}
        \begin{subfigure}[t]{0.29\textwidth}
            \includegraphics[width=.85\textwidth]{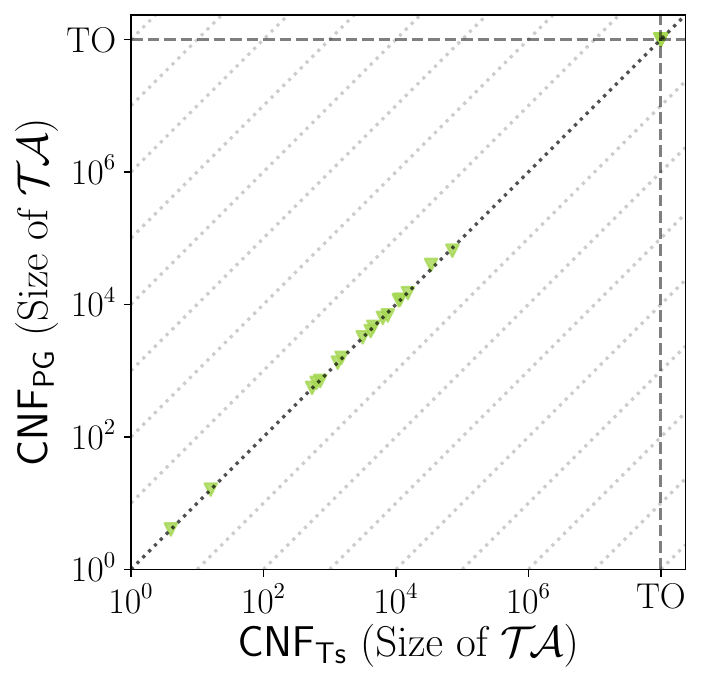}%
            \label{fig:plt:wmi:rep:models:lab_vs_pol}
        \end{subfigure}\hfill
        \begin{subfigure}[t]{0.29\textwidth}
            \includegraphics[width=.85\textwidth]{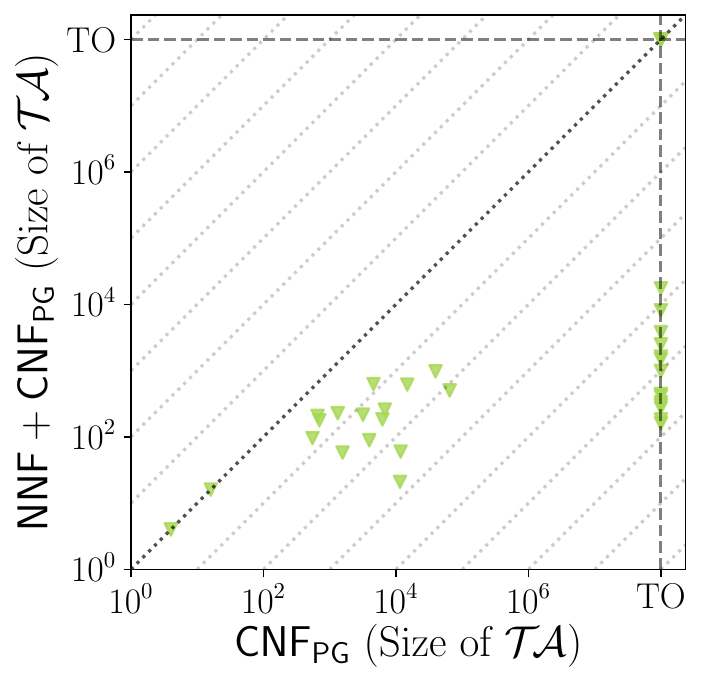}%
            \label{fig:plt:wmi:rep:models:pol_vs_nnfpol}
        \end{subfigure}\hfill
        \begin{subfigure}[t]{0.29\textwidth}
            \includegraphics[width=.85\textwidth]{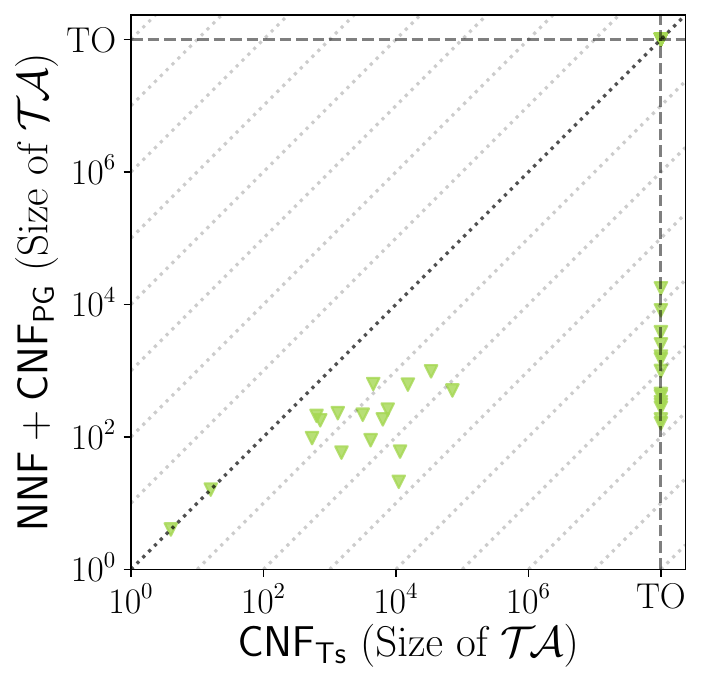}%
            \label{fig:plt:wmi:rep:models:lab_vs_nnfpol}
        \end{subfigure}\hfill
        \begin{subfigure}[t]{0.29\textwidth}
            \includegraphics[width=.85\textwidth]{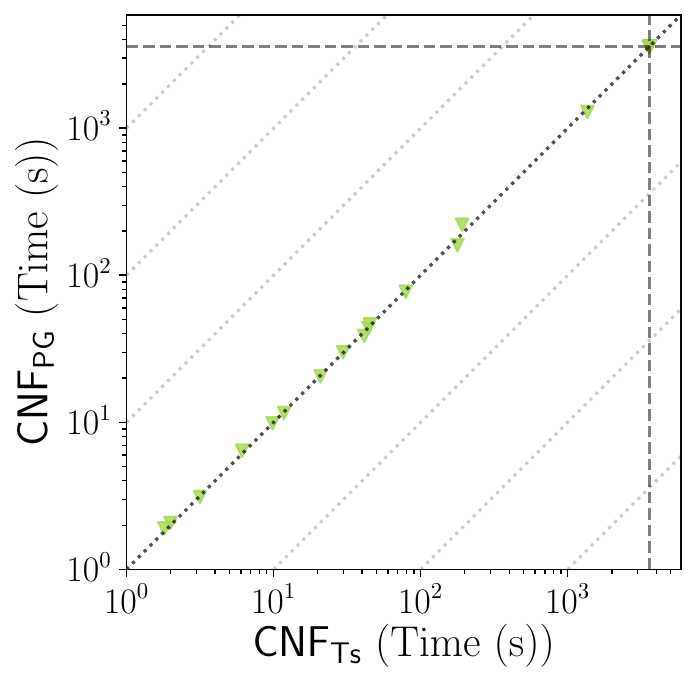}%
            \label{fig:plt:wmi:rep:time:lab_vs_pol}
        \end{subfigure}\hfill
        \begin{subfigure}[t]{0.29\textwidth}
            \includegraphics[width=.85\textwidth]{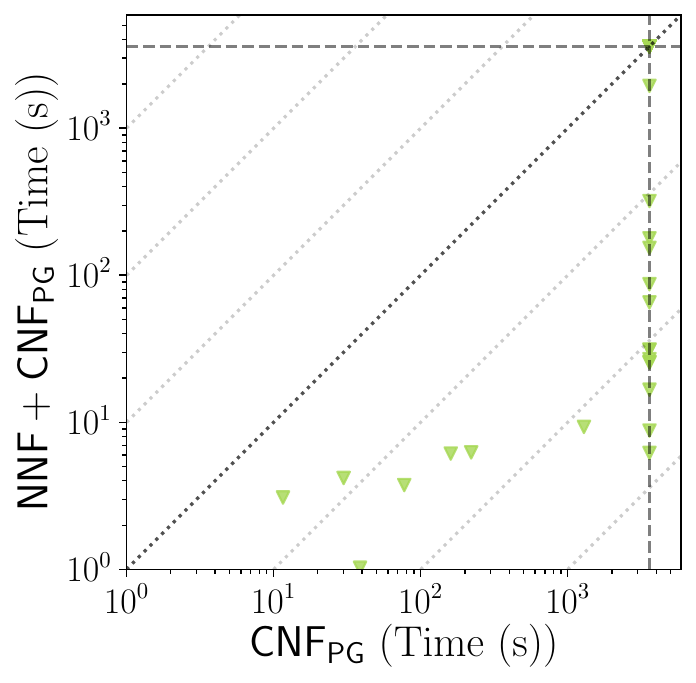}%
            \label{fig:plt:wmi:rep:time:pol_vs_nnfpol}
        \end{subfigure}\hfill
        \begin{subfigure}[t]{0.29\textwidth}
            \includegraphics[width=.85\textwidth]{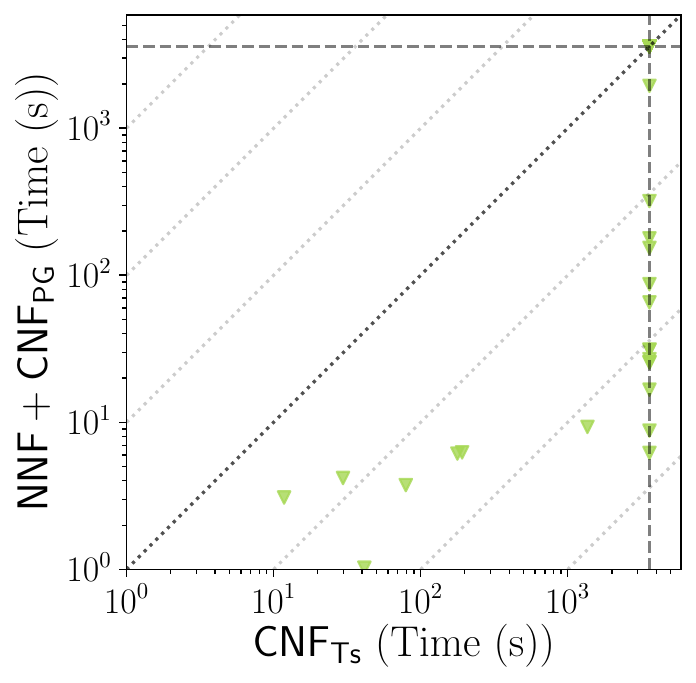}%
            \label{fig:plt:wmi:rep:time:lab_vs_nnfpol}
            % \end{subfigure}
        \end{subfigure}
        \caption{Results for non-disjoint enumeration. %\TseitinCNF{}, \PlaistedCNF{} and $\NNFPlaisted{}$ reported 24, 24 and 10 timeouts, respectively (points on the dashed lines).
        }%
        \label{fig:plt:wmi:rep:scatter}
    \end{subfigure}
    \caption{Results on the WMI benchmarks using \mathsat{}.
        Plots in~\ref{fig:plt:wmi:norep:scatter} and \ref{fig:plt:wmi:rep:scatter} compare CNF-izations by \TAna{} size (first row) and execution time (second row).
        Points on dashed lines represent timeouts, shown in~\ref{tab:timeouts:lra}.
        All axes use a logarithmic scale.}%
    \label{fig:plt:wmi:scatter}
\end{figure}

% ---- AllSMT Tabula ----

\begin{figure}
    \centering
    % \begin{subfigure}[t]{\textwidth}
    \begin{subfigure}[t]{\textwidth}
        \begin{subfigure}[t]{0.29\textwidth}
            \includegraphics[width=.85\textwidth]{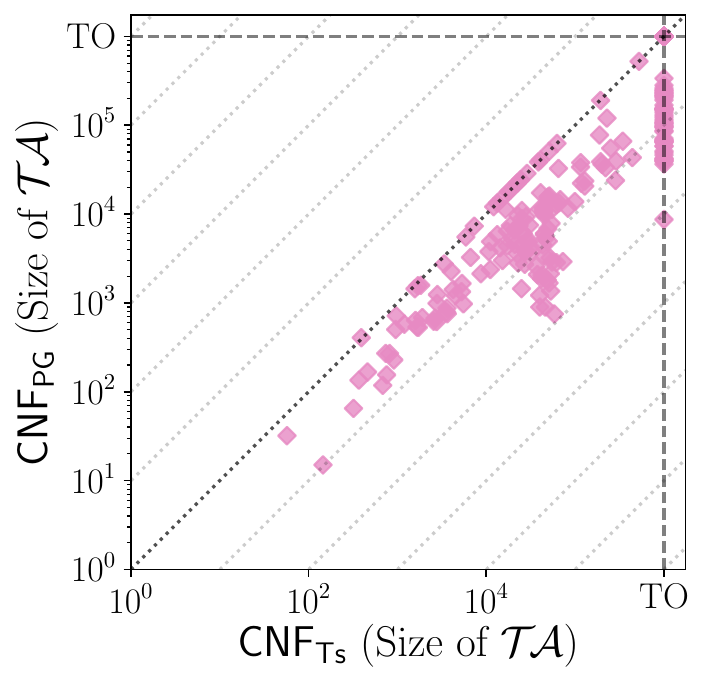}%
            \label{fig:plt:tabula:syn:lra:norep:models:lab_vs_pol}
        \end{subfigure}\hfill
        \begin{subfigure}[t]{0.29\textwidth}
            \includegraphics[width=.85\textwidth]{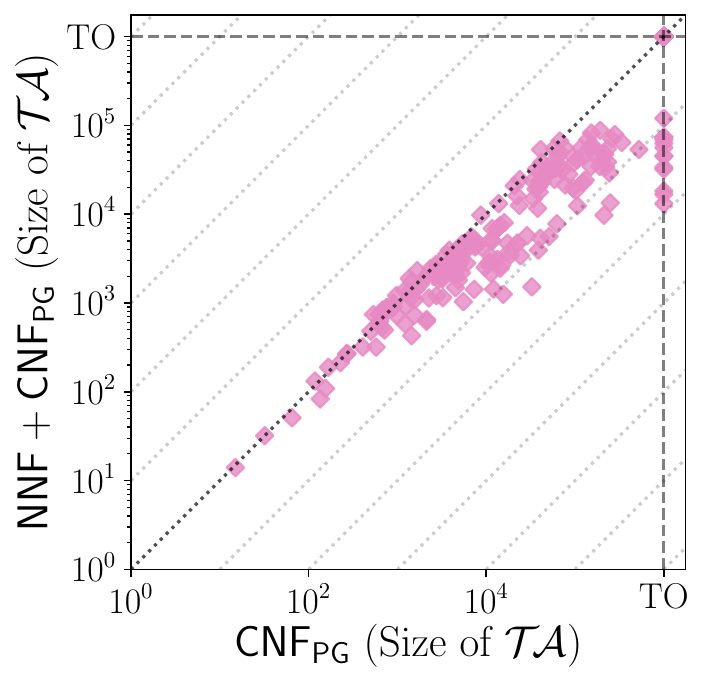}%
            \label{fig:plt:tabula:syn:lra:norep:models:pol_vs_nnfpol}
        \end{subfigure}\hfill
        \begin{subfigure}[t]{0.29\textwidth}
            \includegraphics[width=.85\textwidth]{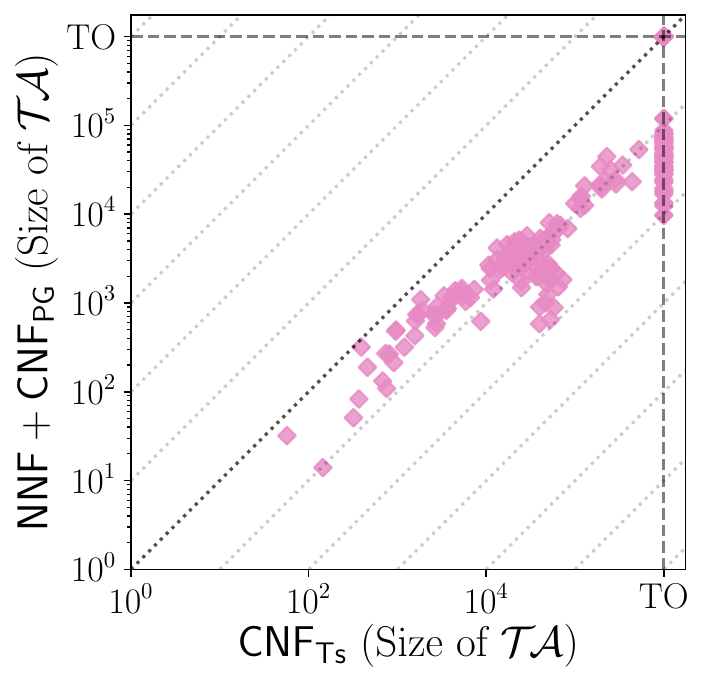}%
            \label{fig:plt:tabula:syn:lra:norep:models:lab_vs_nnfpol}
        \end{subfigure}\hfill
        \begin{subfigure}[t]{0.29\textwidth}
            \includegraphics[width=.85\textwidth]{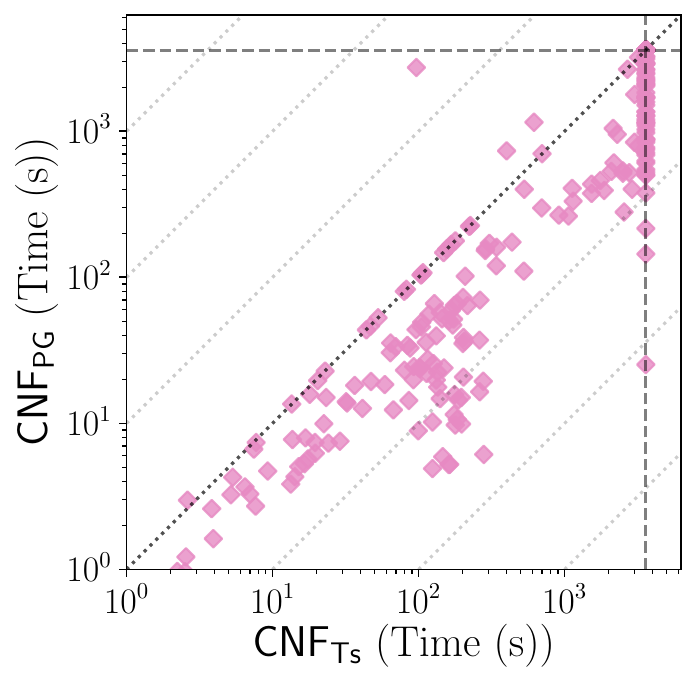}%
            \label{fig:plt:tabula:syn:lra:norep:time:lab_vs_pol}
        \end{subfigure}\hfill
        \begin{subfigure}[t]{0.29\textwidth}
            \includegraphics[width=.85\textwidth]{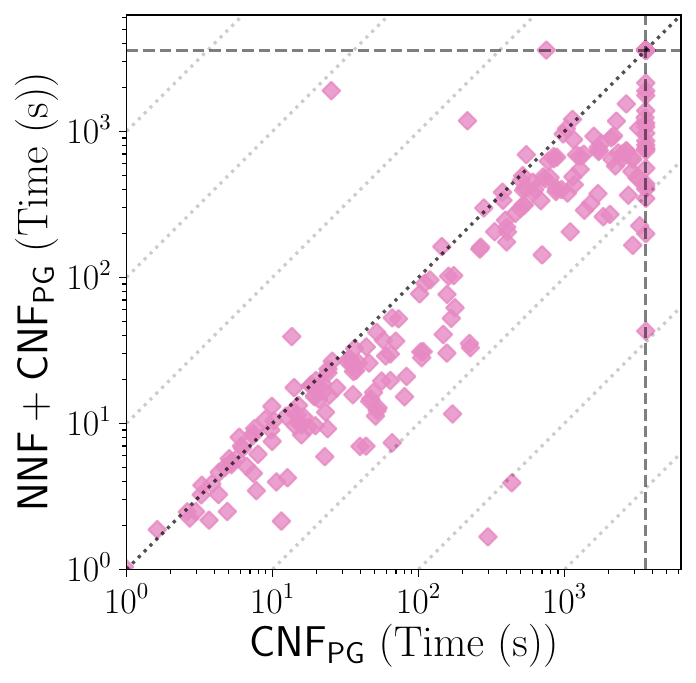}%
            \label{fig:plt:tabula:syn:lra:norep:time:pol_vs_nnfpol}
        \end{subfigure}\hfill
        \begin{subfigure}[t]{0.29\textwidth}
            \includegraphics[width=.85\textwidth]{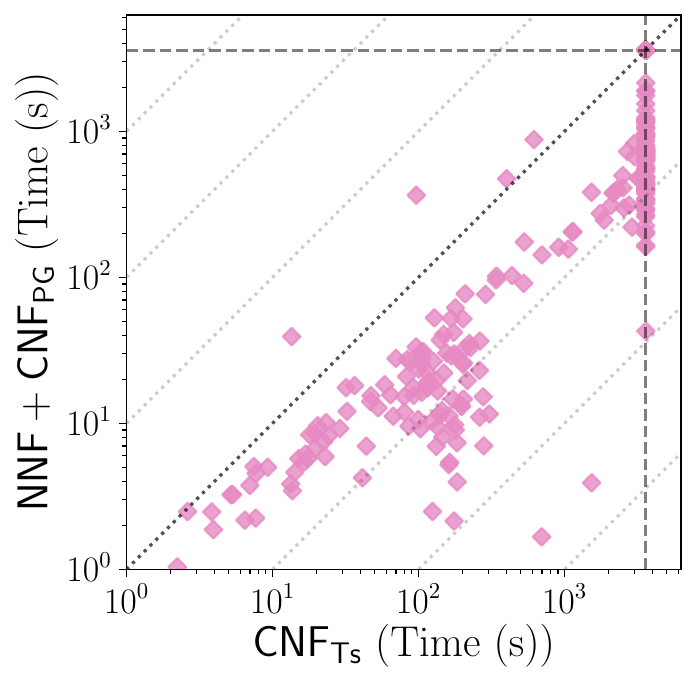}%
            \label{fig:plt:tabula:syn:lra:norep:time:lab_vs_nnfpol}
            % \end{subfigure}
        \end{subfigure}
        \caption{Results for disjoint enumeration. %\TseitinCNF{}, \PlaistedCNF{} and $\NNFPlaisted{}$ reported 163, 94 and 20 timeouts, respectively (points on the dashed lines).
        }%
        \label{fig:plt:tabula:syn:lra:norep:scatter}
    \end{subfigure}
    \caption{Results on the \smtlarat{} synthetic benchmarks using \tabularallsmt{}.
        Plots in~\ref{fig:plt:tabula:syn:lra:norep:scatter} compare CNF-izations by \TAna{} size (first row) and execution time (second row).
        Points on dashed lines represent timeouts, shown in~\ref{tab:timeouts:tabula:lra}.
        All axes use a logarithmic scale.}%
    \label{fig:plt:tabula:syn:lra:scatter}
\end{figure}

\begin{figure}
    \centering
    % \begin{subfigure}[t]{\textwidth}
    \begin{subfigure}[t]{\textwidth}
        \begin{subfigure}[t]{0.29\textwidth}
            \includegraphics[width=.85\textwidth]{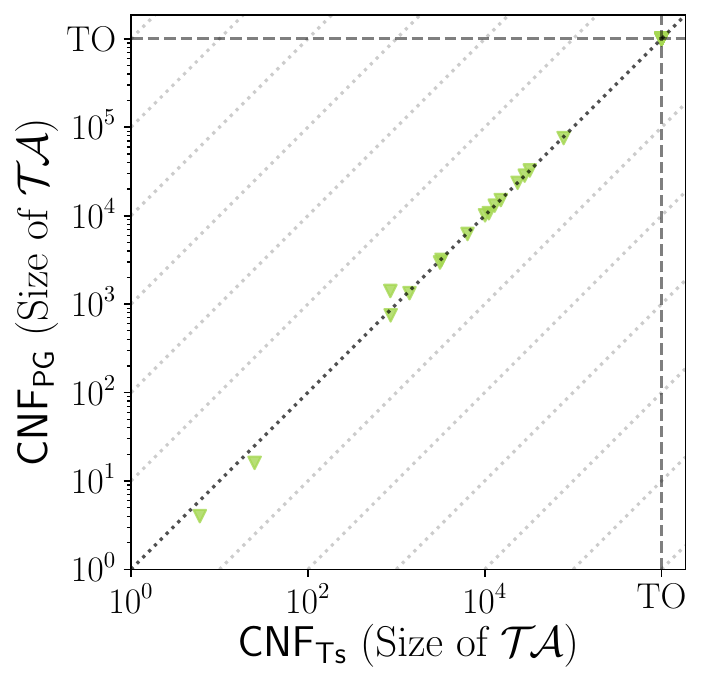}%
            \label{fig:plt:tabula:wmi:norep:models:lab_vs_pol}
        \end{subfigure}\hfill
        \begin{subfigure}[t]{0.29\textwidth}
            \includegraphics[width=.85\textwidth]{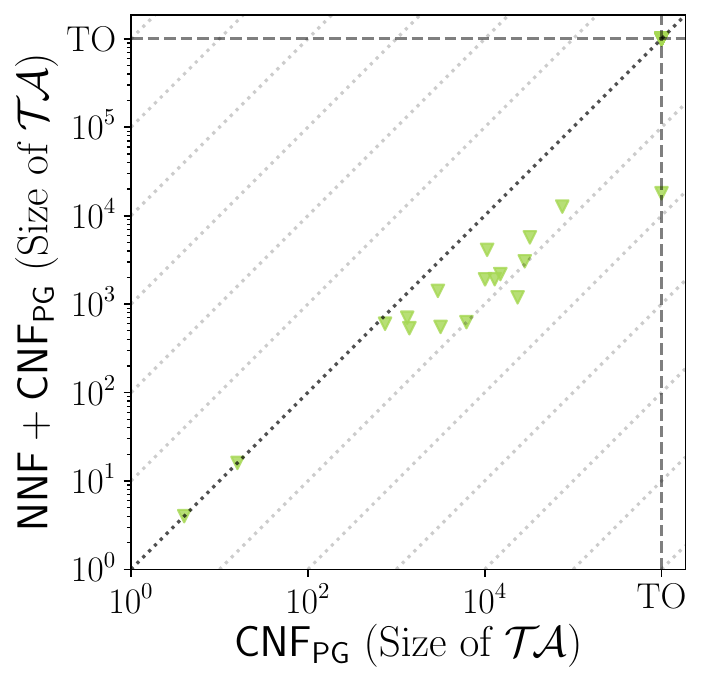}%
            \label{fig:plt:tabula:wmi:norep:models:pol_vs_nnfpol}
        \end{subfigure}\hfill
        \begin{subfigure}[t]{0.29\textwidth}
            \includegraphics[width=.85\textwidth]{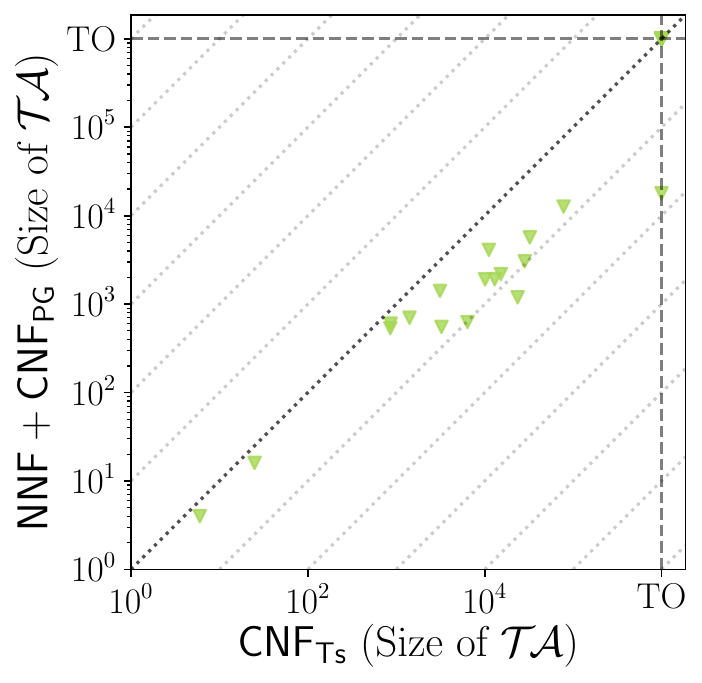}%
            \label{fig:plt:tabula:wmi:norep:models:lab_vs_nnfpol}
        \end{subfigure}\hfill
        \begin{subfigure}[t]{0.29\textwidth}
            \includegraphics[width=.85\textwidth]{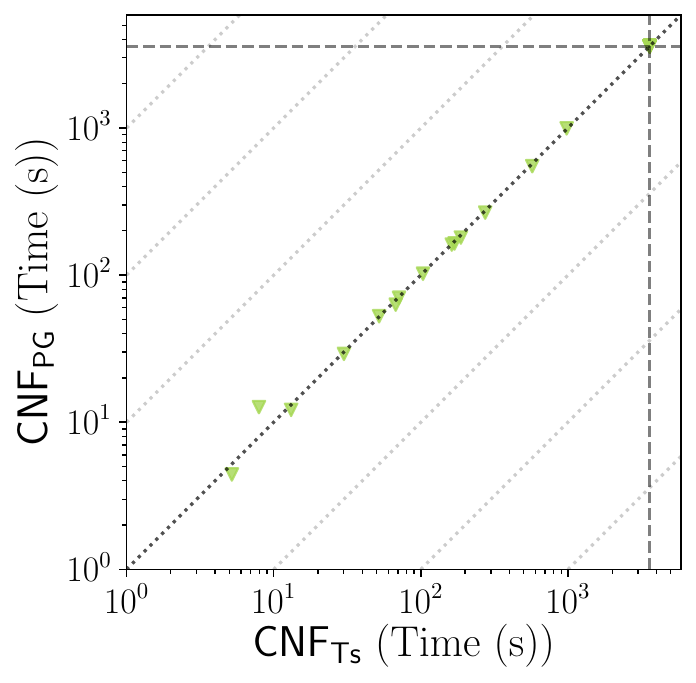}%
            \label{fig:plt:tabula:wmi:norep:time:lab_vs_pol}
        \end{subfigure}\hfill
        \begin{subfigure}[t]{0.29\textwidth}
            \includegraphics[width=.85\textwidth]{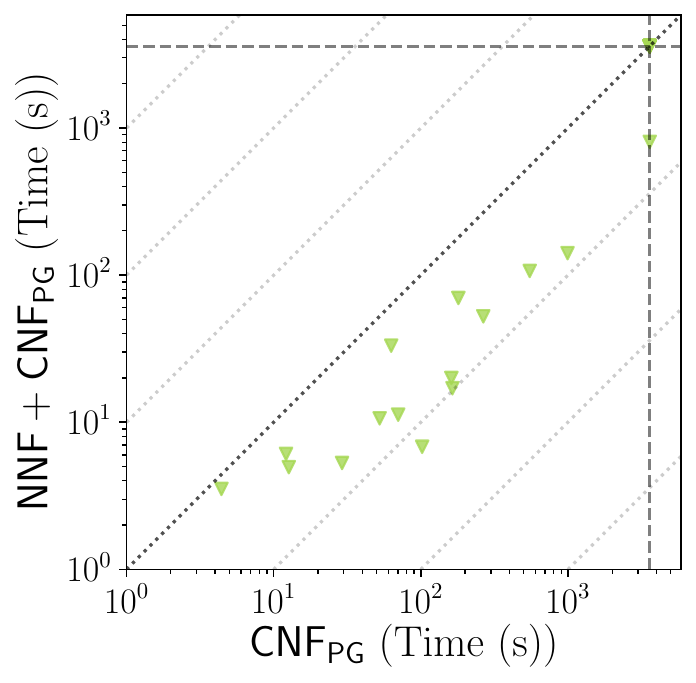}%
            \label{fig:plt:tabula:wmi:norep:time:pol_vs_nnfpol}
        \end{subfigure}\hfill
        \begin{subfigure}[t]{0.29\textwidth}
            \includegraphics[width=.85\textwidth]{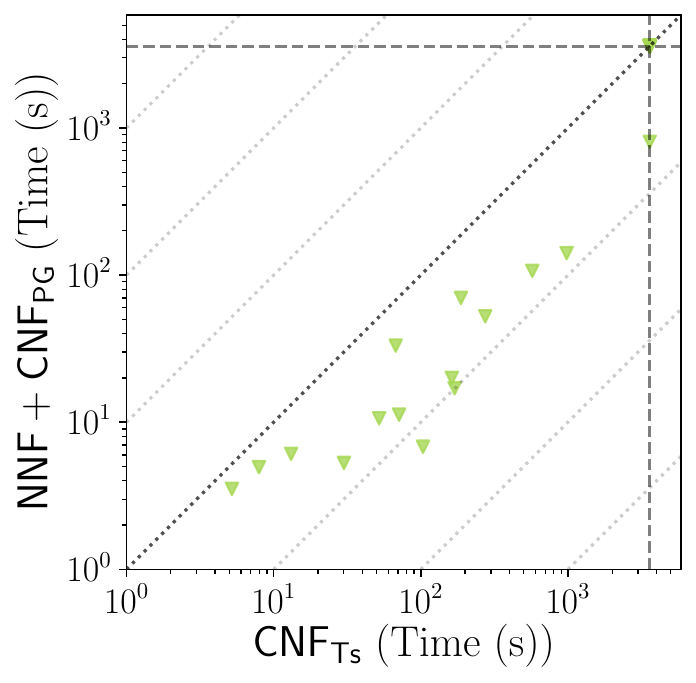}%
            \label{fig:plt:tabula:wmi:norep:time:lab_vs_nnfpol}
            % \end{subfigure}
        \end{subfigure}
        \caption{Results for disjoint enumeration. %\TseitinCNF{}, \PlaistedCNF{} and $\NNFPlaisted{}$ reported 49, 44 and 27 timeouts, respectively (points on the dashed lines).
        }%
        \label{fig:plt:tabula:wmi:norep:scatter}
    \end{subfigure}
    \caption{Results on the WMI benchmarks using \tabularallsmt{}.
        Plots in~\ref{fig:plt:tabula:wmi:norep:scatter} compare CNF-izations by \TAna{} size (first row) and execution time (second row).
        Points on dashed lines represent timeouts, shown in~\ref{tab:timeouts:tabula:lra}.
        All axes use a logarithmic scale.}%
    \label{fig:plt:tabula:wmi:scatter}
\end{figure}
% }

\FloatBarrier
% \vskip 0.2in
\newpage
\printbibliography

\end{document}